\documentclass[10pt,a4paper]{amsart}

\usepackage{amsthm, amsfonts,amssymb,euscript}
\usepackage{amsmath}
\usepackage{bm}
\usepackage[bbgreekl]{mathbbol}
\usepackage{amssymb}
\usepackage{amsthm}
\usepackage{graphicx}
\usepackage{caption}
\usepackage{subcaption}
\usepackage{amsfonts}
\usepackage{float}
\usepackage{color}
\usepackage{slashed}
\usepackage{cite}
\usepackage{booktabs}
\usepackage{array,makecell}
\usepackage{boldline,multirow}

\usepackage[usenames,dvipsnames,svgnames,table]{xcolor}

\usepackage[colorlinks=true]{hyperref}

\hypersetup{linkcolor=BrickRed, urlcolor=green, citecolor=blue, linktoc=page}

\numberwithin{equation}{section}

\newtheorem*{proposition*}{Proposition}
\newtheorem*{theorem*}{Theorem}
\newtheorem*{conjecture*}{Conjecture}
\newtheorem*{claim*}{Claim}
\newtheorem*{lemma*}{Lemma}
\newtheorem*{corollary*}{Corollary}
\newtheorem{theorem}{Theorem}[section]
\newtheorem{proposition}[theorem]{Proposition}

\newtheorem{lemma}[theorem]{Lemma}
\newtheorem{corollary}[theorem]{Corollary}

\newtheorem*{definition*}{Definition}
\newtheorem{definition}{Definition}[section]
\newtheorem*{assumption*}{\mathcal{A}ssumption}

\newtheorem*{remark*}{Remark}
\newtheorem{remark}{Remark}[section]

\newcommand{\R}{\mathbb{R}}
\newcommand{\s}{\mathbb{S}}

\newcommand{\Z}{\mathbb{Z}}
\newcommand{\N}{\mathbb{N}}
\newcommand{\T}{\mathbf{T}}

\newcommand{\snabla}{\slashed{\nabla}}

\newcommand{\sD}{\slashed{\Delta}}

\newcommand{\Lbar}{\underline{L}}
\newcommand{\h}{\mathcal{H}^{+}}
\newcommand{\I}{\mathcal{I}^{+}}

\allowdisplaybreaks
\begin{document}

\title[Late-time asymptotics on ERN]{Late-time asymptotics for the wave equation \\on extremal Reissner--Nordstr\"om backgrounds }
\author{Y. Angelopoulos, S. Aretakis, and D. Gajic}

\maketitle

\begin{abstract}
We derive the precise late-time asymptotics for solutions to the wave equation on extremal Reissner--Nordstr\"{o}m black holes and explicitly express the leading-order coefficients in terms of the initial data. Our method is based on purely physical space techniques. We derive novel weighted energy hierarchies and develop a singular time inversion theory which allow us to uncover the subtle contribution of both the near-horizon and near-infinity regions to the precise asymptotics. 
We introduce a new horizon charge and provide applications pertaining to the interior dynamics of extremal black holes. 
Our work confirms, and in some cases extends, the numerical and heuristic analysis of Lucietti--Murata--Reall--Tanahashi,  Ori--Sela and  Blaksley--Burko.

\end{abstract}

\tableofcontents

\section{Introduction}
\label{sec:intro}

\subsection{Introduction and first remarks }
\label{sec:Introduction}

The existence of \textit{black hole} regions, namely regions of spacetime which are not visible to far away observers, is a celebrated prediction of the \textit{Einstein field equations}. A rigorous understanding of their dynamical properties is of fundamental importance for
addressing several conjectures in general relativity such as \textit{the weak and strong cosmic censorship conjectures} as well as for investigating \textit{the propagation of gravitational waves}. 
Important aspects of the black hole dynamics are captured by the evolution of solutions to the \textit{wave equation}
\begin{equation}
\label{eq:waveequation}
\square_g\psi=0
\end{equation}
on black hole backgrounds. Initial data are prescribed on a Cauchy hypersurface $\Sigma_0$ which intersects \textit{the event horizon} $\h$ and terminates at \textit{the null infinity} $\I$, as in the figure below.  A first step towards the non-linear stability of black hole backgrounds is to establish quantitative \textit{dispersive} estimates in the domain of outer communications up to and including the event horizon.
  \begin{figure}[H]
	\begin{center}
				\includegraphics[scale=0.18]{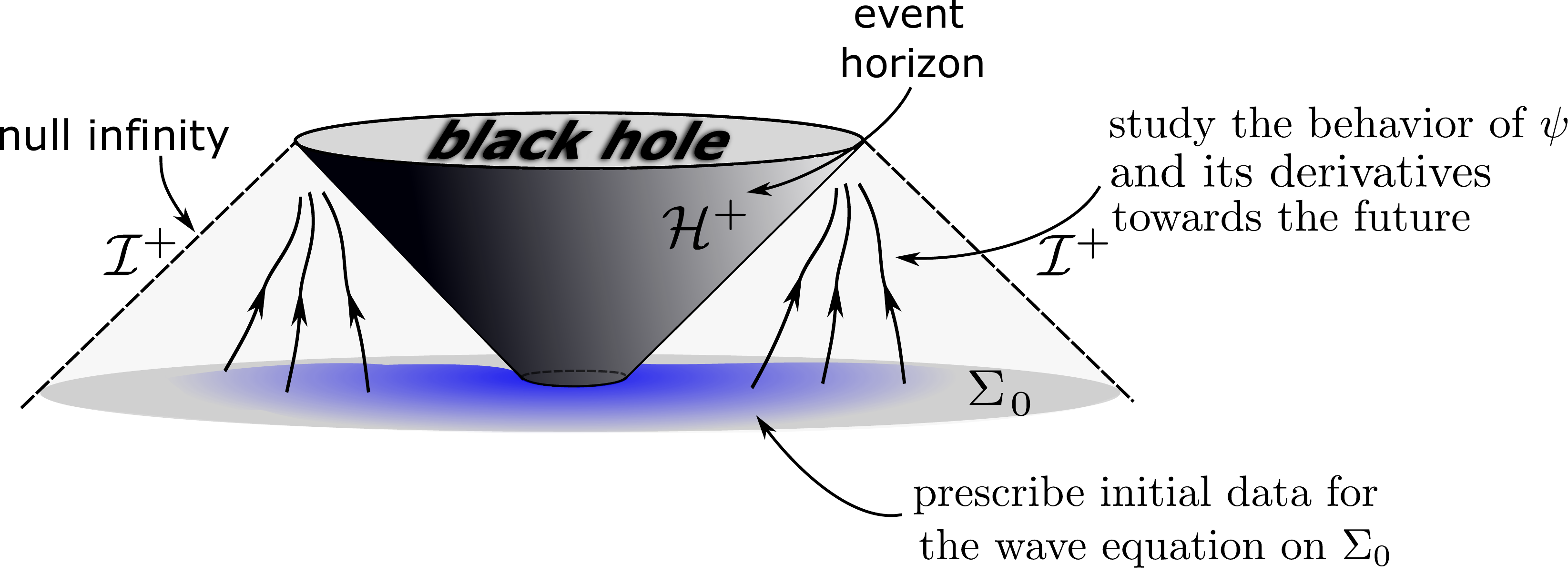}
\end{center}
\vspace{-0.2cm}
\caption{The initial value problem for the wave equation on black hole backgrounds. }
	\label{fig:p455}
\end{figure}
\vspace{-0.3cm} 
\noindent This problem has been extensively studied in both the mathematics and the physics communities.

\textbf{Quantitative decay rates} for \textit{scalar fields satisfying \eqref{eq:waveequation} and all their higher-order derivatives} have been obtained 
for the general \textit{sub-extremal Reissner--Nordstr\"{o}m} and the general \textit{sub-extremal Kerr} families of black hole spacetimes (see \cite{part3}). We refer to  \cite{lecturesMD, MDIR05,  blu1, moschidis1,  volker1, tataru3, metal, blukerr, dssprice,  other1} for additional results in the asymptotically flat setting.

A definitive proof of the \textbf{precise late-time asymptotics} of solutions to the wave equation on the general sub-extremal Reissner--Nordstr\"{o}m backgrounds, including the celebrated \textit{Schwarzschild} family of black holes, was obtained in the recent series of papers \cite{paper1, paper2, paper-bifurcate} confirming, in particular, Price's heuristics (for more details and references see Section \ref{sec:TheWaveEquationOnBlackHolesBackgrounds}).

In the present paper, we focus on another fundamental class of black holes, namely the \textit{extremal} black holes. These are characterized by the \textbf{vanishing} of the \textit{surface gravity} of the event horizon (see also Section \ref{sec:TheERNManifoldFoliationsAndVectorFields}). Geometrically, this condition has to do with the fact that the Killing normal to the event horizon coincides with the affine null normal to the event horizon. 
Extremal black holes play a fundamental role in\footnote{For more details and an exhaustive list of references for works related to extremal black holes see Section \ref{sec:PhysicalImportanceOfExtremalBlackHoles}.}
\begin{itemize}
	\item \textit{astronomy}: according to an abundance of astronomical observations, near-extremal black holes should be ubiquitous in the universe. Such observations concern stellar black holes and supermassive black holes in the centers of galaxies;
	\item \textit{high energy physics}: they allow for the study of supersymmetric theories of gravity, black hole thermodynamics and of  quantum descriptions of gravity;
	\item \textit{classical general relativity}: they saturate various geometric inequalities concerning the mass, angular-momentum and charge. Furthermore, they have intriguing dynamical properties with no analogue in sub-extremal black holes.
	\end{itemize}
	The latter intriguing dynamical properties of extremal black holes are the objects of study in this paper. Specifically, we investigate scalar perturbations of the \textit{extremal Reissner--Nordstr\"{o}m} (ERN) one-parameter family of black hole backgrounds. The corresponding metrics take the following form with respect to the so-called ingoing Eddington--Finkelstein coordinates $(v,r,\theta,\varphi)$:
\begin{equation}
g_{\text{\tiny{ERN}}}=-\left(1-\frac{M}{r}\right)^2dv^2+dv\otimes dr+dr\otimes dv+r^2(d\theta^2+\sin^2\theta d\varphi^2), 
\label{ernmetric}
\end{equation}
where $M>0$ is the mass parameter. 

ERN black hole spacetimes are spherically symmetric, asymptotically flat solutions to the Einstein--Maxwell system
\begin{equation*}
\begin{split}
R_{\mu\nu}(g)-\frac{1}{2}R(g)\cdot g_{\mu\nu}&=T_{\mu\nu}(F),\\
dF=d\star F&=0.
\end{split}
\end{equation*}
Here, $T_{\mu\nu}(F)$ denote the components of the energy-momentum tensor of the electromagnetic field $F$. 

The techniques for establishing decay in time and in fact the precise late-time asymptotics for solutions to the wave equation on sub-extremal backgrounds \textit{break down} in the case of ERN. Indeed, a large portion of the analysis of perturbations of sub-extremal black holes exploits the celebrated \textit{redshift effect} along the event horizon. This \textit{local} version of the redshift effect, which essentially depends on the strict positivity of the surface gravity of the event horizon (see \cite{redshift, lecturesMD}), can be illustrated as follows: consider two observers A and B entering the black hole region such that A crosses the event horizon first (see Figure \ref{fig:p4551ref}). Suppose A emits a light signal that travels along the event horizon and is intercepted by B. Then the frequency of this signal as measured by B will be ``shifted to the red" when compared to the frequency measured by A.
 \begin{figure}[H]
	\begin{center}
				\includegraphics[scale=0.18]{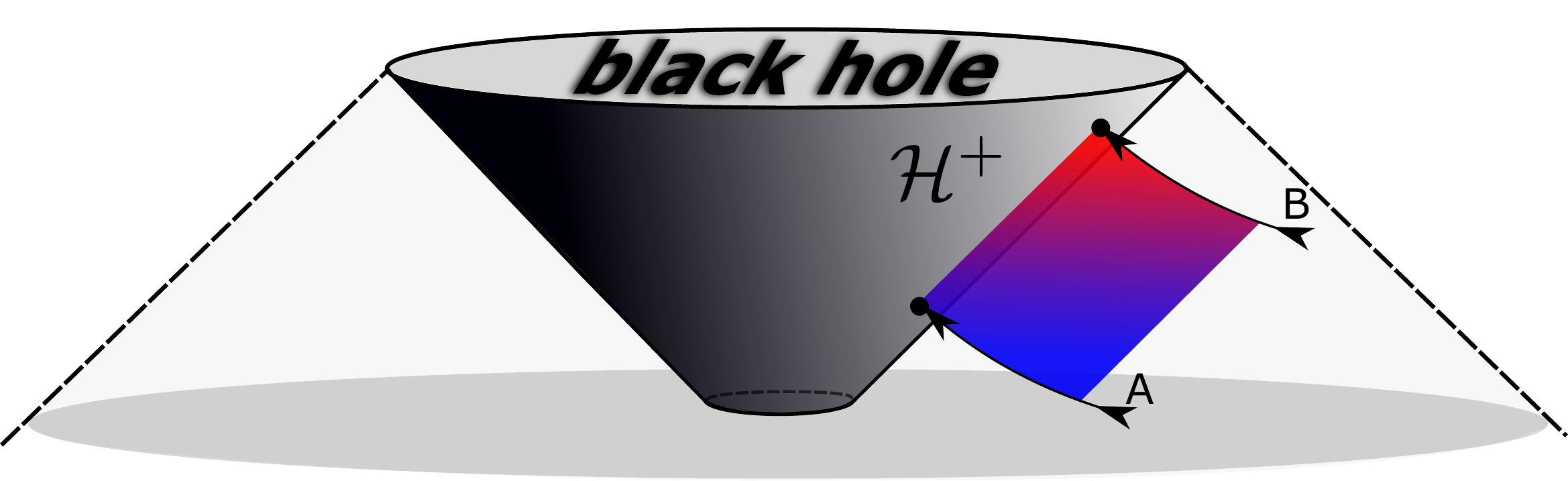}
\end{center}
\vspace{-0.2cm}
\caption{The local redshift effect for sub-extremal horizons. }
	\label{fig:p4551ref}
\end{figure}
\vspace{-0.3cm}

The vanishing of the surface gravity of extremal event horizons means that \textit{the redshift effect along the event horizon degenerates on extremal black holes} and hence cannot be used as a stabilizing mechanism. In fact, it was shown in \cite{SA10, aretakis1, aretakis2} that solutions to the wave equation on ERN satisfy a \textit{conservation law along the event horizon} $\mathcal{H}^{+}$. Consider the  translation-invariant vector field $Y=\partial_r$. Note that $Y$ is transversal to the event horizon as in the figure below. Consider the advanced time parameter  $\tau$ (explicitly defined in Section \ref{sec:TheHyperboloidalFoliation})  which is comparable to the Schwarzschild coordinate time $t$ away from the event horizon and null infinity and  denote by $S_{\tau_0}=\h\cap \{\tau=\tau_0\}$ the corresponding spherical sections of the event horizon. 
 \begin{figure}[H]
\begin{center}
		\includegraphics[scale=0.18]{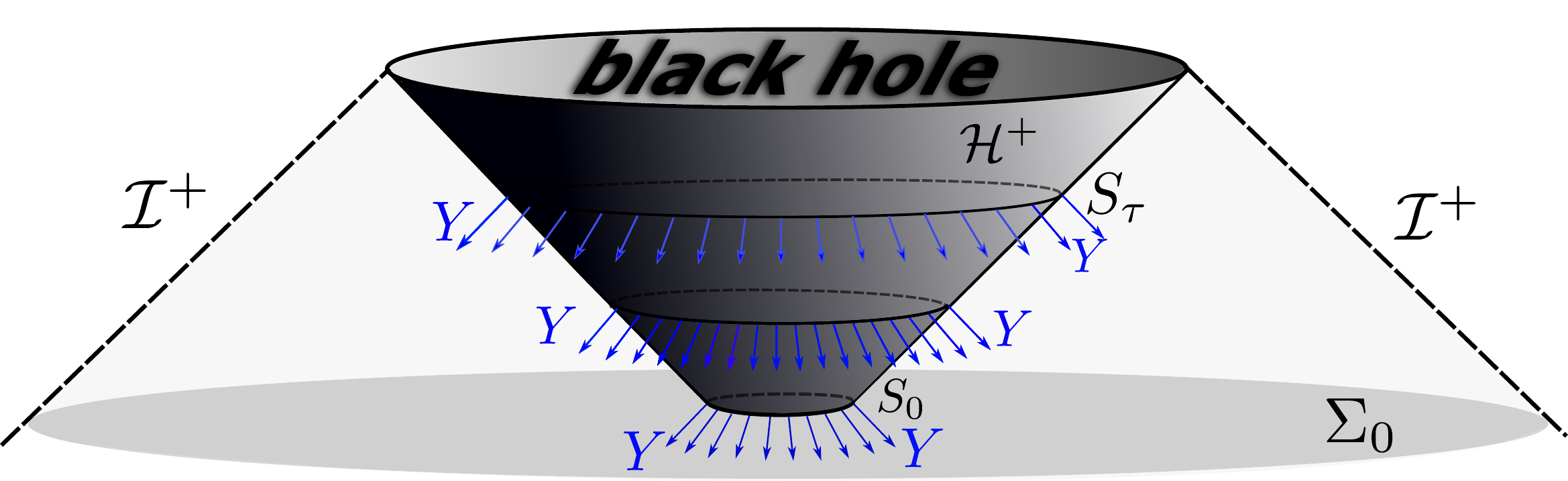}
\end{center}
	\label{fig:p4125}
	\vspace{-0.4cm}
\caption{The sections $S_{\tau}$ of $\mathcal{H}^{+}$ and the transversal to $\mathcal{H}^{+}$ vector field $Y$.}
\end{figure}
	\vspace{-0.3cm}
\noindent Then, the surface integrals
\begin{equation}
H_0[\psi]:=-\frac{M^2}{4\pi}\int_{S_{\tau}}Y(r\psi) \, d\omega
\label{hcons1}
\end{equation}
are \textbf{independent} of $\tau$. Here $\omega=(\theta,\varphi)$ and $d\omega=\sin\theta d\theta d\varphi$.  We will frequently refer to $H_{0}[\psi]$ as a \textit{conserved charge} for $\psi$. This conservation law is certainly an obstruction to decay for generic initial data for which $H_0[\psi]\neq 0$. It can further been shown that higher order derivatives asymptotically \textbf{blow up} along $\h$:
\begin{equation}
|Y^{k}\psi|_{\h}\sim\  c_k H_0[\psi]\cdot \tau^{k-1}\rightarrow +\infty
\label{instaeq1}
\end{equation}for $k\geq 2$ as $\tau\rightarrow +\infty$. Here $c_k$ are constants that depend only on $M,k$.  The growth along $\h$ of (transversal) derivatives yields a genuine \textit{horizon instability of extremal black holes} which can in fact be measured by local observers who cross the event horizon \cite{hm2012,zimmerman2}.   
On the other hand, it can be shown that \textit{away} from the event horizon $\psi$ and \underline{all} its derivatives $Y^{k}\psi$ \textbf{decay} in time. This means that one may regard $H_0$ as a type of \textit{horizon ``hair''} associated to the event horizon. 
 We remark that an analogous version of the horizon instability holds also for scalar perturbations of extremal Kerr \cite{aretakis4, hj2012} and in fact for many other types of perturbations in various settings (see Sections \ref{sec:TheHorizonInstabilityOfExtremalBlackHoles} and \ref{sec:PreviousWorksOnLateTimeAsymptoticsOnERN}).

Returning to decay estimates, the following weak decay rate was rigorously established in \cite{aretakis2} for $\psi$ everywhere in the domain of outer communications up to and including the event horizon:
\begin{equation}
|\psi| \lesssim \frac{1}{\tau^{\frac{3}{5}}}.
\label{bound1}
\end{equation}
Since $\psi$ decays along the event horizon it follows that, in view of the conservation of \eqref{hcons1}, the first-order transversal derivative $Y\psi$ of $\psi$ \textbf{does not decay} along the horizon for initial data for which $H_0[\psi]\neq 0$.

The above results do \emph{not} provide an insight into the precise asymptotic behavior for $\psi$.
 There is extensive work in the physics literature regarding late-time asymptotics for scalar fields on extremal Reissner--Nordstr\"om via heuristic or numerical methods, see for instance \cite{other2, extremal-rn-qnm, Burko2007, hm2012, ori2013, sela, sela2, zimmerman1} and Section \ref{sec:PreviousWorksOnLateTimeAsymptoticsOnERN} for more details. However, there has been no mathematically rigorous proof or derivation of these asymptotics. In fact, the heuristic and numerical predictions in the physics literature did not provide the late-time asymptotics in the \emph{full} spacetime, which remained an open problem and is resolved in the present paper. Before we give a more precise statement of this open problem, we introduce the following definition concerning initial data for \eqref{eq:waveequation}:

\begin{definition}
Initial data on the Cauchy hypersurface $\Sigma_0$ are called \textbf{horizon-penetrating} if they smoothly extend to the event horizon $\mathcal{H}^{+}$ such that the conserved charge $H_0[\psi]\neq 0$.
\label{def1intro}
\end{definition}
The following problem had been left completely open

\vspace{0.1cm}
\textit{Obtain the late-time asymptotics of the radiation field along the null infinity $\mathcal{I}^{+}$ for horizon-penetrating compactly supported initial data.}
\vspace{0.1cm}

\noindent The physical importance of the above problem lies in the fact that these asymptotics capture the observations made by far-away observers of perturbations of the near-horizon region of extremal black holes. This problem is definitively resolved in the present paper. In fact, in this paper: 

\vspace{0.1cm}
 \textit{We derive and rigorously prove the \textbf{precise} late-time asymptotics for scalar fields on ERN \textbf{globally} in the domain of outer communications, for \textbf{a general class of initial data}.}
\vspace{0.1cm}

\noindent In particular, we derive late-time asymptotics along the event horizon $\h$, along constant $r=r_0$ hypersurfaces and along the null infinity $\I$. The \textit{exact coefficient of the leading-order terms in the asymptotic estimate is obtained in terms of explicit expressions of the initial data}. See Section \ref{sec:SummaryOfTheMainResults} for a non-technical summary of the results and Section \ref{subsec:TheMainTheorems} for the precise statements of the main theorems. Our results provide, in particular, sharp upper and lower decay rates for the evolution of scalar fields. Our method is based purely on physical space constructions and avoids explicit representations of solutions to the wave equation.  We establish a novel elliptic estimate and a new class of hierarchies of weighted estimates adapted to the extremal near-horizon geometry. 

Our results provide a rigorous confirmation and proof of the numerics in \cite{Burko2007, hm2012} and heuristics in \cite{ori2013, sela}. For example,  \cite{hm2012} was the first work to  numerically obtain the following late-time asymptotics along the event horizon for horizon-penetrating compactly supported initial data:
\[\psi|_{\h}\sim \frac{2}{M}H_0[\psi]\cdot\frac{1}{\tau}. \]
These asymptotics, which were subsequently heuristically derived in \cite{ori2013,sela}, are indeed rigorously recovered here. Furthermore, as mentioned above, our results extend the works in the physics literature in various directions. Notably, we obtain the asymptotics of the radiation field along the future null infinity $\I$ for horizon-penetrating, compactly supported initial data:
\[r\psi|_{\I}\sim \Big(4MH_0[\psi]-I_0^{(1)}[\psi]\Big)\cdot\frac{1}{\tau^2}. \]
Here $I^{(1)}_0[\psi]$ is the Newman--Penrose constant of a \textit{singular} time integral of $\psi$ and depends on the global properties of the initial data  (see Sections \ref{sec:TheNewHorizonHairH01Psi} and \ref{sec:GeometricOriginOfTheNewHair}). We remark that the horizon charge $H_0[\psi]$ of a scalar perturbation that is initially localized near the event horizon in fact appears in the asymptotic behavior along $\I$. In other words, observations along null infinity (that is, arbitrarily far from the event horizon) can in principle be used to measure the charge $H_0[\psi]$ associated to in-falling observers at the horizon. This might be thought of as a ``leakage'' of horizon information to null infinity and hence could, in principle, be measured by gravitational detectors demonstrating \textit{a new observational signature for extremal black holes }\cite{extremal-prl}.

\paragraph{Outline of the introduction}
\label{sec:OutlineOfTheIntroduction}

We finish this brief introductory subsection with an outline of the remaining sections in the introduction. In Section \ref{sec:TheWaveEquationOnBlackHolesBackgrounds}, we review the key mechanisms behind the existence of late-time tails in the asymptotics of scalar fields on sub-extremal black holes. In Section \ref{sec:PhysicalImportanceOfExtremalBlackHoles} we list various works which emphasize the importance of the dynamics of extremal black holes and hence serve as a motivation for the work of the present paper. In Section \ref{sec:TheHorizonInstabilityOfExtremalBlackHoles}  we provide a review of the horizon instability of extremal black holes and in Section \ref{sec:PreviousWorksOnLateTimeAsymptoticsOnERN} we discuss the physics literature that is relevant to our problem.

\subsection{Asymptotics for the wave equation on sub-extremal black holes}
\label{sec:TheWaveEquationOnBlackHolesBackgrounds}

The following late-time \textit{polynomial} tails for solutions to the wave equation with smooth, \textit{compactly supported} initial data on 
Schwarzschild spacetimes were obtained in a heuristic manner by Price \cite{Price1972} in 1972  along constant radius $r=r_0$ hypersurfaces away from the event horizon
\begin{equation}
\psi|_{r=r_0}(\tau,r=r_0,\omega)\sim \frac{1}{\tau^3}.
\label{eq:1}
\end{equation}
Subsequent heuristic and numerical works \cite{leaver, CGRPJP94b, LB99} suggested the following asymptotics on the event horizon $\h$ and along the null infinity $\mathcal{I}^{+}$:
\begin{equation}
\psi|_{\mathcal{H}^{+}}(\tau,r=2M,\omega) \sim \frac{1}{\tau^3},\ \ \ r\psi|_{\mathcal{I}^{+}}(\tau,r=\infty,\omega) \sim \frac{1}{\tau^2}.
\label{eq:2}
\end{equation}
Here $\tau$ denotes a global time parameter and $\omega\in \mathbb{S}^{2}$. The following global quantitative estimates which establish rigorously the above asymptotics were obtained for general sub-extremal Reissner--Nordstr\"{o}m spacetimes in \cite{paper2, paper-bifurcate}:

\begin{equation}
\left|\psi(\tau, r_0,\cdot)+8 I_{0}^{(1)}[\psi]\cdot\frac{1}{\tau^3}\right| \leq C_{r_0} \cdot \sqrt{E_{\Sigma_0}[\psi]}\cdot \frac{1}{\tau^{3+\epsilon}},
\label{our1}
\end{equation}
\begin{equation}
\left|r\psi|_{\mathcal{I}^{+}}(\tau,\cdot)+2 I_{0}^{(1)}[\psi]\cdot\frac{1}{\tau^2}\right| \leq C \cdot \sqrt{E_{\Sigma_0}[\psi]}\cdot \frac{1}{\tau^{2+\epsilon}}, 
\label{ourrad1}
\end{equation}
where $\sqrt{E_{\Sigma_0}[\psi]}$ are weighted norms of the initial data and the constant $I_{0}^{(1)}$ is given by the following explicit expression of the initial data on $\Sigma_0$: 
\begin{equation}
I_{0}^{(1)}[\psi]=\frac{M}{4\pi} \int_{\Sigma_0\cap\mathcal{H}^{+}}\!\!\psi r^2d\omega+\lim_{r_0\rightarrow \infty}\left(\frac{M}{4\pi}\int_{\Sigma_0\cap\{r\leq r_0\}} n_{\Sigma_0}(\psi)d\mu_{\Sigma_0}+\frac{M}{4\pi}\int_{\Sigma_0\cap\{r=r_0\}}\Big(\psi-\frac{2}{M}r\partial_v(r\psi)\Big)r^2d\omega\right),
\label{i011}
\end{equation}
with $\partial_v$ is an outgoing null derivative and $d\mu_{\Sigma_0}$ denotes the induced volume form on $\Sigma_0$.  We note that for compactly supported initial data on the maximal hypersurface $\{t=0\}$, the  above expression for the coefficient $I_{0}^{(1)}[\psi]$ reduces to 
\[I_0^{(1)}[\psi]= \frac{M}{4\pi}\int_{ S_{\text{BF}}}\!\!\psi \, r^2d\omega+\frac{M}{4\pi}\int_{\{t=0\}}\ \frac{1}{1-\frac{2M}{r}}\partial_t\psi\, r^2 dr d\omega,
\]
where $S_{\text{BF}}$ denotes the bifurcation sphere. 

Generic initial data satisfy $I_{0}^{(1)}[\psi]\neq 0$ and hence give rise to solutions to the wave equation which decay exactly like $\frac{1}{\tau^3}$. This result yielded the first \textit{pointwise lower bounds} for solutions to the wave equation on Schwarzschild backgrounds\footnote{Note that the sharpness of the \emph{decay rate} of the time derivative of $\psi$ along the event horizon was first established by Luk and Oh \cite{luk2015}.}.  In other words, \eqref{our1}, \eqref{ourrad1} and \eqref{i011} provide a complete characterization of all solutions to \eqref{eq:waveequation} which satisfy Price's law as a lower bound. We remark that the study of precise late-time asymptotic expansions is very important in issues related to black hole interior regions and, in particular, in addressing the strong cosmic censorship conjecture \cite{MD03,MD05c, MD12, luk2015, LukSbierski2016, DafShl2016, Hintz2015, Franzen2014, Luk2016a,Luk2016b}.

It is important to emphasize that the approach of  \cite{paper2,paper-bifurcate} is based on purely physical space techniques. On the other hand, the heuristic work of Leaver \cite{leaver} related the late-time power law to the branch point at $\omega=0$ in the Laplace transform of Green's function for each fixed angular frequency. This is consistent with the results of \cite{paper2,paper-bifurcate}, in view of the fact that the geometric origin of the constant $I_{0}^{(1)}[\psi]$ is related to \textit{an obstruction to the invertibility of the time operator} $T=\partial_t$ in a suitable function space (and hence is related to the $\omega=0$ frequency in the Fourier space). Indeed, restricting (strictly) to the future of the bifurcation sphere where $T\neq 0$, we have that \textbf{an obstruction to the invertibility of the operator $T$ is the existence of a conservation law along the null infinity} $\I$: For solutions $\psi$ to the wave equation \eqref{eq:waveequation} on Reissner--Nordstr\"om spacetimes, the limits
\[I_{0}[{\psi}](u):=\frac{1}{4\pi}\lim_{r\rightarrow\infty} \int_{\mathbb{S}^{2}}r^2 \partial_r (r{\psi}) (u,r,\omega) \, d\omega \]
are \textbf{independent} of the retarded time $u$. Here, we consider the standard outgoing Eddington--Finkelstein coordinates $(u,r,\omega)$ (with $\omega \in\mathbb{S}^2$). The associated constant 
\begin{equation}
I_{0}[{\psi}]:=I_{0}[{\psi}](u)
\label{np}
\end{equation}
is called the \textit{Newman--Penrose constant} of ${\psi}$ (see \cite{ np2}).
\begin{figure}[H]
\begin{center}
\includegraphics[width=5cm]{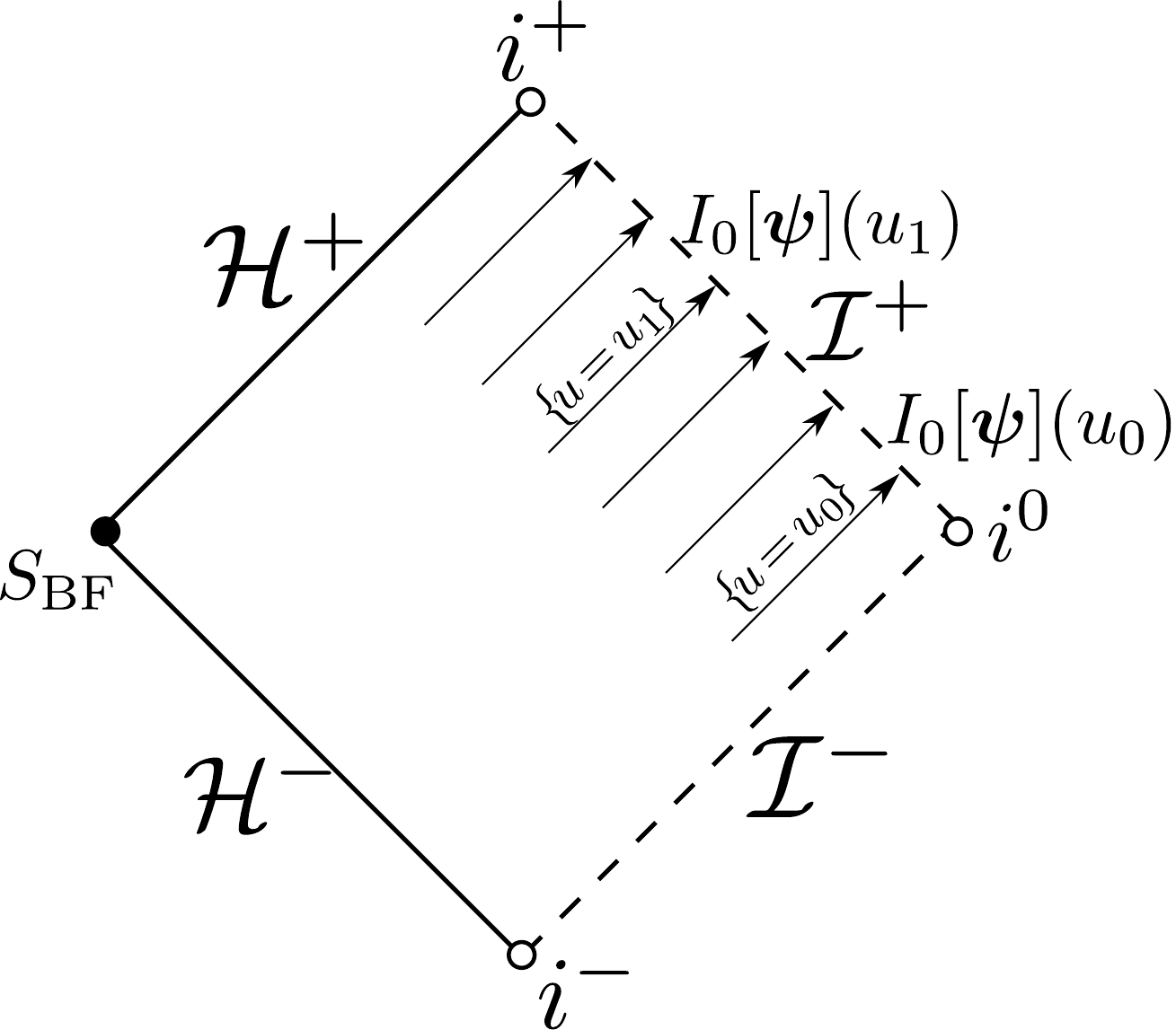}\vspace{-0.25cm} 
\caption{\label{fig:5}The Newman--Penrose constant on $\mathcal{I}^{+}$.}
\end{center}
\end{figure}
The existence of this asymptotic conservation law is an obstruction to inverting the time operator $T$, if the domain of $T$ is taken to be the set of all  smooth solutions $\psi$ to the wave equation which satisfy the condition $|r^2\partial_{r}(r\psi)|\in O_{1}(r^0)$ on the initial hypersurface $\Sigma_0$ (see Section \ref{sec:timeint} for more details).\footnote{We also refer to Section \ref{sec:TheERNManifoldFoliationsAndVectorFields} for a definition of this ``big O notation''.} Indeed, if there is a regular solution $\psi^{(1)}$  to \eqref{eq:waveequation} in the domain of $T$ such that 
\[ T\psi^{(1)}=\psi \]
then we must necessarily have that
\[I_{0}[\psi]=I_0[T\psi^{(1)}]=0.  \]
Conversely, if we consider a smooth initial data on a Cauchy hypersurface $\Sigma_0$ which crosses the event horizon to the future of the bifurcation sphere (see figure below) such that 
\begin{equation}
\lim_{r\rightarrow \infty} \int_{\mathbb{S}^2} r^3 \partial_r(r\psi)|_{\Sigma_0}d\omega<\infty, 
\label{r3condition}
\end{equation}
which in particular implies that $I_{0}[\psi]=0$, then,  by the results in \cite{paper2}, there is a unique smooth spherically symmetric solution $\psi^{(1)}$ to \eqref{eq:waveequation} in the domain of $T$ such that 
\begin{equation}
T\psi^{(1)}=\frac{1}{4\pi}\int_{\mathbb{S}^2} \psi d\omega
\label{timeinvint1}
\end{equation}in $\mathcal{J}^{+}(\Sigma_0)$.
\begin{figure}[H]
\begin{center}
\includegraphics[width=5cm]{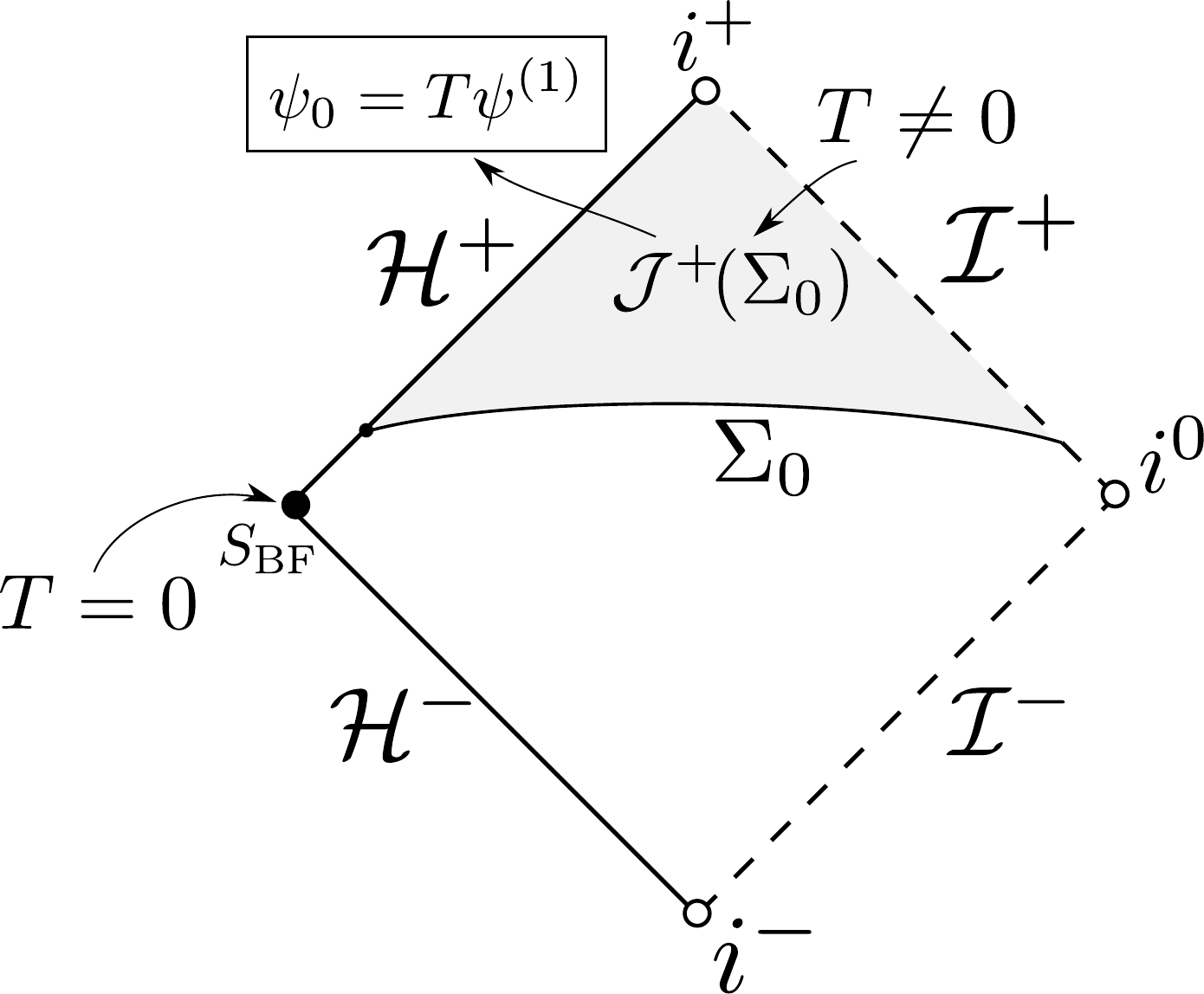}
\caption{\label{fig:2432}Time inversion for the spherical mean $\psi_0$ of $\psi$.}
\end{center}
\end{figure}
Hence, \textbf{$I_0[\psi]$ appears as the unique obstruction to inverting the time operator $T$ on the projection to the spherical mean} of $\psi$. If the Newman--Penrose constant $I_0[\psi]\neq 0$ then \eqref{timeinvint1} has no solution and in this case it is $I_{0}[\psi]$ that appears in the late-time asymptotics of the spherical mean (see \cite{paper2}): for example, at fixed $r=r_0$ we have that
\begin{equation}
\left|\frac{1}{4\pi}\int_{\mathbb{S}^2}\psi(\tau, r_0,\cdot)d\omega-4 I_{0}[\psi]\cdot\frac{1}{\tau^2}\right| \leq C_{r_0} \cdot \sqrt{E_{\Sigma_0}[\psi]}\cdot \frac{1}{\tau^{2+\epsilon}}.
\label{our120}
\end{equation}
We remark that the restriction to the spherical mean is justified by the fact that the non-spherically symmetric projection :
\[\psi_{\ell\geq 1}:=\psi-\frac{1}{4\pi}\int_{\mathbb{S}^2}\psi d\omega,\]
decays at least like $\tau^{-3.5+\epsilon}$ (see \cite{paper2}), for some small  $\epsilon>0$, and hence does \textbf{not} contribute to the leading order terms in the late-time asymptotics. 

If, on the other hand, \eqref{r3condition} holds (and hence $I_0[\psi]=0$) then by the above result $T$ can be inverted to produce the time integral $\psi^{(1)}$. In this case,  the Newman--Penrose constant $I_0[\psi^{(1)}]$ of $\psi^{(1)}$ is an obstruction to acting with $T^{-1}$ on $\psi^{(1)}$, or equivalently, an obstruction to acting with $T^{-2}$ on $\frac{1}{4\pi}\int_{\mathbb{S}^{2}}\psi$. This obstruction is precisely the origin of the coefficient $I_{0}^{(1)}$ in \eqref{our1} and \eqref{ourrad1}, that is
\[I_{0}^{(1)}[\psi]=I_{0}[\psi^{(1)}].\]
Note that $I_{0}^{(1)}[\psi]$ is given in terms of the initial data of $\psi$ by \eqref{i011}. 

Summarizing the asymptotics for the sub-extremal case, we have:
\begin{table}[!ht]{\footnotesize
\begin{center}
    \begin{tabular}{ l | c  }
    \hline
   asymptotics for $\psi$  in $\{r\leq R_0\}$ & origin of the coefficient \\\Xcline{1-2}{0.05cm} 

    $-4I_{0}[\psi]\cdot \frac{1}{\tau^2}$ & $I_{0}[\psi]\neq 0$ unique obstruction to inverting $T$ \\ \hline
   $8I_{0}^{(1)}[\psi]\cdot \frac{1}{\tau^3}$ &  $I_{0}^{(1)}[\psi]\neq 0$ unique obstruction to inverting $T^2$   \\     \hline
  \end{tabular}
\end{center}}
\end{table}

\textbf{In the case of ERN there are additional obstructions to inverting the time operator $T$ which cause many subtle difficulties in obtaining the precise late-time asymptotics} (see Sections \ref{sec:GeometricOriginOfTheNewHair} and \ref{sec:OverviewOfTechniques}).

\subsection{Physical importance of extremal black holes}
\label{sec:PhysicalImportanceOfExtremalBlackHoles}

As has already been mentioned in Section \ref{sec:Introduction}, extremal black holes are of fundamental importance in general relativity. Let us emphasize that an understanding of the dynamical properties of ``exactly'' extremal black holes is relevant also when one is studying the dynamics of ``near-extremal'' black holes over large (but finite) time intervals. In this section we provide a list of references which underpin the intimate connection of extremal black holes with astronomy/astrophysics, high energy physics and classical general relativity.

\vspace{0.2cm}

\noindent{\textbf{Observations of (near-)extremal black holes}}

\vspace{0.15cm}

\noindent

\noindent Astronomical evidence suggests that near-extremal black holes are ubiquitous in the universe.  Various techniques have been developed to analyze the mechanisms for the formation and distribution of near-extremal black holes \cite{near-extremal-accretion, kesden}. It has been suggested that $70\%$ of the stellar black holes, which have been formed from the collapse of massive stars, in the universe are near-extremal \cite{rees2005}. Using techniques from $X$-ray reflection spectroscopy, it has been shown that many supermassive black holes (whose mass is at least 1 billion times the mass of the sun) are near-extremal \cite{brenneman-spin, reynolds-nearextremal}. 
\vspace{0.3cm}

\noindent{\textbf{Observational signatures of extremal black holes}}

\vspace{0.15cm}

\noindent Many astronomical conclusions are based on calculations for exactly Kerr spacetimes. However, time variability might introduce additional observational signatures of extremal black holes, that is features in the observations that are characteristic to the dynamics of extremal black holes. The near-horizon geometry provides a great background for probing such signatures. Such signatures can be divided in two main categories: gravitational signatures \cite{gralla2016, ekerr-plunge} and electromagnetic signatures \cite{extreme-optical2, grallastrominger}. The asymptotics of the present paper derive a new gravitational signature \cite{extremal-prl} (see the discussion
at the end of Section 1.1).  

\vspace{0.3cm}

\noindent{\textbf{Supersymmetry, holography and quantum gravity}}

\vspace{0.15cm}

\noindent Extremal black holes are often supersymmetric as a consequence of the BPS bound. They have zero Hawking temperature and hence play an important role in understanding black hole thermodynamics and the Hawking radiation \cite{haw95}. Quantum considerations of black hole entropy in five-dimensional extremal black holes and applications in string theory can be found in \cite{stromingerextremalentropy, em06}.  One can define a near-horizon limit \cite{gibbonssuper,armen} which yields new solutions to the Einstein equations with conformally invariant properties. These limiting geometries have been classified in \cite{ luciettikund, k09, chrus}. On the other hand, the conformal properties of the near-horizon geometries allow for a description of quantum gravity via a holographic duality \cite{strominger-extremal-holography}.

\vspace{0.3cm}

\noindent{\textbf{Uniqueness and classification of extremal black holes}}

\vspace{0.15cm}

\noindent Extremal event horizon enjoy various rigidity properties \cite{extremalrigidity, hol09}.  Global uniqueness results for extremal black holes in various settings have been obtained in \cite{pauextremalluci, chru, horow}. We also refer to interesting examples of higher dimensional extremal black holes. 

\vspace{0.3cm}

\noindent{\textbf{Extremal black holes as mass minimizers}}

\vspace{0.15cm}

\noindent  Extremal black holes saturate geometric inequalities for the total mass, angular momentum and charge \cite{SD08, chru2008} at higher dimensions \cite{alaee}. They also saturate quasi-local versions of these inequalities for the mass, angular momentum and charge contained in the black hole region \cite{DJR11,  dain2010}.

\vspace{0.3cm}

\noindent{\textbf{Quasinormal modes of extremal black holes}}

\vspace{0.15cm}

\noindent  Starobinski \cite{staro} first investigated the effects of superradiance and extremality. Extensions for quasinormal modes of extremal Kerr were obtained in \cite{detweiler80}  where a sequence of zero damped modes was computed.  Subsequent  analysis was presented in \cite{mhighinsta, glampedakisfull}. The most precise analysis of quasinormal modes in extremal Kerr has been presented in \cite{zeni13}. Gravitational modes of the near extremal Kerr geometry were studied in \cite{harvey-modes-2009}.

\vspace{0.3cm}

\noindent{\textbf{Extremality and non-linear effects}}

\vspace{0.15cm}

\noindent  An intriguing aspect of near-extremal black holes is that they exhibit turbulent gravitational behavior \cite{luis}, that is energy is transferred from high frequencies to low frequencies. Non-linear simulations of formation of binary systems of near-extremal black holes were presented in \cite{extremal-binary-merger}. Furthermore, numerical simulations of the evolutions of the Einstein--Maxwell-scalar field system in a neighborhood of extremal Reissner--Nordstr\"om were studied in \cite{harvey2013}. A general theory of evolution of extremal black holes was developed here \cite{booth16}. For other non-linear works pertaining to the dynamics of extremal black holes we refer to \cite{areangel6, yannis1, bizon-extremal-nonlinear}.

\subsection{The horizon instability of extremal black holes}
\label{sec:TheHorizonInstabilityOfExtremalBlackHoles}

The wave equation on ERN in ingoing Eddington--Finkelstein $(v,r,\theta,\varphi)$ coordinates takes the form
\begin{equation}
\Box_g\psi= D\partial_r\partial_r\psi+2\partial_v\partial_r\psi+\frac{2}{r}\partial_v\psi+R\partial_r\psi+\sD\psi=0,
\label{weern}
\end{equation}
where $D(r)=\left(1-\frac{M}{r}\right)^2$ and $R(r)=\frac{dD}{dr}+\frac{2D}{r}$ and $\sD=\frac{1}{r^2}\sD_{\s^2}$, where $\sD_{\s^2}$ is the standard Laplacian on the round unit sphere $\mathbb{S}^{2}$. 

We will review here the decay, non-decay and blow-up results for \eqref{weern} that were established in \cite{aretakis1,aretakis2, SA10, aag1} and describe the ``\textit{horizon instability of extremal black holes}''. We consider smooth initial data on a spherically symmetric Cauchy hypersurface $\Sigma_0$ which crosses the event horizon and terminates at future null infinity. Recall that the event horizon is the hypersurface given by
\[\h=\{r=M\}. \]
Let $F_{\tau}$ denote the flow of the stationary Killing vector field $T=\partial_v$ and let $\Sigma_{\tau}=F_{\tau}(\Sigma_0)$.

\subsubsection{Conservation laws along the event horizon}
\label{sec:ConservationLawsAlongTheEventHorizon}

Consider the spherical sections $S_{\tau}=\Sigma_{\tau}\cap \h$ of the event horizon. \textit{Restricting to the spherical mean of the wave equation \eqref{weern} on the event horizon} yields
\[ \partial_{v}\left(\int_{S_{\tau}}\left(2\partial_r\psi+2M^{-1}\psi\right)\, M^2\, d\omega\right)=0. \]
Since $\partial_v$ is null and normal to the event horizon $\h$, it immediately follows that  \textit{the surface integrals}
\begin{equation}
H_0[\psi]:=-\frac{M^2}{4\pi}\int_{S_{\tau}}\partial_r(r\psi) \, d\omega
\label{introhorizonH}
\end{equation} 
 \textit{are independent of} $\tau$. Here $d\omega=\sin\theta d\theta d\varphi$ is the volume form of the unit round sphere $\mathbb{S}^2$ with $\omega=(\theta,\varphi)$. This gives rise to a conservation law along the event horizon. Surprisingly, an analogous conservation law holds for each projection on the eigenspace of the angular Laplacian. Indeed, it can be shown that if $\psi_{\ell}$ denotes the projection of $\psi$ on the eigenspace $E_{\ell}$ of $\sD$ with eigenvalue $-\frac{\ell (\ell+1)}{r^2}$, then the following derivative $\psi_\ell$ of order $\ell+1$ that is transversal to $\h$,
\[\partial_{r}^{\ell} \Big( r\partial_r(r \psi_{\ell})\Big), \]
is \textit{constant along the null generators of the event horizon}. 

It is important to emphasize that the derivative $\partial_r$ is translation-invariant (since $[\partial_r,\partial_v]=0$) and hence the above conservation laws provide highly non-trivial \textit{obstructions to decay} for all the geometric quantities associated to a scalar field. Summarizing we have the following: 

\vspace{0.2cm}
\textbf{Hierarchy of conservation laws on ERN:} \textit{for every fixed angular frequency $\ell$ we have a conservation law along the event horizon involving exactly the first $\ell+1$ translation-invariant, transversal derivatives of the scalar field on the event horizon. }

\vspace{0.2cm}
An analogue of this hierarchy for axisymmetric solutions on extremal Kerr was obtained in \cite{aretakis4}. Lucietti and Reall \cite{hj2012} generalized this hierarchy for electromagnetic and gravitational perturbations of extremal Kerr which they used to derive a \textit{gravitational instability of extremal Kerr}. We remark that these conservation laws are a feature characteristic to extremal event horizons. Indeed, it was shown in \cite{aretakisglue} that non-extremal horizons do \textbf{not} admit conservation laws associated to solutions of the wave equation. Further extensions of these conservation laws have recently been provided in \cite{godazgar17}.

\subsubsection{The trapping effect on the event horizon}
\label{sec:TheTrappingEffect}

Let $N$ be a translation-invariant future-directed timelike vector field defined globally in the domain of outer communications up to and including the event horizon. This vector field will be used to measure the energy $E_{\gamma}(s)$ of affinely-parametrized null geodesics  $\gamma(s)$:
\[E_{\gamma}(s)=g\left(\overset{\cdot}{\gamma}(s),N\right),\]
where $\overset{\cdot}{\gamma}(s)=\frac{d\gamma}{ds}(s)$.  A key observation is that for sub-extremal black holes  the energy $E_{\gamma}(s)$ of the null generators of the event horizon with positive surface gravity $\kappa>0$ decays exponentially in $s$. On the other hand, the energy $E_{\gamma}(s)$ of the null generators of the event horizon of ERN remains constant for all $s$. This is intimately related to the geometric characterization of extremal horizons, namely that the Killing normal vector field to the event horizon gives rise to an affine foliation of the event horizon. Sbierski \cite{janpaper} used the Gaussian beam approximation and the above result to show that there are solutions to the wave equation on ERN that are localized in a neighborhood of $\h$ with almost constant energy across $\Sigma_{\tau}$ for arbitrarily large $\tau$. This result immediately yields an obstruction to proving local integrated estimates bounding
\[ \Gamma_1[\psi]=\int_{0}^{\infty}\left(\int_{\Sigma_{\tau}\cap \{r\leq M+\epsilon\}} |\partial\psi|^2 \right)\, d\tau\]
for some arbitrarily small $\epsilon>0$. 
Specifically, Sbierski's result shows that the above integral cannot be bounded purely in terms of the initial energy of $\psi$ on $\Sigma_0$. A Morawetz estimate bounding $\Gamma_1[\psi]$ was established in \cite{aag1} where it was shown that such an estimate requires
\begin{enumerate}
		\item \textit{the finiteness of a weighted higher-order norm of the initial data, and }
 \item\textit{ the vanishing of the conserved charge $H_{0}[\psi]$. }
\end{enumerate}
Furthermore, it was shown that \textit{for smooth and compactly supported initial data, $\Gamma_1[\psi]$ is \underline{infinite} if and only if  $H_0[\psi]\neq 0$.}

The first requirement above is reminiscent to that of the Morawetz estimates  on the photon sphere which accounts for the high-frequency solutions localized on the trapped null geodesics. On the other hand, the second requirement is a global (low-frequency) condition on all of the event horizon, that is on all the null generators of the event horizon. This shows that \textit{the event horizon on ERN exhibits a global trapping effect. }

Another characteristic feature of the event horizon on ERN is the following \textit{stable higher-order trapping effect}:\textit{ For generic smooth and compactly supported initial data with support away from the event horizon, the following higher-order integral }
\[ \Gamma_k[\psi]=\int_{0}^{\infty}\left(\int_{\Sigma_{\tau}\cap \{r\leq M+\epsilon\}} |\partial^k\psi|^2 \right)\, d\tau\]
\textit{is infinite, for all }$k\geq 2$. For more details see \cite{aag1}.

Bounding the integral in time of the energy flux through $\Sigma_{\tau}$ is further obstructed by the standard photon sphere which is an obstruction present for general black hole spacetimes. We refer to \cite{lecturesMD, aretakis2} for the details. 

\subsubsection{Energy and pointwise boundedness and weak decay}
\label{sec:EnergyAndPointwiseBoundednessAndDecay}

An important aspect of ERN is that the Killing vector field $T=\partial_v$ is globally causal. That implies that the conserved energy $T$-fluxes are non-negative definite. However, since $T$ is null at the horizon, the $T$-flux $\mathcal{E}^{T}[\psi]$ along $\Sigma_{\tau}$ \textit{degenerates} at the horizon. Schematically, we have
\[\mathcal{E}^{T}_{\Sigma_{\tau}}[\psi] \sim  \int_{\Sigma_{\tau}} \left(1-\frac{M}{r}\right)^2\cdot |\partial\psi|^2\, d\mu_{\Sigma_{\tau}}.  \] 
Clearly, we have 
\[ \mathcal{E}^{T}_{\Sigma_{\tau}}[\psi]\leq \mathcal{E}^{T}_{\Sigma_{0}}[\psi]. \]
The above estimate was also used in \cite{dd2012} where various boundedness results where shown for the wave equation on ERN away from the event horizon. One can go beyond such boundedness estimates and derive decay for the $T$-flux (see \cite{aretakis2}):
\begin{equation}
\mathcal{E}^{T}_{\Sigma_{\tau}}[\psi] \leq C\cdot \frac{1}{\tau^2}\cdot E[\psi].
\label{energydecaytintro}
\end{equation}
where $E[\psi]$ is an appropriate weighted higher-order energy norm of the initial data.  Using this type of estimate, it can be shown that $\psi$ satisfies the following \textit{pointwise decay estimate}
\begin{equation}
|\psi|_{\Sigma_{\tau\cap \{r\geq r_0\}}}\leq C_{r_0}\cdot \frac{1}{\tau}\cdot E[\psi]
\label{eq:awayintro}
\end{equation}
away from the event horizon $r\geq r_0>M$.

To obtain non-degenerate control of $\psi$ and its derivatives along the event horizon, we consider the energy flux $\mathcal{E}^N[\psi]$ associated to the timelike vector field $N$ which satisfies the \textit{positivity} property
\[\mathcal{E}^{N}_{\Sigma_{\tau}}[\psi]\sim  \int_{\Sigma_{\tau}}  |\partial\psi|^2\, d\mu_{\Sigma_{\tau}}.  \] 
It turns out that there is a uniform positive constant $C$ such that (see \cite{aretakis1})
\[ \mathcal{E}^{N}_{\Sigma_{\tau}}[\psi]\leq C\cdot \mathcal{E}^{N}_{\Sigma_{0}}[\psi]. \]
On the other hand, no decay estimate was known for $J^N$. Nonetheless, via an interpolation argument, it can be shown that $\psi$ does decay along the event horizon: 
\begin{equation}
|\psi|_{\Sigma_{\tau}\cap \h}\leq C \cdot \frac{1}{\tau^{\frac{3}{5}}}\cdot E[\psi]
\label{eq:onintro}
\end{equation}
The decay estimates \eqref{eq:awayintro}, \eqref{eq:onintro} were the only decay rates  that had been proved rigorously for scalar fields $\psi$ on ERN. In this paper, we derive the sharp rates (upper and lower bounds); in fact we derive the precise late-time asymptotics for $\psi$. See Section \ref{sec:SummaryOfTheMainResults}.

\subsubsection{Energy and pointwise blow-up}
\label{sec:EnergyAndPointwiseBlowUp}

As we shall see, the decay rates in \eqref{eq:awayintro}, \eqref{eq:onintro} are not sharp. However, they do suggest that \textit{the decay rate of $\psi$ along the event horizon is slower than the decay rate of $\psi$ away from the horizon}. This statement, which rigorously follows from the main results of the present paper, is a precursor of the horizon instability of ERN. Recall that with respect to the spherical sections $S_{\tau}$ of $\h$, the spherical means  $-\frac{1}{4\pi }\int_{S_{\tau}}\left(\partial_r\psi+M^{-1}\psi\right)\, d\omega$ are conserved. On the other hand, for generic initial data on $\Sigma_0$ we have $-\frac{1}{4\pi }\int_{S_{0}}\left(\partial_r\psi+M^{-1}\psi\right)\,d\omega=\frac{1}{M^3}H_0[\psi]\neq 0$. Hence, in view of the estimate \eqref{eq:onintro}, we conclude the following 
\vspace{0.1cm}

\textbf{Non-decay:} \textit{generically, the spherical mean of the transversal derivative $-\frac{1}{4\pi M^2}\int_{S_{\tau}}\partial_r\psi$ does not decay along the event horizon of ERN. In fact, }
\[ -\frac{1}{4\pi }\int_{S_{\tau}}\!\partial_r\psi \, d\omega\ \rightarrow \frac{1}{M^3}H_0[\psi], \ \text{ as }\tau\rightarrow \infty. \]
The non-decaying transversal derivative along the event horizon accounts for the different decay rates of $\psi$ on and away from the horizon $\h$. On the other hand, it can be shown that $\partial_r\psi$ decays along the hypersurfaces $\{r=r_0>M\}$ away from the event horizon $\h$. It was observed in \cite{hm2012} that the above non-decay result implies that \textit{the component $\T_{rr}[\psi]$ of the energy-momentum tensor of the scalar field $\psi$ does not decay along} $\h$. In fact, we have 
\[\frac{1}{4\pi }\int_{S_{\tau}}\!\T_{rr}[\psi]\, d\omega\rightarrow \frac{1}{M^6}\left(H_{0}[\psi]\right)^2.\] 
Since, $\T_{rr}[\psi]$ is related to the energy density measured by an observer crossing $\h$, the authors of \cite{hm2012} concluded that the conserved charge $H_0[\psi]$ might be thought of as ``hair'' of the extremal event horizon. It is important to remark that \textit{the results of the present paper yield a new way in potentially measuring this hair from observations along null infinity}. See Section \ref{sec:SummaryOfTheMainResults}.

By acting with $\partial_r$ on the wave equation \eqref{weern}, restricting on the event horizon and using the previous results we conclude the following:

\vspace{0.1cm}

\textbf{Blow-up:} \textit{the spherical mean of higher-order transversal derivatives $-\frac{1}{4\pi}\int_{S_{\tau}}\partial_r^k\psi\, d\omega$ with $k\geq 2$ generically blows up along the event horizon of ERN. In fact, }
\[ \frac{1}{4\pi }\int_{S_{\tau}}\!|\partial_r^k\psi| \, d\omega \geq c_k\cdot  H_0[\psi]\cdot \tau^{k-1}, \ \text{ as }\tau\rightarrow \infty. \]
Furthermore, the following \textit{higher-order energy blow-up} result generically holds (see \cite{aretakis2}):
\[\mathcal{E}^{N}_{\Sigma_{\tau}}[N^k\psi]\rightarrow \infty \]
for all $k\geq 1$ as $\tau\rightarrow \infty$.

We remark that an extension of the above instabilities to linearized electromagnetic and gravitational perturbations of ERN was presented in \cite{hm2012} and \cite{sela2}. Nonlinear extensions have been presented in \cite{aretakis2013, harvey2013, bizon-extremal-nonlinear, areangel6, harveyeffective}. For higher-dimensional extensions we refer to \cite{murata2012}. For a more detailed discussion of works in the physics literature, see the next section.

\subsection{Physics literature on the dynamics of extremal Reissner--Nordstr\"{o}m }
\label{sec:PreviousWorksOnLateTimeAsymptoticsOnERN}

In this section we present results in the physics literature which concern the late-time asymptotics for ERN.

\subsubsection{The Blaksley--Burko asymptotic analysis}
\label{sec:TheBlaksleyBurkoAsymptoticAnalysis}

The first work on asymptotics of scalar fields on ERN goes back to 1972 when Bi{\v{c}}{\'a}k suggested in \cite{Bicak1972} that scalar fields $\psi_{\ell}$ on ERN with non-vanishing Newman--Penrose constant and with angular frequency $\ell$ decay with the rate $\frac{1}{t^{\ell+2}}$. However, this result was shown to be false in 2007 when Blaksley and Burko \cite{Burko2007} performed a more accurate heuristic and numerical analysis. Their work considered the following two types of initial data (see Section \ref{sec:Introduction} for the relevant definition):
\begin{itemize}
	\item Type I: horizon-penetrating and null-infinity-extending,
	\item Type II: Supported away from the horizon and compactly supported.
\end{itemize}
Define $\mu\in \{0,1\}$ such that $\mu=0$ for data of Type I and $\mu =1$ for data of type II. The authors argued that the \textit{sharp} decay rates for the scalar field are the following:
\begin{itemize}
	\item Away $\h$ and $\I$:  $\ |\psi_{\ell}|_{r=r_0>M} \ \ \text{decays like} \ \  \frac{1}{\tau^{2\ell+2+\mu}}$,
		\item On $\h$:   $\ |\psi_{\ell}|_{\h}   \ \ \text{decays like} \ \   \frac{1}{\tau^{\ell+1+\mu}}$,
		\item On $\I$: $\ |r\psi_{\ell}|_{\I}  \ \ \text{decays like} \ \   \frac{1}{\tau^{\ell+1+\mu}}$.  
\end{itemize}
Reference \cite{Burko2007} did not obtain the precise late-time asymptotics in the above two cases. 
Moreover, \cite{Burko2007} did not study other types of initial data, and in particular, did not study horizon-penetrating compactly supported initial data.

\subsubsection{The Lucietti--Murata--Reall--Tanahashi asymptotic analysis}
\label{sec:TheWorkOfReallEtAl}

The asymptotic analysis of Lucietti--Murata--Reall--Tanahashi \cite{hm2012} was the first work to numerically investigate the precise late-time asymptotics for scalar fields on ERN. The present paper is highly motivated by \cite{hm2012}. 

A major result of the numerical analysis of \cite{hm2012} is the following precise late-time asymptotic behavior of scalar fields with compactly supported initial data
\begin{equation}
M\cdot\psi|_{\h} \sim  2H_0[\psi]\cdot \frac{1}{\tau}+4MH_0[\psi]\cdot \frac{\log \tau}{\tau^2},  \ \text{ as } \tau\rightarrow \infty.
\label{harveyhorizon}
\end{equation}
Furthermore, the authors suggested, using a near-horizon calculation, that the following precise late-time asymptotic behavior off the horizon along $r=r_0>M$ holds:
\begin{equation}
\psi|_{\{r=r_0\}} \sim  \frac{4M}{r_0-M}H_0[\psi]\cdot\frac{1}{\tau^2},  \ \text{ as } \tau\rightarrow \infty.
\label{harveyaway}
\end{equation}
Moreover, the authors, extrapolating from numerical simulations for the $\ell=1,2$ angular frequencies, suggested the following sharp rate off the horizon along $r=r_0>M$ \:
\begin{equation}
|\psi_{\ell}|_{\{r=r_0\}}  \ \ \text{decays like} \ \   \frac{1}{\tau^{2\ell+2}}.
\label{harveyell}
\end{equation}
On the other hand, the numerics of \cite{hm2012} suggested the following asymptotic expression  in the case of data with $H_0[\psi]=0$
\begin{equation}
\psi|_{\h}\sim \frac{C_0}{\tau^2},  \ \text{ as } \tau\rightarrow \infty.
\label{harvey0hor}
\end{equation}
However, the following points were not addressed in \cite{hm2012}:
\begin{itemize}
	\item The constant $ C_0$ in \eqref{harvey0hor} was not explicitly computed in terms of the initial data. 
	\item The precise asymptotic estimate \eqref{harveyhorizon} was only obtained for compactly supported data, and not, in particular, for data with non-vanishing Newman--Penrose constant. 
	\item The asymptotics of the radiation field $r\psi_{\I}$ along the null infinity $\I$ were not investigated (in the $H_0[\psi]\neq 0$ case).
\end{itemize}
In the present paper, we address all the above issues (see Section \ref{sec:SummaryOfTheMainResults}). 

Another important question that was first raised and investigated in \cite{hm2012} is whether one can trigger the horizon instability using ingoing radiation; that is, using perturbations which are initially supported away from the event horizon and hence necessarily satisfy $H_0[\psi]=0$. The authors found the following stability results
\[|\psi|_{\h}\rightarrow 0,\ \ \ |\partial_{r}\psi|_{\h}\rightarrow 0: \ \text{ along }\h, \]
and uncovered the following (generic) instability behavior
\[ |\partial_{r}^2\psi|_{\h}\nrightarrow 0 \ \ \  |\partial_{r}^3\psi|_{\h}\rightarrow \infty: \ \text{ along }\h. \]
 This instability behavior, which has also been discussed in \cite{bizon2012}, was subsequently rigorously proved in \cite{aretakis2012}. 

Reference \cite{hm2012} also investigated the late-time behavior of massive scalar fields which solve $\Box_g \psi= m^2\psi$. For such massive fields it is widely believed that the late-time behavior is dominated by the $\omega=\pm m$ frequencies (instead of the $\omega=0$ frequency for massless fields on sub-extremal black holes) which results in a damped oscillatory late-time behavior. In particular, massive fields and all their derivatives are expected to decay like ${\tau^{-\frac{5}{6}}}$ in the domain of outer communications (up to and including the event horizon) of a sub-extremal black hole. The results of \cite{hm2012} suggest that this remains true on ERN backgrounds off the horizon (a result that had also been seen in \cite{koyama2001}). On the other hand, \cite{hm2012} found that the horizon instability persists for a \textit{discrete} set of masses $m^2$. Specifically, if  $(mM)^2=n(n+1)$ then the authors argued that 
\[ |\partial_{r}^{n+1}\psi|_{\h}\nrightarrow 0 \ \ \  |\partial_{r}^{n+2}\psi|_{\h}\rightarrow \infty: \ \text{ along }\h. \]
More generally, the numerical analysis of \cite{harvey2013} suggests  the following asymptotic behavior for \textit{general} masses $m^2$:
\[\partial_{r}^{k}\psi\ \ \text{ behaves like } \tau^{k-\frac{1}{2}-\sqrt{(mM)^2+\frac{1}{4}}},\] for all $k\geq 0$. 
A rigorous proof of the above statements for massive fields remains open.

\subsubsection{The Ori--Sela asymptotic analysis}
\label{sec:TheOriSelaAsymptoticAnalysis}

Ori \cite{ori2013} and Sela \cite{sela} used the conservation laws that hold for each fixed angular frequency $\ell$ (see Section \ref{sec:ConservationLawsAlongTheEventHorizon}) to heuristically obtain the precise late-time asymptotics of $\psi_{\ell}$ for horizon-penetrating compactly supported initial data. 
Specifically, they found that along $r=r_0>M$ away from the horizon the following holds:
\[ \psi_{\ell}|_{\{r=r_0\}}\sim (-4)^{\ell+1}eM^{3\ell+2}\frac{r}{(r-M)^{\ell+1}}\cdot \frac{1}{\tau^{2\ell+2}}, \ \ \text{ as }\tau\rightarrow \infty, \]
where $e$ is an explicit expression of the conserved charge $H_{\ell}[\psi_{\ell}]$ for $\psi_{\ell}$. Hence, the above result improves the statement \eqref{harveyell} of \cite{hm2012}. 

Furthermore, Ori and Sela derived  the precise late-time asymptotics of $\psi_{\ell}$ along the horizon
\[\psi_{\ell}|_{\h}\sim e (-M)^{\ell+1}\cdot \frac{1}{\tau^{\ell+1}},\ \ \text{ as }\tau\rightarrow \infty, \]
where $e$ is as above. The recent Fourier based work of Bhattacharjee  et al \cite{ind2018} supported the validity of the above  asymptotics.

On the other hand, no asymptotic estimate was derived for the radiation field $r\psi_{\ell}|_{\I}$ along null infinity. Furthermore, the authors did not obtain precise late-time asymptotics in the case where the initial data are supported away from the event horizon and did not provide an explicit expression for the constant $C_0$ that appears in the asymptotic statement \eqref{harvey0hor}. 

Sela \cite{sela2} subsequently used the decay rates obtained in \cite{ori2013,sela} in order to obtain decay rates for the coupled electromagnetic and gravitational system for ERN.

\subsubsection{The Murata--Reall--Tanahashi spacetimes}
\label{sec:TheReallSpacetimes}

In a very beautiful work \cite{harvey2013}, \textit{Murata, Reall and Tanahashi studied numerically the fully non-linear evolution of the horizon instability of ERN}. Specifically, the authors of \cite{harvey2013} investigated perturbations of ERN in the context of the Cauchy problem for the \textit{spherically symmetric Einstein--Maxwell-(massless) scalar field system}. The authors studied various types of perturbations and obtained a great number of results, all of which are consistent with the linear theory described in the previous sections. Specifically, the authors numerically showed that the maximum size of higher-order derivatives of the first and second order mass perturbations does not go to zero as we let the size of the initial perturbation  go to zero.  They also numerically showed that the situation gets more dramatic when we consider ``smaller'' perturbations of ERN which evolve to what  the authors called   \textit{dynamically extremal} spacetimes. Such spacetimes do not contain trapped surfaces. In this case, the perturbations do grow unboundedly along the event horizon. This may be interpreted as a remnant of the horizon instability in the non-linear theory.

\subsubsection{Addendum: The Casals--Gralla--Zimmerman work on extremal Kerr}
\label{sec:RecentBreakthroughOfCasalsGrallaAndZimmerman}

Casals, Gralla and Zimmerman \cite{zimmerman1} were the first to derive the late-time asymptotics along the event horizon for $\psi_m$ on extremal Kerr, where $\psi_m$ denotes that projection on the $m$'th  the azimuthal frequency. Their semi-analytic work, which is based on the  mode decomposition method of Leaver \cite{leaver}, yielded that $
|\psi_{m}|_{\h}$  {decays like}  $  \frac{1}{\sqrt{\tau}}$ and that  $\partial_r\psi_{m}|_{\h}$  {behaves like}    $\sqrt{\tau}$.  Further extensions have been provided in 
\cite{harveyeffective, zimmerman3, zimmerman2,  zimmerman4, zimmerman5, khanna17}.

\subsection{Outline of the paper}
\label{sec:Outline}

The geometry of ERN along is presented in Section \ref{geometrysection}. The types of initial data for the wave equation that we will consider are introduced in Section \ref{sec:TheTypesOfInitialDataABCD}. A non-technical summary of our results and various applications are presented in Sections \ref{sec:SummaryOfTheMainResults}  and \ref{sec:Remarks}, respectively. The main theorems and an overview of the ideas of the proofs can be found in Section \ref{subsec:TheMainTheorems}. The weighted hierarchies are derived in Sections \ref{sec:rweightest} and \ref{sec:extendhier} and pointwise and energy decay results are obtained in Section \ref{sec:decayest}. The precise late-time asymptotics are derived in Sections \ref{sec:asympnonzeroconst}--\ref{sec:hoasymp}.

\subsection{Acknowledgements}
\label{sec:Acknowledgements}

We would like to thank our mentor Mihalis Dafermos for several insightful discussions. We would also like to thank Harvey Reall, Peter Zimmerman and Samuel Gralla for elucidative conversations.  The second author (S.A.) acknowledges support through NSF grant DMS-1265538, NSERC grant 502581, an Alfred P. Sloan Fellowship in Mathematics and the Connaught Fellowship 503071.

\section{The geometry of ERN}
\label{geometrysection}

\subsection{The ERN metric}
\label{sec:TheERNManifoldFoliationsAndVectorFields}

 The \emph{extremal Reissner--Nordstr\"om} spacetimes  $(\mathcal{M}_M,g_M)$, $M>0$, are given by the following manifold-with-boundary
\begin{equation*}
\mathcal{M}_{M}=\R\times [M ,\infty) \times \s^2,
\end{equation*}
equipped with the coordinate chart $(v,r,\theta,\varphi)$, where $v\in \R$, $r\in [M,\infty)$ and $(\theta,\varphi)$ is the standard spherical coordinate chart on the round 2-sphere $\s^2$.
and the following Lorentzian metric
\begin{equation*}
g_{M}=-D(r)dv^2+2dvdr+r^2(d\theta^2+\sin^2\theta d\varphi^2),
\end{equation*}
where
\begin{equation*}
D(r)=(1-Mr^{-1})^2.
\end{equation*}
We denote the \emph{future event horizon} as the boundary $\mathcal{H}^+=\partial \mathcal{M}_{M}=\{r=M\}$.   We will also denote 
\[T:=\partial_v,\ \ \ Y:=\partial_r.\]
Note that $Y$ is \textit{transversal} to the event horizon. For ERN we have
\[\nabla_{T}T=0\]
on the event horizon. This means that $\h$ has vanishing surface gravity in ERN. 

 We next introduce the \textit{tortoise} coordinate $r^{*}$ by
\[r_*(r)=r-M-M^2(r-M)^{-1}+2M\log\left(\frac{r-M}{M}\right).
\]
The \textit{double null coordinate chart} $(u,v,\theta,\varphi)$ in the manifold $\mathring{\mathcal{M}}:=\mathcal{M}\setminus \partial M$, is given by 
\[u=v-2r_*(r).\]
with $u,v\in \R$. In double null coordinates, the extremal Reissner--Nordstr\"om metric can be expressed as
\begin{equation*}
g_{M}=-D(r)dudv+r^2(d\theta^2+\sin^2\theta d\varphi^2).
\end{equation*}
If we consider the vector fields 
\[L:=\partial_v \ \  \text{ and }\ \ \underline{L}=\partial_u,\]
with respect to the double null coordinates $(u,v,\theta,\varphi)$, then we have the relations
\begin{equation*}
L=\partial_v+\frac{1}{2}D\partial_r, \ \ \underline{L}=-\frac{1}{2}D\partial_r.
\end{equation*}
Finally, we define $t=(u+v)/2$ and introduce the coordinate system $(t,r^{*},\theta,\varphi)$ with respect to which the metric takes the form
\[ g_{M}=-D(r)dt^2+D(r)d(r^{*})^2+r^2(d\theta^2+\sin^2\theta d\varphi^2).
\]
Note that the vector field $T:=\partial_t$ is Killing and timelike everywhere away from the event horizon. 

The null hypersurfaces  $C_{\tau} = \{u = \tau \}$ terminate in the future (as $r,v \rightarrow \infty$) at \textit{future null infinity} $\mathcal{I}^{+}$. We will  occasionally use the notation $v_{r_0}(u')$, with $r_0>M$, to indicate the value of the $v$ coordinate along the hypersurface $\{r=r_0\}$ at $u=u'$, and similarly, $u_{r_0}(v')$ to indicate the value of the $u$ coordinate along the hypersurface $\{r=r_0\}$ at $v=v'$. 

We use the notation $\slashed{\nabla}_{\s^2}$ for the covariant derivative with respect to the metric of the unit round 2-sphere and $\slashed{\Delta}_{\s^2}$ for the corresponding Laplacian.

We will also use the following ``big O notation'' with respect to $u$, $v$ and $r$. We use the notation $O_k(r^{-l})$ to indicate functions $f$ on (a subset of) $\mathcal{M}_M$ that satisfy the behavior $|Y^kf|\leq C r^{-\ell-k}$, where $C>0$ is a constant that is independent of $f$ and $k\in \N_0$, $l\in \Z$. Similarly, we use the notation  $O_k(u^{-l})$ and $O_k(v^{-l})$ when $|\underline{L}^kf|\leq C u^{-\ell-k}$ and  $|L^kf|\leq C v^{-\ell-k}$, respectively. Finally, we will also employ the notations $O_k((v-u)^{-l})$ and $O_k((u-v)^{-l})$ to group functions $f$ that satisfy, for $k_1+k_2=k$, $k_1,k_2\in \N_0 $, $|\underline{L}^{k_1}L^{k_2}f|\leq C |v-u|^{-\ell-k}$ and $|\underline{L}^{k_1}L^{k_2}f|\leq C |u-v|^{-\ell-k}$, respectively. When $k=0$ in the above notation, we will omit the subscript in $O_k$.

\subsection{The spacelike-null foliation}
\label{sec:TheHyperboloidalFoliation}

Let $\Sigma_0$ be a spherically symmetric hypersurface which crosses the event horizon and terminates at null infinity:
\begin{equation*}
\Sigma_0:=\{v=v_{\Sigma_0}(r)\},
\end{equation*}
where $v_{\Sigma_0}: [M,\infty)\to \R$ is a function defined as follows
\begin{equation}
v_{\Sigma_0}(r)=v_{\rm \min}+\int_M^rh(r')\,dr',
\label{definitionh}
\end{equation}
where we take $v_{\rm \min}\in \R_+$ to be a constant and $h: [M,\infty)\to \R_{\geq 0}$ is a non-negative function satisfying 
\begin{align*}
0\leq 2D^{-1}(r)-h(r)&\:=O(r^{-1-\eta}),
\end{align*}
for some constant $\eta>0$. We will take $v_{\Sigma_0}(r)$ to be monotonically increasing function. Moreover, $u_{\Sigma_0}(r):=u|_{\Sigma_0}(r)=v_{\Sigma_0}(r)-2r_*(r)$ satisfies $\frac{du_{\Sigma_0}}{dr}=h(r)-2D^{-1}(r)\leq 0$, so $u_{\Sigma_0}(r)$ is a  monotonically decreasing function.  For convenience, we will assume that $\Sigma_0$ satisfies the following symmetry condition:
\[(t,r^{*})\in \Sigma_0 \Longrightarrow (t,-r^{*}) \in \Sigma_0. \]
This condition is here imposed only for convenience because it simplifies the expressions of various new quantities that we introduce in this paper; our results apply for general initial hypersurfaces $\Sigma_0$ as well. An important example of such a hypersurface is defined as follows: Let $M<r_{\mathcal{H}}<2M$ and $r_{\mathcal{I}}>2M$ such that $r^*(r_{\mathcal{H}})=-r^*(r_{\mathcal{I}})$. Then we may further assume that 
\begin{equation*}
\begin{split}
\Sigma_0\cap \{r\leq r_{\mathcal{H}}\}=&{N}_0^{\mathcal{H}}: =\{v=v_0\}\cap \{r\leq r_{\mathcal{H}}\} ,\\
\Sigma_0\cap \{r\geq r_{\mathcal{I}}\}=&{N}_0^{\mathcal{I}}: =\{u=u_0\}\cap \{r\geq r_{\mathcal{I}}\},
\end{split}
\end{equation*}
with $u_0,v_0>0$.

Let $F_{\tau}$ denote the flow of the stationary vector field $T$ where the \emph{time function} $\tau:J^+({\Sigma_0})\to \R_{\geq 0}$ is defined as follows
\begin{align*}
\tau|_{\Sigma_0}=&\:0,\\
T(\tau)=&\:1.
\end{align*} 
Note that for all $\tau \geq 1$ we have
\[\tau\sim v \text{ for } r\leq r_{\mathcal{H}},\ \  \tau\sim v \sim u \text{ for } r_{\mathcal{H}}\leq r \leq r_{\mathcal{I}}, \ \  \tau \sim u \text{ for } r\geq r_{\mathcal{I}}. \ \]
We  define the following foliation of the future $\mathcal{R}=J^+({\Sigma_0})$ of $\Sigma_0$:
\[\mathcal{R}=\cup_{\tau\geq 0}\Sigma_{\tau}=F_{\tau}(\Sigma_0);  \]
see Figure \ref{fig:sigma0}.

We use the notations $d\mu_{\Sigma_{\tau}}$ and $d\mu_{\tau}$ to indicate the natural volume form on $\Sigma_{\tau}$ with respect to the induced metric, where on the null parts $\mathcal{N}^{\mathcal{H}}_{\tau}$ and $\mathcal{N}^{\mathcal{I}}_{\tau}$ we take this volume form to be $r^2d\omega du$ and $r^2d\omega dv$, respectively, where $d\omega=\sin\theta d\theta d\varphi$. Similarly, we denote the normal vector field to $\Sigma_{\tau}$ with $\mathbf{n}_{\Sigma_{\tau}}$ and $\mathbf{n}_{\tau}$, where we take the normal to $\mathcal{N}^{\mathcal{H}}_{\tau}$ and $\mathcal{N}^{\mathcal{I}}_{\tau}$ to be $\underline{L}$ and $L$, respectively.
\begin{figure}[H]
\begin{center}\vspace{-0.25cm} 
\includegraphics[scale=0.2]{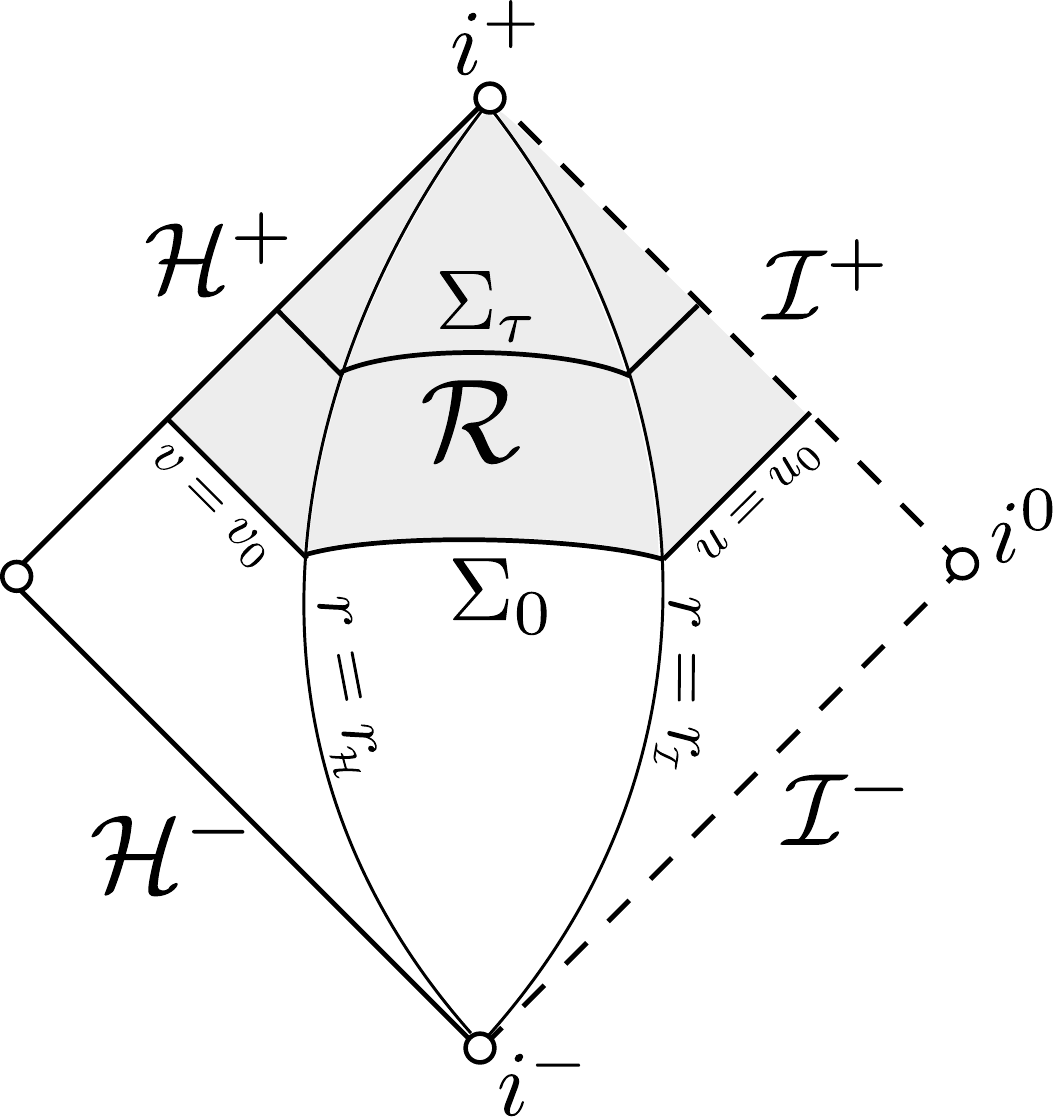}
\caption{The spacelike-null foliation $\Sigma_{\tau}$.}
\label{fig:sigma0}
\end{center}
\end{figure}
\vspace{-0.2cm} 
It will be useful to moreover introduce the following corresponding partition of the spacetime region  $\mathcal{R}$:
\begin{equation*}
\mathcal{R}=\mathcal{A}^{\mathcal{H}}\cup\mathcal{B}\cup \mathcal{A}^{\mathcal{I}},
\end{equation*}
where
\begin{align*}
\mathcal{A}^{\mathcal{H}}:=\mathcal{R}\cap\{r\geq r_{\mathcal{H}}\},  \ \ \mathcal{B}:=\mathcal{R}\cap\{r_{\mathcal{H}}<r<r_{\mathcal{I}}\},\ \ 
\mathcal{A}^{\mathcal{I}}:=\mathcal{R}\cap\{r\leq r_{\mathcal{I}}\};
\end{align*}
see Figure \ref{fig:sigmaabc0}.
\begin{figure}[H]
\begin{center}\vspace{-0.45cm} 
\includegraphics[scale=0.2]{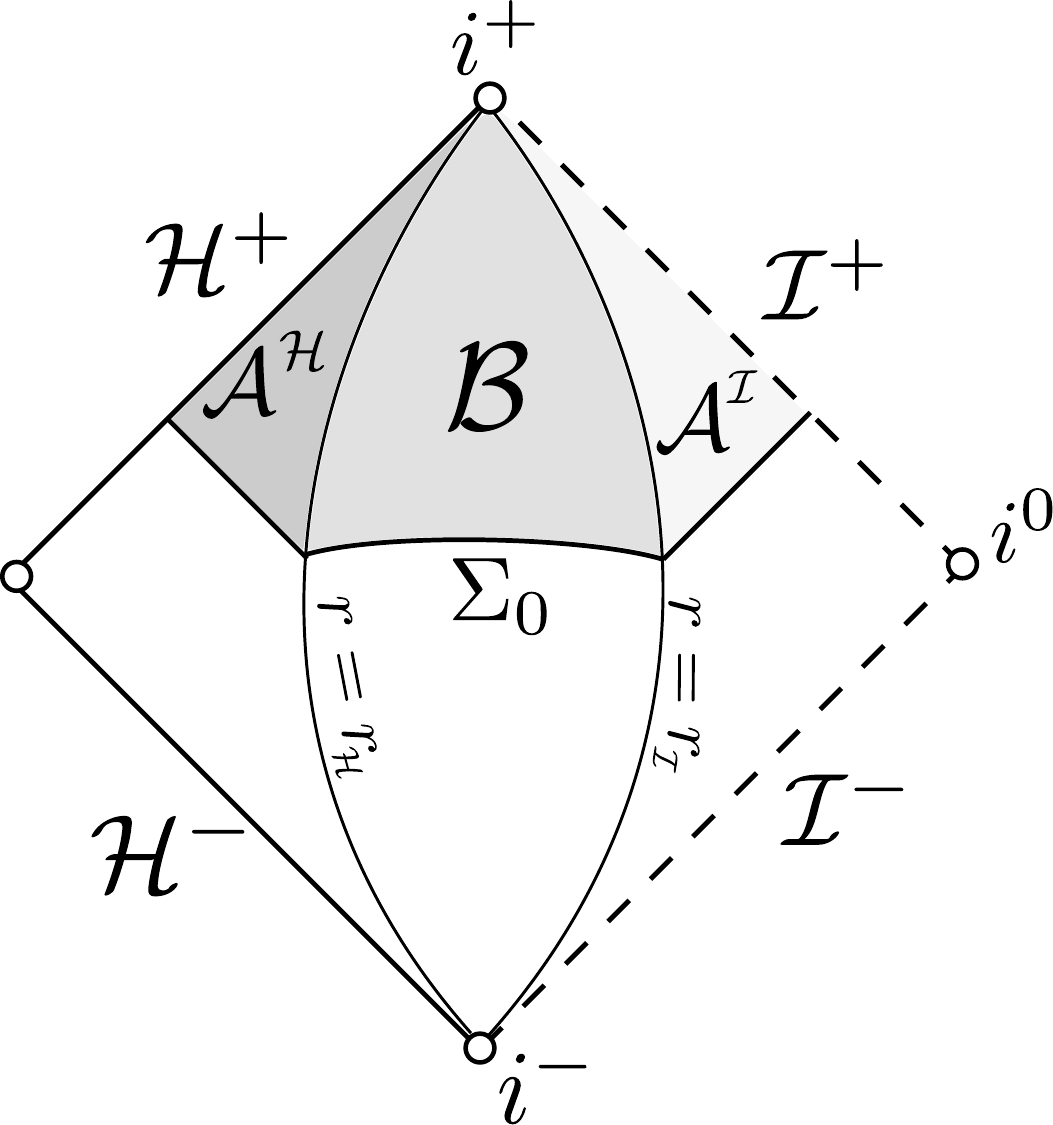}
\caption{The regions $\mathcal{A}^{\mathcal{H}}, \mathcal{B}$ and $\mathcal{A}^{\mathcal{I}}$.}
\label{fig:sigmaabc0}
\end{center}
\end{figure}

\vspace{-0.7cm}

\subsection{Cauchy data of Type \textbf{A}, \textbf{B}, \textbf{C} and  \textbf{D}.}
\label{sec:TheTypesOfInitialDataABCD}

Recall that ERN admits two independent conserved charges: 1) the horizon charge\footnote{The importance of the horizon charge $H_0[\psi]$ for the dynamics of ERN has been discussed in Section \ref{sec:TheHorizonInstabilityOfExtremalBlackHoles}.} $H_0[\psi]$ given by \eqref{introhorizonH}, and 2) the Newman--Penrose constant $I_0[\psi]$ at null infinity given by \eqref{np}. It is important to emphasize that the values of $H_0[\psi]$ and $I_0[\psi]$ depend \textbf{only} on the initial data of $\psi$ at the event horizon $\h$ and at null infinity $\I$, respectively.  Hence, compactly supported initial data necessarily satisfy $I_0[\psi]=0$ whereas data for which $I_0[\psi]\neq 0$ are necessarily not compactly supported. Similarly, data supported away from the horizon necessarily satisfy $H_0[\psi]=0$. Recall Definition \ref{def1intro} according to which data which satisfy $H_0[\psi]\neq 0$ are called horizon-penetrating. We also introduce the following:
\begin{definition}
Initial data on a Cauchy hypersurface $\Sigma_0$ are called \textbf{null-infinity-extending} if the Newman--Penrose constant $I_0[\psi]\neq 0$. 
\label{def1intro2}
\end{definition}

We distinguish the following four types of initial data:

\vspace{0.15cm}

 \textbf{Type A:} \textit{Compactly supported data but horizon-penetrating}.

\vspace{0.15cm}

These data should be thought of as local data in the sense that they reflect perturbations in a neighborhood of the event horizon. 

\vspace{0.15cm}

	\textbf{Type B:} \textit{Compactly supported data that is supported away from the event horizon.}
	
	\vspace{0.15cm}
	
	These data correspond to compact perturbations from afar, that is away from the event horizon. 

\vspace{0.15cm}

\textbf{Type C:} \textit{Null-infinity-extending and horizon-penetrating data.}

\vspace{0.15cm}

These data correspond to global perturbations with non-trivial support across the whole initial hypersurface $\Sigma_0$. In the physics literature, such data are said to have an ``\textit{initial static moment}''.

\vspace{0.15cm}

\textbf{Type D:} \textit{Null-infinity-extending but supported away from the horizon data.} 

\vspace{0.15cm}

These data correspond to perturbations from afar extending all the way to null infinity.

In summary we have the following table:

\begin{table}[H]{\footnotesize
\begin{center}
  \begin{tabular}{ cV{4}c | c }                                \hline
   Data  & $H_0$           & $I_0$                           \\ 
	\Xcline{1-3}{0.05cm} 
  Type \textbf{A} & $\neq 0$        & $=0$      \\ \hline
  Type \textbf{B} &  $=0$           & $=0$             \\ \hline
  Type \textbf{C} &  $\neq 0$       & $\neq 0$             \\ \hline
  Type \textbf{D} &  $=0$           & $\neq 0$             \\ \hline
  \end{tabular}
  \end{center}
 \caption{Types of initial data and the associated conserved charges $H_0, I_0$.}}
\end{table}

\vspace{-0.3cm}
 As we shall see, each of these types requires a separate treatment and exhibits different asymptotic behavior.

\section{A first version of the main results}
\label{sec:SummaryOfTheMainResults}

In this section we will present the main theorems of the present paper. We will first present a non-technical version of the results and then in Section \ref{sec:Remarks} various applications of our results. Finally we will present the rigorous statements of the main theorems in Section \ref{subsec:TheMainTheorems}.  

We first introduce the notion of a new horizon charge which plays a fundamental role in our study of the dynamics of ERN.

\subsection{The new horizon charge $H_{0}^{(1)}[\psi]$ }
\label{sec:TheNewHorizonHairH01Psi}

We introduce the \textit{dual} scalar field $\widetilde{\psi}$ of $\psi$ as follows
\begin{equation}
\widetilde{\psi}(t,r^*,\theta,\phi) = \frac{M}{r-M}\psi(t,-r^*,\theta,\phi).
\label{dual}
\end{equation}
First observe that duality is self-inverse: $\widetilde{\widetilde{\psi}}=\psi$. Furthermore, $\psi$ satisfies the wave equation \eqref{eq:waveequation} if and only if its dual $\widetilde{\psi}$ satisfies \eqref{eq:waveequation}. This duality is motivated by the Couch--Torrence conformal symmetry \cite{couch} of ERN. References \cite{bizon2012,hj2012} showed that this duality can be used to relate the horizon charge with the Newman--Penrose constant as follows:
\[ H_0[\psi] =I_0[\widetilde{\psi}]. \] 
If the Newman--Penrose constant vanishes $I_0[\psi]=0$ then the following expression 
\begin{equation}
I_{0}^{(1)}[\psi]=\frac{M}{4\pi} \int_{\Sigma_0\cap\mathcal{H}^{+}}\!\!\psi r^2d\omega+\lim_{r_0\rightarrow \infty}\left(\frac{M}{4\pi}\int_{\Sigma_0\cap\{r\leq r_0\}} n_{\Sigma_0}(\psi)d\mu_{\Sigma_0}+\frac{M}{4\pi}\int_{\Sigma_0\cap\{r=r_0\}}\Big(\psi-\frac{2}{M}r\partial_v(r\psi)\Big)r^2d\omega\right),
\label{timeinvertednp}
\end{equation}
is finite and conserved\footnote{that is, it is independent of the choice of the hypersurface $\Sigma_0$.}. See the discussion in Section \ref{sec:TheWaveEquationOnBlackHolesBackgrounds} and for more details in \cite{paper-bifurcate}. We refer to  $I_{0}^{(1)}[\psi]$ as the time-inverted Newman--Penrose constant. Note that $I_{0}^{(1)}[\psi]$ is only defined for initial data of type \textbf{A} and \textbf{B}.

In the case where $H_0[\psi]=0$ we introduce the following quantity
\begin{equation}
H_{0}^{(1)}[\psi]= I_{0}^{(1)}[\widetilde{\psi}]. 
\label{h01}
\end{equation} 
 We will refer to $H_{0}^{(1)}[\psi]$ as the \textit{time-inverted horizon charge}. 
Clearly, $H_{0}^{(1)}[\psi]$ is only defined for initial data of Type \textbf{B} and \textbf{D}.

For a discussion on the geometric importance of the constants $H_{0}^{(1)}[\psi]$ and $I_{0}^{(1)}[\psi]$ and their role in the analysis of the present paper see Section \ref{sec:GeometricOriginOfTheNewHair}.  

\subsection{The late-time asymptotics}
\label{sec:TheMainTheoremsintrosummary}

We rigorously derive the late-time asymptotics solutions to the wave equation \eqref{eq:waveequation} on ERN. We next summarize our results. 

\subsubsection{Asymptotics for Type \textbf{C} perturbations}
\label{sec:GlobalPerturbationTypeC}

We first consider {global} perturbations of Type \textbf{C}. These perturbations satisfy $H_0\neq 0$ and $I_0\neq 0$. 
Recall from Section \ref{sec:PreviousWorksOnLateTimeAsymptoticsOnERN} that the heuristic and numerical work \cite{Burko2007}  argued that the decay \textit{rate} of $r\psi$ is $\tau^{-1}$, $\tau^{-2}$ or $\tau^{-1}$ along $\h$, $\{r=r_0\}$ and $\I$, respectively. However, precise late-time asymptotics for this type of perturbations were not known. The non-vanishing of the conserved constants $H_0$ and $I_0$ might seem to suggest that they appear in a potentially complicated way in the asymptotics for $\psi$.  In fact, \cite{hm2012} conjectured that both $H_0$ and $I_0$ appear in the asymptotics of $\psi$ along the event horizon $\h$. In this paper, we derive and rigorously prove the precise late-time asymptotics for scalar perturbations of Type \textbf{C}. We falsify the above conjecture by showing that \textit{the asymptotics along the event horizon are independent of the Newman--Penrose constant $I_0$:}
\begin{equation}
 r\psi|_{\h}\sim 2H_0[\psi]\cdot\frac{1}{\tau}+4MH_0[\psi]\cdot\frac{\log\tau}{\tau^2} \ \ \text{ as } \ \ \tau\rightarrow \infty. 
\label{horasympt}
\end{equation}
On the other hand, we show that \textit{\underline{both} constants $H_0$ and $I_0$ appear in the leading-order terms for the late-time asymptotics of $\psi|_{\{r=r_0\}}$ along $r=r_0$ hypersurfaces away from the event horizon ($r_0>M$)}:
\begin{equation}
\psi|_{\{r=r_0\}} \sim \left(4I_0[\psi]+\frac{4M}{r-M}H_0[\psi]\right)\cdot \frac{1}{\tau^2} \ \ \text{ as } \ \ \tau\rightarrow \infty. 
\label{awayasympt}
\end{equation} 
The proof of \eqref{awayasympt} is particularly subtle  since both the horizon region and the null infinity region  contribute to the asympotics of $\psi|_{\{r=r_0\}}$ via the constants $H_0$ and $I_0$, respectively. This is in stark contrast with the sub-extremal case (see Section \ref{sec:TheWaveEquationOnBlackHolesBackgrounds}) where the dominant terms originate only from the null infinity region. Note  that \textit{the term $\frac{4M}{r-M}$ in front of $H_0$ is itself a static solution of \eqref{eq:waveequation} on ERN}.  We remark that in order to show the asymptotics \eqref{awayasympt}, we need to derive first the asymptotics for the radial  derivative\footnote{with respect to the coordinate system $(\rho=r,\theta,\varphi)$ on $\Sigma_0$} $\partial_{\rho}\psi$ of $\psi$ along $\Sigma_{\tau}$:
\begin{equation}
\partial_{\rho}\psi|_{\{r=r_0\}}\sim -\frac{4M}{(r-M)^2}H_0[\psi]\cdot \frac{1}{\tau^2} \ \ \text{ as } \ \ \tau\rightarrow \infty. 
\label{derivaawayasympt}
\end{equation}
The crucial insight of \eqref{derivaawayasympt} is that\textit{ the leading-order asymptotics of $\partial_{\rho}\psi|_{\{r=r_0\}}$ are independent of $I_0$ for all values of $r_0>M$!} This is somewhat surprising; it shows that, from the point of view of the derivative $\partial_{\rho}\psi$, the event horizon is, in a sense, more relevant than null infinity. Furthermore, note that the decay rate of $\partial_{\rho}\psi|_{\{r=r_0\}}$ is only $\tau^{-2}$ which is equal to the decay rate of $\psi|_{\{r=r_0\}}$. This is again in stark contrast with the sub-extremal case where $\partial_{\rho}\psi|_{\{r=r_0\}}$ decays like $\tau^{-3}$. 

We obtain the following asymptotics along null infinity $\I$:
\begin{equation}
r\psi|_{\I}\sim 2I_0[\psi]\cdot\frac{1}{\tau}+4MI_0[\psi]\cdot\frac{\log\tau}{\tau^2}  \ \ \text{ as } \ \ \tau\rightarrow \infty. 
\label{nullinfasympt}
\end{equation}
Note that these asymptotics are independent of the horizon charge $H_0$.  

\subsubsection{Asymptotics for Type \textbf{A}  perturbations}
\label{sec:AsymptoticsForTypeAPerturbations}

We next consider \textbf{local} horizon-penetrating perturbations of Type \textbf{A}. These perturbations, which satisfy $H_0\neq 0$ and $I_0=0$, are the most physically relevant since they represent local perturbations of ERN. In the physics literature, they are said to describe \textit{outgoing radiation}. 

The asymptotics \eqref{horasympt} along $\h$, and \eqref{awayasympt} and \eqref{derivaawayasympt} along $\{r=r_0\}$ hold in this case as well, where in \eqref{awayasympt} we have to use that $I_0=0$. On the other hand, the asymptotics along null infinity for the radiation field $r\psi_{\I}$ cannot be read off from \eqref{nullinfasympt}. In this case, we derive the following asymptotics
\begin{equation}
r\psi|_{\mathcal{I}^+}\sim \left( 4MH_0[\psi]-2I_{0}^{(1)}[\psi] \right)\cdot \frac{1}{\tau^2} \ \ \text{ as } \ \ \tau\rightarrow \infty. 
\label{typeanullinf}
\end{equation}
Here $I_{0}^{(1)}$ is the time-inverted Newman--Penrose constant given by \eqref{timeinvertednp}. We observe that \textit{for Type \textbf{A} perturbations the dominant term in the asymptotics of the radiation field $r\psi|_{\I}$ contains the horizon charge $H_{0}$}. Such perturbations exhibit the \textit{full} horizon instability
\begin{equation}
\partial_{r}\psi|_{\h}\sim -\frac{1}{M^3}H_{0}[\psi], \ \ \ \partial_{r}^2\psi|_{\h}\sim \frac{1}{M^5} H_{0}[\psi]\cdot\tau \  : \ \text{ along }\h,
\label{typeainsta}
\end{equation}
with respect to $(v,r)$ coordinates, the origin of which is the charge $H_0$ (see Section \ref{sec:TheHorizonInstabilityOfExtremalBlackHoles} for a review). Therefore, the precise asymptotics \eqref{typeanullinf} yield a way to potentially measure the horizon charge $H_0$ and hence detect the horizon instability of extremal black holes from observations in the far away radiation region.

\subsubsection{Asymptotics for Type \textbf{B} perturbations}
\label{sec:AsymptoticsForTypeBPerturbations}

Perturbations of Type \textbf{B} constitute another very important class of physically relevant perturbations. Such perturbations are initially compactly supported and supported away from the horizon and hence satisfy $H_0=0$ and $I_0=0$. They represent \textbf{local} perturbations \textbf{from afar}. In the physics literature, such perturbations are said to describe \textit{ingoing radiation}. 

Recall from Section \ref{sec:TheWorkOfReallEtAl} that Lucietti--Murata--Reall--Tanahashi \cite{hm2012} numerically demonstrated that such perturbations exhibit a \textit{weaker} version of the horizon instability, namely  
\begin{equation}
|\psi|_{\h}\rightarrow 0,\ \ \ |\partial_{r}\psi|_{\h}\rightarrow 0, \ \ \ |\partial_{r}^2\psi|_{\h}\nrightarrow 0 \ \ \  |\partial_{r}^3\psi|_{\h}\rightarrow \infty\ : \ \text{ along }\h.
\label{typebinsta}
\end{equation}
Perturbations of Type \textbf{B} which exhibit the above behavior where rigorously constructed in \cite{aretakis2012}. However, \cite{aretakis2012} did not provide a necessary and sufficient condition for perturbations of Type \textbf{B} so that \eqref{typebinsta} holds. In this paper, we provide an answer to this question. In fact, we provide the precise late-time asymptotics for all perturbations of Type \textbf{B}. 

Recall that the horizon charge $H_{0}^{(1)}[\psi]$ given by \eqref{h01} which is well-defined for all Type \textbf{B} perturbations. We prove that\textit{ the weak horizon instability \eqref{typebinsta} holds if and only if $H_{0}^{(1)}[\psi]\neq 0$.} Specifically,
\begin{equation}
\partial_{r}^2\psi \sim \frac{1}{M^5}H_{0}^{(1)}[\psi], \ \ \ \partial_{r}^3\psi \sim -\frac{3}{M^7}H_{0}^{(1)}[\psi]\cdot\tau \ \ \text{ as } \ \ \tau\rightarrow \infty. 
\label{tbhorinstabi}
\end{equation}
 Furthermore, we prove that $H_{0}^{(1)}[\psi]$ determines the leading-order asymptotics along the event horizon
\begin{equation}
r\psi|_{\h}\sim -2H_{0}^{(1)}[\psi]\cdot \frac{1}{\tau^2}-8MH_0^{(1)}[\psi]\cdot\frac{\log \tau}{\tau^3} \ \ \text{ as } \ \ \tau\rightarrow \infty. 
\label{tbhorasympt}
\end{equation}
On the other hand, the asymptotics of the radiation field depend \textbf{only} on the value of the time-inverted Newman--Penrose constant $I_{0}^{(1)}$:
\begin{equation}
r\psi|_{\I}\sim -2I_{0}^{(1)}[\psi]\cdot \frac{1}{\tau^2}-8MI_0^{(1)}[\psi]\cdot\frac{\log \tau}{\tau^3} \ \ \text{ as } \ \ \tau\rightarrow \infty. 
\label{tbnullasympt}
\end{equation}
Finally, the asymptotics along $\{r=r_0\}$ depend on the value of both constants $H_{0}^{(1)}[\psi]$ and  $I_0^{(1)}[\psi]$:
\begin{equation}
\psi|_{\{r=r_0\}}\sim -8\left(I_{0}^{(1)}[\psi]+\frac{M}{r-M}H_{0}^{(1)}[\psi] \right)\cdot \frac{1}{\tau^3} \ \ \text{ as } \ \ \tau\rightarrow \infty.
\label{tbnullasympt2}
\end{equation}
Note the decay rate of \eqref{tbnullasympt2} agrees with the decay rate of \eqref{our1} for Schwarzschild spacetimes. However, in constrast to Schwarzschild, the coefficient of the asymptotic term in \eqref{tbnullasympt2} depends additionally on the new horizon charge  $H_{0}^{(1)}[\psi]$.

\subsubsection{Asymptotics for Type \textbf{D} perturbations}
\label{sec:AsymptoticsForTypeDPerturbationsintro}

We conclude our discussion with a brief discussion about the asymptotics of Type \textbf{D} perturbations. These are \textbf{non-compact} pertrubations supported \textbf{away} from the event horizon. They satisfy $H_0=0$ and $I_0\neq 0$. Such perturbations have a well-defined $H_{0}^{(1)}$. Our first result for this case shows that the radiation field $r\psi|_{\I}$ and the scalar field $\psi|_{\{r=r_0\}}$ ``see'' to leading order \textbf{only} $I_0$:
\[r\psi|_{\I}\sim 2I_{0}[\psi]\cdot \frac{1}{\tau}, \ \ \ \ \psi|_{\{r=r_0\}}\sim 4I_{0}[\psi]\cdot \frac{1}{\tau^2} \ \ \ \text{ as } \ \ \tau\rightarrow \infty.
\]
On the other hand, the asymptotics along $\h$ to leading order to depend on \textbf{both} $I_0$ and the horizon charge $H_0^{(1)}$:
\[ r\psi|_{\h}\sim \left(4MI_{0} -2H_{0}^{(1)}[\psi]\right)\cdot \frac{1}{\tau^2}\  \ \text{ as } \ \ \tau\rightarrow \infty.\]
Type \textbf{D} perturbations exhibit the weak version of the horizon instability. In this context, it is important to remark that \textit{even though the asymptotic terms along the event horizon contain both $I_0$ and $H_{0}^{(1)}$, the source of the horizon instability in this case originates purely from the horizon charge $H_{0}^{(1)}$.}  In fact, the exact non-decay and blow-up along the event horizon results given by \eqref{tbhorinstabi} hold for Type \textbf{D} perturbations as well. In contrast to Type \textbf{A} perturbations, the derivative $\partial_r\psi$ decays faster than $\psi$ away from the horizon:
\[\partial_r\psi|_{\{r=r_0\}}\sim\left( \frac{8M}{(r-M)^2}\cdot H_0^{(1)}[\psi]+\frac{8(r^2-M^2)}{(r-M)^2}\cdot I_0[\psi]\right)\cdot \tau^{-3}.\]
The  precise asymptotic expression, as above, is derived here for the first time in the literature.

\subsubsection{Summary of the asymptotics}
\label{sec:SummaryOfTheAsymptotics}

We summarize our findings in the table below
\begin{table}[H]{\footnotesize
\begin{center}
    \begin{tabular}{ cV{4}c | c |c} 
		 \Xcline{2-4}{0.01cm} 
		 & \multicolumn{3}{c}{\textbf{Asymptotics}} \\
   \Xcline{1-4}{0.02cm} 
{\textbf{Data}} & $r\psi|_{\mathcal{H}^{+}}$ & $\psi|_{\{r=r_0\}}$ & $r\psi|_{\mathcal{I}^{+}}$  \\ \Xcline{1-4}{0.05cm} 
 {Type \textbf{A}} &   $2H_0\cdot \tau^{-1}$  & $ \frac{4M}{r-M}H_0 \cdot \tau^{-2}$ & $\left(4MH_0-2I_0^{(1)}\right) \cdot \tau^{-2}$\\  \hline 
{Type \textbf{B}} & $-2H_0^{(1)} \cdot \tau^{-2}$  & $-8\left(I_0^{(1)}+\frac{M}{r-M}H_0^{(1)}\right) \cdot \tau^{-3}$ &  $-2I_0^{(1)} \cdot \tau^{-2}$\\    \hline
{Type \textbf{C}}	& $2H_0 \cdot \tau^{-1}$ &  $4\left(I_0+\frac{M}{r-M}H_0\right) \cdot \tau^{-2}$  & $2I_0 \cdot \tau^{-1}$  \\ \hline
{Type \textbf{D}}	& $\left(4MI_0-2H_0^{(1)}\right) \cdot \tau^{-2}$ &  $4I_0\cdot \tau^{-2}$ &  $2I_0\cdot \tau^{-1}$\\ \hline
  \end{tabular}
\end{center}
\caption{The asymptotics for $\psi$ along $\h, \{r=r_0\}$ and $\I$ for data of Type \textbf{A}, \textbf{B}, \textbf{C}, and \textbf{D}.}
\label{summarytable}}
\end{table} 
We remark that we in fact derive the asymptotics for $T^{k}\psi$ for all $k\geq 0$. The relevant asymptotic expressions can be found by taking the $\frac{\partial^{k}}{\partial \tau^{k}}$ derivative of the expressions in Table \ref{summarytable}. Furthermore, we derive the following asymptotics for the transversal derivative $\partial_r\psi$: 
\begin{table}[H]{\footnotesize
\begin{center}
    \begin{tabular}{ cV{3}c | c } 
		\cline{2-3}
		 & \multicolumn{2}{c}{\textbf{Asymptotics}} \\
   \Xcline{1-3}{0.02cm} 
{\textbf{Data}} &  $\partial_r\psi|_{\mathcal{H}^{+}}$ & $\partial_r\psi|_{\{r=r_0\}}$ \\ 
\Xcline{1-3}{0.05cm} 
 {Type \textbf{A}} &   $-\frac{1}{M^3}\cdot H_0$ & $-\frac{4M}{(r-M)^2}\cdot H_0[\psi]\cdot \tau^{-2}$
\\ \hline
 {Type \textbf{B}} & $\frac{2}{M^2}\cdot H_0^{(1)}\cdot {\tau^{-2}} $  & $\frac{8M}{(r-M)^2}\cdot H_0^{(1)}[\psi]\cdot \tau^{-3}$  \\ \hline
 {Type \textbf{C}}	& $-\frac{1}{M^3}\cdot H_0$ &  $-\frac{4M}{(r-M)^2}\cdot H_0[\psi]\cdot \tau^{-2}$   \\ \hline
 {Type \textbf{D}}	& $\frac{2}{M^2}\cdot H_0^{(1)}\cdot {\tau^{-2}} $  &  $\left( \frac{8M}{(r-M)^2}\cdot H_0^{(1)}+\frac{8(r^2-M^2)}{(r-M)^2}\cdot I_0\right)\cdot \tau^{-3}$   \\ \hline
  \end{tabular}
\end{center} 
\caption{The asymptotics for $\partial_r\psi$  on and away from the event horizon}
\label{summarytablerho}}
\end{table}

At the horizon, we have the following asymptotics for the higher order transversal derivatives $\partial_r^k\psi$ revealing the strong horizon instability for Type \textbf{A} and \textbf{C} and the weak horizon instability for Type \textbf{B} and \textbf{D}. 
\begin{table}[H]{\footnotesize
\begin{center}
    \begin{tabular}{ cV{3}c | c | c| c } 
		\cline{2-5}
		 & \multicolumn{4}{c}{\textbf{Asymptotics}} \\
   \Xcline{1-5}{0.02cm} 
{\textbf{Data}} &  $\partial_r\psi|_{\mathcal{H}^{+}}$ & $\partial_r^2\psi|_{\mathcal{H}^{+}}$ & $\partial_r^3\psi|_{\mathcal{H}^{+}}$  & $\ \ \partial_r^k\psi|_{\mathcal{H}^{+}},\, k\geq 2$\\ 
\Xcline{1-5}{0.05cm} 
 {Type \textbf{A}} &  $-\frac{1}{M^3}\cdot H_0$  & $\frac{1}{M^5}\cdot H_0 \cdot\tau$ & $-\frac{3}{2M^7} \cdot H_0 \cdot\tau^2$ & $c_k\cdot H_0\cdot  \tau^{k-1}$
\\ \hline
 {Type \textbf{B}} & $\frac{2}{M^2}\cdot H_0^{(1)}\cdot {\tau^{-2}} $ & $\frac{1}{M^5}\cdot H_0^{(1)}$ & $-\frac{3}{M^7}\cdot H_0^{(1)}\cdot \tau$ &  $a_k\cdot c_{k-1}\cdot H_0^{(1)}\cdot \tau^{k-2}$ \\ \hline
 {Type \textbf{C}}	& $-\frac{1}{M^3}\cdot H_0$  & $\frac{1}{M^5}\cdot H_0 \cdot\tau$ & $-\frac{3}{2M^7} \cdot H_0 \cdot\tau^2$ & $c_k\cdot H_0\cdot  \tau^{k-1}$ \\ \hline
{Type \textbf{D}}	 & $\frac{2}{M^2}\cdot H_0^{(1)}\cdot {\tau^{-2}} $ & $\frac{1}{M^5}\cdot H_0^{(1)}$ & $-\frac{3}{M^7}\cdot H_0^{(1)}\cdot \tau$ &  $a_k\cdot c_{k-1}\cdot H_0^{(1)}\cdot \tau^{k-2}$ \\ \hline
  \end{tabular}
\end{center}
\caption{The horizon instability for transversal derivatives along $\h$.}
\label{horizontable}}
\end{table}

where 
\[a_k= \frac{(-1)}{M^2}\cdot\binom{k}{2}\ \ \ \text{ and } \ \ \ c_k=(-1)^{k}\cdot\frac{1}{M^3}\cdot \frac{1}{\left(2M^2 \right)^{k-1}}\cdot k!. \]

More generally, the late-time asymptotics along $\h$  for $Y^{k}T^{m}\psi$, with $k\geq m+1$ for Type \textbf{A} and \textbf{C} and $k\geq m+2$ for Type \textbf{B} and \textbf{D}, can be informally found by taking the  $\frac{\partial^{m}}{\partial \tau^{m}}$ derivative of the expressions in Table \ref{horizontable} (see also Section \ref{sec:DecayForScalarInvariants}).

\section{Applications and additional remarks}
\label{sec:Remarks}

 In this section we present a few applications and remarks about our results.

\subsection{Singular time inversion and the new horizon charge}
\label{sec:GeometricOriginOfTheNewHair}

Recall from Section \ref{sec:TheWaveEquationOnBlackHolesBackgrounds} that the constants $I_0$ and $I_{0}^{(1)}$ are obstructions to inverting the time operators $T$ and $T^{2}$, respectively. Specifically, $I_0$ and $I_{0}^{(1)}$ are obstructions to defining the operators $T^{-1}$ and  $T^{-2}$, respectively, on solutions the wave equation \eqref{eq:waveequation}, such that  their target functional space consists of  solutions the wave equation which decay appropriately  in $r$ towards null or spacelike infinity. In sub-extremal black holes, $I_0$ and $I_{0}^{(1)}$ are the only such obstructions. However, for ERN we have an additional obstruction that originates from the geometry of the horizon, namely the conserved charge $H_0$. Indeed, for any smooth solution $\psi$ to the wave equation \eqref{eq:waveequation} on ERN we have
\begin{equation}
H_{0}[T\psi]=0.\vspace{-0.12cm}
\label{h0t}
\end{equation}
Hence, \textit{$H_0$ is an obstruction to defining the inverse operator $T^{-1}$ on smooth solutions to \eqref{eq:waveequation} such that the image is also a smooth solution to \eqref{eq:waveequation}}. On the other hand, if $\psi$ is a smooth solution to \eqref{eq:waveequation} with $H_0=0$ then the horizon charge $H_{0}^{(1)}$ is well-defined and satisfies
\[H_{0}^{(1)}[T^2\psi]=0. \vspace{-0.12cm}\]
Hence, \textit{$H_0^{(1)}$ is an obstruction to defining the inverse operator $T^{-2}$ on smooth solutions (with $H_0=0$) to \eqref{eq:waveequation} such that the image is also a smooth solution to \eqref{eq:waveequation}}. The above imply that the horizon associated charges $H_0$ and $H_0^{(1)}$ are related to singularities at time frequencies $\omega \sim 0$. We thus conclude that \textit{the leading order terms in the late-time asymptotic expansion are dominated by the $\omega\sim 0$ frequencies}. 

An important aspect of our analysis is that we invert the operators $T$ and $T^2$ even if the images of $T^{-1}$ and $T^{-2}$ do \underline{not} contain smooth function. This is accomplished by developing a \textbf{singular time inversion theory}. This theory is needed for Type \textbf{A} and Type \textbf{D} perturbations. Let's first consider Type \textbf{A} perturbations. Since such perturbations satisfy $H_0\neq 0$ and $I_0=0$, $I_0^{(1)}$ is well-defined whereas $H_{0}^{(1)}$ is undefined. Clearly, there is no smooth solution $T^{-1}\psi$ to \eqref{eq:waveequation}. Indeed, if a smooth solution $T^{-1}\psi$ to \eqref{eq:waveequation}  existed then by replacing $\psi$ with $T^{-1}\psi$ in \eqref{h0t} we would obtain $H_0[\psi]=H_0[T(T^{-1}\psi)]=0$, which is a contradiction. It turns out that we can still \textit{canonically} define a \textit{singular} time inversion $T^{-1}\psi$ such that 
\begin{itemize}
	\item $T^{-1}\psi\rightarrow 0$ as $r\rightarrow \infty$,
	\item $I_0[T^{-1}\psi]<\infty$, 
	\item $\partial_r(T^{-1}\psi)\sim -2H_{0}[\psi]\cdot \frac{1}{r-M}$ in the region $r\sim M$. 
\end{itemize}
Similar results hold for Type \textbf{D} perturbations. For perturbations of Type \textbf{A} and \textbf{D}, \textit{we develop a low regularity theory which allows us to obtain the precise late-time asymptotics for the singular scalar fields} $T^{-1}\psi$.  We remark that for Type \textbf{B} perturbations we develop a \textbf{regular} time inversion theory, whereas no time inversion is needed for Type \textbf{C} perturbations.  Summarizing,
\begin{table}[H]{\footnotesize
\begin{center}
    \begin{tabular}{ cV{3}c | c| c|cV{3}c  } 
		\cline{2-6}
		 & \multicolumn{5}{c}{\textbf{Time inversion theory}} \\
   \Xcline{1-6}{0.02cm} 
{\textbf{Data}} &  $H_0[\psi]$ & $H_0^{(1)}[\psi]$ & $I_0[\psi]$ & $I_{0}^{(1)}[\psi]$ & {$T^{-1}\psi$} \\ 
\Xcline{1-6}{0.05cm} 
 {Type \textbf{A}} &  $\neq 0$  & $\boldsymbol{=\infty}$ & $=0$ & $<\infty$ & {singular at $\h$} 
\\ \hline
 {Type \textbf{D}} & $=0$  & $<\infty$  & $\neq 0$ & $\boldsymbol{=\infty}$ & {singular at $\I$}   \\ \hline
{Type \textbf{B}}	& $=0$ &  $<\infty$ & $=0$ & $<\infty$ & {regular}  \\ \hline
  \end{tabular}
\end{center}
\caption{The time inversion and its singular support.}
\label{timeinversiontable}}
\end{table}

\vspace{-0.7cm}

\subsection{Decay for scalar invariants}
\label{sec:DecayForScalarInvariants}

Hadar and Reall \cite{harveyeffective}, assuming the asymptotics on ERN (rigorously established in the present paper), showed that the scalar invariants $|\nabla^k\psi|^2$ decay in time. Similar decay results were presented in \cite{khanna17} and in \cite{zimmerman4}.
Let's briefly recall the argument of \cite{harveyeffective}. First of all, note that the Christoffel symbols $\Gamma^{a}_{bc}$, with $a,b,c\in \{ v,r\}$, vanish on the event horizon and, hence, if $\partial_{i^1},\cdots, \partial_{i^k}\in \{\partial_{v},\partial_{r}\}$ then  $\nabla^{k}\psi_{i_1\cdots i_k}=\partial_{i^1}\cdots \partial_{i^k}\psi$ on the event horizon. The following asymptotic decay rates hold along the event horizon for all derivatives:
\[\partial_r^k T^m\psi \sim \tau^{k-m-1-\epsilon(k,m)} \]
where 
\begin{equation}
\epsilon(k,m)=\begin{cases}
  0, \ \text{ if }\  k=0 \ \text{ or } \ k\geq m+1, \\
  1,  \ \text{ if }\  1\leq k \leq m. 
\end{cases}
\label{skipepsilon}
\end{equation} 
Note that the presence of $\epsilon(k,m)$ introduces a \textit{skip} in the decay rates for the derivatives of $\psi$. This skip was also previously observed in \cite{zimmerman1}.

To show that $|\nabla^k\psi|^2$ always decays, it suffices to consider the ``slowest'' case, namely the case of perturbations of Type \textbf{C}. In this case,
\begin{equation*}
\begin{split}
|\nabla^k\psi|^2 &\sim \sum_{k_1+k_2=k}\partial_{r}^{k_1}T^{k_2}\psi\cdot \partial_{r}^{k_2}T^{k_1}\psi\sim \sum_{k_1+k_2=k}\tau^{k_1-k_2-1-\epsilon(k_1,k_2)}\cdot \tau^{k_2-k_1-1-\epsilon(k_2,k_1)}\\
&\sim  \sum_{k_1+k_2=k}\tau^{-2-\epsilon(k_1,k_2)-\epsilon(k_2,k_1)}\sim \tau^{-2},
\end{split}
\end{equation*} 
for all $k\geq 1$, since $\epsilon(k_1,k_2), \epsilon(k_2,k_1)\geq 0$ and $\epsilon(k,0)=\epsilon(0,k)=0$. Note that the decay rate for $|\nabla^k\psi|^2$ is independent of $k$.

\subsection{The interior of black holes and strong cosmic censorship}
\label{sec:TheInteriorOfBlackHolesAndStrongCosmicCensorship}

In this paper we have restricted the analysis of the wave equation to the extremal Reissner--Nordstr\"om \emph{black hole exterior} (the domain of outer communications). One can also extend the initial data hypersurface $\Sigma_0$ into the black hole \emph{interior} (see Section \ref{sec:TheHyperboloidalFoliation} for a precise definition of $\Sigma_0$) and investigate the behavior of solutions to \eqref{eq:waveequation} in the restriction of the domain of dependence of the extended initial data hypersurface to the {black hole interior}. 
\vspace{-0.15cm}
  \begin{figure}[H]
	\begin{center}
				\includegraphics[scale=0.35]{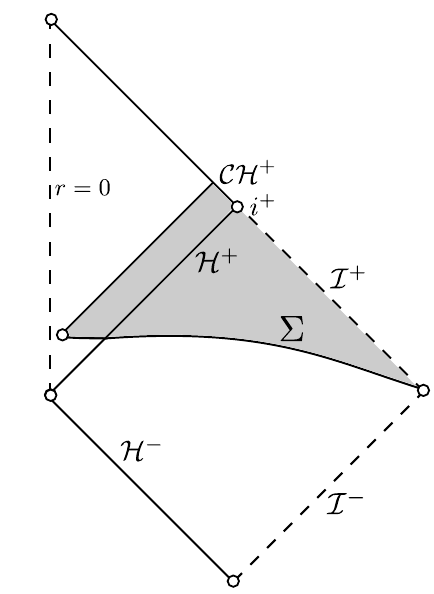}
\end{center}
\vspace{-0.4cm}
\caption{The extended initial value problem that includes the interior region, where $\Sigma$ is the extension of $\Sigma_0$ into the interior.}
	\label{fig:interiop455}
\end{figure}
\vspace{-0.3cm}
An analysis of the behavior of solutions to \eqref{eq:waveequation} in the black hole interior of extremal Reissner--Nordstr\"om was carried out by the third author in \cite{gajic} in the setting of a characteristic initial value problem with initial data imposed on a future geodesically complete segment of the future black hole event horizon and initial data imposed on an ingoing null hypersurface intersecting the event horizon to the past. The late-time behaviour of the solution to \eqref{eq:waveequation} on extremal Reissner--Nordstr\"om along the event horizon was \emph{assumed} to be consistent with the numerical predictions of \cite{hm2012}. The results of \cite{gajic} illustrate a remarkably \emph{delicate} dependence of the qualitative behaviour at the inner horizon in the black hole interior on the \emph{precise} late-time behaviour of the solution to \eqref{eq:waveequation} along the event horizon of extremal Reissner--Nordstr\"om as predicted by numerics and heuristics.

By combining the results stated schematically in Section \ref{sec:SummaryOfTheMainResults} and more precisely in Section \ref{sec:StatementsOfTheTheorems}, that confirm in particular the numerical predictions of \cite{hm2012}, with Theorem 2, 5 and 6 of \cite{gajic}, we conclude that the following theorem holds:
\begin{theorem}
\label{thm:interior}
Solutions $\psi$ to \eqref{eq:waveequation} on extremal Reissner--Nordstr\"om arising from smooth compactly supported data on an extension of $\Sigma_0$ into the black hole interior are extendible across the black hole inner horizon as functions in $C^{0,\alpha}\cap W^{1,2}_{\rm loc}$, with $\alpha<1$. Furthermore, the spherical mean $\frac{1}{4\pi}\int_{\s^2}\psi\,d\omega$ can in fact be extended as a $C^2$ function.
\end{theorem}

\begin{remark}
It follows from Theorem \ref{thm:interior} that for spherically symmetric data one can construct $C^2$ extensions of $\psi$ across the inner horizons that are moreover \emph{classical} solutions to \eqref{eq:waveequation} with respect to a smooth extension of the extremal Reissner--Nordstr\"om metric across the inner horizon. These extensions of $\psi$, much like the smooth extensions of the metric, are \textbf{highly non-unique}!
\end{remark}
\begin{remark}
In order to derive $C^2$ extendibility of the spherical mean across the inner horizon for initial data with $H_0[\psi]\neq 0$, we have to make use of the precise leading order \underline{and next-to-leading order} behavior of $\psi$ along the event horizon in \eqref{horasympt}. 
\end{remark}

See also \cite{gajic2} for extendibility results in the context of \eqref{eq:waveequation} in the interior of extremal Kerr--Newman spacetimes.

The extendibility properties in Theorem \ref{thm:interior} differ drastically from the extendibility properties of solutions to \eqref{eq:waveequation} in the interior of \underline{sub}-extremal Reissner--Nordstr\"om black holes, which are extendible in $C^0$ across the inner (Cauchy) horizon, but \emph{inextendible} in $W^{1,2}_{\rm loc}$, see \cite{Franzen2014, luk2015}. See also \cite{Hintz2015, LukSbierski2016, DafShl2016, gregjan} for extendibility results in sub-extremal Kerr. 

The study of the wave equation in black hole interiors serves as a linear ``toy model'' for the analysis of dynamical black hole interiors, which is closely related to the \emph{Strong Cosmic Censorship Conjecture} (SCC). As formulated in \cite{DC09}, this conjecture states that ``generic'' asymptotically flat initial data for the Einstein vacuum equations have maximal globally hyperbolic developments that are inextendible as a Lorentzian manifold with a continuous metric and locally square integrable Christoffel symbols. See for example \cite{dl-scc} for a more elaborate discussion on SCC. 

Building on the pioneering work of Dafermos \cite{MD03, MD05c, MD12}  and Dafermos--Rodnianski \cite{MDIR05}, Luk--Oh showed in \cite{Luk2016a, Luk2016b} that a $C^2$-version of SCC holds in the spherically symmetric Einstein--Maxwell--scalar field setting: they proved inextendibility of the metric in $C^2$ for generic asymptotically flat two-ended data.

We next consider the case of ``dynamical extremal black holes'', i.e.~black hole spacetime solutions which approach the extremal Reissner--Nordstr\"om suitably rapidly along the event horizon. We remark that in \cite{harvey2013} dynamical extremal black hole spacetimes are defined as being black hole spacetimes without trapped surfaces and it is shown numerically that the solutions under consideration actually approach an extremal Reissner Nordstr\"om solution along the event horizon. Conversely, it can be shown in this setting (massless, uncharged scalar field) that black hole spacetimes that approach extremal Reissner--Nordstr\"om along the event horizon will have no radial trapped surfaces in the black hole interior, provided none of the round spheres foliating the event horizon are (marginally) \emph{anti-trapped}; see also Remark 1.7 in \cite{dejanjon1}.

In this dynamical setting, the interior dynamics are significantly different for dynamical extremal black holes compared to the ``dynamical sub-extremal black holes'' that arise from asymptotically flat two-ended data, as shown in \cite{Luk2016a, Luk2016b}. Indeed, in \cite{dejanjon1} it was shown that the dynamical extremal black holes in consideration are extendible across the inner horizon as \emph{weak} solutions to the spherically symmetric Einstein--Maxwell--(charged) scalar field system of equations (in particular, with Christoffel symbols in $L^2_{\rm loc}$). Hence, \textit{in contrast to the sub-extremal case, dynamical extremal black holes do not conform to the inextendibility properties stated in SCC}. The only way that SCC can therefore still be valid, at least in the spherically symmetric setting under consideration, is if \emph{dynamical extremal black holes do not arise from ``generic'' initial data}. Analogous numerical results  were presented in \cite{harvey2013} that are moreover compatible with the non-genericity of dynamical extremal black holes (see also Section \ref{sec:TheReallSpacetimes} for a discussion on the results of  \cite{harvey2013} in the black hole exterior).

\section{The main theorems and ideas of the proofs}
\label{subsec:TheMainTheorems}

\subsection{Statements of the theorems}
\label{sec:StatementsOfTheTheorems}

We  consider smooth solutions $\psi$ to the wave equation \eqref{eq:waveequation} on ERN and we derive global late-time asymptotic estimates for $\psi$. We will mostly express our main theorems in terms of the double null coordinates $(u,v)$ (see Section \ref{sec:TheERNManifoldFoliationsAndVectorFields} for a review of the geometry of ERN). The constants $H_0$ and $H_0^{(1)}$ are defined in Sections \ref{sec:ConservationLawsAlongTheEventHorizon} and \ref{sec:TheNewHorizonHairH01Psi}, respectively. The constants $I_0$ and $I_0^{(1)}$ are defined in Section \ref{sec:TheWaveEquationOnBlackHolesBackgrounds}. The perturbations of Type \textbf{A}, \textbf{B}, \textbf{C} and \textbf{D} are introduced in Section \ref{sec:TheTypesOfInitialDataABCD}. For simplicity, the initial data norms on the right hand side of the estimates of the main theorems are not presented here and instead are presented in the relevant sections where these theorems are proved. We mention below the exact section where each theorem is proved. Moreover, the quantities $\eta>0$ and $\epsilon>0$ below are suitably small constants, $\beta\in (0,1]$ appears in the initial data norms and $k\in\mathbb{N}_{0}$. Furthermore, $C=C(M,\Sigma_0,r_\mathcal{H},r_\mathcal{I},\eta,\epsilon,\beta, k)>0$ is a universal constant.


\begin{theorem}[{\textbf{Asymptotics for Type \textbf{C} perturbations}}]
\label{prop:asympsitheo} Assume that the initial data of $\psi$ are of Type \textbf{C} and that $\psi$ solves the wave equation \eqref{eq:waveequation}. The following {global} estimate holds:
\begin{equation}
\label{eq:asympsiintro}
\begin{split}
\Bigg|T^k\psi_0(u,v)&-4\left(I_0[\psi]+ \frac{M}{r\sqrt{D}}H_0[\psi]\right)T^k\left(\frac{1}{u\cdot v}\right)\Bigg|\\
\leq&\: C\left(\sqrt{E_{0;k+1}^{\epsilon}[\psi]+\sum_{|\alpha|\leq 2}E_{1;k}^{\epsilon}[\Omega^{\alpha}\psi]}+I_0[\psi]+P_{I_0,\beta;k}[\psi]\right)v^{-1}u^{-1-k-\eta}\\
&+C\left(\sqrt{E_{0;k+1}^{\epsilon}[\psi]+\sum_{|\alpha|\leq 2}E_{1;k}^{\epsilon}[\Omega^{\alpha}\psi]}+H_0[\psi]+P_{H_0,1;k}[\psi]\right)D^{-\frac{1}{2}}u^{-1}v^{-1-k-\eta}.
\end{split}
\end{equation}
\label{ctheo}
\end{theorem}

\begin{theorem}[{\textbf{Asymptotics for Type \textbf{A} perturbations}}] Assume that the initial data of $\psi$ are of Type \textbf{A} and that $\psi$ solves the wave equation \eqref{eq:waveequation}. The following {global} estimate holds:
\begin{equation*}
\begin{split}
\Bigg|T^k\psi_0(u,v)&-4\left[ I_0^{(1)}[\psi]T^{k+1}\left(\frac{1}{u\cdot v}\right)+\frac{M}{r \sqrt{D}}H_0[\psi]T^{k}\left(\frac{1}{u(v+4M-2r)}\right)\right]\Bigg|\\
\leq&\: C\Bigg[\sqrt{E^{\epsilon}_{0, \mathcal{I}; k+1}[\psi]+\sum_{|\alpha|\leq 2}E_{1,\mathcal{I};k}^{\epsilon}[\Omega^{\alpha}\psi]}+P_{I_0,\beta;k+1}[\psi]+P_{H_0,1;k}[\psi]+H_0[\psi]+I_0^{(1)}[\psi] \Bigg]\\
&\cdot \left( v^{-1}u^{-2-k-\eta}+D^{-\frac{1}{2}}u^{-1}v^{-1-k-\eta}\right).
\end{split}
\end{equation*}
\label{atheo}
\end{theorem}
\vspace{-0.7cm}
\begin{theorem}[{\textbf{Asymptotics for Type \textbf{D} perturbations}}] Assume that the initial data of $\psi$ are of Type \textbf{D} and that $\psi$ solves the wave equation \eqref{eq:waveequation}. 
The following {global} estimate holds:
\begin{equation*}
\begin{split}
\Bigg|T^k\psi_0(u,v)&-4\left[ \frac{1}{\sqrt{D}}H_0^{(1)}[\psi]T^{k+1}\left(\frac{1}{u\cdot v}\right)+I_0[\psi]T^{k}\left(\frac{1}{v(u+2M-2M^2(r-M)^{-1})}\right)\right]\Bigg|\\
\leq&\: C\Bigg[\sqrt{E^{\epsilon}_{0, \mathcal{H}; k+1}[\psi]+\sum_{|\alpha|\leq 2}E_{1,\mathcal{H};k}^{\epsilon}[\Omega^{\alpha}\psi]}+P_{I_0,\beta;k}[\psi]+P_{H_0,1;k+1}[\psi]+I_0[\psi]+H_0^{(1)}[\psi] \Bigg]\\
&\cdot \left( v^{-1}u^{-1-k-\eta}+D^{-\frac{1}{2}}u^{-1}v^{-2-k-\eta}\right).
\end{split}
\end{equation*}
\label{dtheo}
\end{theorem}
\vspace{-0.7cm}
\begin{theorem}[{\textbf{Asymptotics for Type \textbf{B} perturbations}}]
Assume that the initial data of $\psi$ are of Type \textbf{B} and that $\psi$ solves the wave equation \eqref{eq:waveequation}. The following {global} estimate holds:
\begin{equation}
\label{eq:asympsizeroIHtheo}
\begin{split}
\Bigg|T^k\psi_0(u,v)&-4\left(I_0^{(1)}[\psi]+ \frac{M}{r\sqrt{D}}H_0^{(1)}[\psi]\right)T^{k+1}\left(\frac{1}{v\cdot u}\right)\Bigg|\\
\leq&\: C\left(\sqrt{E_{0, \mathcal{H};k+1}^{\epsilon}[\psi]+E_{0, \mathcal{I};k+1}^{\epsilon}[\psi]+\sum_{|\alpha|\leq 2}E_{2;k+1}^{\epsilon}[\Omega^{\alpha}\psi]}+I_0^{(1)}[\psi]+P_{I_0,\beta;k+1}[\psi]\right)v^{-1}u^{-2-k-\eta}\\
&+C\left(\sqrt{E_{0, \mathcal{H};k+1}^{\epsilon}[\psi]+E_{0, \mathcal{I};k+1}^{\epsilon}[\psi]+\sum_{|\alpha|\leq 2}E_{2;k+1}^{\epsilon}[\Omega^{\alpha}\psi]}+H_0^{(1)}[\psi]+P_{H_0,1;k+1}[\psi]\right)D^{-\frac{1}{2}}u^{-1}v^{-2-k-\eta}.
\end{split}
\end{equation}
\label{btheo}
\end{theorem}

\begin{theorem}[{\textbf{Logarithmic corrections for Type \textbf{C} perturbations}}] Assume that the initial data of $\psi$ are spherically symmetric and of Type \textbf{C} and that $\psi$ solves the wave equation \eqref{eq:waveequation}. Then, the following estimate holds on $\I$
\begin{equation}
\label{eq:2ndasymphinullinfintro}
\left|r\psi|_{\mathcal{I}^+}(u)-2I_0[\psi]u^{-1}+4MI_0[\psi]u^{-2}\log u\right|\leq C\left(I_0[\psi]+H_0[\psi]+\sqrt{E^{\epsilon}_{0; 1}[\psi]}+P_{\mathcal{H}}[\psi]+P_{\mathcal{I}}[\psi]\right)u^{-2}
\end{equation}
and the following estimate holds on $\mathcal{H}^{+}$
\begin{equation}
\label{eq:2ndasymphihointro}
\left|r\psi|_{\mathcal{H}^+}(v)-2H_0[\psi]v^{-1}+4MH_0[\psi]v^{-2}\log v\right|\leq  C\left(I_0[\psi]+H_0[\psi]+\sqrt{E^{\epsilon}_{0;1}[\psi]}+P_{\mathcal{H}}[\psi]+P_{\mathcal{I}}[\psi]\right)v^{-2}.
\end{equation}
\label{logtheo}
\end{theorem}

\begin{theorem}[{\textbf{Logarithmic corrections for Type \textbf{B} perturbations}}] Assume that the initial data of $\psi$ are spherically symmetric and of Type \textbf{B} and that $\psi$ solves the wave equation \eqref{eq:waveequation}. Then, the following estimate holds on $\I$
\begin{equation}
\label{eq:2ndasymphinullinftimeintintro}
\begin{split}
\Big|&r\psi|_{\mathcal{I}^+}(u)+2I_0^{(1)}[\psi]u^{-2}-8MI_0^{(1)}[\psi]u^{-3}\log u\Big|\\
\leq&\: C\left(I_0^{(1)}[\psi]+H_0^{(1)}[\psi]+\sqrt{E^{\epsilon}_{0,\mathcal{H}; 1}[\psi]+E^{\epsilon}_{0,\mathcal{I}; 1}[\psi]}+P_{\mathcal{H}, T}[\psi]+P_{\mathcal{I}, T}[\psi]\right)u^{-3},
\end{split}
\end{equation}
and the following estimate holds on $\mathcal{H}^{+}$
\begin{equation}
\label{eq:2ndasymphihotimintintro}
\begin{split}
\Big|&r\psi|_{\mathcal{H}^+}(v)+2H_0^{(1)}[\psi]v^{-2}-4MH_0^{(1)}[\psi]v^{-3}\log v\Big|\\
\leq&\:  C\left(I_0^{(1)}[\psi]+H_0^{(1)}[\psi]+\sqrt{E^{\epsilon}_{0,\mathcal{H}; 1}[\psi]+E^{\epsilon}_{0,\mathcal{I}; 1}[\psi]}+P_{\mathcal{H}, T}[\psi]+P_{\mathcal{I}, T}[\psi]\right)v^{-3}.
\end{split}
\end{equation}
\label{logtheoB}
\end{theorem}

\vspace{-0.5cm}
\noindent We split $\psi=\psi_0+\psi_{\geq 1}$ and prove the appropriate decay estimates for $\psi_{\geq 1}$ in Section \ref{sec:pdecayest}. We can then replace $\psi$ with $\psi_0$ in the theorem statements: Theorem \ref{ctheo} is proved in Section \ref{sec:asympnonzeroconst}, Theorem \ref{atheo} is proved in Section \ref{sec:asympzeroconst} and Theorems \ref{dtheo} and \ref{btheo} are proved in Section \ref{sec:AsymptoticsForTypeDPerturbations}. Finally, Theorem \ref{logtheo} and \ref{logtheoB} are proved in Section \ref{sec:hoasymp}. 

\subsection{Overview of techniques}
\label{sec:OverviewOfTechniques}
In this section we will give an overview of the main steps and methods involved in proving the theorems stated in Section \ref{sec:StatementsOfTheTheorems}. We will moreover highlight the key new ideas that play a role in the proofs. 
\subsubsection{The zeroth step}
\label{sec:TheZerothStep}

Deriving the precise late-time asymptotics requires obtaining decay rates for weighted energy fluxes and pointwise norms that are as sharp as possible. Our strategy is based on the integrated $r^{p}$-weighted energy decay approach of Dafermos--Rodnianski \cite{newmethod} and its extension presented in \cite{paper1}.  The main idea is to derive energy decay by first establishing boundedness for suitable (weighted) \textit{spacetime} integrals. For ERN, the ``zeroth'' step is the Morawetz estimate of the form (see Appendix \ref{sec:EnergyBounds})
\begin{equation}
\int_{\tau_1}^{\tau_2}\int_{\Sigma_{\tau}}\left(1-\frac{M}{r}\right)^{\sigma_1}\cdot \frac{1}{r^{\sigma_2}}\cdot J^{T}[\psi]  \, d\tau\  \lesssim \ \int_{\Sigma_{\tau_1} } J^{T}[\psi]+J^{T}[T\psi],
\label{morasche}
\end{equation}
with $\sigma_1,\sigma_2>2$ sufficient large constants. Here, $J^{T}[\psi]$ denotes the standard $T$-energy current through $\Sigma_{\tau}$. From now on, if the volume form is missing in the integrals, it is implied that we consider the standard volume form with respect to the induced metric on the corresponding hypersurface. The higher-order terms on the right hand side account for \textit{the high-frequency trapping effect on the photon sphere} at $\{r=2M\}$. The $r^{-\sigma_2}$ degenerate coefficient is related to the asymptotic flatness of the spacetime and is present in the analogous estimate for Minkowski spacetime. On the other hand, the degenerate factor $(r-M)^{\sigma_1}$ accounts for \textit{the global trapping effect on the extremal event horizon}, a feature characteristic to ERN (see Section \ref{sec:TheTrappingEffect}). 

Clearly, one needs to remove the degenerate factors from \eqref{morasche} in order to prove decay for the energy flux 
\begin{equation}
\mathcal{E}^{T}_{\Sigma_{\tau}}[\psi]:= \int_{\Sigma_{\tau}}\!J^{T}[\psi].
\label{tfluxsigmatau}
\end{equation} Dafermos and Rodnianski \cite{newmethod} and subsequently Moschidis \cite{moschidis1} showed that the weight at infinity $r^{-\sigma_2}$ can be removed for general asymptotically flat spacetimes by introducing appropriate \textit{growing} $r$ weights on the right hand side yielding \textit{a hierarchy of two $r$-weighted estimates}. In view of the degenerate factors both at the horizon and at infinity in the Morawetz estimate \eqref{morasche} on ERN, one needs to obtain \textit{an analogue of the Dafermos--Rodnianski hierarchy both at the near-infinity region $\mathcal{A}^{\mathcal{I}}$ and at the near-horizon region $\mathcal{A}^{\mathcal{H}}$ }(see Section \ref{sec:TheHyperboloidalFoliation} for the relevant definitions). This was accomplished in \cite{aretakis2}. We denote
\[N_{\tau}^{\mathcal{I}}= \Sigma_{\tau}\cap \mathcal{A}^{\mathcal{I}}, \ \ \text{ and }  \ \ N_{\tau}^{\mathcal{H}}= \Sigma_{\tau}\cap \mathcal{A}^{\mathcal{H}}. \]
\vspace{-0.4cm}
  \begin{figure}[H]
	\begin{center}
				\includegraphics[scale=0.2]{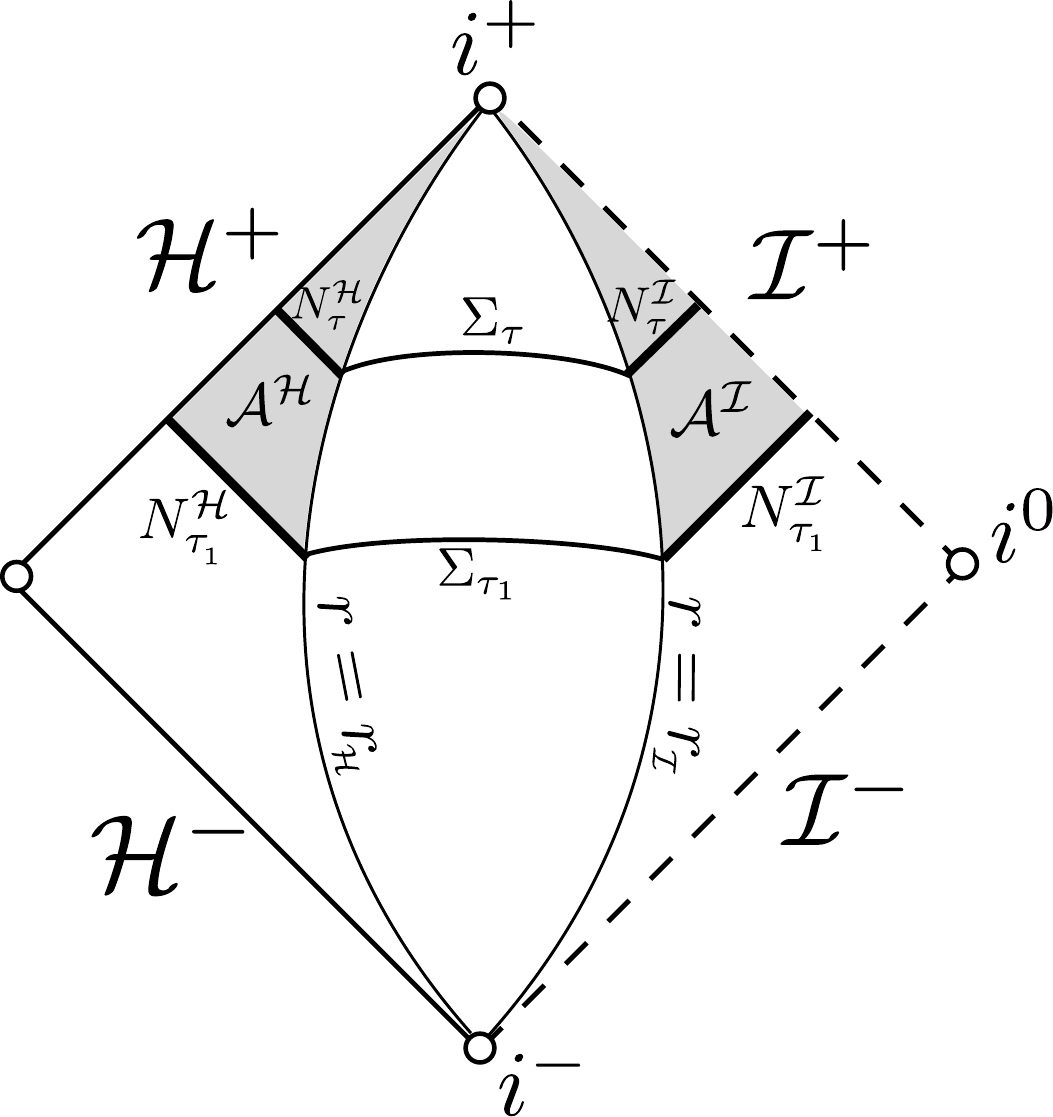}
\end{center}
\vspace{-0.4cm}
\caption{The hypersurfaces $N_{\tau}^{\mathcal{H}}$ and $N_{\tau}^{\mathcal{I}}$.}
	\label{fig:interiop455ntau}
\end{figure}
\vspace{-0.3cm}
\noindent The following  $\I-$\textbf{localized hierarchy} holds in $\mathcal{A}^{\mathcal{I}}$ for all $0\leq \tau_1<\tau_2$:
\begin{equation}
\begin{split}
\int_{\tau_1}^{\tau_2} \left[\int_{{N}^{\mathcal{I}}_{\tau}} J^{T}[\psi]\right] d\tau &\lesssim \int_{{N}^{\mathcal{I}}_{\tau_1}} r\cdot (\partial_v(r\psi))^2\,d\omega dv+\text{l.o.t.},\\
\int_{\tau_1}^{\tau_2} \left[\int_{{N}^{\mathcal{I}}_{\tau}} r\cdot (\partial_v(r\psi))^2\,d\omega dv\right]d\tau &\lesssim \int_{{N}^{\mathcal{I}}_{\tau_1}} r^2\cdot (\partial_v(r\psi))^2\,d\omega dv+\text{l.o.t.},
\label{hierioverview}
\end{split}
\end{equation}
and the following $\h-$\textbf{localized hierarchy} holds in $\mathcal{A}^{\mathcal{H}}$ for all $0\leq \tau_1<\tau_2$:
\begin{equation}
\begin{split}
\int_{\tau_1}^{\tau_2} \left[\int_{{N}^{\mathcal{H}}_{\tau}} J^{T}[\psi]\right] d\tau &\lesssim \int_{{N}^{\mathcal{H}}_{\tau_1}} (r-M)^{-1}\cdot (\partial_u(r\psi))^2\,d\omega du+\text{l.o.t.},\\
\int_{\tau_1}^{\tau_2} \left[\int_{{N}^{\mathcal{H}}_{\tau}} (r-M)^{-1}\cdot (\partial_u(r\psi))^2\,d\omega du\right]d\tau &\lesssim \int_{{N}^{\mathcal{H}}_{\tau_1}} (r-M)^{-2}\cdot (\partial_u(r\psi))^2\,d\omega du+\text{l.o.t.}.
\label{hierhoverview}
\end{split}
\end{equation}
The integral on the right hand side of the second estimate of the $\I-$localized hierarchy corresponds to \textit{the conformal energy near $\I$}. Similarly, the integral on the right hand side of the second estimate of the $\h-$localized hierarchy corresponds to \textit{the conformal energy near $\h$}. We denote 
\begin{equation}
\hspace{-1.6cm}\text{\textbf{Conformal energy near }}\I: \ \ \ \mathcal{C}_{N_{\tau}^{\mathcal{I}}}[\psi]= \int_{{N}^{\mathcal{I}}_{\tau}} r^2\cdot (\partial_v(r\psi))^2\,d\omega dv
\label{confenergyI}
\end{equation}
and
\begin{equation}
\text{\textbf{Conformal energy near }}\h: \ \ \ \mathcal{C}_{N_{\tau}^{\mathcal{H}}}[\psi]= \int_{{N}^{\mathcal{H}}_{\tau}} (r-M)^{-2}\cdot (\partial_u(r\psi))^2\,d\omega du.
\label{nondegnH}
\end{equation}
It is important to note that $du=-2\left(1-\frac{M}{r}\right)^{-2}dr$ on $\Sigma_{\tau}$ and $\partial_{u}=-\frac{1}{2}\left(1-\frac{M}{r}\right)^2Y$, where $Y=\partial_r$ is a regular vector field at the horizon. Hence, the conformal flux near $\h$ $\mathcal{C}_{N_{\tau}^{\mathcal{H}}}[\psi]\sim \int_{{N}^{\mathcal{H}}_{\tau}}(Y\psi)^2$ is at the level of the \textit{non-degenerate} energy. 

If both of the energies \eqref{confenergyI} and \eqref{nondegnH} are initially finite, then, by using a standard application of the mean value theorem on dyadic time intervals and the boundedness of the $T$-energy flux, we obtain the decay rate $\tau^{-2}$ for the $T$-energy flux $\mathcal{E}^{T}_{\Sigma_{\tau}}[\psi]$. This decay rate however is quite weak. Faster decay rates for the higher order flux $\mathcal{E}^{T}_{\Sigma_{\tau}}[T\psi]$ were obtained for sub-extremal black holes by Schlue \cite{volker1} and Moschidis \cite{moschidis1}. Their method used $\partial_v$, $r\partial_v$ as commutator vector fields in the near-infinity region. Nonetheless, their approach does not yield faster decay for the $T$-flux $\mathcal{E}^{T}_{\Sigma_{\tau}}[\psi]$ itself.

\subsubsection{Commuted hierarchies in the regions $\mathcal{A}^{\mathcal{H}}$ and $\mathcal{A}^{\mathcal{I}}$}
\label{sec:CommutedHierarchies}

Our strategy for obtaining further decay for $\mathcal{E}^{T}_{\Sigma_{\tau}}[\psi]$ on ERN is to establish \textit{integrated decay estimates for the conformal fluxes}\footnote{Note that (non-degenerate) integrated decay estimates for the fluxes $\mathcal{C}_{N_{\tau}^{\mathcal{I}}}[\psi]$ and $\mathcal{C}_{N_{\tau}^{\mathcal{H}}}[\psi]$ on ERN are closely related to the \textit{trapping effect} at $\I$ and at $\h$.}  $\mathcal{C}_{N_{\tau}^{\mathcal{I}}}[\psi]$ and $\mathcal{C}_{N_{\tau}^{\mathcal{H}}}[\psi]$, extending thereby the $\I-$localized and $\h-$localized hierarchies \eqref{hierioverview} and \eqref{hierhoverview}. However, it is not possible to further extend of \eqref{hierioverview} and \eqref{hierhoverview} by considering larger powers of $r$ and $(r-M)^{-1}$, respectively. Instead, motivated by the following Hardy inequality (see also Section \ref{sec:HardyInequalities} in the Appendix):
\begin{equation}
\int_{0}^{\infty}x^2\cdot \left(\partial_{x}f\right)^{2} \,dx \leq C\int_{0}^{\infty}\Big(\partial_{x}\left(\boldsymbol{x^{2}\partial_{x}}f\right)\Big)^2 \,dx,
\label{basichardy}
\end{equation}
applied to $f=r\psi$, with $x=r$, $\partial_x=\partial_v$ in $\mathcal{A}^{\mathcal{I}}$, and $x=(r-M)^{-1}$, $\partial_x=\partial_u$ in $\mathcal{A}^{\mathcal{H}}$,  
we introduce the following $n$-\textbf{commuted} quantities:
\[ \Phi_{(n)}:=(r^2\partial_v)^n(r\psi), \hspace{1.5cm}\underline{\Phi}_{(n)}:=Y^n(r\psi)  \sim \Big(-(r-M)^{-2}\partial_{u}\Big)^n(r\psi), \]
where $n\in \N_0$. The idea therefore is to derive  $\I-$localized and $\h-$localized \textit{commuted} hierarchies which yield decay for weighted fluxes of the commuted functions $\Phi_{(n)}$ and $\underline{\Phi}_{(n)}$, respectively. As we shall see, these hierarchies involve growing $r$ and $(r-M)^{-1}$ weights. Partial results for $\Phi_{(n)}$ and $\underline{\Phi}_{(n)}$ were previously obtained in \cite{aretakis2, paper1}. 

If $\psi$ solves the wave equation \eqref{eq:waveequation} on ERN then for all $n\geq 0$ and \textit{for all} $p\in \R$ the commuted quantities $\Phi_{(n)}$ and $\underline{\Phi}_{(n)}$ satisfy the following \textit{key identities} in $\mathcal{A}^{\mathcal{I}}$ and $\mathcal{A}^{\mathcal{H}}$ regions, respectively (see Section \ref{sec:mainest1}): 

\medskip
\noindent\textbf{Near-infinity identity:}
\begin{equation}
\label{eq:keyid}
\begin{split}
\int_{\s^2}& \partial_u\left(r^p(\partial_v\Phi_{(n)})^2\right)+\partial_v\left( r^{p-2}|\snabla_{\s^2}\Phi_{(n)}|^2-n(n+1)r^{p-2}\Phi_{(n)}^2\right)\,d\omega\\
&+\int_{\s^2}  (p+4n)r^{p-1} (\partial_v\Phi_{(n)})^2+(2-p)r^{p-3}\left( |\snabla_{\s^2}\Phi|^2-n(n+1)\Phi_{(n)}^2\right)\,d\omega\\
=&\: n\cdot\sum_{k=0}^{\max\{0,n-1\}} \int_{\s^2}O(r^{p-2})\cdot \Phi_{(k)}\cdot \partial_v\Phi_{(n)}\,d\omega+\text{l.o.t.},
\end{split}
\end{equation}
\textbf{Near-horizon identity:}
\begin{equation}
\label{eq:keyidhor}
\begin{split}
\int_{\mathbb{S}^{2}}& \partial_v\left((r-M)^{-p}(\partial_u\underline{\Phi}_{(n)})^2\right)+\partial_u\left( (r-M)^{-p+2}|\snabla_{\s^2}\underline{\Phi}_{(n)}|^2-n(n+1)(r-M)^{-p+2}\underline{\Phi}_{(n)}^2\right)\,d\omega\\
&+\int_{\mathbb{S}^{2}}  (p+4n)(r-M)^{-p+1} (\partial_u\underline{\Phi}_{(n)})^2+(2-p)(r-M)^{-p+3}\left( |\snabla_{\mathbb{S}^{2}}\underline{\Phi}_{(n)}|^2-n(n+1)\underline{\Phi}_{(n)}^2\right)\,d\omega\\
=&\: n\cdot\sum_{k=0}^{\max\{0,n-1\}} \int_{\mathbb{S}^{2}}O((r-M)^{-p+2})\cdot \underline{\Phi}_{(k)}\cdot \partial_u\underline{\Phi}_{(n)}\,d\omega+\text{l.o.t.}
\end{split}
\end{equation}
Note that \eqref{eq:keyidhor} is of the same form as \eqref{eq:keyid}, but with $u$ and $v$ reversed and $r$ replaced by $(r-M)^{-1}$. This is of course related to the existence of the Couch--Torrence conformal inversion of ERN.

 After integrating in $u$ and $v$, the ``\textit{error}'' terms that appear on the right-hand sides of the corresponding spacetime identities can be controlled via \textit{Morawetz and Hardy inequalities} for the following range of weights\footnote{For spherically symmetric solutions (with harmonic mode number $\ell=0$) we only take $n=0$.}: 
\begin{equation}
-4n<p\leq 2.
\label{pbound}
\end{equation} We arrive at the following inequalities (see Section \ref{sec:mainest})

\medskip
\noindent \textbf{$\I-$localized $n-$commuted $p-$inequalities for $\Phi_{(n)}$:}
\begin{equation}
\begin{split}
\label{eq:introrweightestI}
&\int_{{N}^{\mathcal{I}}_{\tau_2}} r^p\left(\partial_v\Phi_{(n)}\right)^2\,d\omega dv\\+& \int_{\tau_1}^{\tau_2} \int_{{N}^{\mathcal{I}}_{\tau}} (p+4n)r^{p-1} \left(\partial_v\Phi_{(n)}\right)^2+ (2-p)r^{p-3}\left( |\snabla_{\s^2}\Phi_{(n)}|^2-n(n+1)\Phi_{(n)}^2\right)\,d\omega dv d\tau\\
&\lesssim_p\: \int_{{N}^{\mathcal{I}}_{\tau_1}} r^p\left(\partial_v\Phi_{(n)}\right)^2\,d\omega dv+\ldots,
\end{split}
\end{equation}
\noindent \textbf{$\h-$localized $n-$commuted $p-$inequalities for $\underline{\Phi}_{(n)}$:}
\begin{equation}
\begin{split}
\label{eq:introrweightestH}
&\int_{{N}^{\mathcal{H}}_{\tau_2}} (r-M)^{-p}\left(\partial_u \underline{\Phi}_{(n)}\right)^2\,d\omega du\\
+& \int_{\tau_1}^{\tau_2} \int_{{N}^{\mathcal{H}}_{\tau}} (p+4n)(r-M)^{-p+1} \left(\partial_u\underline{\Phi}_{(n)}\right)^2+ (2-p)(r-M)^{-p+3}\left({|\snabla_{\s^2}\underline{\Phi}_{(n)}|^2-n(n+1)\underline{\Phi}_{(n)}^2}\right)\,d\omega du d\tau\\
&\lesssim_p \: \int_{{N}^{\mathcal{H}}_{\tau_1}}(r-M)^{-p}\left(\partial_u \underline{\Phi}_{(n)}\right)^2\,d\omega du+\ldots
\end{split}
\end{equation}
These inequalities hold for \textbf{all} $n$, as long as $p$ satisfies \eqref{pbound}.
In order to turn these inequalities into actual estimates we need to guarantee the non-negativity of the terms 
$|\snabla_{\s^2}{\Phi}_{(n)}|^2-n(n+1){\Phi}_{(n)}^2$ and $|\snabla_{\s^2}\underline{\Phi}_{(n)}|^2-n(n+1)\underline{\Phi}_{(n)}^2$. In view of the \textit{Poincar\'e inequality} on $\s^2$ (see Section \ref{sec:HardyInequalities} in the Appendix), these terms are non-negative if $\psi$ is supported on angular frequencies $\ell$ such that 
\begin{equation}
\ell \geq n.
\label{nell}
\end{equation} 
In other words, \emph{we can commute the wave equation $n$ times and obtain two estimates for $\Phi_{(n)}$ and two estimates for $\underline{\Phi}_{(n)}$ for each $n$, as long as $n$ is less or equal than the lowest harmonic mode that is present in a harmonic mode expansion of $\psi$.} The two estimates correspond to the values $p=1$ and $p=2$. 

It is worth mentioning that the estimates \eqref{eq:introrweightestH} can be thought of as  \emph{degenerate remnants} of the red shift estimates. Note that the degeneracy of the (higher-order) red shift effect is manifested in the additional factor of $(r-M)$ that appears in the spacetime integral of $(\partial_u\underline{\Phi}_{(n)})^2$ on the left-hand side of \eqref{eq:introrweightestH}.

The table below summarizes the number of the $\h-$localized $n-$commuted estimates and the $\I-$localized $n-$commuted estimates for each fixed $n$ as well as the total number of estimates available in the \textbf{total hierarchy} over all admissible values of $n$. 

\medskip

\underline{Definition}: \textit{ We define the \textbf{length of a hierarchy} to be equal to the number of available and useful integrated estimates in the hierarchy. Useful here means that the exponents $p$ of the weights  increase by an integer number or by an almost (modulo $\epsilon>0$) integer number.} 
 		\begin{table}[H]{\footnotesize
\begin{center}
    \begin{tabular}{ cV{5}c | c |c} 
		\cline{2-4}
		 & \multicolumn{3}{c}{\textbf{Commuted hierarchies}} \\
  \hline
\multirow{2}{*}{\textbf{Harmonic mode}} &   \multicolumn{2}{c|}{\textbf{Fixed} ${n}$ \textbf{commuted}} & {\textbf{Total hierarchy}}  \\
	\cline{2-4}
  &   \multicolumn{1}{c|}{${n}$}  & \textbf{Length}  &  \textbf{Length}   \\
 \Xcline{1-4}{0.06cm} 
 {$\ell=0$} & {$0$} & $2$ & $2$  \\  
 \Xcline{1-4}{0.04cm} 
 \multirow{2}{*}{$\ell=1$} & {$0$}  & $2$ &   \multirow{2}{*}{$4$}  \\    		
\cline{2-3}
 & {1} &2&  \\  
 \Xcline{1-4}{0.04cm} 
\multirow{3}{*}{$\ell\geq 2$} & {$0$}  & $2$ & \multirow{3}{*}{$6$}   \\    	
	\cline{2-3}
 & {1} &2& \\    
	\cline{2-3}
 & {2} &2& \\   
 \Xcline{1-4}{0.04cm} 
  \end{tabular}
\end{center}
\caption{The length of the commuted hierarchies for $\ell=0$, $\ell=1$ and $\ell\geq 2$.}
\label{summarytablebasic}}
\end{table} 
\subsubsection{Improved hierarchies for $\ell=0,1$}
\label{sec:ImprovedHierarchiesForEll01}

The harmonic projections $\psi_{\ell=0}$ and $\psi_{\ell=1}$ of $\psi$ satisfy only two and four estimates in the total hierarchy, respectively, as in Table \ref{summarytablebasic}. When dealing with $\ell=0$ (and hence $n=0$) separately, we show in Section \ref{sec:mainestextended} that the range of $p$ can actually be \textit{extended} to $0<p<3$ for both the $\h-$localized and the $\I-$localized hierarchies. Note that even though we cannot take $p=3$ exactly in this case, we can take $p=3-\epsilon$ for $\epsilon>0$ arbitrarily small. Additionally, we show that 
\begin{itemize}
	\item if  $I_{0}[\psi]=0$ then we can take $0<p<5$ in the $\I-$localized hierarchy, and 
\item if  $H_{0}[\psi]=0$ then we can take $0<p<5$ in the $\h-$localized hierarchy.
\end{itemize}
Similarly as above, even though we cannot take $p=5$ exactly, we will take $p=5-\epsilon$ for $\epsilon>0$. In this sense, the lengths of the above hierarchies (under the vanishing assumptions) is indeed five. Moreover, these hierarchies are \textit{inextendible} (consistent with the horizon instability results of Section \ref{sec:TheHorizonInstabilityOfExtremalBlackHoles}) and hence their length is sharp. It is important to observe that, based on the above result, the lengths of the total hierarchies depend on the type of data. These are summarized in the table below. 

\underline{Convention:} By $\mathcal{R}-$\textbf{global hierarchy} we mean the hierarchy that arises for weighted fluxed on $\Sigma_{\tau}$ by adding the $\h-$localized hierarchy (in region $\mathcal{A}^{\mathcal{H}}$), the $\I-$localized hierarchy (in region $\mathcal{A}^{\mathcal{I}}$) and the higher-order Morawetz estimates (in region $\mathcal{B}$; see Appendix \ref{sec:EnergyBounds}). Recall that $\mathcal{R}=\mathcal{A}^{\mathcal{H}}\cup \mathcal{A}^{\mathcal{I}}\cup \mathcal{B}$.

		\begin{table}[H]{\footnotesize
\begin{center}
    \begin{tabular}{ cV{3}c | c |c } 
		\cline{2-4}
		 & \multicolumn{3}{c}{\textbf{Improved  hierarchies for} $\ell=0$  \textbf{with} $n=0$} \\
  \hline
{\textbf{Data}} &   $\h-$\textbf{localized} & $\I-$\textbf{localized} & $\mathcal{R}-$\textbf{global}  \\
 \Xcline{1-4}{0.05cm} 
{Type \textbf{A}} &  {$3$} & {$5$} & {$3$} \\  
  \hline
{Type \textbf{B}}  &  {$5$} & {$5$} & {$5$} \\  
  \hline
{Type \textbf{C}}  & {$3$} & {$3$} & {$3$} \\  
  \hline
{Type \textbf{D}} & {$5$} & {$3$} & {$3$} \\  
\hline
  \end{tabular}
\end{center}
\caption{Lengths of improved hierarchies for $\ell=0$.}
\label{summarytablel0}}
\end{table}\vspace{-0.3cm}

In order to extend the length of the hierarchies for $\ell=1$ we introduce the following ``modified'' variants of $\Phi_{(1)}$ and $\underline{\Phi}_{(1)}$ (with $n=1$): 
\[\widetilde{\Phi}=\widetilde{\Phi}_{(1)}:=r(r-M)\partial_v(r\psi_{\ell=1}),  \hspace{1.5cm} \underline{\widetilde{\Phi}}=\underline{\widetilde{\Phi}}_{(1)}:=r\cdot Y(r\psi_{\ell=1}).\]
The derivatives $r^2\partial_v\widetilde{\Phi}_{(1)}$ and $(r-M)^{-2}\partial_u\underline{\widetilde{\Phi}}_{(1)}$ are conserved along the spheres foliating future null infinity $\mathcal{I}^+$ and the future event horizon $\mathcal{H}^+$, respectively. This follows from Lemma \ref{lm:maineq}.\footnote{Note that in the Minkowski spacetime, the quantity $r^2\partial_v\Phi_{(1)}$ is actually conserved along future null infinity, but in asymptotically flat spacetimes with non-zero mass, like extremal Reissner--Nordstr\"om, one has to replace $\Phi_{(1)}$ with $\widetilde{\Phi}_{(1)}$ to obtain a conservation law. The same modification of $\Phi_{(1)}$ appears in \cite{paper2}.}

We obtain the following improved identities for $\psi_{\ell=1}$ (see Section \ref{sec:mainestextended})
\begin{equation}
\label{eq:keyidmod}
\begin{split}
\int_{\mathbb{S}^{2}}& \partial_u\left(r^p(\partial_v\widetilde{\Phi})^2\right)\,d\omega+\int_{\mathbb{S}^{2}}  (p+4n)r^{p-1} (\partial_v\widetilde{\Phi})^2\,d\omega=\int_{\mathbb{S}^{2}}\boldsymbol{O(r^{p-3})} \cdot r\psi\cdot \partial_v\widetilde{\Phi}\,d\omega+\text{l.o.t}
\end{split}
\end{equation}
and
\begin{equation}
\label{eq:keyidhormod}
\begin{split}
\!\int_{\mathbb{S}^{2}}& \partial_v\left((r-M)^{-p}(\partial_u\widetilde{\underline{\Phi}})^2\right)+\!\int_{\mathbb{S}^{2}}  (p+4n)(r-M)^{-p+1} (\partial_u\widetilde{\underline{\Phi}})^2\,d\omega=\!\int_{\mathbb{S}^{2}}\boldsymbol{O((r-M)^{-p+3})}\cdot\! r\psi\cdot \partial_u\widetilde{\underline{\Phi}}\,d\omega+\text{l.o.t}.
\end{split}
\end{equation}
Note that the error terms (in bold) are now of lower order compared to the error terms in \eqref{eq:keyid} and \eqref{eq:keyidhor}. This allows us to obtain versions of \eqref{eq:introrweightestI} and \eqref{eq:introrweightestH} with $\Phi_{(1)}$ and $\underline{\Phi}_{(1)}$ replaced by $\widetilde{\Phi}$ and $\underline{\widetilde{\Phi}}$, respectively, where the range of $p$ can be \emph{extended} to either $0<p<3$. We further obtain that:
\begin{itemize}
	\item the range of the $\I-$localized hierarchy can be further extended to $0<p<4$ if ${\Phi}$ decays sufficiently fast towards $\I$, and 
\item 	the range of the $\h-$localized hierarchy can be further extended to $0<p<4$ if $\underline{\Phi}$ decays sufficiently fast towards $\h$.
\end{itemize}
Again, we cannot take $p=3$ or 4, but we will take $p=3-\epsilon$ or $4-\epsilon$. 

The results for $\ell=1$ are summarized in the table below.


			\begin{table}[H]{\footnotesize
\begin{center}
    \begin{tabular}{ cV{5}c | c |cV{3}c|c|cV{3}c} 
		\cline{2-8}
		 & \multicolumn{7}{c}{\textbf{Improved hierarchies for} $\ell=1$} \\
		\cline{2-8}
&  \multicolumn{3}{cV{3}}{$\h-$\textbf{localized}} &    \multicolumn{3}{cV{3}}{$\I-$\textbf{localized}} & $\mathcal{R}-$\textbf{global}  \\
 \Xcline{1-8}{0.04cm} 
 \multirow{2}{*}{\textbf{Data}}  &   \multicolumn{2}{c|}{${n}-$\textbf{commuted}}  &  \multirow{2}{*}{\textbf{Total length}}  &   \multicolumn{2}{c|}{${n}-$\textbf{commuted}}  &  \multirow{2}{*}{\textbf{Total length}}  & \multirow{2}{*}{\textbf{Total length}}  \\
		\cline{2-3}  \cline{5-6}
  &   \multicolumn{1}{c|}{${n}$} & \textbf{Length} &   & $n$&  \textbf{Length}  & &  \\
 \Xcline{1-8}{0.06cm} 
  \multirow{2}{*}{Type \textbf{A}} & \multicolumn{1}{c|}{$0$} & $2$ & \multirow{2}{*}{$5$} & {$0$} & {$2$} & \multirow{2}{*}{$6$} & \multirow{2}{*}{$5$} \\  
	\cline{2-3}  \cline{5-6}
 & \multicolumn{1}{c|}{1} &3& & 1 & 4& &\\ 
 \Xcline{1-8}{0.04cm} 
 \multirow{2}{*}{Type \textbf{B}} & {$0$}  & $2$ &   \multirow{2}{*}{$6$} & {$0$}  &  {$2$} & \multirow{2}{*}{$6$}& \multirow{2}{*}{$6$} \\    		
\cline{2-3} \cline{5-6}
 &{1} &4& & 1 & 4 & &\\  
 \Xcline{1-8}{0.04cm} 
 \multirow{2}{*}{Type \textbf{C}} & {$0$}  & $2$ &   \multirow{2}{*}{$5$} & 0 & 2 & \multirow{2}{*}{$5$} & \multirow{2}{*}{$5$} \\    		
\cline{2-3} \cline{5-6}
 & {1} &3& & 1 &3 & &\\  
 \Xcline{1-8}{0.04cm} 
 \multirow{2}{*}{Type \textbf{D}} & {$0$}  & $2$ &   \multirow{2}{*}{$6$} & 0 & 2  & \multirow{2}{*}{$5$} &\multirow{2}{*}{$5$}\\    		
\cline{2-3} \cline{5-6}
 & {1} &4& & 1 &3 & &\\  
 \Xcline{1-8}{0.04cm} 
\hline
  \end{tabular}
\end{center}
\caption{Lengths of improved hierarchies for $\ell=1$.}
\label{summarytablel1}}
\end{table}\vspace{-0.3cm}

\underline{\emph{Remark: additionally extended hierarchies for time-derivatives}}\\
\\
Schlue \cite{volker1} and Moschidis \cite{moschidis1} obtained improved energy decay estimates for the time derivative $T\psi$ by considering $r$-weighted estimates for the quantities $\partial_v(r\psi)$ or $r\partial_v(r\psi)$. We generalize in Section \ref{sec:extendhier} their approach by establishing estimates for $\partial_v^k\Phi_{(n)}$ in the near-infinity region $\mathcal{A}^{\mathcal{I}}$ and for $\partial_u^k\underline{\Phi}_{(n)}$ in the near-horizon region $\mathcal{A}^{\mathcal{H}}$  (with $n$ as above), where $k\in \N$ takes \textit{any} positive value $k\geq 1$.  This yields the following: \textit{for each time derivative that we take, we gain two more estimates in the $\I-$localized hierarchy and two more estimates in the $\h-$localized hierarchy.} These improvements play an important role in the subsequent subsections.

\subsubsection{Energy and pointwise decay}
\label{sec:EnergyAndPointwiseDecay}

The total (that is, over all admissible $n$) $\I-$localized and $\h-$localized hierarchies give quantitative decay rates for the conformal fluxes $\mathcal{C}_{N_{\tau}^{\mathcal{I}}}[\psi]$, given by \eqref{confenergyI}, and  $\mathcal{C}_{N_{\tau}^{\mathcal{H}}}[\psi]$, given by \eqref{nondegnH}, respectively. This is easily obtained via successive application of the mean value theorem in dyadic intervals and  the Hardy inequality \eqref{basichardy}. The rule is the following: 
\[\hspace{-0.2cm}\text{\textit{decay rate of the conformal flux }} \mathcal{C}_{N_{\tau}^{\mathcal{I}}}[\psi]=\text{\textit{length}}\left(\I\!\!-\!\text{\textit{localized hierarchy}}\right)-2-\epsilon,\]
and
\[\text{\textit{decay rate of the conformal flux }} \mathcal{C}_{N_{\tau}^{\mathcal{H}}}[\psi]=\text{\textit{length}}\left(\h\!\!-\!\text{\textit{localized hierarchy}}\right)-2-\epsilon\]
for any sufficiently small $\epsilon>0$. The $\epsilon$ loss here has to do with the fact that the maximum value of $p$ in the extended improved hierarchies for $\ell=0$ and $\ell=1$ is not an exact integer. 

Having obtained decay rates for the conformal fluxes we can proceed to obtain decay rate for the global $T-$flux $\mathcal{E}^{T}_{\Sigma_{\tau}}[\psi]$. For this we revisit the $\h-$localized and $\I-$localized hierarchies; we add the  $\h-$localized hierarchy (in region $\mathcal{A}^{\mathcal{H}}$), the $\I-$localized hierarchy (in region $\mathcal{A}^{\mathcal{I}}$) and the higher-order Morawetz estimates (in region $\mathcal{B}$). Using again successively the mean value theorem in dyadic intervals and appropriate Hardy inequalities we obtain decay estimates for the $T$-energy flux. The rule here is the following:
\[\text{\textit{decay rate of the energy flux }} {\mathcal{E}}^{T}_{\Sigma_{\tau}}[\psi]=\text{\textit{decay rate of \underline{slowest} conformal flux }}+2.\]
Unlike the sub-extremal case, in ERN there are \textbf{two} independent conformal fluxes that contribute to the decay rate for the energy flux. This feature of ERN creates further complications later in the derivation of the precise asymptotics.

As an illustration of our techniques, let us consider initial data for $\psi$ of Type \textbf{A}. As we can see in Table \ref{summarytablel0}, the length of the total $\I-$localized hierarchy and total $\h-$localized hierarchy is 5 and 3 for $\ell=0$ , respectively. Hence, we obtain schematically the decay estimates for the conformal fluxes (see Section \ref{sec:decayest}):
  \begin{align*}
       \mathcal{C}_{N_{\tau}^{\mathcal{H}}}[\psi_{\ell=0}]\ \ \lesssim&\ \ \: E_{\ell=0}\cdot \tau^{-1+\epsilon},\\
  \mathcal{C}_{N_{\tau}^{\mathcal{I}}}[\psi_{\ell=0}]\ \ \lesssim&\ \ \:  E_{\ell=0}\cdot \tau^{-3+\epsilon}.
  \end{align*}
	Furthermore, from  Tables \ref{summarytablebasic} and \ref{summarytablel1} we have that the length of the total $\I-$localized hierarchy and total $\h-$localized hierarchy is  6 and 5 for $\ell\geq 1$, respectively. Hence, 
	  \begin{align*}
  \mathcal{C}_{N_{\tau}^{\mathcal{H}}}[\psi_{\ell\geq 1}]\ \ \lesssim&\ \ \:  E_{\ell\geq 1}\cdot \tau^{-3+\epsilon},\\
       \mathcal{C}_{N_{\tau}^{\mathcal{I}}}[\psi_{\ell\geq 1}]\ \ \lesssim&\ \ \:  E_{\ell\geq 1}\cdot \tau^{-4+\epsilon}.
  \end{align*}
	We conclude the following decay estimate for the $T-$energy flux:
  \begin{align*}
 \mathcal{E}^{T}_{\Sigma_{\tau}}[\psi_{\ell=0}]  \ \  \lesssim&\ \ \:  E_{\ell=0}\cdot \tau^{-3+\epsilon},\\
     \mathcal{E}^{T}_{\Sigma_{\tau}}[\psi_{\ell\geq 1}] \ \  \lesssim&\ \ \:  E_{\ell\geq 1} \cdot \tau^{-5+\epsilon},
  \end{align*}
  where $E_{\ell=0}$ and $E_{\ell\geq 1}$ denote (higher-order, weighted) initial data energy norms. Furthermore,
    \begin{align*}
\mathcal{E}^{T}_{\Sigma_{\tau}}[T^k\psi_{\ell=0}]  \ \  \lesssim&\ \ \:  E_{\ell=0; k}\cdot \tau^{-3-2k+\epsilon},\\
     \mathcal{E}^{T}_{\Sigma_{\tau}}[T^k\psi_{\ell\geq 1}] \ \  \lesssim&\ \ \:  E_{\ell\geq 1; k} \cdot \tau^{-5-2k+\epsilon},
  \end{align*}
for all $k\geq 1$,   where $E_{\ell=0;k}$ and $E_{\ell\geq 1;k}$ denote (higher-order, weighted) initial data energy norms.

We next proceed with deriving pointwise decay estimates (see Section \ref{sec:decayest}). We will use the following Hardy estimates
\begin{align*}
 \int_{\mathbb{S}^{2}} (r\psi)^2\,d\omega \ \ &\lesssim  \ \ \sqrt{\mathcal{C}_{N_{\tau}^{\mathcal{H}}}[\psi] }\cdot \sqrt{ \mathcal{E}^{T}_{\Sigma_{\tau}}[\psi]}  \ \ \ \  \textnormal{ in }\ \ \mathcal{A}^{\mathcal{H}}, \\
 \int_{\mathbb{S}^{2}} (r\psi)^2\,d\omega\ \ &\lesssim   \ \  \sqrt{\mathcal{C}_{N_{\tau}^{\mathcal{I}}}[\psi] }\cdot \sqrt{ \mathcal{E}^{T}_{\Sigma_{\tau}}[\psi]}  \ \ \ \ \textnormal{in }\ \ \mathcal{A}^{\mathcal{I}},\\
 \int_{\mathbb{S}^{2}} (r-M)\!\cdot\!\psi^2\,d\omega \ \ &\lesssim  \ \  \sqrt{ \mathcal{E}^{T}_{\Sigma_{\tau}}[\psi]}  \ \ \ \ \text{ on } \ \Sigma_{\tau}. 
\end{align*}
For initial data of Type \textbf{A}, using the above decay estimates for the conformal energies and the $T$-energy flux, we obtain
\begin{align*}
\int_{\mathbb{S}^{2}} (r\psi_{\ell=0})^2\,d\omega  \ \  \lesssim&  \ \ \:  E_{\ell=0}\cdot \tau^{-2+\epsilon}\:  \ \  \ \  \textnormal{in }\ \ \mathcal{A}^{\mathcal{H}}, \\
\int_{\mathbb{S}^{2}} (r\psi_{\ell=0})^2\,d\omega \ \  \lesssim&\: \ \   E_{\ell=0} \cdot \tau^{-3+\epsilon} \: \ \ \ \ \textnormal{in }\ \ \mathcal{A}^{\mathcal{I}},\\
\int_{\mathbb{S}^{2}} (r-M)\cdot (\psi_{\ell=0})^2\,d\omega  \ \ \lesssim&\: \ \  E_{\ell=0}\cdot \tau^{-3+\epsilon} \:   \ \ \  \text{ on } \ \ \Sigma_{\tau}. 
\end{align*}
Using the standard Sobolev estimates on $\s^2$ we immediately obtain $L^{\infty}$ decay estimates for $r\psi_{\ell=0}$ in $\mathcal{A}^\mathcal{H}$, $r\psi_{\ell=0}$ in $\mathcal{A}^\mathcal{I}$ and $\sqrt{r-M}\cdot\psi_{\ell=0}$ on $\Sigma_{\tau}$, with the decaying factors $\tau^{-1+{\epsilon}}$, $\tau^{-\frac{3}{2}+{\epsilon}}$ and $\tau^{-\frac{3}{2}+{\epsilon}}$, respectively.
Similarly, 
\begin{align*}
\int_{\mathbb{S}^{2}} (r\psi_{\ell\geq 1})^2\,d\omega  \ \  \lesssim&  \ \ \:  E_{\ell=0}\cdot \tau^{-4+\epsilon}\:  \ \  \ \  \textnormal{in }\ \ \mathcal{A}^{\mathcal{H}}, \\
\int_{\mathbb{S}^{2}} (r\psi_{\ell\geq 1})^2\,d\omega \ \  \lesssim&\: \ \   E_{\ell=0} \cdot \tau^{-\frac{9}{2}+\epsilon} \: \ \ \ \ \textnormal{in }\ \ \mathcal{A}^{\mathcal{I}},\\
\int_{\mathbb{S}^{2}} (r-M)\cdot (\psi_{\ell\geq 1})^2\,d\omega  \ \ \lesssim&\: \ \  E_{\ell=0}\cdot \tau^{-5+\epsilon} \:   \ \ \  \text{ on } \ \ \Sigma_{\tau}. 
\end{align*}
As above, $L^{\infty}$ decay estimates for $r\psi_{\ell\geq 1}$ in $\mathcal{A}^\mathcal{H}$, $r\psi_{\ell\geq 1}$ in $\mathcal{A}^\mathcal{I}$ and $\sqrt{r-M}\cdot\psi_{\ell\geq 1}$ on $\Sigma_{\tau}$, with the decaying factors $\tau^{-2+{\epsilon}}$, $\tau^{-\frac{9}{4}+{\epsilon}}$ and $\tau^{-\frac{5}{2}+{\epsilon}}$, respectively.

The above estimates illustrate another deviation from the sub-extremal analysis in \cite{paper1,paper2}: for Type \textbf{A} initial data, the decay rate of $r\psi_{\ell=0}$ in $\mathcal{A}^\mathcal{I}$ \textit{is a power $\frac{1}{2}+\epsilon$ \underline{away} from the sharp decay rate}, whereas in the sub-extremal case, the analogous estimate results in a decay rate that is almost sharp, in other words only $\epsilon$ away from the sharp decay rate. In this case it is the non-vanishing of $H_0$ and hence the slow decay for the conformal energy in the near-horizon region that forms the ``bottleneck'' for the maximal length of the global hierarchy of weighted estimates for $\psi_{\ell=0}$.

The energy and pointwise decay rates are summarized in the two tables below (see Section \ref{sec:decayest}). 
		\begin{table}[H]{\footnotesize
\begin{center}
    \begin{tabular}{ cV{5}c | c |cV{3}c|c|c} 
			\cline{2-7}
		 & \multicolumn{6}{c}{\textbf{Decay rates for $\ell=0$}}\\
				\cline{2-7}
				 & \multicolumn{3}{cV{3}}{\textbf{Energy flux decay}} &  \multicolumn{3}{c}{\textbf{Pointwise decay}}\\
	   \Xcline{1-7}{0.03cm}
{\textbf{Data}} &$\mathcal{E}^{T}_{\ \Sigma_{\tau}}[\psi] $& $\mathcal{C}_{N_{\tau}^{\mathcal{H}}}[\psi]$ & $\mathcal{C}_{N_{\tau}^{\mathcal{I}}}[\psi]$ &$r\psi|_{\mathcal{H}^{+}}$ & $\psi|_{\{r=r_0\}}$ & $r\psi|_{\mathcal{I}^{+}}$ \\
 \Xcline{1-7}{0.05cm} 
 {Type \textbf{A}} &   $\tau^{-3+\epsilon}$  & $\tau^{-1+\epsilon}$ & $\tau^{-3+\epsilon}$ &$\tau^{-1+\epsilon}$& $\boldsymbol{\tau^{-\frac{3}{2}+\epsilon}}$ &$\boldsymbol{\tau^{-\frac{3}{2}+\epsilon}}$  \\  \hline 
{Type \textbf{B}} & $\tau^{-5+\epsilon}$  & $\tau^{-3+\epsilon}$ &  $\tau^{-3+\epsilon}$   &$\tau^{-2+\epsilon}$&$\boldsymbol{\tau^{-\frac{5}{2}+\epsilon}}$  &$\tau^{-2+\epsilon}$  \\    \hline
 {Type \textbf{C}}	& $\tau^{-3+\epsilon}$ &  $\tau^{-1+\epsilon}$  & $\tau^{-1+\epsilon}$ &$\tau^{-1+\epsilon}$&$\boldsymbol{\tau^{-\frac{3}{2}+\epsilon}}$ &$\tau^{-1+\epsilon}$ \\ \hline
 {Type \textbf{D}}	& $\tau^{-3+\epsilon}$ & $\tau^{-3+\epsilon}$ & $\tau^{-1+\epsilon}$&$\boldsymbol{\tau^{-\frac{3}{2}+\epsilon}}$ &$\boldsymbol{\tau^{-\frac{3}{2}+\epsilon}}$  &$\tau^{-1+\epsilon}$ \\ \hline
  \end{tabular}
\end{center} 
\caption{Decay rates for $\ell=0$. All are almost sharp except the bold rates.}
\label{summarytablenew}}
\end{table}\vspace{-0.3cm}

		\begin{table}[H]{\footnotesize
\begin{center}
    \begin{tabular}{ cV{5}c | c |cV{3}c|c|c} 
			\cline{2-7}
		 & \multicolumn{6}{c}{\textbf{Decay rates for $\ell\geq 1$}}\\
				\cline{2-7}
				 & \multicolumn{3}{cV{3}}{\textbf{Energy flux decay}} &  \multicolumn{3}{c}{\textbf{Pointwise decay}}\\
	   \Xcline{1-7}{0.03cm}
{\textbf{Data}} &$\mathcal{E}^{T}_{\ \Sigma_{\tau}}[\psi] $& $\mathcal{C}_{N_{\tau}^{\mathcal{H}}}[\psi]$ & $\mathcal{C}_{N_{\tau}^{\mathcal{I}}}[\psi]$ &$r\psi|_{\mathcal{H}^{+}}$ & $\psi|_{\{r=r_0\}}$ & $r\psi|_{\mathcal{I}^{+}}$ \\
 \Xcline{1-7}{0.05cm} 
 {Type \textbf{A}} &   $\tau^{-5+\epsilon}$  & $\tau^{-3+\epsilon}$ & $\tau^{-4+\epsilon}$ &$\tau^{-2+\epsilon}$&$\tau^{-\frac{5}{2}+\epsilon}$ &$\tau^{-\frac{9}{4}+\epsilon}$  \\  \hline 
 {Type \textbf{B}} & $\tau^{-6+\epsilon}$  & $\tau^{-4+\epsilon}$ &  $\tau^{-4+\epsilon}$   &$\tau^{-\frac{5}{2}+\epsilon}$&
\cellcolor{gray!25}$\tau^{-3+\epsilon}$  &$\tau^{-\frac{5}{2}+\epsilon}$  \\    \hline
 {Type \textbf{C}}	& $\tau^{-5+\epsilon}$ &  $\tau^{-3+\epsilon}$  & $\tau^{-3+\epsilon}$ &$\tau^{-2+\epsilon}$&$\tau^{-\frac{5}{2}+\epsilon}$  &$\tau^{-2+\epsilon}$ \\ \hline
 {Type \textbf{D}}	& $\tau^{-5+\epsilon}$ & $\tau^{-4+\epsilon}$ & $\tau^{-3+\epsilon}$&$\tau^{-\frac{5}{2}+\epsilon}$ &$\tau^{-\frac{5}{2}+\epsilon}$  &$\tau^{-2+\epsilon}$ \\ \hline
  \end{tabular}
\end{center} 
\caption{Decay rates for $\ell \geq 1$. All are sub-dominant except the one in the shaded cell.  }
\label{summarytablenew1}}
\end{table}\vspace{-0.3cm}
Note that the decay rates for $\psi_{\{r=r_0\}}$ apply for $\sqrt{r-M}\psi$ for all $r>M$.

\subsubsection{An elliptic estimate for $\ell \geq 1$}
\label{sec:EllipticEstimates}

The decay rates for $\psi_{\{r=r_0\}}$, in the $\ell=0$ case, in the  Table \ref{summarytablenew} are $\frac{1}{2}+\epsilon$ away from sharp.  Furthermore, the decay rate for $\psi_{\{r=r_0\}}$, in the $\ell\geq 1$ case, as in Table \ref{summarytablenew1}, is slower than the corresponding expected sharp rate for the $\ell=0$ case. For obtaining late-time asymptotics, the $\ell\geq 1$ rate must be improved. 

The desired improvement of the decay rate of $\psi_{\{r=r_0\}}$ will be achieved using an elliptic estimate and the improved decay rates for $T\psi$. The challenge for obtaining the elliptic estimate is that, in contrast to the sub-extremal case, the decaying global energy flux ${\mathcal{E}}^{T}_{\Sigma_{\tau}}$ is highly \textit{degenerate} at the event horizon. Indeed, recall that ${\mathcal{E}}^{T}_{\Sigma_{\tau}}[\psi]\sim \int_{\Sigma_{\tau}}\left(1-\frac{M}{r}\right)^2\cdot |\partial\psi|^{2}$. In other words, we need to obtain a degenerate elliptic estimate on ERN. It turns out that such an estimate is \textbf{not} possible for $\ell=0$ and hence we will need to derive the precise asymptotics using the aforementioned weak rates (see the next subsection). On the other hand, we can establish such a degenerate elliptic estimate for $\ell \geq 1$ (see Section \ref{sec:ellpest}) which, coupled with a Hardy inequality, schematically gives:
\begin{equation}
\int_{\Sigma_{\tau}}\left(1- \frac{M}{r}\right)^4 \cdot \left(\partial_{\rho}\psi_{\ell \geq 1}\right)^2\cdot r^{-2}\,  d\mu_{\Sigma_{\tau}}\ \ \lesssim \ \ \int_{\Sigma_{\tau}} \left(1- \frac{M}{r}\right)^2 \cdot \left(\partial_{\rho} T \psi_{\ell \geq 1}\right)^2\, d\mu_{\Sigma_{\tau}},
\label{scheelli}
\end{equation}
where $\partial_{\rho}$ denotes the radial ($SO(3)-$invariant) vector field tangent to $\Sigma_{\tau}$. Consequently, recalling that $D=\left(1-\frac{M}{r}\right)^2$ and using a standard Hardy inequality  and the improved  energy decay estimates for $T\psi$ we obtain for Type \textbf{B} data:
\begin{equation*}
\begin{split}
\int_{\mathbb{S}^2}  \left(\psi_{\ell\geq 1}\right)^2\,d\omega\ \ \lesssim\ \  &\:\ \ \frac{1}{{D}}\sqrt{ \int_{\Sigma_{\tau}} D^2 \cdot\left(\partial_{\rho} \psi_{\ell\geq 1}\right)^2 \cdot r^{-2}\,d\mu_{\Sigma_{\tau}}}\cdot \sqrt{\int_{\Sigma_{\tau}} \psi_{\ell\geq 1}^2\cdot r^{-2}\,d\mu_{\Sigma_{\tau}}}\\\
\overset{\eqref{scheelli}}{\lesssim} &\:\ \ \frac{1}{{D}}\sqrt{ \int_{\Sigma_{\tau}} D \cdot\left(\partial_{\rho} T \psi_{\ell\geq 1}\right)^2 \,d\mu_{\Sigma_{\tau}}}\cdot \sqrt{\int_{\Sigma_{\tau}} D \cdot\left(\partial_{\rho} \psi_{\ell\geq 1}\right)^2\,d\mu_{\Sigma_{\tau}}}\\
= \ \ &\: \ \ \frac{1}{{D}}\sqrt{\mathcal{E}^{T}_{\Sigma_{\tau}}[T\psi_{\ell\geq 1}]}\cdot \sqrt{\mathcal{E}^{T}_{\Sigma_{\tau}}[\psi_{\ell\geq 1}]}\\
\lesssim\ \  &\: \ \ \frac{1}{{D}} \sqrt{E_{\ell\geq 1;1}}\cdot \sqrt{E_{\ell\geq 1}}\cdot\tau^{-7+\epsilon},
\end{split}
\end{equation*}
where used the decay rates in Table \ref{summarytablenew1}. This yields that  $\left(1-\frac{M}{r}\right)\cdot \psi_{\ell \geq 1}$ decays with a rate $\tau^{-\frac{7}{2}+\frac{\epsilon}{2}}$. This rate is now indeed sub-dominant (i.e.\ strictly faster than $\tau^{-3}$). We summarize our findings in the table below:
		\begin{table}[H]{\footnotesize
\begin{center}
    \begin{tabular}{ cV{5}c | c |cV{3}c|c|c} 
			\cline{2-7}
		 & \multicolumn{6}{c}{\textbf{Decay rates for $\ell\geq 1$}}\\
				\cline{2-7}
				 & \multicolumn{3}{cV{3}}{\textbf{Energy flux decay}} &  \multicolumn{3}{c}{\textbf{Pointwise decay}}\\
	   \Xcline{1-7}{0.03cm}
{\textbf{Data}} &$\mathcal{E}^{T}_{\ \Sigma_{\tau}}[\psi] $& $\mathcal{C}_{N_{\tau}^{\mathcal{H}}}[\psi]$ & $\mathcal{C}_{N_{\tau}^{\mathcal{I}}}[\psi]$ &$r\psi|_{\mathcal{H}^{+}}$ & $\psi|_{\{r=r_0\}}$ & $r\psi|_{\mathcal{I}^{+}}$ \\
 \Xcline{1-7}{0.05cm} 
 {Type \textbf{A}} &   $\tau^{-5+\epsilon}$  & $\tau^{-3+\epsilon}$ & $\tau^{-4+\epsilon}$ &$\tau^{-2+\epsilon}$&$\tau^{-\frac{5}{2}+\epsilon}$ &$\tau^{-\frac{9}{4}+\epsilon}$  \\  \hline 
 {Type \textbf{B}} & $\tau^{-6+\epsilon}$  & $\tau^{-4+\epsilon}$ &  $\tau^{-4+\epsilon}$   &$\tau^{-\frac{5}{2}+\epsilon}$&
\cellcolor{gray!25}$\tau^{-\frac{7}{2}+\epsilon}$  &$\tau^{-\frac{5}{2}+\epsilon}$  \\    \hline
 {Type \textbf{C}}	& $\tau^{-5+\epsilon}$ &  $\tau^{-3+\epsilon}$  & $\tau^{-3+\epsilon}$ &$\tau^{-2+\epsilon}$&$\tau^{-\frac{5}{2}+\epsilon}$  &$\tau^{-2+\epsilon}$ \\ \hline
 {Type \textbf{D}}	& $\tau^{-5+\epsilon}$ & $\tau^{-4+\epsilon}$ & $\tau^{-3+\epsilon}$&$\tau^{-\frac{5}{2}+\epsilon}$ &$\tau^{-\frac{5}{2}+\epsilon}$  &$\tau^{-2+\epsilon}$ \\ \hline
  \end{tabular}
\end{center} 
\caption{The decay in the shaded cell, obtained using the elliptic estimate, is sub-dominant.  }
\label{summarytablenewafterelliptic}}
\end{table}\vspace{-0.3cm}

\subsubsection{Late-time asymptotics }
\label{sec:LateTimeAsymptotics}

In this section we will provide a summary of the mechanism that gives rise to the precise leading-order asymptotics for $\psi$. Our discussion here complements that of Section \ref{sec:SummaryOfTheMainResults}. The complete proofs cover Sections \ref{sec:asympnonzeroconst}--\ref{sec:hoasymp}. We claim that the decay rates for $\psi_{\ell\geq 1}$ as in Table \ref{summarytablenewafterelliptic} are faster than the sharp decay rates for $\psi_{\ell=0}$. Based on this claim, we will derive first the precise late-time asymptotics (and hence the sharp rates) for $\psi_{\ell=0}$.  For this reason, we will assume in the rest of this section that $\psi$ is a spherically symmetric (and hence supported only on $\ell=0$) solution to the wave equation \eqref{eq:waveequation} on ERN. 

We need to overcome the following difficulties

\begin{itemize}
\item \textbf{Difficulty 1:} Find spacetime regions in which asymptotics can be derived \textit{independently} of their complement in $\mathcal{R}$. An obstruction here is that the decay rates that we have already obtained (as summarized in the previous subsections) are a power $\frac{1}{2}+\epsilon$ from the sharp values in the region $\mathcal{B}=\{r_{\mathcal{H}}\leq r\leq r_{\mathcal{I}}\}$. Compare the rates in Tables \ref{summarytable} and \ref{summarytablenew}.

\item \textbf{Difficulty 2:} Propagate the above asymptotics \textit{globally} in the region $\mathcal{R}$. The main obstruction here is that for data of Type \textbf{A}, \textbf{B} and \textbf{C} the radial (tangential to $\Sigma_{\tau}$) derivative $\partial_{\rho}\psi$ decays only as fast as $\psi$ itself and hence the corresponding decay estimates cannot be easily integrated to propagate the asymptotics of $\psi$, without first deriving the precise asymptotics of $\partial_{\rho}\psi$. Compare the rates in Tables \ref{summarytable} and \ref{summarytablerho}. We remark that this is not the case in sub-extremal black holes where radial derivatives decay faster than the scalar field itself. 

\end{itemize}
We consider the timelike hypersurfaces $\gamma^{\mathcal{I}}$ and $\gamma^{\mathcal{H}}$ such that $(v-u) |_{\gamma^{\mathcal{I}}}\sim u^{\alpha}$ and $(u-v) |_{\gamma^{\mathcal{H}}}\sim v^{\alpha}$  where $0<\alpha<1$ is a constant, and we define the following subsets of the near-infinity region $\mathcal{A}^\mathcal{I}$ and the near-horizon region $\mathcal{A}^\mathcal{H}$: $\mathcal{A}^\mathcal{I}_{\gamma^{\mathcal{I}}}:=\mathcal{A}^\mathcal{I}\cap \{r\geq r|_{\gamma^{\mathcal{I}}}\}$ and $\mathcal{A}^\mathcal{H}_{\gamma^{\mathcal{H}}}=\mathcal{A}^\mathcal{H}\cap \{r\leq r|_{\gamma^{\mathcal{H}}}\}$. Note that $(r-M)|_{\gamma^{\mathcal{H}}} \sim r|_{\gamma^{\mathcal{I}}}\sim\tau^{\alpha}$.

\vspace{-0.15cm}
  \begin{figure}[H]
	\begin{center}
				\includegraphics[scale=0.4]{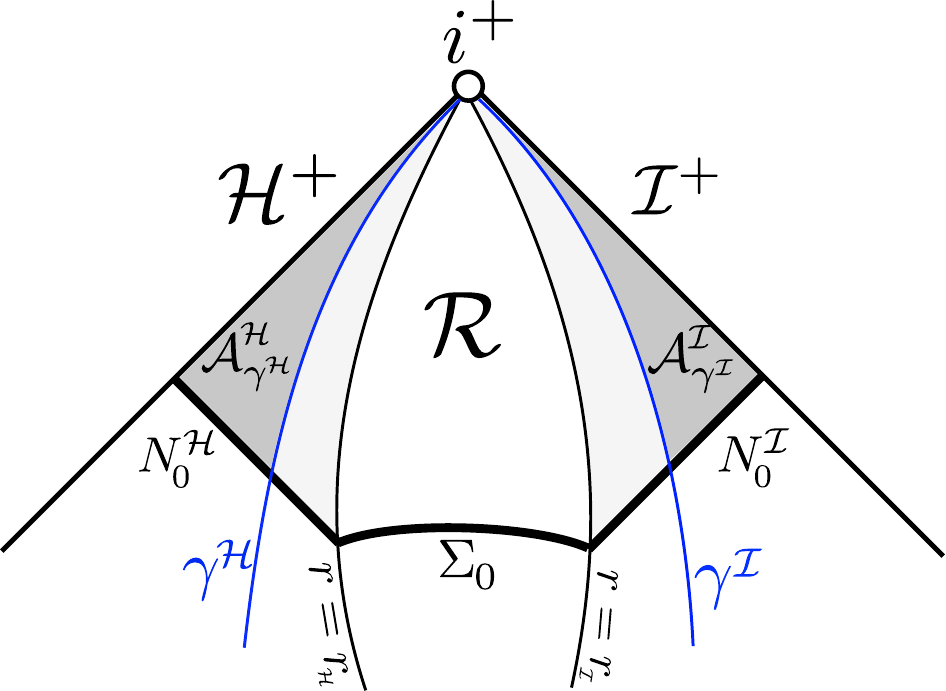}
\end{center}
\vspace{-0.4cm}
\caption{The curves $\gamma^{\mathcal{H}}, \gamma^{\mathcal{I}}$ and the regions $\mathcal{A}^{\mathcal{H}}_{\gamma^{\mathcal{H}}}, \mathcal{A}^{\mathcal{I}}_{\gamma^{\mathcal{I}}}$.}
	\label{fig:gammaa}
\end{figure}
\vspace{-0.3cm}
\noindent  We will summarize the resolutions to the above  difficulties mainly for initial data of Type \textbf{C} and \textbf{A} and make a few concluding comments for data of Type \textbf{B} and \textbf{D}. 

\medskip

\noindent\textit{\underline{Late-time asymptotics for Type \textbf{C} data}}

\medskip

\noindent\textbf{Resolution of difficulty 1} 

\smallskip

For Type \textbf{C} data we derive the leading-order asymptotics of $\psi$ in the near-horizon region $\mathcal{A}^{\mathcal{H}}_{\gamma^{\mathcal{H}}}$ and separately and independently in the near-infinity region  $\mathcal{A}^{\mathcal{I}}_{\gamma^{\mathcal{I}}}$. This derivation distinguishes the extremal case from the sub-extremal case treated in \cite{paper2}, where the asymptotics at the near-infinity region can be propagated all the way to the event horizon using that the radial derivative $\partial_{\rho}\psi$ decays faster than $\psi$. The reason we can independently derive the asymptotics in  the regions $\mathcal{A}^{\mathcal{H}}_{\gamma^{\mathcal{H}}}$ and $\mathcal{A}^{\mathcal{I}}_{\gamma^{\mathcal{I}}}$ in the extremal case has to do with the existence of the two (independent) conserved charges $H_0$ and $I_0$; moreover, for Type \textbf{C} data they are both non-zero, i.e.\ $H_0\neq 0$ and $I_0\neq 0$. To obtain the  precise asymptotics in $\mathcal{A}^{\mathcal{I}}_{\gamma^{\mathcal{I}}}$ and $\mathcal{A}^{\mathcal{H}}_{\gamma^{\mathcal{H}}}$ we propagate the following $v$-asymptotics and $u$-asymptotics of the initial data on  $N_{0}^{\mathcal{I}}$ and $N_{0}^{\mathcal{H}}$, respectively,
\begin{equation}\begin{split}
\partial_v(r\psi)|_{N_{0}^{\mathcal{I}}}=&2I_0v^{-2}+O(v^{-2}),\\
\partial_u(r\psi)|_{N_{0}^{\mathcal{H}}}=&2H_0u^{-2}+O(u^{-2})
\end{split}
\label{uvinitialasymptotics}
\end{equation}
everywhere in $\mathcal{A}^{\mathcal{I}}_{\gamma^{\mathcal{I}}}$ and $\mathcal{A}^{\mathcal{H}}_{\gamma^{\mathcal{H}}}$, respectively. This can be achieved for $\alpha<1$, but sufficiently close to 1.  We next integrate the resulting estimates for $\partial_v(r\psi)$ and $\partial_u(r\psi)$ starting from $\gamma^{\mathcal{I}}$ and $\gamma^{\mathcal{H}}$, respectively, to obtain the asymptotics for $r \psi$, and consequently $\psi$, in appropriate sub-regions $\mathcal{A}^{\mathcal{I}}_{{\gamma'}^{\mathcal{I}}}$ and $\mathcal{A}^{\mathcal{H}}_{{\gamma'}^{\mathcal{H}}}$ of $\mathcal{A}^{\mathcal{I}}_{\gamma^{\mathcal{I}}}$ and $\mathcal{A}^{\mathcal{H}}_{\gamma^{\mathcal{H}}}$ obtained by replacing $\alpha$ with appropriate $\alpha'$ such that  $\alpha< \alpha'<1$. A crucial observation is that  the previously derived decay rates for $\sqrt{r-M}\cdot\psi_{{\gamma'}^{\mathcal{I}}}$ and  $\sqrt{r-M}\cdot\psi_{{\gamma'}^{\mathcal{H}}}$ are almost sharp\footnote{Note, however, that the relevant decay rates for $\psi$, without the $\sqrt{r-M}$ weight, are \underline{not} almost sharp; see Table \ref{summarytablenew}.} and hence strong enough to close this argument by showing that, as long as $a<1$, the terms $r\psi|_{\gamma^{\mathcal{I}}}$ and $r\psi|_{\gamma^{\mathcal{H}}}$ decay faster than, say $r\psi|_{{\gamma'}^{\mathcal{I}}}$ and $r\psi|_{{\gamma'}^{\mathcal{H}}}$, and hence are lower order terms.

\medskip

\noindent\textbf{Resolution of difficulty 2} 

\smallskip

Ideally, we would like to propagate to the left of ${\gamma'}^{\mathcal{I}}$ the asymptotics for $\psi_{{\gamma'}^{\mathcal{I}}}$. In the sub-extremal case this would follow using that $\alpha'<1$ and that the radial derivative $\partial_{\rho}\psi$ decays \emph{faster} that $\psi$. This approach however breaks down in the extremal case in view of the fact that the expected sharp decay rate for $\partial_{\rho}\psi$ is now the \emph{same} as the expected sharp rate for $\psi$. 

Instead we obtain first the precise asymptotic behavior of the radial derivative $\partial_{\rho}\psi$. We commute by $T$ and reproduce the above argument to derive the precise late-time asymptotics for $T(r\psi)$ in the near-horizon region $\mathcal{A}^{\mathcal{H}}_{\gamma^{\mathcal{H}}}$. The crucial observation here is that \textit{the asymptotics for $\partial_{\rho}\psi$ in the region $\{M<r\leq r_{{\mathcal{I}}}\}$ depend only on the asymptotics of $T\psi$ along the event horizon, which in turn depend only on $H_0$}! We next derive sharp decay estimates (with growing $r$ weights in the error terms) for $\partial_{\rho}\psi$ up to the curve ${\gamma'}^{\mathcal{I}}$, that is in the region $\{r_{{\mathcal{I}}}\leq r\leq r_{{\gamma'}^{\mathcal{I}}}\}$.  The latter step would fail if we were to take $\alpha'=1$. We can next derive the  asymptotics for $\psi$ in $\{M<r\leq r_{{\gamma'}^{\mathcal{I}}}\}$  by integrating the estimate for  $\partial_{\rho}\psi$ in the same region backwards from ${\gamma'}^{\mathcal{I}}$.  In the last step we crucially use again that $\alpha' <1$ and that we have already computed the asymptotics for $\psi|_{{\gamma'}^{\mathcal{I}}}$. Global asymptotics follow using a dual argument from infinity and the asymptotics in $\mathcal{A}^{\mathcal{H}}_{{\gamma'}^{\mathcal{H}}}$. See Section \ref{sec:asympnonzeroconst}. Higher order logarithmic corrections are derived in Section \ref{sec:hoasymp}.

\medskip

\noindent\textit{\underline{Late-time asymptotics for Type \textbf{A} data}}

\medskip

\noindent\textbf{Resolution of difficulty 1} 

\smallskip

For Type \textbf{A} data we can derive the leading-order asymptotics of $\psi$, and crucially of $T\psi$, in the near-horizon region $\mathcal{A}^{\mathcal{H}}_{\gamma^{\mathcal{H}}}$ as in the Type \textbf{C} case, but in contrast to the Type \textbf{C} case, {we \underline{cannot} obtain independently the asymptotics in the near-infinity region}  $\mathcal{A}^{\mathcal{I}}_{\gamma^{\mathcal{I}}}$ since the first equation of \eqref{uvinitialasymptotics} does not provide exact asymptotics anymore, given that $I_0=0$. 

\medskip

\noindent\textbf{Resolution of difficulty 2} 

\smallskip

As in the Type \textbf{C} case, we can obtain the precise asymptotics for $\partial_{\rho}\psi$ in the region $\{M<r\leq r_{{\mathcal{I}}}\}$ using the asymptotics of $T\psi$ along the event horizon. However, like before, these asymptotics for $\partial_{\rho}\psi$ do not yield asymptotics for $\psi$ away from $\mathcal{A}^{\mathcal{H}}_{\gamma^{\mathcal{H}}}$. The main idea is that we can, however, derive the precise asymptotics exactly on $\gamma^{\mathcal{I}}$. In other words, equipped with the asymptotics for $\psi$ in $\mathcal{A}^{\mathcal{H}}_{\gamma^{\mathcal{H}}}$ we can next obtain asymptotics only along $\gamma^{\mathcal{I}}$ (and not to the right or to the left of $\gamma^{\mathcal{I}}$ as the asymptotics in these regions will only be derived at a later step). In order to derive asymptotics for $\psi|_{\gamma^{\mathcal{I}}}$ we need to analyze the contributions from the left side (horizon side)  and the right side (infinity side) of $\gamma^{\mathcal{I}}$. As we shall see, in order to capture the precise contributions from both sides we will need to make crucial use of $I_0=0$. It turns out that we can only capture the precise contributions at one level of differentiability higher using the following splitting identity
\begin{equation}
\Big.\frac{D}{2} r \psi\Big|_{\gamma^{\mathcal{I}}}\hspace{0.5cm}= \Big.\underbrace{r\partial_v (r\psi)\Big|_{\gamma^{\mathcal{I}}}}_{\substack{\text{contribution from} \\ \text{the right side of $\gamma^{\mathcal{I}}$}}}-\Big.\underbrace{r^2 \partial_v\psi\Big|_{\gamma^{\mathcal{I}}}}_{\substack{\text{contribution from} \\ \text{the left side of $\gamma^{\mathcal{I}}$}}}
\label{leftright}
\end{equation}
The point behind using this identity is that we can derive asymptotics for certain first-order derivative quantities independently near the horizon and infinity whereas the coupling between horizon and infinity prevents us from doing so directly for the zeroth-order derivatives. 

\textit{Contribution from the right side of }$\gamma^{\mathcal{I}}$: Recall that we want to show that $r\psi|_{{\gamma^{\mathcal{I}}}}$ decays like $\tau^{-2}$ (see Table \ref{summarytable}) and hence all error terms must decay like $\tau^{-2-\epsilon}$. Now propagating in $\mathcal{A}^{\mathcal{I}}_{\gamma^{\mathcal{I}}}$ the first of \eqref{uvinitialasymptotics} only yields an $\epsilon$ improvement for $\partial_{v}(r\psi)|_{\gamma^{\mathcal{I}}}$, that is $r\partial_{v}(r\psi)|_{\gamma^{\mathcal{I}}}\sim r\tau^{-2-\epsilon}\sim \tau^{-2-\epsilon +\alpha}$ which is not fast enough since $\alpha$ is close to 1. To circumvent this difficulty, we need to introduce a new technique which we call \textbf{the singular time inversion} (see Section \ref{sec:timeint}). Specifically, we construct the  time integral $\psi^{(1)}$ of $\psi$ that solves the wave equation  $\square_g\psi^{(1)}=0$ and satisfies $T\psi^{(1)}=\psi$. Note that if $H_0[\psi]\neq 0$ then $\psi^{(1)}$ is \emph{singular} at the horizon and in fact satisfies
\[(r-M)\cdot \partial_{\rho}\psi^{(1)} =-\frac{2}{M}\cdot H_0[\psi]+O(r-M) \]
for $r$ close to  $M$,  but is smooth away from $r=M$. Using appropriate low regularity estimates we can obtain global-in-time decay estimates for $\psi^{(1)}$ to the right of $\gamma^{\mathcal{I}}$. Moreover, using that for $\psi^{(1)}$ has a well-defined Newman--Penrose constant $I_0[\psi^{(1)}]<\infty$, we can propagate \eqref{uvinitialasymptotics} for $\psi^{(1)}$ which yields $\partial_{v}(r\psi^{(1)})|_{\gamma^{\mathcal{I}}}\sim \tau^{-2}$ and hence $r\partial_{v}(r\psi)|_{\gamma^{\mathcal{I}}}\sim r\tau^{-3}\sim \tau^{-3+\alpha}$ which shows that this term does \underline{not} contribute to the asymptotics for $r \psi|_{\gamma^{\mathcal{I}}}$.

\textit{Contribution from the left side of }$\gamma^{\mathcal{I}}$: This is the side that fully contributes to the asymptotics for $r \psi|_{\gamma^{\mathcal{I}}}$ via the term $r^2 \partial_v\psi|_{\gamma^{\mathcal{I}}}$. For we will derive the precise asymptotics for $r^2 \partial_v\psi|_{\gamma^{\mathcal{I}}}$. We make use of the \textit{improved decay rates for the conformal flux} $\mathcal{C}_{N_{\tau}^{\mathcal{I}}}[T\psi]$ (see Table \ref{summarytablenew}; Type \textbf{A}) which, upon integrating the wave equation on $N_{\tau}^{\mathcal{I}}$, yield that the asymptotics for $r^2\partial_{v}\psi|_{\gamma^{\mathcal{I}}}$ can be derived from the asymptotics of $\partial_{\rho}\psi|_{\{r=r_{\mathcal{I}}\}}$ which we already derived (and recall they depend only on $H_0[\psi]$). Hence, the asymptotics for $r^2\partial_{v}\psi|_{\gamma^{\mathcal{I}}}$ depend only on $H_{0}$ and the respective rate is $\tau^{-2}$.  

Concluding, \textit{the precise asymptotics for $r\psi_{\gamma^{\mathcal{I}}}$ depend only on the horizon charge $H_0[\psi]$} and the respective rate is $\tau^{-2}$. The estimate for the conformal flux, as above, in fact yields the asymptotics for $r^2\partial_{v}\psi$ in $\{M< r\leq r_{\gamma^{\mathcal{I}}} \}$ which we can now integrate backwards from $\gamma^{\mathcal{I}}$ (using the asymptotics for $r\psi|_{\gamma^{\mathcal{I}}}$!) to obtain the asymptotics for $r\psi$ in whole region $\{M< r\leq r_{\gamma^{\mathcal{I}}}\}$. It remains to find the asymptotics of $r\psi$  to the right of $\gamma_{\mathcal{\mathcal{I}}}$ all the way up to null infinity. For this, we use the singular time integral $\psi^{(1)}$ once again. Specifically, using the time decay estimates for $\psi^{(1)}$ and that it generically satisfies $I_0[\psi^{(1)}]\neq 0$ we derive the asymptotics of $T(r\psi^{(1)})-T(r\psi^{(1)})|_{\gamma^{\mathcal{I}}}=r\psi-r\psi|_{\gamma^{\mathcal{I}}}$ in $\mathcal{A}^{\mathcal{I}}_{\gamma^{\mathcal{I}}}$ in terms of $I_0[\psi^{(1)}]$. Combined with the above asymptotics for $r\psi|_{\gamma^{\mathcal{I}}}$ we obtain the asymptotics of $r\psi$ in $\mathcal{A}^{\mathcal{I}}_{\gamma^{\mathcal{I}}}$. Note that this shows that both the near-horizon region and the near-infinity region contribute to the asymptotics for the radiation field $r\psi|_{\I}$. This completes the derivation of the asymptotics for $\psi$ everywhere in $\mathcal{R}$. See Section \ref{sec:asympzeroconst}.

\medskip

\noindent\textit{\underline{Late-time asymptotics for Type \textbf{B} data}}

\medskip

In the case of Type \textbf{B} initial data the time integral $\psi^{(1)}$ extends smoothly to the horizon, so we can apply the same procedure as for Type \textbf{C} data to $\psi^{(1)}$ to derive the global late-time asymptotics of $\psi^{(1)}$ and of $T\psi^{(1)}=\psi$. See Section \ref{sec:AsymptoticsForTypeDPerturbations}.

\medskip

\noindent\textit{\underline{Late-time asymptotics for Type \textbf{D} data}}

\medskip
A modified variant of the methods for Type \textbf{A} data can be applied for initial data of Type \textbf{D}. In this case $\partial_{\rho}\psi$ decays faster than $\psi$ itself. In order to obtain the asymptotics for $\partial_{\rho}\psi$ we need to first obtain the asymptotics for the weighted derivative $\partial_{\rho}\big((r-M)\psi\big)$, which in fact decays as fast as $\psi$, by starting from null infinity and propagating up to $\gamma^{\mathcal{H}}$. Once we obtain the asymptotics for $\psi$ and its time derivatives then a posteriori we obtain the asymptotics for $\partial_{\rho}\psi$. See Section \ref{sec:AsymptoticsForTypeDPerturbations}.

\section{The $\h-$localized and $\I-$localized hierarchies}
\label{sec:rweightest}
In this section we will derive the main hierarchies of commuted $r^p$-weighted estimates near $\mathcal{I}^+$ and the analogous commuted ``$(r-M)^{-p}$-weighted'' estimates  $\mathcal{H}^+$ that lie at the heart of the energy and pointwise decay estimates in the subsequent sections.

\subsection{The commutator vector fields and basic estimates}
\label{sec:maineq}
We define the main quantities obtained from $\psi$ for which we will derive $r$-weighted estimates.
\begin{definition}
\label{def:horadfields}
We introduce the following \emph{higher-order radiation fields}: let $n\in \N_0$ and let $\phi=r\cdot \psi$, with $\psi$ a solution to \eqref{eq:waveequation}, then define
\begin{align*}
\Phi_{(n)}:=&\:(2D^{-1}r^2L)^n\phi,\\
\widetilde{\Phi}_{(1)}:=&\:2r(r-M)D^{-1}L\phi,\\
\underline{\Phi}_{(n)}:=&\:(2D^{-1}r^2\underline{L})^n\phi,\\
\widetilde{\underline{\Phi}}_{(1)}:=&\:2rD^{-1}\underline{L}\phi.
\end{align*}
Denote moreover $\widetilde{\underline{\Phi}}_{(0)}:=\phi$ and $\widetilde{\Phi}_{(0)}:=\phi$.
\end{definition}
The lemma below provides the equations for the higher-order radiation fields that are central in deriving the $r$-weighted estimates in a neighbourhood of $\mathcal{H}^+$ and $\mathcal{I}^+$.
\begin{lemma}
\label{lm:maineq}
Let $\psi$ be a smooth solution to \eqref{eq:waveequation}, then for all $n\in \N_0$, we have that
\begin{align}
 \label{eq:maincommeqI}
   4L\underline{L}\Phi_{(n)}=&\:Dr^{-2}\slashed{\Delta}_{\s^2}\Phi_{(n)}+[-4n r^{-1}+O(r^{-2})] L\Phi_{(n)}+ D\left[n(n+1)r^{-2} +O(r^{-3})\right]\Phi_{(n)}\\ \nonumber
    &+n\sum_{k=0}^{\max\{0,n-1\}} O(r^{-2}) \Phi_{(k)},\\
     \label{eq:maincommeqH}
      4\underline{L}L\underline{\Phi}_{(n)}=&\:Dr^{-2}\slashed{\Delta}_{\s^2}\underline{\Phi}_{(n)}+[-4M^{-2}n (r-M)+O((r-M)^{2})] \underline{L}\underline{\Phi}_{(n)}\\ \nonumber
      &+ \left[n(n+1)M^{-4}(r-M)^{2} +O((r-M)^{3})\right]\underline{\Phi}_{(n)}\\ \nonumber
    &+n\sum_{k=0}^{\max\{0,n-1\}} O((r-M)^{2}) \underline{\Phi}_{(k)}.
     \end{align}
Furthermore, 
\begin{align}
 \label{eq:maincommeqL1I}
   4L\underline{L}\widetilde{\Phi}_{(1)}=&\:Dr^{-2}\slashed{\Delta}_{\s^2}\widetilde{\Phi}_{(1)}+[-4 r^{-1}+O(r^{-2})] L\widetilde{\Phi}_{(1)}+ D\left[2r^{-2} +O(r^{-3})\right]\widetilde{\Phi}_{(1)}\\ \nonumber
    &+[2Mr^{-2}+O(r^{-3})]\phi+[Mr^{-2}+O(r^{-3})] \slashed{\Delta}_{\s^2}\phi,\\
     \label{eq:maincommeqL1H}
      4\underline{L}L\widetilde{\underline{\Phi}}_{(1)}=&\:Dr^{-2}\slashed{\Delta}_{\s^2}\widetilde{\underline{\Phi}}_{(1)}+[-4M^{-2} (r-M)+O((r-M)^{2})] \underline{L}\widetilde{\underline{\Phi}}_{(1)}\\ \nonumber
      &+ \left[2M^{-4}(r-M)^{2} +O((r-M)^{3})\right]\widetilde{\underline{\Phi}}_{(1)}\\ \nonumber
    &+[-2M^{-2}(r-M)^2+O((r-M)^{3})]\phi+[-M^{-2}(r-M)^2+O((r-M)^{3})] \slashed{\Delta}_{\s^2}\phi.
     \end{align}
\end{lemma}
\begin{proof}
The equations \eqref{eq:maincommeqI} and \eqref{eq:maincommeqH} can easily be derived inductively, using that the $n=0$ case follows directly from rewriting \eqref{eq:waveequation} as the following equation for $\phi$:
\begin{equation}
\label{eq:eqradield}
4\underline{L}L\phi=Dr^{-2}\slashed{\Delta}_{\s^2}\phi-\frac{DD'}{r}\phi.
\end{equation}
\end{proof}

In the following proposition, we establish finiteness of certain limits of the higher-order radiation fields $\Phi_{(n)}$ at $\mathcal{I}^+$.

\begin{proposition}
\label{prop:radfieldsinf}
Let $\psi$ be a smooth solution to \eqref{eq:waveequation}. 
Let $n\in \N$ and assume that
\begin{equation*}
\int_{\Sigma_0}J^T[\psi]\cdot \mathbf{n}_{\Sigma_0}\,d\mu_{\Sigma_0}<\infty.
\end{equation*}
\begin{itemize}
\item[$\rm (i)$]
If we assume that
\begin{equation*}
\sum_{0\leq k\leq n}\left[\int_{\s^2} |\slashed{\Delta}_{\s^2}^{n-k} \Phi_{(k)}|^2|_{{N}_{0}^{\mathcal{I}}}\,d\omega\right](v)<\infty
\end{equation*}
then for all $u\geq u_0$, we have that there exists a constant $C_u=C_u(M,N,{\Sigma_0},u)>0$, such that
\begin{equation}
\label{eq:udepradfieldsinft}
\sum_{0\leq k\leq n}\sup_{u_0\leq u'\leq u}\left[\int_{\s^2} |\slashed{\Delta}_{\s^2}^{n-k} \Phi_{(k)}|^2|_{{N}^{\mathcal{I}}}\,d\omega\right](v)<C_u\sum_{0\leq k\leq n}\left[\int_{\s^2} |\slashed{\Delta}_{\s^2}^{n-k} \Phi_{(k)}|^2|_{{N}_{0}^{\mathcal{I}}}\,d\omega\right](v).
\end{equation}
\item[$\rm (ii)$]
If we make the stronger assumption that for all $0\leq k \leq n$:
\begin{equation*}
\lim_{v\to \infty}\left[\int_{\s^2} |\slashed{\Delta}_{\s^2}^{n-k} \Phi_{(k)}|^2|_{{N}_{0}^{\mathcal{I}}}\,d\omega\right](v)<\infty,
\end{equation*}
then, for all $u\geq u_0$ and $0\leq k \leq n$, we have that
\begin{equation}
\label{eq:limitradfieldsinft}
\lim_{v\to \infty}\left[\int_{\s^2} |\slashed{\Delta}_{\s^2}^{n-k} \Phi_{(k)}|^2|_{{N}_{u}^\mathcal{I}}\,d\omega\right](v)<\infty.
\end{equation}
\end{itemize}
\end{proposition}
\begin{proof}
We will prove inductively that \eqref{eq:limitradfieldsinft} holds. The $n=0$ case follows directly from Proposition 3.4 of \cite{paper1}, so we will omit the derivation here.  Let $N\in \N_0$ and suppose \eqref{eq:limitradfieldsinft} holds for all $0\leq n \leq N$. Then, by applying the fundamental theorem of calculus together with \eqref{eq:maincommeqI}, we have that
\begin{equation}
\label{eq:transpest}
\begin{split}
\Phi_{(N+1)}(u,v,\theta,\varphi)=&\:\Phi_{(N+1)}(u_0,v,\theta,\varphi)+\int_{u_0}^u \underline{L}(2r^2D L\Phi_{(N)})(u',v,\theta,\varphi)\,du'\\
=&\:\Phi_{(N+1)}(u_0,v,\theta,\varphi)\\
&+\int_{u_0}^u \left[O(r^{-1})\Phi_{(N+1)}+O(r^0)\slashed{\Delta}_{\s^2}\Phi_{(N)}+O(r^0)\sum_{k=0}^N \Phi_{(k)}\right](u',v,\theta,\varphi)\,du'
\end{split}
\end{equation}
By applying a Gr\"onwall inequality in $u$, we can therefore estimate
\begin{equation*}
\sup_{u_0\leq u'\leq u}|\Phi_{(N+1)}|^2(u',v,\theta,\varphi)\leq C(u)\left( |\Phi_{(N+1)}(u_0,v,\theta,\varphi)|^2+ \sup_{u_0\leq u' \leq u}\left[|\slashed{\Delta}_{\s^2}\Phi_{(N)}|^2+ \sum_{k=0}^N |\Phi_{(k)}|^2\right](u',v,\theta,\varphi)\right).
\end{equation*}
The above equation integrated over $\s^2$, together with \eqref{eq:transpest}, gives the following estimate:
\begin{equation*}
\begin{split}
\Bigg|&\int_{\s^2}|\Phi_{(N+1)}|^2(u,v,\theta,\varphi)d\omega-\int_{\s^2}|\Phi_{(N+1)}|^2(u_0,v,\theta,\varphi)d\omega\\
&-\int_{\s^2}\left(\int_{u_0}^u \left[O(r^0)\slashed{\Delta}_{\s^2}\Phi_{(N)}+O(r^0)\sum_{k=0}^N \Phi_{(k)}\right](u',v,\theta,\varphi)\,du\right)^2\,d\omega\Bigg|\\
\leq& C(u) r^{-1}(u,v) \sup_{u_0\leq u' \leq u}\int_{\s^2}|\Phi_{(N+1)}(u_0,v,\theta,\varphi)|^2+ |\slashed{\Delta}_{\s^2}\Phi_{(N)}|^2(u',v,\theta,\varphi)+ \sum_{k=0}^N |\Phi_{(k)}|^2(u',v,\theta,\varphi)\,d\omega.
\end{split}
\end{equation*}
By the inductive step, together with the additional fact that the equation for the inductive step immediately holds for $\Phi_{(N)}$ replaced by $\slashed{\Delta}_{\s^2}\Phi_{(N)}$, since $[\slashed{\Delta}_{\s^2},\square_g]=0$, we can infer
that \eqref{eq:udepradfieldsinft} holds and morever, the right-hand side of the equation above goes to zero as $v\to \infty$ and therefore, using once more the inductive step (commuted with $\slashed{\Delta}_{\s^2}$), we conclude that
\begin{equation*}
\lim_{v\to \infty}\int_{\s^2}|\Phi_{(N+1)}|^2(u,v,\theta,\varphi)\,d\omega<\infty,
\end{equation*}
which allows us to obtain \eqref{eq:limitradfieldsinft}.
\end{proof}

\begin{remark}
Since $\widetilde{\Phi}_{(1)}=(1+Mr^{-1})\Phi_{(1)}$, \eqref{eq:udepradfieldsinft} and \eqref{eq:limitradfieldsinft} with $n=1$ automatically hold when $\Phi_{(1)}$ is replaced by $\widetilde{\Phi}_{(1)}$.
\end{remark}

\subsection{The key identities}
\label{sec:mainest1}
In order to establish the $r^p$- and $(r-M)^{-p}$-weighted estimates below, we first derive the following $r^p$- and $(r-M)^{-p}$-weighted \emph{identities} in the following \textbf{key lemma}:
\begin{lemma}
\label{lm:rweightidentitiess2v1} Let $\psi$ be a smooth solution to \eqref{eq:waveequation} and $p\in \R$. Then the following identities hold for all $n\in \N$:
\begin{equation}
\label{eq:mainidentitysphereinf}
\begin{split}
\int_{\s^2}&\underline{L}\left(r^p (L\Phi_{(n)})^2\right)\,d\omega + \frac{1}{2}\int_{\s^2}[(p+4n)r^{p-1}+O(r^{p-2})](L\Phi_{(n)})^2\,d\omega\\ 
&+\frac{1}{8}\int_{\s^2} (2-p)r^{p-3}D\left(|\slashed{\nabla}_{\s^2}\Phi_{(n)}|^2-n(n+1)\Phi_{(n)}^2\right)\,d\omega\\ 
=&\:  \int_{\s^2} \frac{1}{4}L\left(n(n+1)r^{p-2}\Phi_{(n)}^2-r^{p-2}|\slashed{\nabla}_{\s^2}\Phi_{(n)}|^2\right)\,d\omega\\ 
+&\int_{\s^2} O(r^{p-3}) \Phi_{(n)}\cdot L\Phi_{(n)}+O(r^{p-3})\slashed{\Delta}_{\s^2}\Phi_{(n)}\cdot L\Phi_{(n)}+n\sum_{k=0}^{\max\{0,n-1\}}\int_{\s^2} O(r^{p-2}) \Phi_{(k)}\cdot L\Phi_{(n)}\,d\omega,
\end{split}
\end{equation}
and
\begin{equation}
\label{eq:mainidentityspherehor}
\begin{split}
\int_{\s^2}&{L}\left((r-M)^{-p} (\underline{L}\underline{\Phi}_{(n)})^2\right)\,d\omega \\ 
&+ \frac{1}{2}\int_{\s^2}(p+4n)M^{-2}[(r-M)^{1-p}+O((r-M)^{2-p})](\underline{L}\underline{\Phi}_{(n)})^2\,d\omega\\ 
&+\frac{1}{8}\int_{\s^2} (2-p)M^{-6}(r-M)^{3-p}\left(|\slashed{\nabla}_{\s^2}\underline{\Phi}_{(n)}|^2-n(n+1)\underline{\Phi}_{(n)}^2\right)\,d\omega\\ 
=&\:  \int_{\s^2} \frac{1}{4}\underline{L}\left(n(n+1)M^{-4}(r-M)^{2-p}\underline{\Phi}_{(n)}^2-M^{-4}(r-M)^{2-p}|\slashed{\nabla}_{\s^2}\underline{\Phi}_{(n)}|^2\right)\,d\omega\\ 
+&\int_{\s^2} O((r-M)^{3-p}) \underline{\Phi}_{(n)}\cdot \underline{L}\underline{\Phi}_{(n)}+O((r-M)^{3-p})\slashed{\Delta}_{\s^2}\underline{\Phi}_{(n)}\cdot \underline{L}\underline{\Phi}_{(n)}\,d\omega\\
&+n\int_{\s^2}\sum_{k=0}^{\max\{0,n-1\}}\int_{\s^2} O((r-M)^{2-p}) \underline{\Phi}_{(k)}\cdot \underline{L}\underline{\Phi}_{(n)}\,d\omega.
\end{split}
\end{equation}
\end{lemma}
\begin{proof}
The identities follow by applying \eqref{eq:maincommeqI} and \eqref{eq:maincommeqH} and integrating by parts on $\s^2$.
\end{proof}

We will make use of the following orthogonal projections
\begin{equation*}
P_{\ell},P_{\leq \ell}, P_{\geq \ell}:L^2(\s^2)\to L^2(\s^2),
\end{equation*}
with $\ell\in \N_0$, which are defined as follows: let $f\in L^2(\s^2)$, then
\begin{align*}
P_{\ell}f=\:f_{\ell},\quad P_{\leq \ell}f=\sum_{\ell'=0}^{\ell}f_{\ell'},\quad P_{\geq \ell}f=\sum_{\ell'=\ell}^{\infty}f_{\ell'},
\end{align*}
where $f_{\ell'}$ is the $\ell'$-th angular mode.

In the lemma below, we prove similar identities for the orthogonal projection $P_1\widetilde{\Phi}_{(1)}$, but we exploit \textbf{crucial cancellations} occurring when we apply standard Poincar\'e inequalities on $\s^2$.
\begin{lemma}
\label{lm:rweightidentitiess2v2} Let $\psi$ be a smooth solution to \eqref{eq:waveequation}. 
The following identities hold for all $p\in \R$:
\begin{align}
\label{eq:rweightidentitiess2v2I}
\int_{\s^2}&\underline{L}\left(r^p (LP_1\widetilde{\Phi}_{(1)})^2\right)\,d\omega + \frac{1}{2}\int_{\s^2}[(p+4)r^{p-1}+O(r^{p-2})](LP_1\widetilde{\Phi}_{(1)})^2\,d\omega\\ \nonumber
=&\:  \int_{\s^2} O(r^{p-3}) P_1\widetilde{\Phi}_{(1)}\cdot LP_1\widetilde{\Phi}_{(1)}+\int_{\s^2} O(r^{p-3}) P_1\phi\cdot LP_1\widetilde{\Phi}_{(1)}\,d\omega,\\
\label{eq:rweightidentitiess2v2H}
\int_{\s^2}&{L}\left((r-M)^{-p} (\underline{L}P_1\widetilde{\underline{\Phi}}_{(1)})^2\right)\,d\omega+ \frac{1}{2}\int_{\s^2}(p+4)M^{-2}[(r-M)^{1-p}+O((r-M)^{2-p})](\underline{L}P_1\widetilde{\underline{\Phi}}_{(1)})^2\,d\omega\\ \nonumber
=&\:  \int_{\s^2} O((r-M)^{3-p}) P_1\widetilde{\underline{\Phi}}_{(1)}\cdot \underline{L}P_1\widetilde{\underline{\Phi}}_{(1)}\,d\omega+\int_{\s^2} O((r-M)^{3-p}) P_1\phi\cdot \underline{L}P_1\widetilde{\underline{\Phi}}_{(1)}\,d\omega.
\end{align}
\end{lemma}
\begin{proof}
The proof proceeds exactly as the proof of Lemma \ref{lm:rweightidentitiess2v1} with $n=1$, but we additionally use the standard Poincar\'e inequalities on $\s^2$ together with the fact that
\begin{equation*}
\slashed{\Delta}_{\s^2}P_1\psi+2P_1\psi=0
\end{equation*}
to arrive at additional cancellations resulting in additional factors of $\frac{1}{r}$ and $r-M$ on the right-hand sides of the identities for $P_1\widetilde{\Phi}_{(1)}$ and $P_1\widetilde{\underline{\Phi}}_{(1)}$, respectively.
\end{proof}

\subsection{The main commuted hierarchies}
\label{sec:mainest}
We have the following
\begin{proposition}
\label{prop:generalrpest} Let $\psi$ be a smooth solution to \eqref{eq:waveequation}. 
Fix $n\in \N_0$ and assume that for all $0\leq k\leq \min\{n-1,0\}$ and $0\leq j \leq n-k$
\begin{equation}
\label{assm:id1}
\lim_{v\to \infty}\left(\int_{\s^2}|\snabla_{\s^2}\slashed{\Delta}_{\s^2}^jP_{\geq 1}\Phi_{(k)}|^2\,d\omega\right)(u_0,v)<\infty.
\end{equation}

Let $\epsilon>0$ be arbitrarily small, then there exists $r_{\mathcal{I}}>0$ sufficiently large, such that for $p\in (-4n,2]$ and for all $0\leq u_1\leq u_2$:
\begin{equation}
\label{eq:rpestinf}
\begin{split}
(1-\epsilon)\int_{{N}^{\mathcal{I}}_{u_2}}& r^p(LP_{\geq 1}\Phi_{(n)})^2\,d\omega dv+ \frac{1}{2}(1-\epsilon)\int_{u_1}^{u_2} \int_{{N}^{\mathcal{I}}_u}(p+4n)r^{p-1}(LP_{\geq 1}\Phi_{(n)})^2\,d\omega dv du\\
&+\frac{1}{4}\int_{\mathcal{I}^+(u_1,u_2)}\left[r^{p-2} |\snabla_{\s^2}P_{\geq 1}\Phi_{(n)}|^2-n(n+1)r^{p-2}(P_{\geq 1}\Phi_{(n)})^2\right]\,d\omega du\\
&+\frac{1}{8}\int_{u_1}^{u_2} \int_{{N}^{\mathcal{I}}_u}(2-p)r^{p-3}D\left(|\snabla_{\s^2}P_{\geq 1}\Phi_{(n)}|^2-n(n+1)P_{\geq 1}\Phi_{(n)}^2\right)\,d\omega dv du\\
\leq&\: C\int_{{N}^{\mathcal{I}}_{u_1}}r^p(LP_{\geq 1}\Phi_{(n)})^2\,d\omega dv+ C\sum_{k\leq n}\int_{\Sigma_{u_1}} J^T[T^k\psi]\cdot \mathbf{n}_{u_1}\,d\mu_{\Sigma_{u_1}},
\end{split}
\end{equation}
where $C=C(n,M, {\Sigma_0},r_{\mathcal{H}})>0$ is a constant and we can take $r_{\mathcal{I}}=(p+4n)^{-1}R_0(n,M)>0$. 

Furthermore, there exists $r_{\mathcal{H}}>M$, with $r_{\mathcal{H}}-M$ suitably small, such that for $p\in (-4n,2]$ and for all $0\leq u_1\leq u_2$:
\begin{equation}
\label{eq:rpesthor}
\begin{split}
(1-\epsilon)\int_{{N}^{\mathcal{H}}_{v_2}}& (r-M)^{-p}(\underline{L}P_{\geq 1}\underline{\Phi}_{(n)})^2\,d\omega du+ \frac{1}{2}(1-\epsilon)\int_{v_1}^{v_2} \int_{{N}^{\mathcal{H}}_u}(p+4n)(r-M)^{1-p}(\underline{L}P_{\geq 1}\underline{\Phi}_{(n)})^2\,d\omega du dv\\
&+\frac{1}{4}M^{-4}\int_{\mathcal{H}^+(v_1,v_2)}\left[(r-M)^{2-p} |\snabla_{\s^2}\underline{\Phi}_{(n)}|^2-n(n+1)(r-M)^{2-p}\underline{\Phi}_{(n)}^2\right]\,d\omega dv\\
&+\frac{1}{8}\int_{v_1}^{v_2} \int_{{N}^{\mathcal{H}}_v}(2-p)(r-M)^{3-p}M^{-6}\left(|\snabla_{\s^2}P_{\geq 1}\underline{\Phi}_{(n)}|^2-n(n+1)P_{\geq 1}\underline{\Phi}_{(n)}^2\right)\,d\omega du dv\\
\leq&\: C\int_{{N}^{\mathcal{H}}_{v_1}}r^p(\underline{L}P_{\geq 1}\underline{\Phi}_{(n)})^2\,d\omega du+ C\sum_{k\leq n}\int_{\Sigma_{v_1}} J^T[T^k\psi]\cdot \mathbf{n}_{v_1}\,d\mu_{\Sigma_{v_1}},
\end{split}
\end{equation}
where $C=C(n,D,r_{\mathcal{H}})>0$ is a constant and we can take $(r_{\mathcal{H}}-M)^{-1}=(p+4n)^{-1}M^{-2}R_0(n,M)>0$. 
\end{proposition}
\begin{proof}
Observe first of all that the assumption \eqref{assm:id1} together with the smoothness assumption of the initial data on $\Sigma_0$ imply that
\begin{equation*}
\int_{\Sigma_0} J^T[\psi_{\geq 1}]\cdot \mathbf{n}_{\Sigma_0}\,d\mu_{\Sigma_0}<\infty.
\end{equation*}
We can therefore appeal to Proposition \ref{prop:radfieldsinf} with regards to the limiting behaviour of $P_{\geq 1}\Phi_{(k)}$ at $\mathcal{I}^+$.

We will first derive the estimate \eqref{eq:rpestinf}. In all the estimates in this proof, we assume for notational convenience that $\int_{\s^2}\psi\,d\omega =0$. We introduce a smooth cut-off function $\chi: \R \to \R$ such that $\chi(r)=0$ for all $r\leq r_{\mathcal{I}}$ and $\chi(r)=1$ for all $r\geq r_{\mathcal{I}}+M$. We will choose $r_{\mathcal{I}}$ appropriately large.

We now integrate both sides of \eqref{eq:mainidentitysphereinf} in the $u$ and $v$ directions to obtain
\begin{equation}
\label{eq:maineqrpestinf}
\begin{split}
\int_{{N}^{\mathcal{I}}_{u_2}}& r^p(L(\chi\Phi_{(n)}))^2\,d\omega dv+ \frac{1}{2}\int_{u_1}^{u_2} \int_{{N}^{\mathcal{I}}_u}(p+4n)r^{p-1}(L(\chi\Phi_{(n)}))^2\,d\omega dv du\\
&+\frac{1}{4}\int_{\mathcal{I}^+(u_1,u_2)}\left[r^{p-2} |\snabla_{\s^2}\Phi_{(n)}|^2-n(n+1)r^{p-2}\Phi_{(n)}^2\right]\,d\omega du\\
&+\frac{1}{8}\int_{u_1}^{u_2} \int_{{N}^{\mathcal{I}}_u}(2-p)r^{p-3}D\chi^2\left(|\snabla_{\s^2}\Phi_{(n)}|^2-n(n+1)\Phi_{(n)}^2\right)\,d\omega dv du\\
=&\: \int_{{N}^{\mathcal{I}}_{u_1}} r^p(L(\chi\Phi_{(n)}))^2\,d\omega dv\\
&+J_1+J_2+J_3+\sum_{|\alpha|\leq 1} \int_{u_1}^{u_2} \int_{{N}^{\mathcal{I}}_u}r^{p-2}\underline{L} (\chi\Phi_{(n)}) \cdot R_{\chi}[\partial^{\alpha}\Phi_{(n)}]\,d\omega dudv,
\end{split}
\end{equation}
where we use the notation $R_{\chi}[f]$ for terms that are compactly supported in $r_{\mathcal{I}}\leq r\leq r_{\mathcal{I}}+M$ and are linear in the function $f$, and we define
\begin{align*}
J_1:=&\: \int_{u_1}^{u_2} \int_{{N}^{\mathcal{I}}_u}O(r^{p-2})(L(\chi \Phi_{(n)}))^2\,d\omega du dv,\\
J_2:=&\: \int_{u_1}^{u_2} \int_{{N}^{\mathcal{I}}_u}O(r^{p-3})\chi \Phi_{(n)}\cdot L(\chi \Phi_{(n)})+ O(r^{p-3})\slashed{\Delta}_{\s^2}(\chi\Phi_{(n)})\cdot L(\chi \Phi_{(n)})\,d\omega du dv,\\
J_3:=&\: n\sum_{k=0}^{n-1}\int_{u_1}^{u_2} \int_{{N}^{\mathcal{I}}_u} O(r^{p-2}) \chi \Phi_{(k)}\cdot L(\chi \Phi_{(n)})\,d\omega dv du.
\end{align*}
In order to obtain \eqref{eq:maineqrpestinf} we used that $r^{p-2}\chi^2\Phi_{(n)}^2$ and $r^{p-2}\chi^2\Phi_{(n)}^2$ vanish on $\{r=r_{\mathcal{I}}\}$.

First of all, by \eqref{eq:iledwayps} and the compactness of the support of $R_{\chi}$ it follows that there exists a constant $C(M,{\Sigma_0},r_{\mathcal{I}})>0$ such that
\begin{equation}
\label{est:cptsuppremainder}
\sum_{|\alpha|\leq 1} \int_{u_1}^{u_2} \int_{{N}^{\mathcal{I}}_u}r^{p-2}\underline{L} (\chi\Phi_{(n)}) \cdot R_{\chi}[\partial^{\alpha}\Phi_{(n)}]\,d\omega dudv\leq C\sum_{k\leq n}\int_{\Sigma_{u_1}} J^T[T^k\psi]\cdot \mathbf{n}_{u_1}\,d\mu_{\Sigma_{u_1}}.
\end{equation}
The strategy for the remainder of the proof will therefore be to absorb $J_1+J_2+J_3$ into the second term on the left-hand side of \eqref{eq:maineqrpestinf}.\\
\\
\underline{\textbf{Estimating $J_1$}}\\
\\
We apply Young's inequality with weights in $\epsilon$ to estimate
\begin{equation*}
|J_1|\leq \int_{u_1}^{u_2} \int_{{N}^{\mathcal{I}}_u} \epsilon (p+4n) r^{p-1} (L(\chi \Phi_{(n)}))^2+C\epsilon^{-1}(p+4n)^{-1}r^{p-5}\chi^2 \Phi_{(n)}^2 \,d\omega du dv,
\end{equation*}
for some $\epsilon>0$ to be chosen suitably small. We absorb the first term into the left-hand side of \eqref{eq:maineqrpestinf}. We apply \eqref{eq:hardyinf} to further estimate
\begin{equation*}
\begin{split}
\epsilon^{-1}(p+4n)^{-1}&\int_{u_1}^{u_2} \int_{{N}^{\mathcal{I}}_u}r^{p-5}\chi^2 \Phi_{(n)}^2\,d\omega dudv \leq C\epsilon^{-1}(p+4n)^{-1} \int_{{N}^{\mathcal{I}}_u}r^{p-3} (L(\chi\Phi_{(n)}))^2\,d\omega dudv\\
\leq &\: C\epsilon^{-1}(p+4n)^{-1}r_{\mathcal{I}}^{-2}\int_{{N}^{\mathcal{I}}_u}r^{p-1} (L(\chi\Phi_{(n)}))^2\,d\omega dudv.
\end{split}
\end{equation*}
For $r_{\mathcal{I}}^2>0$ suitably large (depending linearly on $(p+4n)^{-1}$), we can therefore also absorb the term above in to the second integral on the left-hand side of \eqref{eq:maineqrpestinf}.\\
\\
\underline{\textbf{Estimating $J_2$}}\\
\\

To estimate $J_2$, we first consider $O(r^{p-3})\chi \Phi_{(n)}\cdot L(\chi \Phi_{(n)})$ and apply Young's inequality to obtain:
\begin{equation*}
\begin{split}
\int_{u_1}^{u_2} \int_{{N}^{\mathcal{I}}_u}O(r^{p-3})\chi \Phi_{(n)}\cdot L(\chi \Phi_{(n)}) \,d\omega du dv\leq&\: \int_{u_1}^{u_2} \int_{{N}^{\mathcal{I}}_u}\epsilon (p+4n)r^{p-1}(L(\chi \Phi_{(n)})^2 \,d\omega du dv\\
&+C {\epsilon}^{-1} (p+4n)^{-1}\int_{u_1}^{u_2} \int_{{N}^{\mathcal{I}}_u} r^{p-5}\chi^2\Phi_{(n)}^2\,d\omega du dv.
\end{split}
\end{equation*}
The first term on the right-hand side can be absorbed into the left-hand side of \eqref{eq:maineqrpestinf} and the second term on the right-hand side can be absorbed into $J_1$, as above.

In order to estimate $O(r^{p-3})\chi \slashed{\Delta}_{\s^2}\Phi_{(n)}\cdot L(\chi \Phi_{(n)})$ we first rearrange the terms in \eqref{eq:maincommeqI} to obtain:
\begin{equation}
\label{eq:equationforsphlaplacian}
\frac{1}{2}D \chi \slashed{\Delta}_{\s^2} \Phi_{(n)}=2r^2\underline{L}L(\chi \Phi_{(n)})+(2nr+O(r^0))L(\chi \Phi_{(n)})+\sum_{k=0}^n O(r^0) \chi \Phi_{(k)}+\sum_{|\alpha|\leq 1}R_{\chi}[\partial^{\alpha}\Phi_{(n)}],
\end{equation}
so that
\begin{equation*}
\begin{split}
\left|\int_{u_1}^{u_2} \int_{{N}^{\mathcal{I}}_u}O(r^{p-3})\chi\slashed{\Delta}_{\s^2}  \Phi_{(n)}\cdot L(\chi \Phi_{(n)}) \,d\omega du dv\right|\leq&\: \left|\int_{u_1}^{u_2} \int_{{N}^{\mathcal{I}}_u}O(r^{p-1}) \underline{L}L(\chi \Phi_{(n)}) \cdot L(\chi \Phi_{(n)}) \,d\omega du dv\right|\\
&+\left|\int_{u_1}^{u_2} \int_{{N}^{\mathcal{I}}_u}O(r^{p-2}) (L(\chi \Phi_{(n)}))^2 \,d\omega du dv\right|\\
&+\left|\sum_{k=0}^n\int_{u_1}^{u_2} \int_{{N}^{\mathcal{I}}_u}O(r^{p-3}) \chi\Phi_{(k)}\cdot L(\chi \Phi_{(n)}) \,d\omega du dv\right|\\
&+ C\sum_{k\leq n}\int_{\Sigma_{u_1}} J^T[T^k\psi]\cdot \mathbf{n}_{u_1}\,d\mu_{\Sigma_{u_1}}
\end{split}
\end{equation*}
Note that we can absorb the second integral on the right-hand side into $J_1$ and we can group the third integral with the $O(r^{p-3})\chi \Phi_{(n)}\cdot L(\chi \Phi_{(n)})$ term of $J_2$ and with $J_3$ (which we estimate below). It remains to estimate the integral of $O(r^{p-1}) \underline{L}L(\chi \Phi_{(n)}) \cdot L(\chi \Phi_{(n)})$. 

We first integrate by parts in the $\underline{L}$ direction:
\begin{equation*}
\begin{split}
\int_{u_1}^{u_2}& \int_{{N}^{\mathcal{I}}_u}O(r^{p-1}) \underline{L}L(\chi \Phi_{(n)}) \cdot L(\chi \Phi_{(n)}) \,d\omega du dv= \int_{u_1}^{u_2} \int_{{N}^{\mathcal{I}}_u}\underline{L}(O(r^{p-1}) (L(\chi \Phi_{(n)}))^2) \,d\omega du dv\\
&+(p-1)\int_{u_1}^{u_2} \int_{{N}^{\mathcal{I}}_u}O(r^{p-2}) (L(\chi \Phi_{(n)}))^2 \,d\omega du dv\\
=&\: \int_{{N}^{\mathcal{I}}_{u_2}}O(r^{p-1}) (L(\chi \Phi_{(n)}))^2 \,d\omega dv- \int_{{N}^{\mathcal{I}}_{u_1}}O(r^{p-1}) (L(\chi \Phi_{(n)}))^2 \,d\omega dv\\
&+(p-1)\int_{u_1}^{u_2} \int_{{N}^{\mathcal{I}}_u}O(r^{p-2}) (L(\chi \Phi_{(n)}))^2 \,d\omega du dv.
\end{split}
\end{equation*}
We can absorb the third term on the very right-hand side above into $J_1$ and we can absorb the absolute values of the remaining terms into the integrals over ${N}^{\mathcal{I}}_{u_2}$ and ${N}^{\mathcal{I}}_{u_1}$  that appear in \eqref{eq:maineqrpestinf} (after taking $r_{\mathcal{I}}>0$ suitably large).\\
\\
\underline{\textbf{Estimating $J_3$}}\\
\\
If $n=0$, there is nothing to estimate. Suppose therefore that $n\geq 1$. \textbf{It is only in this step that we will make use of the assumption} $\int_{\s^2}\psi\,d\omega=0$. That is to say, using this assumption it follows that there exist functions $f_{(k)}$, with $0\leq k\leq n$, such that
\begin{equation*}
\slashed{\Delta}_{\s^2}f_{(k)}=\Phi_{(k)}.
\end{equation*}
for all $0\leq k\leq n$. We can then estimate $J_3$ as by integrating by parts twice on $\s^2$ and then applying once more \eqref{eq:equationforsphlaplacian} to obtain:
\begin{equation*}
\begin{split}
J_3= &\: \sum_{k=0}^{n-1}\int_{u_1}^{u_2} \int_{{N}^{\mathcal{I}}_u} O(r^{p-2}) \chi \slashed{\Delta}_{\s^2}\Phi_{(k)}\cdot L(\chi f_{(n)})\,d\omega dv du \\
=&\: \sum_{k=0}^{n-1}\int_{u_1}^{u_2} \int_{{N}^{\mathcal{I}}_u} O(r^{p}) \chi \underline{L}L\Phi_{(k)}\cdot L(\chi f_{(n)})+O(r^{p-1})\chi L\Phi_{(k)}\cdot L(\chi f_{(n)})\\
&+\sum_{m=0}^{k-1} O(r^{p-2})\chi \Phi_{(m)}\cdot L(\chi f_{(n)}) \,d\omega dv du.
\end{split}
\end{equation*}
We integrate by parts in the $\underline{L}$ direction to obtain:
\begin{equation*}
\begin{split}
\sum_{k=0}^{n-1}\int_{u_1}^{u_2} \int_{{N}^{\mathcal{I}}_u} O(r^{p}) \chi \underline{L}L\Phi_{(k)}\cdot L(\chi f_{(n)})\,d\omega dudv=&\:\sum_{k=0}^{n-1}\int_{u_1}^{u_2} \int_{{N}^{\mathcal{I}}_u} \underline{L}\left(O(r^{p}) \chi L\Phi_{(k)}\cdot L(\chi f_{(n)})\right)\,d\omega dudv\\
&+p\sum_{k=0}^{n-1}\int_{u_1}^{u_2} \int_{{N}^{\mathcal{I}}_u} O(r^{p-1}) \chi L\Phi_{(k)}\cdot L(\chi f_{(n)})\,d\omega dudv+\ldots\\
=&\: \int_{{N}_{u_2}} O(r^{p}) \chi L\Phi_{(k)}\cdot L(\chi f_{(n)})\,d\omega dv\\
&+\int_{{N}_{u_1}} O(r^{p}) \chi L\Phi_{(k)}\cdot L(\chi f_{(n)})\,d\omega dv\\
&+p\sum_{k=0}^{n-1}\int_{u_1}^{u_2} \int_{{N}^{\mathcal{I}}_u} O(r^{p-1}) \chi L\Phi_{(k)}\cdot L(\chi f_{(n)})\,d\omega dudv+\ldots,
\end{split}
\end{equation*}
where for the sake of brevity we employ the schematic notation $\ldots$ to denote all integral terms that are supported in $r_{\mathcal{I}}\leq r\leq r_{\mathcal{I}}+M$. 

Note that by applying the standard Poincar\'e inequality on $\s^2$, we can estimate
\begin{equation}
\label{eq:poincaref}
\int_{\s^2} f_{(n)}^2\,d\omega\leq \frac{1}{2} \int_{\s^2} |\slashed{\nabla}_{\s^2}f_{(n)}|^2\,d\omega\leq \frac{1}{4}\int_{\s^2} \Phi_{(n)}^2\,d\omega
\end{equation}
and hence,
\begin{align*}
\left|\int_{{N}_{u}} O(r^{p}) \chi L\Phi_{(k)}\cdot L(\chi f_{(n)})\,d\omega dv\right|\leq& \:\int_{{N}_{u}} \epsilon r^{p}  (L\chi\Phi_{(n)})^2 +C\epsilon^{-1}r^{p-4} \chi^2 (\Phi_{(k+1)})^2\,d\omega dv,\\
\left|\int_{u_1}^{u_2} \int_{{N}^{\mathcal{I}}_u} O(r^{p-1}) \chi L\Phi_{(k)}\cdot L(\chi f_{(n)})\,d\omega dudv \right| \leq& \:\int_{u_1}^{u_2} \int_{{N}_{u}} \epsilon r^{p-1}  (L\chi\Phi_{(n)})^2 +C\epsilon^{-1}r^{p-5} \chi^2 (\Phi_{(k+1)})^2\,d\omega dv.
\end{align*}

By applying \eqref{eq:hardyinf}, with $r_1=r_{\mathcal{I}}$ suitably large and $r_2=\infty$, a number $n-k$ times to the second term on the right-hand side above and moreover \eqref{eq:udepradfieldsinft} (which holds by the assumption \eqref{assm:id1}) we can conclude that all the boundary terms appearing in \eqref{eq:hardyinf} vanish, so that we can further estimate for $p<3$:
\begin{equation*}
\begin{split}
\int_{{N}_{u}} \epsilon r^{p}  (L\chi\Phi_{(n)})^2 +C\epsilon^{-1}r^{p-4} \chi^2 (\Phi_{(k+1)})^2\,d\omega dv \leq &\:C\epsilon\int_{{N}_{u}}  r^{p}  (L\chi\Phi_{(n)})^2\,d\omega dv+C\sum_{m\leq n}\int_{\Sigma_{u}} J^T[T^m\psi]\cdot n_{u}\,d\mu_{\Sigma_{u}},
\end{split}
\end{equation*}
where we take $u=u_1$ or $u=u_2$.

Similarly,
\begin{equation*}
\int_{u_1}^{u_2} \int_{{N}_{u}} \epsilon r^{p-1}  (L\chi\Phi_{(n)})^2 +C\epsilon^{-1}r^{p-5} \chi^2 (\Phi_{(k+1)})^2\,d\omega dv \leq C\epsilon \int_{u_1}^{u_2}\int_{{N}_{u}}  r^{p}  (L\chi\Phi_{(n)})^2\,d\omega dvdu +\ldots.
\end{equation*}

Putting the above estimates together, we therefore obtain:
\begin{equation*}
\begin{split}
|J_3|\leq &\: C\int_{u_1}^{u_2} \int_{{N}^{\mathcal{I}}_u}\epsilon r^{p-1}(L\chi\Phi_{(n)})^2+C\epsilon\int_{{N}_{u_2}}  r^{p}  (L\chi\Phi_{(n)})^2\,d\omega dv+C\epsilon\int_{{N}_{u_1}}  r^{p}  (L\chi\Phi_{(n)})^2\,d\omega dv\\
&+C\sum_{k\leq n}\int_{\Sigma_{u_1}} J^T[T^k\psi]\cdot \mathbf{n}_{u_1}\,d\mu_{\Sigma_{u_1}}.
\end{split}
\end{equation*}
The first integral on the right-hand side can be absorbed into $J_1$ and the second integral on the right-hand side of the above equation can be absorbed into the left-hand side of \eqref{eq:maineqrpestinf}.

Hence, we arrive at \eqref{eq:rpestinf} with $\Phi_{(n)}$ replaced by $\chi \Phi_{(n)}$. In order to remove the cut-off function $\chi$ on the right-hand side of \eqref{eq:maineqrpestinf}, we estimate:
\begin{equation*}
\begin{split}
\int_{{N}^{\mathcal{I}}_{u_1}} r^{p}  (L\chi\Phi_{(n)})^2\,d\omega dv \leq&\: C\int_{{N}^{\mathcal{I}}_{u_1}} r^{p}  (L\Phi_{(n)})^2\,d\omega dv +C \int_{{N}^{\mathcal{I}}_{u_1}\cap\{r_{\mathcal{I}}\leq r\leq r_{\mathcal{I}}+M\}} \Phi_{(n)}^2\,d\omega dv\\
\leq&\: C\int_{{N}^{\mathcal{I}}_{u_1}} r^{p}  (L\Phi_{(n)})^2\,d\omega dv +C \sum_{k\leq n-1} \int_{{N}^{\mathcal{I}}_{u_1}\cap\{r_{\mathcal{I}}\leq r\leq r_{\mathcal{I}}+M\}} J^T[T^k\psi]\cdot L \,d\omega dv,
\end{split}
\end{equation*}
where we applied  \eqref{eq:hardyinf} together with \eqref{eq:udepradfieldsinft} and a standard elliptic estimate on ${N}_{u_1}$ to arrive at the second inequality. We can similarly estimate
\begin{equation*}
\begin{split}
\int_{{N}^{\mathcal{I}}_{u_2}} r^{p}  (L\Phi_{(n)})^2\,d\omega dv \leq&\: C\int_{{N}^{\mathcal{I}}_{u_2}\cap\{r\geq r_{\mathcal{I}}+M\}} r^{p}  (L\chi \Phi_{(n)})^2\,d\omega dv +C \int_{{N}^{\mathcal{I}}_{u_2}\cap\{r_{\mathcal{I}}\leq r\leq r_{\mathcal{I}}+M\}} (L\Phi_{(n)})^2\,d\omega dv\\
\leq&\: C\int_{{N}^{\mathcal{I}}_{u_2}} r^{p}  (L\chi \Phi_{(n)})^2\,d\omega dv +C \sum_{k\leq n} \int_{{N}^{\mathcal{I}}_{u_2}} J^T[T^k\psi]\cdot L \,d\omega dv\\
\leq&\: C\int_{{N}^{\mathcal{I}}_{u_2}} r^{p}  (L\chi \Phi_{(n)})^2\,d\omega dv +C \sum_{k\leq n} \int_{\Sigma_{u_1}} J^T[T^k\psi]\cdot \mathbf{n}_{u_1} \,d\mu_{\Sigma_{u_1}}
\end{split}
\end{equation*}
and, by applying moreover \eqref{eq:iledwayps}, we also obtain
\begin{equation*}
\begin{split}
\int_{u_1}^{u_2}\int_{{N}^{\mathcal{I}}_{u}} r^{p-1}  (L\Phi_{(n)})^2\,d\omega dvdu \leq&\: C\int_{u_1}^{u_2}\int_{{N}^{\mathcal{I}}_{u}\cap\{r\geq r_{\mathcal{I}}+M\}} r^{p-1}  (L\chi \Phi_{(n)})^2\,d\omega dvdu \\
&+C \int_{u_1}^{u_2}\int_{{N}^{\mathcal{I}}_{u}\cap\{r_{\mathcal{I}}\leq r\leq r_{\mathcal{I}}+M\}} (L\Phi_{(n)})^2\,d\omega dvdu\\
\leq&\:C\int_{u_1}^{u_2}\int_{{N}^{\mathcal{I}}_{u}} r^{p-1}  (L\chi \Phi_{(n)})^2\,d\omega dvdu  +C \sum_{k\leq n} \int_{\Sigma_{u_1}} J^T[T^k\psi]\cdot \mathbf{n}_{u_1} \,d\mu_{\Sigma_{u_1}}.
\end{split}
\end{equation*}

In order to derive \eqref{eq:rpesthor} we introduce a different smooth cut-off function $\chi: \R \to \R$ (we use the same notation for this cut-off function for the sake of convenience) such that $\chi(r)=0$ for all $r\geq r_{\mathcal{H}}$ and $\chi(r)=1$ for all $M\leq r\leq M+\frac{1}{2}(r_{\mathcal{H}}-M)$, with $r_{\mathcal{H}}<2M$ and $r_{\mathcal{H}}-M$ appropriately small. We integrate both sides of \eqref{eq:mainidentityspherehor} in the $u$ and $v$ directions to obtain:
\begin{equation}
\label{eq:maineqrpesthor}
\begin{split}
\int_{{N}^{\mathcal{H}}_{v_2}}& (r-M)^{-p}(\underline{L}(\chi\underline{\Phi}_{(n)}))^2\,d\omega du+ \frac{1}{2}M^{-2}\int_{v_1}^{v_2} \int_{{N}^{\mathcal{H}}_v}(p+4n)(r-M)^{1-p}(L(\chi\underline{\Phi}_{(n)}))^2\,d\omega du dv\\
&+\frac{1}{4}M^{-4}\int_{\mathcal{H}^+(v_1,v_2)}\left[(r-M)^{2-p} |\snabla_{\s^2}\underline{\Phi}_{(n)}|^2-n(n+1)(r-M)^{2-p}\underline{\Phi}_{(n)}^2\right]\,d\omega dv\\
&+\frac{1}{8}M^{-6}\int_{v_1}^{v_2} \int_{{N}^{\mathcal{H}}_v}(2-p)(r-M)^{3-p}\chi^2\left(|\snabla_{\s^2}\underline{\Phi}_{(n)}|^2-n(n+1)\underline{\Phi}_{(n)}^2\right)\,d\omega du dv\\
=&\: \int_{{N}^{\mathcal{H}}_{v_1}} (r-M)^{-p}(L(\chi\underline{\Phi}_{(n)}))^2\,d\omega du\\
&+J_1+J_2+J_3+\sum_{|\alpha|\leq 1} \int_{v_1}^{v_2} \int_{{N}^{\mathcal{H}}_u}r^{p-2}\underline{L} (\chi\underline{\Phi}_{(n)}) \cdot R_{\chi}[\partial^{\alpha}\underline{\Phi}_{(n)}]\,d\omega dudv,
\end{split}
\end{equation}
where we now use the notation $R_{\chi}[f]$ for terms that are compactly supported in $M+\frac{1}{2}(r_{\mathcal{H}}-M)\leq r\leq r_{\mathcal{H}}$ and are linear in the function $f$, and we define
\begin{align*}
\underline{J}_1:=&\: \int_{v_1}^{v_2} \int_{{N}^{\mathcal{H}}_v}O((r-M)^{2-p})(\underline{L}(\chi \underline{\Phi}_{(n)}))^2\,d\omega dv du,\\
\underline{J}_2:=&\: \int_{v_1}^{v_2} \int_{{N}^{\mathcal{H}}_v}O((r-M)^{3-p})\chi \underline{\Phi}_{(n)}\cdot \underline{L}(\chi \underline{\Phi}_{(n)})+ O((r-M)^{3-p})\slashed{\Delta}_{\s^2}(\chi\underline{\Phi}_{(n)})\cdot \underline{L}(\chi \underline{\Phi}_{(n)})\,d\omega dv du,\\
\underline{J}_3:=&\: n\sum_{k=0}^{n-1}\int_{v_1}^{v_2} \int_{{N}^{\mathcal{H}}_v} O((r-M)^{2-p}) \chi \underline{\Phi}_{(k)}\cdot \underline{L}(\chi \underline{\Phi}_{(n)})\,d\omega dv du.
\end{align*}
We can absorb $\underline{J}_1+\underline{J}_2+\underline{J}_3$ into the left-hand side of \eqref{eq:rpesthor} by repeating the estimates in $\mathcal{A}^{\mathcal{I}}$ above in the region $\mathcal{A}^{\mathcal{H}}$, using \eqref{eq:hardyinf} instead of \eqref{eq:hardyinf}. Note that in the $\mathcal{A}^{\mathcal{H}}$ case there is no need to apply \eqref{eq:udepradfieldsinft} as the analogous estimate at $\mathcal{H}^+$ follows immediately from the smoothness of $\psi$ at $\mathcal{H}^+$ (as we consider smooth initial data).

\end{proof}

\begin{remark}
The third and fourth integrals on the right-hand sides of \eqref{eq:rpestinf} and \eqref{eq:rpesthor} have a \textbf{positive sign} if we consider $P_{\geq n}\psi$ rather than the more general $P_{\geq 1}\psi$. This follows directly from the standard Poincar\'e inequality on $\s^2$.
\end{remark}

\subsection{The improved hierarchies for $\ell =0$ and $\ell=1$}
\label{sec:mainestextended}
The next proposition yields improved hierarchies for the harmonic mode numbers $\ell=0$ and $\ell=1$. 
\begin{proposition}
\label{prop:rpestell01} Let $\psi$ be a smooth solution to \eqref{eq:waveequation}. 
Fix $n\in \{0,1\}$ and assume that on ${N}^{\mathcal{I}}$:
\begin{equation}
\label{asm:radfieldl1}
\lim_{v\to \infty}\left(\int_{\s^2}\phi_1^2\,d\omega\right)(u_0,v)<\infty.
\end{equation}

Then there exists an $r_{\mathcal{I}}>0$, such that for $p\in (-4n,4)$ and for all $0\leq u_1\leq u_2$:
\begin{equation}
\label{eq:rpestinftilde}
\begin{split}
\int_{{N}^{\mathcal{I}}_{u_2}}& r^p(LP_{n}\widetilde{\Phi}_{(n)})^2\,d\omega dv+ \frac{1}{2}\int_{u_1}^{u_2} \int_{{N}^{\mathcal{I}}_u}(p+4n)r^{p-1}(LP_{n}\widetilde{\Phi}_{(n)})^2\,d\omega dv du\\
\leq&\: C\int_{{N}^{\mathcal{I}}_{u_1}}r^p(LP_{n}\widetilde{\Phi}_{(n)})^2\,d\omega dv+ C\sum_{k\leq n}\int_{\Sigma_{u_1}} J^T[T^k\psi]\cdot \mathbf{n}_{u_1}\,d\mu_{\Sigma_{u_1}},
\end{split}
\end{equation}
where $C=C(M,{\Sigma_0},r_{\mathcal{I}})>0$ is a constant and we can take $r_{\mathcal{I}}=(p+4n)^{-1}R_0(n,M)>0$.

Furthermore, there exists an $r_{\mathcal{H}}>M$, such that for $p\in (-4n,4)$ and for all $0\leq u_1\leq u_2$:
\begin{equation}
\label{eq:rpeshortilde}
\begin{split}
\int_{{N}^{\mathcal{H}}_{v_2}}& (r-M)^{-p}(\underline{L}P_n\widetilde{\underline{\Phi}}_{(n)})^2\,d\omega du+ \frac{1}{2}\int_{v_1}^{v_2} \int_{{N}^{\mathcal{H}}_u}(p+4n)(r-M)^{1-p}(\underline{L}P_n\widetilde{\underline{\Phi}}_{(n)})^2\,d\omega du dv\\
\leq&\: C\int_{{N}^{\mathcal{H}}_{v_1}}(r-M)^{-p}(\underline{L}P_n\widetilde{\underline{\Phi}}_{(n)})^2\,d\omega du+ C\sum_{k\leq n}\int_{\Sigma_{v_1}} J^T[T^k\psi]\cdot \mathbf{n}_{v_1}\,d\mu_{\Sigma_{v_1}},
\end{split}
\end{equation}
where $C=C(M,{\Sigma_0},r_{\mathcal{H}})>0$ is a constant and we can take $(r_{\mathcal{H}}-M)^{-1}=(p+4n)^{-1}R_0(n,M)>0$.

If we additionally assume that there exist constants $\eta>0$ and $\mathcal{E}_{\eta}>0$ such that
\begin{align}
\label{addasmforp5I}
\sum_{k\leq n}\int_{\s^2}|P_n\widetilde{\Phi}_{(k)}|^2\,d\omega\lesssim \mathcal{E}_{\eta}\cdot  (1+u)^{-2-\eta}\quad \textnormal{in $\mathcal{A}^{\mathcal{I}}$},\\
\label{addasmforp5H}
\sum_{k\leq n}\int_{\s^2}|P_n\widetilde{\underline{\Phi}}_{(k)}|^2\,d\omega\lesssim \mathcal{E}_{\eta} \cdot (1+v)^{-2-\eta}\quad \textnormal{in $\mathcal{A}^{\mathcal{H}}$},
\end{align}
then we can obtain for $0<p<5$:
\begin{equation}
\label{eq:rpestinfp5}
\begin{split}
\int_{{N}^{\mathcal{I}}_{u_2}}& r^p(LP_n\widetilde{\Phi}_{(n)})^2\,d\omega dv+ \frac{1}{2}\int_{u_1}^{u_2} \int_{{N}^{\mathcal{I}}_u}(p+4n)r^{p-1}(LP_n\widetilde{\Phi}_{(n)})^2\,d\omega dv du\\
\leq&\: C\int_{{N}^{\mathcal{I}}_{u_1}}r^p(LP_n\widetilde{\Phi}_{(n)})^2\,d\omega dv+ C\sum_{k\leq n}\int_{\Sigma_{u_1}} J^T[T^k\psi]\cdot \mathbf{n}_{u_1}\,d\mu_{\Sigma_{u_1}}+C(p-5)^{-1}\mathcal{E}_{\eta},
\end{split}
\end{equation}
and 
\begin{equation}
\label{eq:rpeshorp5k}
\begin{split}
\int_{{N}^{\mathcal{H}}_{v_2}}& (r-M)^{-p}(\underline{L}P_n\widetilde{\underline{\Phi}}_{(n)})^2\,d\omega du+ \frac{1}{2}\int_{v_1}^{v_2} \int_{{N}^{\mathcal{H}}_u}(p+4n)(r-M)^{1-p}(\underline{L}P_n\widetilde{\underline{\Phi}}_{(n)})^2\,d\omega du dv\\
\leq&\: C\int_{{N}^{\mathcal{H}}_{v_1}}(r-M)^{-p}(\underline{L}P_n\widetilde{\underline{\Phi}}_{(n)})^2\,d\omega du+ C\sum_{k\leq n}\int_{\Sigma_{v_1}} J^T[T^k\psi]\cdot \mathbf{n}_{v_1}\,d\mu_{\Sigma_{v_1}}+C(p-5)^{-1}\mathcal{E}_{\eta}.
\end{split}
\end{equation}
\end{proposition}
\begin{proof}
For $p<4$, the proof proceeds in a similar manner to the proof of Proposition \ref{prop:generalrpest}, using in the $n=1$ case the identities in Lemma \ref{lm:rweightidentitiess2v2} rather than those in Lemma \ref{lm:rweightidentitiess2v1}. Note in particular that due to the \emph{lower} powers in $r$ and $(r-M)^{-1}$ appearing in the identities in Lemma \ref{lm:rweightidentitiess2v2} and in Lemma \ref{lm:rweightidentitiess2v1} if $n=0$ (compared to the general $n$ case), we are able to increase the range of $p$. We omit the details of these steps.

We will now show how we can extend the range of $p$ to $0<p<5$, after invoking the additional assumptions \eqref{addasmforp5I} and \eqref{addasmforp5H}. We restrict to $\mathcal{A}^{\mathcal{I}}$, because the argument in $\mathcal{A}^{\mathcal{H}}$ proceeds analogously.

By Lemma \ref{lm:rweightidentitiess2v1} in the $n=0$ case and Lemma \ref{lm:rweightidentitiess2v2} in the $n=1$ case, we only need to estimate
\begin{equation*}
\int_{u_1}^{u_2} \int_{{N}^{\mathcal{I}}_u} r^{p-3}|\phi| |LP_n \widetilde{\Phi}_{(n)}|+ r^{p-3}|P_n \widetilde{\Phi}_{(n)}||LP_n \widetilde{\Phi}_{(n)}|\,d\omega dv du.
\end{equation*}
We apply Young's inequality to obtain
\begin{equation*}
\begin{split}
\int_{u_1}^{u_2}& \int_{{N}^{\mathcal{I}}_u} r^{p-3}|\phi| |LP_n \widetilde{\Phi}_{(n)}|+ r^{p-3}|P_n \widetilde{\Phi}_{(n)}||LP_n \widetilde{\Phi}_{(n)}|\,d\omega dv du\\
 \leq&\: \epsilon \int_{u_1}^{u_2} \int_{{N}^{\mathcal{I}}_u} u^{-1-\frac{\eta}{2}}\cdot r^p(LP_n \widetilde{\Phi}_{(n)})^2\,d\omega dv du+ \frac{1}{4\epsilon} \int_{u_1}^{u_2} \int_{{N}^{\mathcal{I}}_u} u^{1+\frac{\eta}{2}}r^{p-6}\left[(\phi)^2+(P_n \widetilde{\Phi}_{(n)})^2 \right]\,d\omega dv du\\
 \leq&\: C\epsilon u_0^{-\frac{\eta}{2}} \sup_{u_1\leq u\leq u_2} \int_{{N}^{\mathcal{I}}_u} r^p(LP_n \widetilde{\Phi}_{(n)})^2\,d\omega dv + \frac{C}{\epsilon}(p-5)^{-1} u_0^{-\frac{\eta}{2}} E_{\eta}.
\end{split}
\end{equation*}
For $\epsilon>0$ suitably small, we can absorb the first term on the right-hand side above into the left-hand side of the spacetime integral of the identities \eqref{eq:mainidentitysphereinf} (for $n=0$) and \eqref{eq:rweightidentitiess2v2H} (for $n=1$), where we take a supremum in $u$ on the left-hand side.
\end{proof}

\section{Extended hierarchies for $T^k\psi$}
\label{sec:extendhier}

\subsection{The preliminary extended identities}
\label{sec:TheLAndUnderlineLCommutedIdenties}

In order to obtain \emph{improved} estimates for time-derivatives of $\psi$, which are essential for deriving the late-time asymptotics for $\psi$ itself, we will derive additional hierarchies of $r^p$ and $(r-M)^{-p}$ weighted estimates for $T^k\psi$, with $k\geq 1$. As a first step, we will derive additional hierarchies for $L^k\Phi_{(n)}$ and $L^k\phi_0$ in $\mathcal{A}^{\mathcal{I}}$ and $\underline{L}^k\underline{\Phi}_{(n)}$ and $\Lbar^k\phi_0$ in $\mathcal{A}^{\mathcal{H}}$. We start with the following identities.

\begin{lemma}
\label{lm:hoequations}
Let $n\in \N_0$ and $k\in \N$. Then:
\begin{equation}
\label{eq:hoequations1}
\begin{split}
4 L\underline{L}(L^k\Phi_{(n)})=&\:Dr^{-2}\slashed{\Delta}_{\s^2}L^k\Phi_{(n)}+[-kDr^{-3}+O(r^{-4})]\slashed{\Delta}_{\s^2}L^{k-1}\Phi_{(n)}+[-4nr^{-1}+O(r^{-2})]L^{k+1}\Phi_{(n)}\\
&+D[n(n+1+2k)r^{-2}+O(r^{-3})]L^k\Phi_{(n)}+k(k-1)\sum_{j=\min\{k,2\}}^k O(r^{-2-j})\slashed{\Delta}_{\s^2}L^{k-j}\Phi_{(n)}\\
&+\sum_{j=1}^k O(r^{-2-j})L^{k-j}\Phi_{(n)}+n\sum_{m=0}^{\max\{0,n-1\}}\sum_{j=0}^k O(r^{-2-j})L^{k-j}\Phi_{(m)}
\end{split}
\end{equation}
and
\begin{equation}
\label{eq:hoequations2}
\begin{split}
4 L\underline{L}(\underline{L}^k\underline{\Phi}_{(n)})=&\:Dr^{-2}\slashed{\Delta}_{\s^2}\underline{L}^k\underline{\Phi}_{(n)}+[-kM^{-6}(r-M)^3+O((r-M)^4)]\slashed{\Delta}_{\s^2}\Lbar^{k-1}\underline{\Phi}_{(n)}\\
&+[-4M^{-2}n(r-M)+O((r-M)^2)]\Lbar^{k+1}\underline{\Phi}_{(n)}\\
&+[n(n+1+2k)M^{-4}(r-M)^2+O((r-M)^3)]\Lbar^k\underline{\Phi}_{(n)}\\
&+k(k-1)\sum_{j=\min\{k,2\}}^k O((r-M)^{2+j})\slashed{\Delta}_{\s^2}\Lbar^{k-j}\underline{\Phi}_{(n)}\\
&+\sum_{j=1}^k O((r-M)^{2+j})\Lbar^{k-j}\underline{\Phi}_{(n)}+n\sum_{m=0}^{\max\{0,n-1\}}\sum_{j=0}^k O((r-M)^{2+j})\Lbar^{k-j}\underline{\Phi}_{(m)}.
\end{split}
\end{equation}
\end{lemma}
\begin{proof}
The identities \eqref{eq:hoequations1} and \eqref{eq:hoequations2} follow from a straightforward induction argument, where we apply Lemma \ref{lm:maineq} for the $k=0$ case.
\end{proof}

When we restrict the  spherical mean $\psi_0$, the analog of Lemma \ref{lm:hoequations} (with $n=0$) simplifies significantly.
\begin{lemma}
\label{lm:hoequationsl0}
Consider $\psi_0$. Then for all $k\in \N$,
\begin{align}
\label{eq:hoequations1l0}
4 L\underline{L}(L^k\phi_0)=&\:O(r^{-2})L^{k+1}\phi_0+O(r^{-3})L^k\phi_0+\sum_{j=1}^k O(r^{-3-j})L^{k-j}\phi_0,\\
\label{eq:hoequations2l0}
4 \Lbar L(\Lbar^k\phi_0)=&\:O((r-M)^{2})\Lbar \Lbar^k\phi_0+O((r-M)^{3})\Lbar^k\phi_0+\sum_{j=1}^k O((r-M)^{3+j})\Lbar^{k-j}\phi_0.
\end{align}
\end{lemma}
\begin{proof}
Equations \eqref{eq:hoequations1l0} and \eqref{eq:hoequations2l0} follow from a standard induction argument, where we apply Lemma \ref{lm:maineq} to obtain the $k=0$ case.
\end{proof}

Before we derive $r$-weighted estimates for $L^k\Phi_{(n)}$ and $\underline{L}^k\underline{\Phi}_{(n)}$, with $k>0$, we need to establish appropriate ($u$-dependent) boundedness estimates near $\mathcal{I}^+$. Recall that in the $k=0$ case these were obtained in Proposition \ref{prop:radfieldsinf}.
\begin{proposition}
\label{prop:radfieldsinfLk}
Let $n\in \N$ and $J\in \N_0$ and assume that
\begin{equation*}
\int_{\Sigma_0}J^T[\psi]\cdot \mathbf{n}_{\Sigma_0}\,d\mu_{\Sigma_0}<\infty.
\end{equation*}
\begin{itemize}
\item[$\rm (i)$]
For all $u\geq u_0$, we have that there exists a constant $C_u=C_u(M,J,{\Sigma_0},u)>0$, such that
\begin{equation}
\label{eq:udepradfieldsinftLk}
\begin{split}
\sum_{0\leq k\leq n}\sum_{0\leq j\leq J}&\sup_{u_0\leq u'\leq u}\left[\int_{\s^2} r^{2j}|\slashed{\Delta}_{\s^2}^{n-k} L^j\Phi_{(k)}|^2|_{{N}^{\mathcal{I}}}\,d\omega\right](v)\\
<&\:C_u\sum_{0\leq k\leq n}\sum_{0\leq j\leq J}\left[\int_{\s^2}r^{2j} |\slashed{\Delta}_{\s^2}^{n-k} L^j\Phi_{(k)}|^2|_{{N}^{\mathcal{I}}}\,d\omega\right](v).
\end{split}
\end{equation}
\item[$\rm (ii)$]
If we assume that for all $0\leq k \leq n$ and $0\leq j\leq J$
\begin{equation*}
\lim_{v\to \infty}\left[\int_{\s^2}r^{2j} |\slashed{\Delta}_{\s^2}^{n-k} L^j\Phi_{(k)}|^2|_{{N}^{\mathcal{I}}}\,d\omega\right](v)<\infty,
\end{equation*}
then, for all $u\geq u_0$, $0\leq k \leq n$ and $0\leq j\leq J$, we have that
\begin{equation}
\label{eq:limitradfieldsinftLk}
\lim_{v\to \infty}\left[\int_{\s^2}r^{2j} |\slashed{\Delta}_{\s^2}^{n-k} L^j\Phi_{(k)}|^2|_{{N}_{u}^\mathcal{I}}\,d\omega\right](v)<\infty.
\end{equation}
\end{itemize}
\end{proposition}
\begin{proof}
The proof proceeds inductively in $J$. The $J=0$ case is proved in Proposition \ref{prop:radfieldsinf}. We then suppose that \eqref{eq:udepradfieldsinftLk} holds for $J$ and prove that it must also hold for $J$ replaced by $J+1$ by applying the same arguments as in the proof of Proposition \ref{prop:radfieldsinf}, using equation \eqref{eq:hoequations1}.
\end{proof}

We now state the key lemma containing the $r$-weighted \emph{identities} for $L^J\Phi_{(n)}$ and $\underline{L}^j \underline{\Phi}_{(n)}$ with $J>1$. Recall that the $J=0$ case was obtained in  Lemma \ref{lm:rweightidentitiess2v1}.
\begin{lemma}
\label{lm:rweightidentitiesho}
Let $p\in \R$. Then the following identities hold for all $N\in \N_0$ and $J\in \N$:
\begin{equation}
\label{eq:mainidentitysphereinfho}
\begin{split}
\int_{\s^2}&\underline{L}\left(r^p (L^{J+1}\Phi_{(n)})^2\right)\,d\omega+\int_{\s^2} L\left(\frac{1}{4}r^{p-2}|\slashed{\nabla}_{\s^2}L^J\Phi_{(n)}|^2\right)\,d\omega\\
&+\frac{1}{2}\int_{\s^2}L\left(\left[-JDr^{p-3}+O(r^{p-4})\right]\slashed{\nabla}_{\s^2}L^J\Phi_{(n)}\cdot \snabla_{\s^2}L^{J-1}\Phi_{(n)}\right)\,d\omega \\
&+ \frac{1}{2}\int_{\s^2}[(p+4N)r^{p-1}+O(r^{p-2})](L^{J+1}\Phi_{(n)})^2\,d\omega\\ 
&+\frac{1}{8}\int_{\s^2} [(2+4J-p)r^{p-3}+O(r^{p-4})]D|\slashed{\nabla}_{\s^2}L^J\Phi_{(n)}|^2+(p-2)N(N+1+2J)r^{p-3}D(L^J\Phi_{(n)})^2\,d\omega\\ 
=&\:  \int_{\s^2} \frac{1}{4}L\left(N(N+1+2J)r^{p-2}(L^J\Phi_{(n)})^2\right)\,d\omega\\ 
&+\int_{\s^2} O(r^{p-3}) L^J\Phi_{(n)}\cdot L^{J+1}\Phi_{(n)}\,d\omega +\sum_{j=1}^J\int_{\s^2}O(r^{p-2-j})L^{J-j}\Phi_{(n)}\cdot L^{J+1}\Phi_{(n)}\\
&+J(J-1)\int_{\s^2}\sum_{j=\min\{J,2\}}^J O(r^{p-2-j})\slashed{\Delta}_{\s^2}L^{J-j}\Phi_{(n)}\cdot L^{J+1}\Phi_{(n)}\,d\omega\\
&+\int_{\s^2} O(r^{p-3})\slashed{\Delta}_{\s^2}L^J\Phi_{(n)}\cdot L^{J+1}\Phi_{(n)}\,d\omega+N\sum_{n=0}^{\max\{0,N-1\}}\sum_{j=0}^J\int_{\s^2} O(r^{p-2-j}) L^{J-j}\Phi_{(n)}\cdot L^{J+1}\Phi_{(n)}\,d\omega\\
&+\int_{\s^2}O(r^{p-4}) \slashed{\nabla}_{\s^2}L^{J-1}\Phi_{(n)} \cdot \slashed{\nabla}_{\s^2}L^J\Phi_{(n)}\,d\omega,
\end{split}
\end{equation}
and
\begin{equation}
\label{eq:mainidentityspherehorho}
\begin{split}
\int_{\s^2}&L\left((r-M)^{-p} (\Lbar^{J+1}\underline{\Phi}_{(n)})^2\right)\,d\omega+\int_{\s^2} \Lbar\left(\frac{1}{4}M^{-4}(r-M)^{2-p}|\slashed{\nabla}_{\s^2}\Lbar^J\underline{\Phi}_{(n)}|^2\right) \,d\omega\\
&+\frac{1}{2}\int_{\s^2}L\left(\left[-JM^{-6}D(r-M)^{3-p}+O((r-M)^{4-p})\right]\snabla_{\s^2}\Lbar^J\underline{\Phi}_{(n)}\cdot \slashed{\nabla}_{\s^2}\Lbar^{J-1}\underline{\Phi}_{(n)}\right)\,d\omega \\
&+ \frac{1}{2}\int_{\s^2}[(p+4N)M^{-2}(r-M)^{1-p}+O((r-M)^{2-p})](\Lbar^{J+1}\underline{\Phi}_{(n)})^2\,d\omega\\ 
&+\frac{1}{8}\int_{\s^2}[ (2+4J-p)M^{-6}(r-M)^{3-p}+O((r-M)^{4-p})]|\slashed{\nabla}_{\s^2}\Lbar^J\underline{\Phi}_{(n)}|^2\\
&+(p-2)M^{-6}N(N+1+2J)(r-M)^{3-p}(\Lbar^J\underline{\Phi}_{(n)})^2\,d\omega\\ 
=&\:  \int_{\s^2} \frac{1}{4}M^{-4}\Lbar\left(N(N+1+2J)(r-M)^{2-p}(\Lbar^J\underline{\Phi}_{(n)})^2-(r-M)^{2-p}|\slashed{\nabla}_{\s^2}\Lbar^J\underline{\Phi}_{(n)}|^2\right)\,d\omega\\ 
&+\int_{\s^2} O((r-M)^{3-p}) \Lbar^J\underline{\Phi}_{(n)}\cdot \Lbar^{J+1}\underline{\Phi}_{(n)}\,d\omega +\sum_{j=1}^J\int_{\s^2}O((r-M)^{2+j-p}))\Lbar^{J-j}\underline{\Phi}_{(n)}\cdot \Lbar^{J+1}\underline{\Phi}_{(n)}\\
&+J(J-1)\int_{\s^2}\sum_{j=\min\{J,2\}}^J O((r-M)^{2+j-p})\slashed{\Delta}_{\s^2}\Lbar^{J-j}\underline{\Phi}_{(n)}\cdot \Lbar^{J+1}\underline{\Phi}_{(n)}\,d\omega\\
&+\int_{\s^2} O((r-M)^{3-p})\slashed{\Delta}_{\s^2}\Lbar^J\underline{\Phi}_{(n)}\cdot \Lbar^J\underline{\Phi}_{(n)}\,d\omega\\
&+N\sum_{n=0}^{\max\{0,N-1\}}\sum_{j=0}^J\int_{\s^2} O((r-M)^{2+j-p}) \Lbar^{J-j}\underline{\Phi}_{(n)}\cdot \Lbar^{J+1}\underline{\Phi}_{(n)}\,d\omega\\
&+\int_{\s^2}O((r-M)^{4-p}) \slashed{\nabla}_{\s^2}\Lbar^{J-1}\underline{\Phi}_{(n)} \cdot \slashed{\nabla}_{\s^2}\Lbar^J\underline{\Phi}_{(n)}\,d\omega.
\end{split}
\end{equation}
\end{lemma}
\begin{proof}
The proof of \eqref{eq:mainidentitysphereinfho} and \eqref{eq:mainidentityspherehorho} proceeds in an analogous fashion to the proof of Lemma \ref{lm:rweightidentitiess2v1}, where we use the more general equations in Lemma \ref{lm:hoequations}, rather than the equations in Lemma \ref{lm:maineq}, and we  integrate by parts appropriately in the $L$ direction and on $\s^2$ and in the $\Lbar$ direction and on $\s^2$, respectively. 
\end{proof}

\subsection{The preliminary extended hierarchies}
\label{sec:TheExtendedHierarchies}

We now obtain the \emph{higher-order} (with respect to $L$ or $\underline{L}$ derivation) $r$-weighted estimates.

\begin{proposition}
\label{prop:generalrpestLkpsi}
Fix $N\in \N_0$ and $J\in \N_0$ and assume that there exists a constant $C_0>0$ such that on ${N}^{\mathcal{I}}$: for all $0\leq n\leq \min\{N-1,0\}$ and $0\leq j\leq J$
\begin{equation}
\label{assm:id1b}
\lim_{v\to \infty} \left(\int_{\s^2}r^{2j}|\snabla_{\s^2}\slashed{\Delta}_{\s^2}^nL^jP_{\geq 1}\Phi_{(N-n)}|^2\,d\omega\right)(u_0,v)<\infty.
\end{equation}
Then there exists $r_{\mathcal{I}}>0$ sufficiently large, such that for $p\in (-4N+2J,2+2J]$ and for all $0\leq u_1\leq u_2$:
\begin{equation}
\label{eq:rpestinfLkpsi}
\begin{split}
\int_{{N}^{\mathcal{I}}_{u_2}}& r^p(LP_{\geq 1}L^J\Phi_{(N)})^2\,d\omega dv+ \int_{u_1}^{u_2} \int_{{N}^{\mathcal{I}}_u}(p+4N)r^{p-1}(LP_{\geq 1}L^J\Phi_{(N)})^2\,d\omega dv du\\
&+\int_{u_1}^{u_2} \int_{{N}^{\mathcal{I}}_u}r^{p-3}D|\snabla_{\s^2}P_{\geq 1}L^J\Phi_{(N)}|^2\,d\omega dv du\\
\leq&\: C\sum_{0\leq j\leq J}\int_{{N}^{\mathcal{I}}_{u_1}}r^{p-2j}(LP_{\geq 1}L^{J-j}\Phi_{(N)})^2\,d\omega dv+ C\sum_{m\leq N+J}\int_{\Sigma_{u_1}} J^T[T^m\psi]\cdot \mathbf{n}_{u_1}\,d\mu_{\Sigma_{u_1}}\\
&+C\cdot N \int_{\mathcal{I}^+(u_1,u_2)}r^{p-2-2J}(P_{\leq \max\{N-1,1\}}P_{\geq 1}\Phi_{(n)})^2\,d\omega du\\
&+C\cdot N\int_{u_1}^{u_2} \int_{{N}^{\mathcal{I}}_u} (2+2J-p)r^{p+1-2J}(LP_{\leq \max\{N-1,1\}}P_{\geq 1}\Phi_{(\max\{N-1,0\})})^2\,d\omega dv du,
\end{split}
\end{equation}
where $C=C(N,J,M,{\Sigma_0},r_{\mathcal{I}})>0$ is a constant and we can take $r_{\mathcal{I}}=(p-2J+4N)^{-1}R_0(N,J,M)>0$. 

Furthermore, there exists $r_{\mathcal{H}}>M$, with $r_{\mathcal{H}}-M$ suitably small, such that for $p\in (-4N+2J,2+2J]$ and for all $0\leq u_1\leq u_2$:
\begin{equation}
\label{eq:rpesthorLbarkpsi}
\begin{split}
\int_{{N}^{\mathcal{H}}_{v_2}}& (r-M)^{-p}(\underline{L}P_{\geq 1}\Lbar^J\underline{\Phi}_{(N)})^2\,d\omega du+\int_{v_1}^{v_2} \int_{{N}^{\mathcal{H}}_u}(p+4N)(r-M)^{1-p}(\underline{L}P_{\geq 1}\Lbar^J\underline{\Phi}_{(N)})^2\,d\omega du dv\\
&+\int_{v_1}^{v_2} \int_{{N}^{\mathcal{H}}_v}(r-M)^{3-p}D|\snabla_{\s^2}P_{\geq 1}\Lbar^J\underline{\Phi}_{(N)}|^2\,d\omega du dv\\
\leq&\: C\sum_{0\leq j\leq J}\int_{{N}^{\mathcal{H}}_{v_1}}(r-M)^{-p+2j}(\underline{L}P_{\geq 1}\Lbar^{J-j}\underline{\Phi}_{(N)})^2\,d\omega du+ C\sum_{m\leq N+J}\int_{\Sigma_{v_1}} J^T[T^m\psi]\cdot \mathbf{n}_{v_1}\,d\mu_{\Sigma_{v_1}}\\
&+C\cdot N \int_{\mathcal{H}^+(v_1,v_2)}(r-M)^{-p+2J+2}(P_{\leq \max\{N-1,1\}}P_{\geq 1}\Phi_{(n)})^2\,d\omega dv\\
&+C\cdot N\int_{v_1}^{v_2} \int_{{N}^{\mathcal{H}}_v} (2+2J-p)(r-M)^{-p+2J-1}(\Lbar P_{\leq \max\{N-1,1\}}P_{\geq 1}\underline{\Phi}_{(\max\{N-1,0\})})^2\,d\omega du dv,
\end{split}
\end{equation}
where $C=C(N,J,M,{\Sigma_0},r_{\mathcal{H}})>0$ is a constant and we can take $(r_{\mathcal{H}}-M)^{-1}=(p-2J+4N)^{-1}M^{-2}R_0(N,J,M)>0$. 
\end{proposition}
\begin{proof}
We prove \eqref{eq:rpestinfLkpsi} and \eqref{eq:rpesthorLbarkpsi} inductively in $J$. The $J=0$ case follows immediately from the estimates \eqref{eq:rpestinf} and \eqref{eq:rpesthor}. Now suppose \eqref{eq:rpestinfLkpsi} and \eqref{eq:rpesthorLbarkpsi} hold for all $0\leq J\leq J'$. Then we need to show that they must also hold with $J$ replaced by $J'+1$. In order to show this, we use the identity \eqref{eq:mainidentitysphereinfho}, where we either absorb all terms without a sign into the terms with a good sign or we use the induction step to estimate them, applying the Hardy inequalities \eqref{eq:hardyinf} and \eqref{eq:hardyhor1} where necessary, after introducing a cut-off as in the proof of Proposition \ref{prop:generalrpest}. In addition, we integrate by parts the terms with a $\slashed{\Delta}_{\s^2}$ derivative that arise on the right-hand side of \eqref{eq:mainidentitysphereinfho}. We refer to Proposition 4.6 in \cite{paper1} for additional details in an analogous proof (where $N=2$).
\end{proof}

\begin{remark}
Let us emphasize that for $N\geq 2$ the estimates \eqref{eq:rpestinfLkpsi} and \eqref{eq:rpesthorLbarkpsi} applied to $P_{\leq N-1} \psi$ have integrals on the right-hand side that are \underline{not} solely supported on $\Sigma_{u_1}$ or $\Sigma_{v_1}$. We will need \emph{different} $r$-weighted estimates to be able to estimate these terms further. More precisely, we will apply Proposition \ref{prop:rpestell01} in the case where $N=2$.
\end{remark}

We can similarly extend the estimates for $n=0$ in Proposition \ref{prop:rpestell01} to the higher-order quantities $L^k\phi_0$ and $\underline{L}^k\phi_0$ as follows:

\begin{proposition}
\label{prop:rpestell01Lkpsi}
Let $k\in \N_0$ and consider the following assumptions for all $0\leq j\leq k-1$,
\begin{align}
\label{eq:asmLjphi_01}
\lim_{v\to \infty} r^{j+2} L^{j+1}\phi_0(u_0,v)<&\infty, \\
\label{eq:asmLjphi_02}
\lim_{v\to \infty} r^{j+3} L^{j+1}\phi_0(u_0,v)<&\infty.
\end{align}
Then, if we assume \eqref{eq:asmLjphi_01}, there exists an $r_{\mathcal{I}}>0$, such that for $p\in (-4n+2k,4+2k)$ and for all $0\leq u_1\leq u_2$:
\begin{equation}
\label{eq:rpestinfLkpsi0}
\begin{split}
\int_{{N}^{\mathcal{I}}_{u_2}}& r^p(L^{k+1}\phi_0)^2\,d\omega dv+\int_{u_1}^{u_2} \int_{{N}^{\mathcal{I}}_u}p r^{p-1}(L^{k+1}\phi_0)^2\,d\omega dv du\\
\leq&\: C\int_{{N}^{\mathcal{I}}_{u_1}}r^p(L^{k+1}\phi_0)^2\,d\omega dv+ C\sum_{m\leq k}\int_{\Sigma_{u_1}} J^T[T^m\psi]\cdot \mathbf{n}_{u_1}\,d\mu_{\Sigma_{u_1}},
\end{split}
\end{equation}
where $C=C(k,M,{\Sigma_0},r_{\mathcal{I}})>0$ is a constant and we can take $r_{\mathcal{I}}=(p-2k)^{-1}R_0(n,M)>0$, and we additionally assume \eqref{eq:asmLjphi_02} when $p\geq 3+k$.

Furthermore, there exists an $r_{\mathcal{H}}>M$, such that for $p\in (2k-4n,4+2k)$ and for all $0\leq u_1\leq u_2$:
\begin{equation}
\label{eq:rpeshorLkpsi0}
\begin{split}
\int_{{N}^{\mathcal{H}}_{v_2}}& (r-M)^{-p}(\Lbar^{k+1}\phi_0)^2\,d\omega du+\int_{v_1}^{v_2} \int_{{N}^{\mathcal{H}}_u}p(r-M)^{1-p}(\Lbar^{k+1}\phi_0)^2\,d\omega du dv\\
\leq&\: C\int_{{N}^{\mathcal{H}}_{v_1}}(r-M)^{-p}(\Lbar^{k+1}\phi_0)^2\,d\omega du+ C\sum_{m\leq k}\int_{\Sigma_{v_1}} J^T[T^m\psi]\cdot \mathbf{n}_{v_1}\,d\mu_{\Sigma_{v_1}},
\end{split}
\end{equation}
where $C=C(k,M,{\Sigma_0},r_{\mathcal{H}})>0$ is a constant and we can take $(r_{\mathcal{H}}-M)^{-1}=(p+4n)^{-1}R_0(n,M)>0$.

If we additionally assume that there exists constants $\eta>0$ and $\mathcal{E}_{k,\eta}>0$ such that
\begin{align}
\label{eq:hoauxdecayassm1}
\sum_{j\leq k}r^{2j}(L^j\phi_0)^2\lesssim \mathcal{E}_{k,\eta}\cdot  (1+u)^{-2-\eta}\quad \textnormal{in $\mathcal{A}^{\mathcal{I}}$},\\
\label{eq:hoauxdecayassm2}
\sum_{j\leq k}(r-M)^{-2j}(\underline{L}^j\phi_0)^2\lesssim \mathcal{E}_{k,\eta} \cdot (1+v)^{-2-\eta}\quad \textnormal{in $\mathcal{A}^{\mathcal{H}}$},\\
\end{align}
and we assume \eqref{eq:asmLjphi_02}, then we can obtain for $2k<p<5+2k$:
\begin{equation}
\label{eq:rpestinfp5k}
\begin{split}
\int_{{N}^{\mathcal{I}}_{u_2}}& r^p(L^{k+1}\phi_0)^2\,d\omega dv+ \int_{u_1}^{u_2} \int_{{N}^{\mathcal{I}}_u}pr^{p-1}(L^{k+1}\phi_0)^2\,d\omega dv du\\
\leq&\: C\int_{{N}^{\mathcal{I}}_{u_1}}r^p(L^{k+1}\phi_0)^2\,d\omega dv+ C\sum_{m\leq k}\int_{\Sigma_{u_1}} J^T[T^m\psi]\cdot \mathbf{n}_{u_1}\,d\mu_{\Sigma_{u_1}}+C(p+2k-5)^{-1}\mathcal{E}_{k,\eta},
\end{split}
\end{equation}
and 
\begin{equation}
\label{eq:rpeshorp5}
\begin{split}
\int_{{N}^{\mathcal{H}}_{v_2}}& (r-M)^{-p}(\Lbar^{k+1}\phi_0)^2\,d\omega du+\int_{v_1}^{v_2} \int_{{N}^{\mathcal{H}}_u}p(r-M)^{1-p}(\Lbar^{k+1}\phi_0)^2\,d\omega du dv\\
\leq&\: C\int_{{N}^{\mathcal{H}}_{v_1}}(r-M)^{-p}(\Lbar^{k+1}\phi_0)^2\,d\omega du+ C\sum_{m\leq k}\int_{\Sigma_{v_1}} J^T[T^m\psi]\cdot \mathbf{n}_{v_1}\,d\mu_{\Sigma_{v_1}}+C(p+2k-5)^{-1}\mathcal{E}_{k,\eta}.
\end{split}
\end{equation}
\end{proposition}
\begin{proof}
The estimates follow from an induction argument. The $k=0$ case follows from Proposition \ref{prop:rpestell01} with $n=0$ and we use the identities in Lemma \ref{lm:rweightidentitiesho} (with $n=0$) in the induction step, as in the proof of Proposition \ref{prop:generalrpestLkpsi}. We moreover use that for all $0\leq j \leq k-1$ and $0<p<5$:
\begin{equation*}
\int_{{N}^{\mathcal{I}}_{u}} r^{p+2j} (L^{j+1}\phi_0)^2\,d\omega dv\lesssim \int_{{N}^{\mathcal{I}}_{u}} r^{p+2j+2} (L^{j+2}\phi_0)^2\,d\omega dv,
\end{equation*}
by \eqref{eq:hardyinf} together with the assumption $\lim_{v\to \infty}r^{\frac{p+1}{2}+j} L^{j+1}\phi_0(u,v)<\infty$, which follows from the assumption \eqref{eq:asmLjphi_01} in the $p<3$ case and \eqref{eq:asmLjphi_02} in the $p<5$ case after a straightforward propagation to $u\geq u_0$ as in the proof of Proposition \ref{prop:radfieldsinfLk}. A similar inequality holds along ${N}^{\mathcal{H}}_{v}$, where there is no need for additional initial data assumptions by the assumption of smoothness of the initial data.
\end{proof}

\subsection{The hierarchies for $T^{k}\psi$}
\label{sec:TheHierarchiesForTKPsi}

In order to use the extended $r$-weighted hierarchies for $L^k\Phi_{(n)}$ and $\underline{L}^k\underline{\Phi}_{(n)}$ from Proposition \ref{prop:generalrpestLkpsi} and \ref{prop:rpestell01Lkpsi} to obtain additional $r$-weighted estimates for $T^k\Phi_{(n)}$ and $T^k\Phi_{(n)}$ compared to $\Phi_{(n)}$ and $\underline{\Phi}_{(n)}$, we first \emph{relate} $r$-weighted estimates for $T$-derivatives to $r$-weighted estimates for $L$- or $\underline{L}$-derivatives.
\begin{lemma}
\label{lm:intestTpsi}
Let $n\in \N_0$, $k\geq1$ and $p\in (2k,2k+2]$, then we can estimate
\begin{equation}
\label{eq:intestTpsiinf}
\begin{split}
\int_{u_1}^{u_2}&\int_{{N}^{\mathcal{I}}_{u}} r^{p-1}(L L^{k-1}T\Phi_{(n)})^2\,d\omega dv du\\
\leq &\:C\sum_{|\alpha|=1}\int_{u_1}^{u_2}\int_{{N}^{\mathcal{I}}_{u}} r^{p-1}(LL^{k}\Phi_{(n)})^2+r^{p-3}|\snabla_{\s^2}L^{k-1}\Omega^{\alpha}\Phi_{(n)}|^2\,d\omega dv du\\
&+ C\sum_{m\leq k+n}\int_{\Sigma_{u_1}} J^T[T^m\psi]\cdot \mathbf{n}_{u_1}\,d\mu_{\Sigma_{u_1}}
\end{split}
\end{equation}
and
\begin{equation}
\label{eq:intestTpsihor}
\begin{split}
\int_{v_1}^{v_2}&\int_{{N}^{\mathcal{H}}_{v}} (r-M)^{1-p}(\Lbar \Lbar^{k-1}T\underline{\Phi}_{(n)})^2\,d\omega du dv\\
\leq &\:C\sum_{|\alpha|=1} \int_{v_1}^{v_2}\int_{{N}^{\mathcal{I}}_{v}} (r-M)^{1-p}(\Lbar \Lbar^{k}\underline{\Phi}_{(n)})^2+(r-M)^{3-p} |\snabla_{\s^2} \Lbar^{k-1}\Omega^{\alpha}\underline{\Phi}_{(n)}|^2\,d\omega du dv\\
&+ C\sum_{m\leq k+n}\int_{\Sigma_{v_1}} J^T[T^m\psi]\cdot \mathbf{n}_{v_1}\,d\mu_{\Sigma_{v_1}}.
\end{split}
\end{equation}
Furthermore, if we let $p\in (2k,2k+5)$, and we assume that for all $0\leq j\leq k$,
\begin{align*}
\lim_{v\to \infty} r^{j+2} L^{j+1}\phi_0(u_0,v)<&\infty \quad \textnormal{when $p<3+2k$ and} \\
\lim_{v\to \infty} r^{j+3} L^{j+1}\phi_0(u_0,v)<&\infty\quad \textnormal{when $p<5+2k$,}
\end{align*}
 then we can estimate
\begin{equation}
\label{eq:intestTpsiinfl0}
\begin{split}
\int_{u_1}^{u_2}&\int_{{N}^{\mathcal{I}}_{u}} r^{p-1}(L L^{k-1}T\phi_0)^2\,d\omega dv du\\
\leq &\:C\int_{u_1}^{u_2}\int_{{N}^{\mathcal{I}}_{u}} r^{p-1}(LL^{k}\phi_0)^2\,d\omega dv du+ C\sum_{m\leq k}\int_{\Sigma_{u_1}} J^T[T^m\psi]\cdot \mathbf{n}_{u_1}\,d\mu_{\Sigma_{u_1}}
\end{split}
\end{equation}
and
\begin{equation}
\label{eq:intestTpsihorl0}
\begin{split}
\int_{v_1}^{v_2}&\int_{{N}^{\mathcal{H}}_{v}} (r-M)^{1-p}(\Lbar \Lbar^{k-1}T\phi_0)^2\,d\omega du dv\\
\leq &\:C\int_{v_1}^{v_2}\int_{{N}^{\mathcal{I}}_{v}} (r-M)^{1-p}(\Lbar \Lbar^{k}\phi_0)^2\,d\omega du dv+ C\sum_{m\leq k}\int_{\Sigma_{v_1}} J^T[T^m\psi]\cdot \mathbf{n}_{v_1}\,d\mu_{\Sigma_{v_1}}.
\end{split}
\end{equation}
\end{lemma}
\begin{proof}
All the estimates in the lemma follow directly from the identities in Lemma \ref{lm:hoequations} and \ref{lm:hoequationsl0} together with the equality $T=L+\underline{L}$. Note that we have made use of standard relations between angular derivatives and angular momentum operators $\Omega_i$, $i=1,2,3$, to replace $\slashed{\Delta}_{\s^2} L^{k-1}\Phi_{(n)}$ appearing in the integral by $\sum_{|\alpha|=1}\snabla_{\s^2}\Omega^{\alpha}L^{k-1}\Phi_{(n)}$ (and similarly for $\Phi_{(n)}$ replaced with $\underline{\Phi}_{(n)}$).
\end{proof}

\section{Time decay estimates}
\label{sec:decayest}
In this section we will derive time decay estimates for $\psi$. First, we will \emph{convert} hierarchies of $r$-weighted estimates from Section \ref{sec:rweightest} and \ref{sec:extendhier} into time decay estimates for various ($r$-weighted) energy quantities in Section \ref{sec:edecayest} and \ref{sec:hoedecayest}. Then, we will use these energy decay estimates to obtain \emph{pointwise} time decay estimates in Section \ref{sec:pdecayest} by applying in addition certain elliptic estimates that are derived in Section \ref{sec:ellpest}. 
\subsection{Energy decay estimates}
\label{sec:edecayest}
We start by deriving separately energy decay estimates for $\psi_0$ and $\psi_{\geq 1}$.
\begin{proposition}
\label{prop:edaypsi0}
For all $\epsilon>0$ there exists a constant $C=C(M,\Sigma_0,\epsilon)>0$ such that
\begin{align}
\label{eq:edecayl01}
\int_{\Sigma_{\tau}} J^T[ \psi_{0}]\cdot \mathbf{n}_{\tau}\,d\mu_{\tau}\leq &\:C\cdot E^{\epsilon}_{0}[\psi]  (1+\tau)^{-3+\epsilon},\\
\label{eq:r2edecayl01}
\int_{{N}^{\mathcal{I}}_{\tau}} r^2\cdot (L\phi_0)^2\,d\omega dv+\int_{{N}^{\mathcal{I}}_{\tau}} (r-M)^{-2}\cdot (\underline{L}\phi_0)^2\,d\omega dv\leq &\:C\cdot E^{\epsilon}_{0}[\psi]  (1+\tau)^{-1+\epsilon}.
\end{align}
with
\begin{equation*}
E^{\epsilon}_{0}[\psi]:= \int_{{N}^{\mathcal{I}}} r^{3-\epsilon}(L\phi_0)^2\,d\omega dv+\int_{{N}^{\mathcal{H}}} (r-M)^{-3+\epsilon}(\underline{L}\phi_0)^2\,d\omega du+\int_{\Sigma_0} J^T[\psi_0]\cdot \mathbf{n}_{\Sigma_0}\,d\mu_{\Sigma_0}.
\end{equation*}
\end{proposition}
\begin{proof}
Energy decay follows from the hierarchies of $r$-weighted estimates in Proposition \ref{prop:rpestell01} with $n=0$ by a repeated use of the \emph{mean value theorem along dyadic time intervals} (sometimes called ``the pigeonhole principle''). For additional details, see also the analogous Proposition 7.1 in \cite{paper1} for the sub-extremal case.
\end{proof}
Before we proceed with proving energy decay for the $\psi_{\geq 1}$ part of the solution, we will show that we can obtain \emph{additional improvements} to the estimates in  Proposition \ref{prop:edaypsi0} in the cases where $I_0[\psi]=0$ or $H_0[\psi]=0$, after proving the following \emph{auxiliary} decay lemma:

\begin{lemma}
\label{lm:auxedaypsi0}
For all $\epsilon>0$ there exists a constant $C=C(M,\Sigma_0,\epsilon)>0$ such that
\begin{align}
\label{eq:auxdecayl01}
\int_{{N}^{\mathcal{I}}_{\tau}} r^2\cdot (L\phi_0)^2\,d\omega dv \leq&\: C\cdot \left[E^{\epsilon}_{0}[\psi]+ \int_{{N}^{\mathcal{I}}_{0}} r^{4-\epsilon}\cdot (L\phi_0)^2\,d\omega dv\right] (1+\tau)^{-2+\epsilon},\\
\label{eq:auxdecayl02}
 \int_{{N}^{\mathcal{H}}_{\tau}} (r-M)^{-2}\cdot (\underline{L}\phi_0)^2\,d\omega du\leq &\: C\cdot \left[E^{\epsilon}_{0}[\psi]+  \int_{{N}^{\mathcal{H}}_{0}} (r-M)^{-4+\epsilon}\cdot (\underline{L}\phi_0)^2\,d\omega dv\right] (1+\tau)^{-2+\epsilon}.
\end{align}
Furthermore,
\begin{align}
\label{eq:auxpointdecayl01}
||\phi_0||^2_{L^{\infty}({N}^{\mathcal{I}}_{\tau})}\leq&\: C\cdot \left[E^{\epsilon}_{0}[\psi]+ \int_{{N}^{\mathcal{I}}_{0}} r^{4-\epsilon}\cdot (L\phi_0)^2\,d\omega dv\right] (1+\tau)^{-\frac{5}{2}+\epsilon},\\
\label{eq:auxpointdecayl02}
||\phi_0||^2_{L^{\infty}({N}^{\mathcal{I}}_{\tau})}\leq &\: C\cdot \left[E^{\epsilon}_{0}[\psi]+  \int_{{N}^{\mathcal{H}}_{0}} (r-M)^{-4+\epsilon}\cdot (\underline{L}\phi_0)^2\right] (1+\tau)^{-\frac{5}{2}+\epsilon}.
\end{align}
\end{lemma}
\begin{proof}
The estimates \eqref{eq:auxdecayl01} and \eqref{eq:auxdecayl02} follow after applying, in addition to the estimates in the proof of Proposition \ref{prop:edaypsi0}, \emph{either} \eqref{eq:rpestinftilde} \emph{or} \eqref{eq:rpeshortilde} with $n=0$ and $p=4-\epsilon$ and $p=3-\epsilon$ (and applying the mean value theorem twice in an analogous fashion to the  proof of Proposition \ref{prop:edaypsi0}. The pointwise decay estimates then follow from a standard application of the fundamental theorem of calculus; see the proof of Proposition \ref{prop:pointdecay} for explicit details of this type of computation.
\end{proof}
With the $L^{\infty}$ estimates in Lemma \ref{lm:auxedaypsi0}, we can recover the assumptions \eqref{addasmforp5I} \emph{or} \eqref{addasmforp5H} and make use of the \emph{full} hierarchy of $r$-weighted estimates from Proposition \ref{prop:rpestell01}.

\begin{proposition}
\label{prop:extraendecay}
For all $\epsilon>0$ there exists a constant $C=C(M,{\Sigma_0},\epsilon)>0$ such that
\begin{align}
\label{eq:optr2decay1}
\int_{{N}_{\tau}^\mathcal{I}}r^{2-\epsilon} (L \phi_0)^2\,d\omega dv\leq&\:  C \cdot E^{\epsilon}_{0,\mathcal{I}}[\psi] \cdot (1+\tau)^{-3+\epsilon},\\
\label{eq:optr2decay2}
\int_{{N}^{\mathcal{H}}}(r-M)^{-2+\epsilon} (\underline{L}\phi_0)^2\,du\leq&\:C \cdot E^{\epsilon}_{0,\mathcal{H}}[\psi] \cdot (1+\tau)^{-3+\epsilon},
\end{align}
where
\begin{align*}
E^{\epsilon}_{0,\mathcal{I}}[\psi]:=&\: E^{\epsilon}_{0}[\psi]+\int_{{N}^{\mathcal{I}}} r^{5-\epsilon} (L  \phi_0)^2\,d\omega dv,\\
E^{\epsilon}_{0,\mathcal{H}}[\psi]:=&\: E^{\epsilon}_{0}[\psi]+\int_{{N}^{\mathcal{H}}} (r-M)^{-5+\epsilon} ( \underline{L}\phi_0)^2\,d\omega du.
\end{align*}
\end{proposition}
\begin{proof}
We repeat the proof of Lemma \ref{lm:auxedaypsi0}, but we \underline{additionally} apply either \eqref{eq:rpestinftilde} or \eqref{eq:rpeshortilde} with $n=0$ and $p=5-\epsilon$, using the $L^{\infty}$ estimates \eqref{eq:rpestinftilde} or \eqref{eq:rpeshortilde}.
\end{proof}

\begin{proposition}
\label{prop:edaypsi1}
Assume that for $n=0,1$ and $0\leq k \leq 2-n$:
\begin{equation*}
\lim_{v\to \infty} |\snabla_{\s^2}\slashed{\Delta}_{\s^2}^kP_{\geq 1}\Phi_{(n)}|^2\,d\omega<\infty.
\end{equation*}
Then, for all $\epsilon>0$ there exists a constant $C=C(M,{\Sigma_0},\epsilon)>0$ such that
\begin{align}
\label{eq:edecayl11}
\int_{\Sigma_{\tau}} J^T[ \psi_{\geq 1}]\cdot \mathbf{n}_{\tau}\,d\mu_{\tau}\leq &\:C\cdot E^{\epsilon}_{1}[\psi]  (1+\tau)^{-5+\epsilon},\\
\label{eq:edecaylr211}
\int_{{N}^{\mathcal{I}}_{\tau}} r^2(L \phi_{\geq 1})^2\,d\omega dv+\int_{{N}^{\mathcal{H}}_{\tau}} (r-M)^{-2}(\underline{L} \phi_{\geq 1})^2\,d\omega du \leq &\: C\cdot E^{\epsilon}_{1}[\psi]  (1+\tau)^{-3+\epsilon},
\end{align}
with
\begin{equation*}
\begin{split}
E^{\epsilon}_{1}[\psi]:=&\: \int_{{N}^{\mathcal{I}}} r^{1-\epsilon}(LP_{\geq 1}\Phi_{(2)})^2+\sum_{n=0}^1 \sum_{m=0}^{3-2n} \int_{{N}^{\mathcal{I}}} r^{2}(LT^mP_{\geq 1}\Phi_{(n)})^2+r^{1}(LT^{1+m}P_{\geq 1}\Phi_{(n)})^2\,d\omega dv\\
&+\int_{{N}^{\mathcal{H}}}(r-M)^{-1+\epsilon}(\underline{L}P_{\geq 1}\underline{\Phi}_{(2)})^2 +\sum_{n=0}^1 \sum_{m=0}^{3-2n} \int_{{N}^{\mathcal{H}}} (r-M)^{-2}(\underline{L}T^mP_{\geq 1}\underline{\Phi}_{(n)})^2\\
&+(r-M)^{-1}(\underline{L}T^{1+m}P_{\geq 1}\underline{\Phi}_{(n)})^2\,d\omega du+\sum_{m=0}^{5}\int_{\Sigma_0} J^T[T^m\psi_{\geq 1}]\cdot \mathbf{n}_{\Sigma_0}\,d\mu_{\Sigma_0}.
\end{split}
\end{equation*}
\end{proposition}
\begin{proof}
In order to prove \eqref{eq:edecayl01} we follow a similar strategy to the proof of Proposition \ref{prop:edaypsi0}: we apply the mean value theorem on dyadic intervals. However, in this case we appeal to the hierarchies of $r$-weighted estimates in Proposition \ref{prop:generalrpest}, where we take $n=0,1,2$ and we \emph{relate} the $r$-weighted estimates at different $n$ via the following estimate: for $p<4$, we have that
\begin{equation*}
\begin{split}
\int_{u_1}^{u_2}\int_{{N}^{\mathcal{I}}_{u}} r^{p-1} (L \chi \Phi_{(n)})^2\,d\omega dvdu\lesssim&\: \int_{u_1}^{u_2}\int_{{N}^{\mathcal{I}}_{u}} r^{p-5}  \chi^2 (\Phi_{(n+1)})^2\,d\omega dv+ \sum_{k=0}^n \int_{\Sigma_{u_1}}\int_{\Sigma_{u_1}} J^T[T^k\psi ]\cdot \mathbf{n}_{\Sigma_{u_1}}\,d\mu_{\Sigma_{u_1}}\\
\lesssim&\: \int_{u_1}^{u_2}\int_{{N}^{\mathcal{I}}_{u}} r^{p-3}  (L\chi \Phi_{(n+1)})^2\,d\omega dvdu+ \sum_{k=0}^n \int_{\Sigma_{u_1}}\int_{\Sigma_{u_1}} J^T[T^k\psi ]\cdot \mathbf{n}_{\Sigma_{u_1}}\,d\mu_{\Sigma_{u_1}},
\end{split}
\end{equation*}
where $\chi$ is the cut-off function that appears in the estimates in $\mathcal{A}^{\mathcal{I}}$ in the proof of Proposition  \ref{prop:generalrpest} and the first inequality on the right-hand side above follows from the Morawetz inequality \eqref{eq:iledwayps}, whereas the second inequality follows from the Hardy inequality \eqref{eq:hardyinf}. We can similarly estimate for $p<4$:
\begin{equation*}
\begin{split}
\int_{v_1}^{v_2}\int_{{N}^{\mathcal{H}}_{v}} (r-M)^{-p+1} (\underline{L} \chi \underline{\Phi}_{(n)})^2\,d\omega dudv\lesssim&\: \int_{v_1}^{v_2}\int_{{N}^{\mathcal{H}}_{v}} (r-M)^{-p+3}   (\underline{L} \chi \underline{\Phi}_{(n+1)})^2\,d\omega dudv\\
&+ \sum_{k=0}^n \int_{\Sigma_{u_1}}\int_{\Sigma_{u_1}} J^T[T^k\psi ]\cdot \mathbf{n}_{\Sigma_{u_1}}\,d\mu_{\Sigma_{u_1}},
\end{split}
\end{equation*}
where $\chi$ here denotes the cut-off function that appears in the estimates in $\mathcal{A}^{\mathcal{H}}$ in the proof of Proposition \ref{prop:generalrpest}.

In the $n=2$ case, the left-hand side of the $r$-weighted estimates \eqref{eq:rpestinf} and \eqref{eq:rpesthor} are not positive definite due to the presence of the spacetime integrals of the terms
\begin{align*}
(2-p)&r^{p-3}\left[|\snabla_{\s^2}P_{\geq 1} \Phi_{(2)}|^2-6(P_{\geq 1} \Phi_{(2)})^2\right]\quad\textnormal{and}\\
(2-p)&(r-M)^{-p+3}\left[|\snabla_{\s^2}P_{\geq 1} \underline{\Phi}_{(2)}|^2-6(P_{\geq 1} \underline{\Phi}_{(2)})^2\right].
\end{align*}
Note however that, after applying the Poincar\'e inequality it follows that the above terms \emph{are} positive definite for $P_{\geq 2}\Phi_{(2)}$. To deal with the $\ell=1$ case, we instead use that for $p<2$:
\begin{equation*}
\begin{split}
\int_{u_1}^{u_2}\int_{{N}^{\mathcal{I}}_{u}} r^{p-3}\chi^2(P_{1}\Phi_{(2)})^2\,d\omega dv du \lesssim&\: \int_{u_1}^{u_2}\int_{{N}^{\mathcal{I}}_{u}} r^{p+1}(L\chi P_{1} \widetilde{\Phi}_{(1)})^2+r^{p-3}\chi^2{\Phi}_{(1)}^2\,d\omega dv du\\
&+\sum_{k=0}^1 \int_{\Sigma_{u_1}}\int_{\Sigma_{u_1}} J^T[T^k\psi ]\cdot \mathbf{n}_{\Sigma_{u_1}}\,d\mu_{\Sigma_{u_1}}\\
\lesssim&\:  \int_{u_1}^{u_2}\int_{{N}^{\mathcal{I}}_{u}} r^{p+1}(L\chi P_{1} \widetilde{\Phi}_{(1)})^2+r^{p-1}\chi^2(LP_{1} \Phi_{(1)})^2\,d\omega dv\\
&+\sum_{k=0}^1 \int_{\Sigma_{u_1}}\int_{\Sigma_{u_1}} J^T[T^k\psi ]\cdot \mathbf{n}_{\Sigma_{u_1}}\,d\mu_{\Sigma_{u_1}}\\
\lesssim&\:  \int_{u_1}^{u_2}\int_{{N}^{\mathcal{I}}_{u}} r^{p+1}(L\chi P_{1} \widetilde{\Phi}_{(1)})^2+r^{p-5}\chi^2(P_{1}\Phi_{(2)})^2\,d\omega dv du\\
&+\sum_{k=0}^1 \int_{\Sigma_{u_1}}\int_{\Sigma_{u_1}} J^T[T^k\psi ]\cdot \mathbf{n}_{\Sigma_{u_1}}\,d\mu_{\Sigma_{u_1}},
\end{split}
\end{equation*}
where we made use of \eqref{eq:iledwayps} and \eqref{eq:hardyinf}. Note that the second term on the very right-hand side above can be absorbed into the left-hand side for sufficiently large $r_{\mathcal{I}}$, similarly, for $r_{\mathcal{H}}-M$ suitably small, we have that for $p<2$:
\begin{equation*}
\begin{split}
\int_{v_1}^{v_2}\int_{{N}^{\mathcal{H}}_{v}} (r-M)^{-p+3}\chi^2(P_{1}\underline{\Phi}_{(2)})^2\,d\omega du dv 
\lesssim&\:  \int_{v_1}^{v_2}\int_{{N}^{\mathcal{H}}_{v}} (r-M)^{-p-1}(L\chi P_{1} \widetilde{\underline{\Phi}}_{(1)})^2\,d\omega dudv\\
&+\sum_{k=0}^1 \int_{\Sigma_{v_1}}\int_{\Sigma_{v_1}} J^T[T^k\psi ]\cdot \mathbf{n}_{\Sigma_{v_1}}\,d\mu_{\Sigma_{v_1}}.
\end{split}
\end{equation*}
The above estimate therefore allow us to use the $r$-weighted estimates in Proposition \ref{prop:rpestell01} with $n=1$ and $p<3$ to estimate the integrals of $r^{p-3}\chi^2(P_{1}\Phi_{(2)})^2$ and $(r-M)^{-p+3}\chi^2(P_{1}\underline{\Phi}_{(2)})^2$ \emph{first} and then combine these with the $n=2$ estimates in Proposition \ref{prop:generalrpest} for $p<1$.

Finally, we note that in order to estimate a \emph{global} integrated energy, we moreover apply the Morawetz estimate \eqref{eq:iledlossder} which has a loss of $T$-derivatives on the right-hand side, and therefore the initial data norms that appear on the right-hand side of the time decay estimates will have \emph{additional} $T$-derivatives.
\end{proof}

We will also need the following energy decay estimate that involves higher-order weights in $r$ or $(r-M)^{-1}$ in the initial data norms.
\begin{proposition}
\label{prop:edaypsi1b}
Assume that for $n=0,1$ and $0\leq k \leq 2-n$:
\begin{equation*}
\lim_{v\to \infty} |\snabla_{\s^2}\slashed{\Delta}_{\s^2}^kP_{\geq 1}\Phi_{(n)}|^2\,d\omega<\infty.
\end{equation*}
Then, for all $\epsilon>0$ there exists a constant $C=C(M,{\Sigma_0},\epsilon)>0$ such that
\begin{align}
\label{eq:extrarweightsl1inf}
\int_{{N}^{\mathcal{I}}_{\tau}}& r^2(L \phi_{\geq 1})^2\,d\omega dv \leq  C\cdot  E^{\epsilon}_{1,\mathcal{I}}[\psi] (1+\tau)^{-4+\epsilon},\\
\label{eq:extrarweightsl1hor}
\int_{{N}^{\mathcal{H}}_{\tau}}& (r-M)^{-2}(\underline{L} \phi_{\geq 1})^2\,d\omega du \leq   C\cdot E^{\epsilon}_{1,\mathcal{H}}[\psi](1+\tau)^{-4+\epsilon},\\
\label{eq:edecayl11b}
\int_{\Sigma_{\tau}}& J^T[ \psi_{\geq 1}]\cdot \mathbf{n}_{\tau}\,d\mu_{\tau}\leq  C\cdot{E}^{\epsilon}_{2}[\psi]  (1+\tau)^{-6+\epsilon}.
\end{align}
with
\begin{align*}
E^{\epsilon}_{1,\mathcal{I}}[\psi]:=&\: E^{\epsilon}_{1}[\psi]+  \int_{{N}^{\mathcal{I}}_0} r^{2-\epsilon}(LP_{\geq 1}\Phi_{(2)})^2+r^{4-\epsilon}(LP_1\widetilde{\Phi}_{(1)})^2\,d\omega dv,\\
E^{\epsilon}_{1,\mathcal{H}}[\psi]:=&\: E^{\epsilon}_{1}[\psi]+ \int_{{N}^{\mathcal{H}}_0}(r-M)^{-2+\epsilon}(\underline{L}P_{\geq 1}\underline{\Phi}_{(2)})^2\,d\omega du,\\
{E}^{\epsilon}_{2}[\psi]:=&\: \int_{{N}^{\mathcal{I}}} r^{2-\epsilon}(LP_{\geq 1}\Phi_{(2)})^2+r^{1-\epsilon}(LTP_{\geq 1}\Phi_{(2)})^2+r^{4-\epsilon}(LP_1\widetilde{\Phi}_{(1)})^2+r^{3-\epsilon}(LTP_1\widetilde{\Phi}_{(1)})^2\,d\omega dv\\
&+ \sum_{n=0}^1 \sum_{m=0}^{4-2n} \int_{{N}^{\mathcal{I}}} r^{2}(LT^mP_{\geq 1}\Phi_{(n)})^2+r^{1}(LT^{1+m}P_{\geq 1}\Phi_{(n)})^2\,d\omega dv\\
&+\int_{{N}^{\mathcal{H}}}(r-M)^{-2+\epsilon}(\underline{L}P_{\geq 1}\underline{\Phi}_{(2)})^2+(r-M)^{-1+\epsilon}(\underline{L}TP_{\geq 1}\underline{\Phi}_{(2)})^2+(r-M)^{-4+\epsilon}(LP_1\widetilde{\underline{\Phi}}_{(1)})^2\\
&+(r-M)^{-3+\epsilon}(\underline{L}TP_1\widetilde{\underline{\Phi}}_{(1)})^2\,d\omega du\\
&+ \sum_{n=0}^1 \sum_{m=0}^{4-2n} \int_{{N}^{\mathcal{H}}} (r-M)^{-2}(\underline{L}T^mP_{\geq 1}\underline{\Phi}_{(n)})^2+(r-M)^{-1}(\underline{L}T^{1+m}P_{\geq 1}\underline{\Phi}_{(n)})^2\,d\omega du\\
&+\sum_{m=0}^{6}\int_{\Sigma_0} J^T[T^m\psi_{\geq 1}]\cdot \mathbf{n}_{\Sigma_0}\,d\mu_{\Sigma_0}.
\end{align*}
\end{proposition}
\begin{proof}
In order to prove \eqref{eq:edecayl11b} we proceed as in the proof of Proposition \ref{prop:edaypsi1}, but we consider additionally the $r$-weighted estimates in Proposition \ref{prop:generalrpest}  for $n=2$ with $p=2-\epsilon$ and the $r$-weighted estimates in Proposition \ref{prop:rpestell01} for $n=1$ with $p=4-\epsilon$, resulting in an additional power in the energy decay rate.

In the process of deriving \eqref{eq:edecayl11b}, we also obtain \eqref{eq:extrarweightsl1inf} and \ref{eq:extrarweightsl1hor}, but we require a weaker energy norm on the right-hand side. That is to say, the energy norm will contain higher powers in \underline{\emph{either}} r \underline{\emph{or}} $(r-M)^{-1}$ on the right hand side (and not both, as in the norm ${E}^{\epsilon}_{2}[\psi]$).
\end{proof}

\subsection{Improved energy decay estimates for time derivatives}
\label{sec:hoedecayest}
In this section, we obtain improved decay estimates for the time-derivatives $T^J\psi$.
\begin{proposition}
\label{prop:hoedaypsi0}
Let $J\in \N_0$ and assume that for all $0\leq j\leq J-1$,
\begin{align*}
\lim_{v\to \infty} r^{j+2} L^{j+1}\phi_0(u_0,v)<&\infty.
\end{align*}
Then, for all $\epsilon>0$, there exists a constant $C=C(M,{\Sigma_0},\epsilon,J)>0$ such that
\begin{align}
\label{eq:hoedecayl01a}
\int_{\Sigma_{\tau}}& J^T[ T^J\psi_{0}]\cdot \mathbf{n}_{\tau}\,d\mu_{\tau}\leq C\cdot E^{\epsilon}_{0; J}[\psi]  (1+\tau)^{-3-2J+\epsilon}, \\
\label{eq:hor2edecayl01}
\sum_{\substack{j_1+j_2=J\\ j_1\geq 0, j_2\geq 0}}\int_{{N}^{\mathcal{I}}_{\tau}}& r^{2+2j_1}\cdot (L^{1+j_1}T^{j_2}\phi_0)^2\,d\omega dv+\int_{{N}^{\mathcal{H}}_{\tau}} (r-M)^{-2-2j_2}\cdot (\underline{L}^{j_1+1} T^{j_2}\phi_0)^2\,d\omega du\\ \nonumber
\leq&\: C\cdot E^{\epsilon}_{0;J}[\psi]  (1+\tau)^{-1-2j_2+\epsilon}.
\end{align}
with
\begin{equation*}
E^{\epsilon}_{0;J}[\psi]:= \int_{{N}^{\mathcal{I}}} r^{3+2J-\epsilon}(L^{1+J}\phi_0)^2\,d\omega dv+\int_{{N}^{\mathcal{H}}} (r-M)^{-3-2J+\epsilon}(\underline{L}^{1+J}\phi_0)^2\,d\omega du+\sum_{j=0}^J\int_{\Sigma_0} J^T[T^j\psi_0]\cdot \mathbf{n}_{\Sigma_0}\,d\mu_{\Sigma_0}.
\end{equation*}
\end{proposition}
\begin{proof}
We follow the steps in the proof of Proposition \ref{prop:edaypsi0}, applied to $T^J\psi$ instead of $\psi$ and we \emph{extend} the hierarchy of $r$-weighted estimates in the case $J\geq 1$ by   employing the additional $r$-weighted estimates in Proposition \ref{prop:rpestell01Lkpsi} and relating them to the original hierarchy of estimates (for $J=0$) via \eqref{eq:intestTpsihor} and \eqref{eq:intestTpsiinfl0}.
\end{proof}

In analogy with Lemma \ref{lm:auxedaypsi0}, we now obtain auxiliary decay estimates for $J\geq 1$ that will be useful when improving the estimates Proposition \ref{prop:hoedaypsi0} in the setting where either $H_0[\psi]=0$ or $I_0[\psi]=0$. 

\begin{lemma}
\label{lm:hoauxedaypsi0}
Let $J\in \N_0$ and assume that for all $0\leq j\leq J-1$,
\begin{align*}
\lim_{v\to \infty} r^{j+3} L^{j+1}\phi_0(u_0,v)<&\infty.
\end{align*}
Then, for all $\epsilon>0$, there exists a constant $C=C(M,{\Sigma_0},\epsilon,J)>0$ such that
\begin{align}
\label{eq:hoauxdecayl01}
\sum_{\substack{j_1+j_2=J\\ j_1\geq 0, j_2\geq 0}}&\int_{{N}^{\mathcal{I}}_{\tau}}r^{2+2j_1}\cdot (L^{1+j_1}T^{j_2}\phi_0)^2\,d\omega dv\\
\leq&\: C\cdot \left[E^{\epsilon}_{0;J}[\psi]+ \int_{{N}^{\mathcal{I}}_{0}} r^{4+2J-\epsilon}\cdot (L^{1+J}\phi_0)^2\,d\omega dv\right] (1+\tau)^{-2-2j_2+\epsilon},\\
\label{eq:hoauxdecayl02}
\sum_{\substack{j_1+j_2=J\\ j_1\geq 0, j_2\geq 0}}&\int_{{N}^{\mathcal{H}}_{\tau}} (r-M)^{-2-2j_2}\cdot (\underline{L}^{j_1+1} T^{j_2}\phi_0)^2\,d\omega du\\
\leq &\: C\cdot \left[E^{\epsilon}_{0;J}[\psi]+\int_{{N}^{\mathcal{H}}_{0}} (r-M)^{-4-2J+\epsilon}\cdot (\underline{L}^{1+J} \phi_0)^2\,d\omega du\right] (1+\tau)^{-2-2j_2+\epsilon}.
\end{align}
Furthermore,
\begin{align}
\label{eq:hoauxpointdecayl01}
||r^{J}L^J\phi_0||^2_{L^{\infty}({N}^{\mathcal{I}}_{\tau})}\leq&\: C\cdot \left[E^{\epsilon}_{0}[\psi]+\int_{{N}^{\mathcal{I}}_{0}} r^{4+2J-\epsilon}\cdot (L^{1+J}\phi_0)^2\,d\omega dv\right] (1+\tau)^{-\frac{5}{2}+\epsilon},\\
\label{eq:hoauxpointdecayl02}
||r^{J}L^J\phi_0||^2_{L^{\infty}({N}^{\mathcal{I}}_{\tau})}\leq &\: C\cdot \left[E^{\epsilon}_{0}[\psi]+ \int_{{N}^{\mathcal{H}}_{0}} (r-M)^{-4-2J+\epsilon}\cdot (\underline{L}^{1+J} \phi_0)^2\,d\omega du\right] (1+\tau)^{-\frac{5}{2}+\epsilon}.
\end{align}
\end{lemma}
\begin{proof}
We obtain \eqref{eq:hoauxdecayl01} and \eqref{eq:hoauxdecayl02} by adding one estimate to the hierarchies of Proposition \ref{prop:hoedaypsi0}, in analogy to the proof of  Lemma \ref{lm:auxedaypsi0}. In order to obtain \eqref{eq:hoauxpointdecayl01}, we need additionally use the estimate
\begin{equation*}
\int_{{N}^{\mathcal{I}}_{\tau}}r^{1+J}\cdot (L^{1+J}\phi_0)^2\,d\omega dv\leq C\cdot \left[E^{\epsilon}_{0;J}[\psi]+\int_{{N}^{\mathcal{I}}_{0}} r^{4+2J-\epsilon}\cdot (L^{1+J}\phi_0)^2\,d\omega dv\right] (1+\tau)^{-3-2J+\epsilon}.
\end{equation*}
In order to prove \eqref{eq:hoauxdecayl01}, we then apply the fundamental theorem of calculus (together with the Morawetz estimate \eqref{eq:iledwayps} and the Hardy inequality \eqref{eq:hardyinf}) to estimate $r^{2J} \chi^2 (L^J\Phi_0)^2$ using the estimate above together with \eqref{eq:hoauxdecayl01}.  We arrive at \eqref{eq:hoauxpointdecayl02} by applying similar arguments in the region $\mathcal{A}^{\mathcal{H}}$.
\end{proof}
                
The $L^{\infty}$ estimates in Lemma \ref{lm:hoauxedaypsi0} allow us to retrieve the assumptions \eqref{eq:hoauxdecayassm1} and \eqref{eq:hoauxdecayassm2} so that we can use the \emph{full} hierarchy of $r$-weighted estimates in Proposition \ref{prop:hoedaypsi0} to obtain the following analog of Proposition \ref{prop:extraendecay}:

\begin{proposition}
\label{prop:hoextraendecay}
Let $J\in \N_0$. Then, for all $\epsilon>0$, there exists a constant $C=C(M,{\Sigma_0},\epsilon,J)>0$ such that
\begin{align}
\label{eq:hooptr2decay1}
\sum_{\substack{j_1+j_2=J\\ j_1\geq 0, j_2\geq 0}}\int_{{N}^{\mathcal{I}}_{\tau}}r^{2+2j_1}\cdot (L^{1+j_1}T^{j_2}\phi_0)^2\,d\omega dv \leq&\: C\cdot E^{\epsilon}_{0,\mathcal{I};J}[\psi] (1+\tau)^{-3-2j_2+\epsilon},\\ 
\label{eq:hooptr2decay2}
\sum_{\substack{j_1+j_2=J\\ j_1\geq 0, j_2\geq 0}}\int_{{N}^{\mathcal{H}}_{\tau}} (r-M)^{-2-2j_2}\cdot (\underline{L}^{1+j_1} T^{j_2}\phi_0)^2\,d\omega du \leq&\:  C\cdot E^{\epsilon}_{0,\mathcal{H};J}[\psi] (1+\tau)^{-3-2j_2+\epsilon},
\end{align}
where
\begin{align*}
E^{\epsilon}_{0,\mathcal{I};J}[\psi]=:&\: E^{\epsilon}_{0;J}[\psi]+ \int_{{N}^{\mathcal{I}}_{0}} r^{5+2J-\epsilon}\cdot (L^{1+J}\phi_0)^2\,d\omega dv,\\
E^{\epsilon}_{0,\mathcal{H};J}[\psi]=:&\: E^{\epsilon}_{0;J}[\psi]+ \int_{{N}^{\mathcal{H}}_{0}} (r-M)^{-5-2J+\epsilon}\cdot (\underline{L}^{1+J} \phi_0)^2\,d\omega du,
\end{align*}
and we assume \eqref{eq:asmLjphi_02} for \eqref{eq:hooptr2decay1} and \eqref{eq:asmLjphi_01} for \eqref{eq:hooptr2decay2}.
\end{proposition}

We now define the following higher-order weighted energy norms for $\psi_{\geq 1}$ and $J\in \N_0$.

\begin{equation*}
\begin{split}
E^{\epsilon}_{1;J}[\psi]:=&\: \sum_{\substack{0\leq j+|\alpha|\leq J\\ |\alpha|\leq J}} \int_{{N}^{\mathcal{I}}} r^{1-\epsilon}|\snabla_{\s^2}^{\alpha}LP_{\geq 1}T^j\Phi_{(2)}|^2\,d\omega dv\\
&+ \sum_{n=0}^1 \sum_{m=0}^{3-2n+2J} \int_{{N}^{\mathcal{I}}} r^{2}(LT^mP_{\geq 1}\Phi_{(n)})^2+r^{1}(LT^{1+m}P_{\geq 1}\Phi_{(n)})^2\,d\omega dv\\
&+\sum_{\substack{|\alpha|\leq \max\{0,J-1\}\\ i\leq \max\{J-1,0\}\\ 0\leq j+|\alpha|\leq J-2i+\min\{J,1\}}} \int_{{N}^{\mathcal{I}}} r^{1+2i-\epsilon}|\snabla_{\s^2}^{\alpha}L^{1+i}P_{\geq 1}T^j\Phi_{(2)}|^2\, d\omega dv\\
&+\sum_{\substack{|\alpha|\leq \max\{0,J-1\}\\ i\leq J\\ 0\leq j+|\alpha|\leq 2J-2i+1}} \int_{{N}^{\mathcal{I}}} r^{2i-\epsilon}|\snabla_{\s^2}^{\alpha}L^{1+i}P_{\geq 1}T^j\Phi_{(2)}|^2\, d\omega dv\\
&+\int_{{N}^{\mathcal{I}}} r^{2+2J-\epsilon}(L^{1+J}P_{\geq 1}\Phi_{(2)})^2\,d\omega dv\\
&+\sum_{\substack{0\leq j+|\alpha|\leq J\\ |\alpha|\leq  J}} \int_{{N}^{\mathcal{H}}} (r-M)^{-1+\epsilon}|\snabla_{\s^2}^{\alpha}\underline{L}P_{\geq 1}T^j\underline{\Phi}_{(2)}|^2\,d\omega du\\
&+ \sum_{n=0}^1 \sum_{m=0}^{3-2n+2J} \int_{{N}^{\mathcal{H}}} (r-M)^{-2}(\underline{L}T^mP_{\geq 1}\underline{\Phi}_{(n)})^2+(r-M)(\underline{L}T^{1+m}P_{\geq 1}\underline{\Phi}_{(n)})^2\,d\omega du\\
&+\sum_{\substack{|\alpha|\leq \max\{0,J-1\}\\ i\leq \max\{J-1,0\}\\ 0\leq j+|\alpha|\leq J-2i+\min\{J,1\}}} \int_{{N}^{\mathcal{H}}} (r-M)^{-1-2i+\epsilon}|\snabla_{\s^2}^{\alpha}\underline{L}^{1+i}P_{\geq 1}T^j\underline{\Phi}_{(2)}|^2\, d\omega du\\
&+\sum_{\substack{|\alpha|\leq \max\{0,J-1\}\\ i\leq J\\ 0\leq j+|\alpha|\leq 2J-2i+1}} \int_{{N}^{\mathcal{H}}} (r-M)^{-2i+\epsilon}|\snabla_{\s^2}^{\alpha}\underline{L}^{1+i}P_{\geq 1}T^j\underline{\Phi}_{(2)}|^2\, d\omega du\\
&+\int_{{N}^{\mathcal{H}}} (r-M)^{-1-2J+\epsilon}(\underline{L}^{1+J}P_{\geq 1}\underline{\Phi}_{(2)})^2\,d\omega du\\
&+\sum_{\substack{j+|\alpha|\leq 5+3k\\ |\alpha|\leq k}}\int_{\Sigma_0} J^T[T^j\Omega^{\alpha}\psi_{\geq 1}]\cdot \mathbf{n}_{\Sigma_0}\,d\mu_{\Sigma_0}.
\end{split}
\end{equation*}
and
\begin{equation*}
\begin{split}
E^{\epsilon}_{2;J}[\psi]:=&\: \sum_{\substack{0\leq j+|\alpha|\leq J\\ |\alpha|\leq J}} \int_{{N}^{\mathcal{I}}} r^{2-\epsilon}|\snabla_{\s^2}^{\alpha}LP_{\geq 1}T^j\Phi_{(2)}|^2+r^{1-\epsilon}|\snabla_{\s^2}^{\alpha}LT^{1+j}P_{\geq 1}\Phi_{(2)}|^2+r^{4-\epsilon}(LP_1T^j\widetilde{\Phi}_{(1)})^2\,d\omega dv\\
&+ \sum_{n=0}^1 \sum_{m=0}^{4-2n+2J} \int_{{N}^{\mathcal{I}}} r^{2}(LT^mP_{\geq 1}\Phi_{(n)})^2+r^{1}(LT^{1+m}P_{\geq 1}\Phi_{(n)})^2\,d\omega dv\\
&+\sum_{\substack{|\alpha|\leq \max\{0,J-1\}\\ i\leq \max\{J-1,0\}\\ 0\leq j+|\alpha|\leq J-2i+\min\{J,1\}}} \int_{{N}^{\mathcal{I}}} r^{2+2i-\epsilon}|\snabla_{\s^2}^{\alpha}L^{1+i}P_{\geq 1}T^j\Phi_{(2)}|^2\, d\omega dv\\
&+\sum_{\substack{|\alpha|\leq \max\{0,J-1\}\\ i\leq J\\ 0\leq j+|\alpha|\leq 2J-2i+1}} \int_{{N}^{\mathcal{I}}} r^{1+2i-\epsilon}|\snabla_{\s^2}^{\alpha}L^{1+i}P_{\geq 1}T^j\Phi_{(2)}|^2\, d\omega dv\\
&+\int_{{N}^{\mathcal{I}}} r^{2+2J-\epsilon}(L^{1+J}P_{\geq 1}\Phi_{(2)})^2\,d\omega dv\\
&+\sum_{\substack{0\leq j+|\alpha|\leq J\\ |\alpha|\leq J}} \int_{{N}^{\mathcal{H}}} (r-M)^{-2+\epsilon}|\snabla_{\s^2}^{\alpha}\underline{L}P_{\geq 1}T^j\underline{\Phi}_{(2)}|^2+(r-M)^{-1+\epsilon}|\snabla_{\s^2}^{\alpha}\underline{L}T^{1+j}P_{\geq 1}\underline{\Phi}_{(2)}|^2\\
&+(r-M)^{-4+\epsilon}(\underline{L}P_1T^j\widetilde{\underline{\Phi}}_{(1)})^2\,d\omega du\\
&+ \sum_{n=0}^1 \sum_{m=0}^{4-2n+2J} \int_{{N}^{\mathcal{H}}} (r-M)^{-2}(\underline{L}T^mP_{\geq 1}\underline{\Phi}_{(n)})^2+(r-M)(\underline{L}T^{1+m}P_{\geq 1}\underline{\Phi}_{(n)})^2\,d\omega du\\
&+\sum_{\substack{|\alpha|\leq \max\{0,J-1\}\\ i\leq \max\{J-1,0\}\\ 0\leq j+|\alpha|\leq J-2i+\min\{J,1\}}} \int_{{N}^{\mathcal{H}}} (r-M)^{-2-2i+\epsilon}|\snabla_{\s^2}^{\alpha}\underline{L}^{1+i}P_{\geq 1}T^j\underline{\Phi}_{(2)}|^2\, d\omega du\\
&+\sum_{\substack{|\alpha|\leq \max\{0,J-1\}\\ i\leq J\\ 0\leq j+|\alpha|\leq 2J-2i+1}} \int_{{N}^{\mathcal{H}}} (r-M)^{-1-2i+\epsilon}|\snabla_{\s^2}^{\alpha}\underline{L}^{1+i}P_{\geq 1}T^j\underline{\Phi}_{(2)}|^2\, d\omega du\\
&+\int_{{N}^{\mathcal{H}}} (r-M)^{-2-2J+\epsilon}(\underline{L}^{1+J}P_{\geq 1}\underline{\Phi}_{(2)})^2\,d\omega du\\
&+\sum_{\substack{j+|\alpha|\leq 6+3k\\ |\alpha|\leq k}}\int_{\Sigma_0} J^T[T^j\Omega^{\alpha}\psi_{\geq 1}]\cdot \mathbf{n}_{\Sigma_0}\,d\mu_{\Sigma_0}.
\end{split}
\end{equation*}
We note that the integrals appearing in the energy norm $E^{\epsilon}_{1;J}[\psi]$ that are supported \emph{away} from $\mathcal{H}^+$ are similar to the energy norms defined in Proposition 7.7 of \cite{paper1} and similarly, the integrals supported away from the horizon in $E^{\epsilon}_{2;J}[\psi]$ appear in the norms defined in Appendix A of \cite{paper2}.

We obtain energy decay estimates for the time derivatives of $\psi_{\geq 1}$.
\begin{proposition}
\label{prop:hoedaypsi1}
Let $J\in \N_0$. Assume that for $n=0,1$ and for all $0\leq k\leq 2-n$ and $0\leq j\leq J$:
\begin{equation*}
\lim_{v\to \infty}\int_{\s^2}r^{2j}|\snabla_{\s^2}\slashed{\Delta}_{\s^2}^kL^jP_{\geq 1}\Phi_{(n)}|^2\,d\omega|_{u=u_0}<\infty.
\end{equation*}
Then, for all $\epsilon>0$ there exists a constant $C=C(M,{\Sigma_0},\epsilon,J)>0$ such that
\begin{align}
\label{eq:hoedecayl01}
\int_{\Sigma_{\tau}} J^T[ T^J\psi_{\geq 1}]\cdot \mathbf{n}_{\tau}\,d\mu_{\tau}\leq&\: C\cdot E^{\epsilon}_{1;J}[\psi]  (1+\tau)^{-5-2J+\epsilon},\\
\label{eq:hoedecaylr211}
\int_{{N}^{\mathcal{I}}_{\tau}} r^2(LT^J\phi_{\geq 1})^2\,d\omega dv+\int_{{N}^{\mathcal{H}}_{\tau}} (r-M)^{-2}(\underline{L}T^J\phi_{\geq 1})^2\,d\omega du \leq&\: C\cdot E^{\epsilon}_{1;J}[\psi]  (1+\tau)^{-3-2J+\epsilon}.
\end{align}
\end{proposition}
\begin{proof}
The $J=0$ case follows from Proposition \ref{prop:edaypsi1}. In order to obtain the estimates for $J\geq 1$ we repeat the steps in the proof of Proposition \ref{prop:edaypsi1} applied to $T^J\phi_{\geq 1}$, but we use additionally that for $J\geq 1$ we can extend the hierarchy of $r$-weighted estimates by using the $r$-weighted estimates for $L^J\Phi_{(2)}$ that follow from Proposition \ref{prop:generalrpestLkpsi} and combining them with the estimates in Lemma \ref{lm:intestTpsi} to extend the hierarchy for $T^J\Phi_{(2)}$. Note that we need to distinguish between $P_{\geq 2}\Phi_{(2)}$ and $P_{1}\Phi_{(2)}$ when estimating $J=0$ terms, so for $J\geq 1$ there is no further need to perform the splitting $P_{\geq 2}L^J\Phi_{(2)}=P_{1}L^J\Phi_{(2)}+P_{\geq 2}L^J\Phi_{(2)}$. We omit further details of the proof and refer to the proof of Proposition 7.7 in \cite{paper1} for more details on estimates that are analogous to the estimates in $\mathcal{A}^{\mathcal{I}}$ in extremal Reissner--Nordstr\"om. The estimates in $\mathcal{A}^{\mathcal{H}}$ follow via very similar arguments.
\end{proof}

\begin{proposition}
\label{prop:hoedaypsi1b}
Let $J\in \N_0$. Assume that
\begin{equation*}
\lim_{v\to \infty}\sum_{j=0}^J\sum_{n=0}^2 \sum_{k=0}^{2-n}\int_{\s^2}r^{2j}|\snabla_{\s^2}\slashed{\Delta}_{\s^2}^kL^jP_{\geq 1}\Phi_{(n)}|^2\,d\omega|_{u=u_0}<\infty.
\end{equation*}
Then, for all $\epsilon>0$ there exists a constant $C=C(M,{\Sigma_0},\epsilon,J)>0$ such that
\begin{align}
\label{eq:hoedecaylr211inf}
\int_{{N}^{\mathcal{I}}_{\tau}} r^2(LT^J\phi_{\geq 1})^2\,d\omega dv\leq&\: C\cdot E^{\epsilon}_{1,\mathcal{I};J}[\psi]  (1+\tau)^{-4-2J+\epsilon},\\
\label{eq:hoedecaylr211hor}
\int_{{N}^{\mathcal{H}}_{\tau}} (r-M)^{-2}(\underline{L}T^J\phi_{\geq 1})^2\,d\omega du \leq&\: C\cdot E^{\epsilon}_{1,\mathcal{H};J}[\psi]  (1+\tau)^{-4-2J+\epsilon},\\
\label{eq:hoedecayl01b}
\int_{\Sigma_{\tau}} J^T[ T^J\psi_{\geq 1}]\cdot \mathbf{n}_{\tau}\,d\mu_{\tau}\leq&\: C\cdot E^{\epsilon}_{2;J}[\psi]  (1+\tau)^{-6-2J+\epsilon},\\
\end{align}
where
\begin{align*}
E^{\epsilon}_{1,\mathcal{I};J}[\psi]:=&\: E^{\epsilon}_{1;J}[\psi] +  \sum_{\substack{0\leq j+|\alpha|\leq J\\ |\alpha|\leq J}} \int_{{N}^{\mathcal{I}}} r^{2-\epsilon}|\snabla_{\s^2}^{\alpha}LP_{\geq 1}T^j\Phi_{(2)}|^2+r^{4-\epsilon}(LP_1T^j\widetilde{\Phi}_{(1)})^2\,d\omega dv  \\
&+\sum_{\substack{|\alpha|\leq \max\{0,J-1\}\\ i\leq \max\{J-1,0\}\\ 0\leq j+|\alpha|\leq J-2i+\min\{J,1\}}} \int_{{N}^{\mathcal{I}}} r^{2+2i-\epsilon}|\snabla_{\s^2}^{\alpha}L^{1+i}P_{\geq 1}T^j\Phi_{(2)}|^2\, d\omega dv\\
&+\int_{{N}^{\mathcal{I}}} r^{2+2J-\epsilon}(L^{1+J}P_{\geq 1}\Phi_{(2)})^2\,d\omega dv,\\
E^{\epsilon}_{1,\mathcal{H};J}[\psi]:=&\: E^{\epsilon}_{1;J}[\psi] +  \sum_{\substack{0\leq j+|\alpha|\leq J\\ |\alpha|\leq J}} \int_{{N}^{\mathcal{H}}} (r-M)^{-2+\epsilon}|\snabla_{\s^2}^{\alpha}\underline{L}P_{\geq 1}T^j\underline{\Phi}_{(2)}|^2+(r-M)^{-4+\epsilon}(\underline{L}P_1T^j\widetilde{\underline{\Phi}}_{(1)})^2\,d\omega du  \\
&+\sum_{\substack{|\alpha|\leq \max\{0,J-1\}\\ i\leq \max\{J-1,0\}\\ 0\leq j+|\alpha|\leq J-2i+\min\{J,1\}}} \int_{{N}^{\mathcal{H}}} (r-M)^{-2-2i+\epsilon}|\snabla_{\s^2}^{\alpha}\underline{L}^{1+i}P_{\geq 1}T^j \underline{\Phi}_{(2)}|^2\, d\omega du\\
&+\int_{{N}^{\mathcal{H}}} (r-M)^{-2-2J+\epsilon}(\underline{L}^{1+J}P_{\geq 1}\underline{\Phi}_{(2)})^2\,d\omega du.\\
\end{align*}
\end{proposition}
\begin{proof}
For $J=0$, the estimates in the proposition are contained in Proposition \ref{prop:edaypsi1b}. In order to prove the estimates in the $J\geq 1$ case, we proceed as in Proposition \ref{prop:hoedaypsi1} but we moreover apply the \emph{extended} hierarchy in the $J=0$ case in $\mathcal{A}^{\mathcal{I}}$ or $\mathcal{A}^{\mathcal{H}}$ (or in both regions) that is used in the proof of Proposition \ref{prop:edaypsi1b}.
\end{proof}

\subsection{Degenerate elliptic estimates for $\psi_{\geq 1}$}
\label{sec:ellpest}
In this section we derive a degenerate elliptic estimate for solutions $\psi_{\geq 1}$ to \eqref{eq:waveequation} that will be used to obtain pointwise decay estimates for $\psi_{\geq 1}$.

\begin{proposition}
\label{prop:degenelliptic}
Let $\psi$ be a solution to \eqref{eq:waveequation}. Assume moreover that
\begin{align*}
\lim_{\rho\to \infty}r^{\frac{1}{2}}T\psi|_{\Sigma_{\tau}}=&0,\\
\lim_{\rho\to \infty}r^{\frac{1}{2}}L\psi|_{\Sigma_{\tau}}=&0.
\end{align*}

Then there exists a $C=C(M,\Sigma_0)>0$, such that, with respect to $(\rho,\theta,\varphi)$ coordinates,
\begin{equation}
\label{eq:ellipticpsi}
\begin{split}
\int_{M}^{\infty}&\int_{\s^2} Dr^2(\partial_{\rho}( D\partial_\rho \psi_{\geq 1}))^2+D^2|\snabla_{\s^2} \partial_{\rho}\psi_{\geq 1}|^2+Dr^{-2}(\slashed{\Delta}_{\s^2}\psi_{\geq 1})^2\,d\omega d\rho\\
\leq&\: C\int_{M}^{\infty}\int_{\s^2}Dr^2(\partial_{\rho} T \psi_{\geq 1})^2+Dh(r)^2r^2(T^2\psi_{\geq 1})^2\,d\omega  d\rho.
\end{split}
\end{equation}
\end{proposition}
\begin{proof}
For the sake of convenience, we will assume that $\int_{\s^2}\psi\,d\omega=0$ so that $\psi=\psi_{\geq 1}$. By Lemma 7.9 in \cite{paper1} it follows that \eqref{eq:waveequation} on (extremal) Reissner--Nordstr\"om reduces to the following equation
\begin{equation}
\label{eq:waveeqrhocoord}
\partial_{\rho}(Dr^2\partial_{\rho}\psi)+\slashed{\Delta}_{\s^2}\psi=-2(1-h\cdot D)r^2 \partial_{\rho}T\psi+(2-h\cdot D)r^2hT^2\psi+((hDr^2)' -2r)T\psi.
\end{equation}
By squaring both sides of \eqref{eq:waveeqrhocoord} and multiplying the resulting equation with the factor $Dr^{-2}$, we obtain the following estimate:
\begin{equation}
\label{eq:waveeqestext}
\begin{split}
\int_{M}^{\infty}\int_{\s^2} &r^{-2}D(\partial_{\rho}(Dr^2\partial_{\rho}\psi))^2+r^{-2}D(\slashed{\Delta}_{\s^2}\psi)^2+2r^{-2}D\partial_{\rho}(Dr^2\partial_{\rho}\psi)\slashed{\Delta}_{\s^2}\psi\,d\omega d\rho\\
\leq&\: C\int_{M}^{\infty}\int_{\s^2}r^2D(\partial_{\rho} T \psi)^2+h(r)^2Dr^2(T^2\psi)^2+D(T\psi)^2\,d\omega d\rho.
\end{split}
\end{equation}

By applying \eqref{eq:hardyhor2} again as follows
\begin{equation*}
\int_{M}^{\infty}(T\psi)^2\,d\rho\leq 4\int_{M}^{\infty}(r-M)^2(\partial_{\rho} T\psi)^2\,d\rho,
\end{equation*}
so that
\begin{equation}
\label{eq:waveeqestextv2}
\begin{split}
\int_{M}^{\infty}\int_{\s^2} &r^{-2}D(\partial_{\rho}(Dr^2\partial_{\rho}\psi))^2+r^{-2}D(\slashed{\Delta}_{\s^2}\psi)^2+2r^{-2}D\partial_{\rho}(Dr^2\partial_{\rho}\psi)\slashed{\Delta}_{\s^2}\psi\,d\omega d\rho\\
\leq&\: C\int_{M}^{\infty}\int_{\s^2}Dr^2(\partial_{\rho} T \psi)^2+Dh(r)^2r^2(T^2\psi)^2\,d\omega d\rho.
\end{split}
\end{equation}
We first consider the mixed derivative term on the left-hand side of \eqref{eq:waveeqestextv2}. We integrate over $\s^2$ and integrate by parts in $\rho$ and the angular variables:
\begin{equation}
\begin{split}
\label{eq:ibpangularext}
\int_{M}^{\infty} \int_{\s^2}2Dr^{-2}\partial_{\rho}(Dr^2\partial_{\rho}\psi)\slashed{\Delta}_{\s^2}\psi\,d\omega d\rho=&\:\int_{M}^{\infty} \int_{\s^2}-2Dr^{2}\partial_r(Dr^{-2})\partial_{\rho}\psi\slashed{\Delta}_{\s^2}\psi-2D^2\partial_{\rho}\psi\slashed{\Delta}_{\s^2}\partial_{\rho}\psi\,d\omega d\rho\\
=&\:\int_{M}^{\infty} \int_{\s^2}-2Dr^{2}\partial_r(Dr^{-2})\partial_{\rho}\psi\slashed{\Delta}_{\s^2}\psi+2D^2|\snabla_{\s^2}\partial_{\rho}\psi|^2\,d\omega d\rho,
\end{split}
\end{equation}
where we used that all resulting boundary terms vanish.

Note that
\begin{equation*}
\begin{split}
\partial_r(Dr^{-2})&=\partial_r((r^{-1}-Mr^{-2})^2)=-2(r^{-1}-Mr^{-2})(r^{-2}-2Mr^{-3})\\
&=-2r^{-3}\left(1-\frac{M}{r}\right)\left(1-\frac{2M}{r}\right).
\end{split}
\end{equation*}

We now apply Cauchy--Schwarz and the standard Poincar\'e inequality on $\s^2$ to estimate the first term inside the integral on the very right-hand side above:
\begin{equation*}
\begin{split}
\int_{M}^{\infty}\int_{\s^2}\left|2r^{2}D\partial_r(Dr^{-2})\partial_{\rho}\psi\slashed{\Delta}_{\s^2}\psi\right|\,d\omega d\rho\leq&\: \int_{M}^{\infty}Dr^{6}(\partial_r(Dr^{-2}))^2(\partial_{\rho}\psi)^2+r^{-2}D(\slashed{\Delta}_{\s^2}\psi)^2\,d\omega d\rho\\
= &\: \int_{M}^{\infty}4D^2\left(1-\frac{2M}{r}\right)^2(\partial_{\rho}\psi)^2+r^{-2}D(\slashed{\Delta}_{\s^2}\psi)^2\,d\omega d\rho\\
\leq &\:\int_{M}^{\infty}\int_{\s^2}2D^2|\snabla_{\s^2}\partial_{\rho}\psi|^2+r^{-2}D(\slashed{\Delta}_{\s^2}\psi)^2\,d\omega d\rho,
\end{split}
\end{equation*}
where we used moreover that $\left|1-\frac{2M}{r}\right|\leq 1$.

We use the above estimates together with \eqref{eq:waveeqestextv2} to estimate:
\begin{equation}
\label{eq:mainineqellipticext}
\begin{split}
\int_{M}^{\infty}\int_{\s^2} &Dr^{-2}(\partial_{\rho}(Dr^2\partial_{\rho}\psi))^2+r^{-2}D(1-|D|)(\slashed{\Delta}_{\s^2}\psi)^2\,d\omega d\rho\\
\leq&\: C\int_{M}^{\infty}D(\partial_{\rho} T \psi)^2r^2+Dh(r)^2(T^2\psi)^2\,d\omega d\rho.
\end{split}
\end{equation}

The above estimate allows us to conclude that
\begin{equation}
\label{eq:mainineqellipticext3}
\begin{split}
\int_{M}^{\infty}&\int_{\s^2} Dr^{-2}(\partial_{\rho}(Dr^2\partial_{\rho}\psi))^2+Dr^{-2}(\slashed{\Delta}_{\s^2}\psi)^2+D^2|\snabla_{\s^2}\partial_{\rho}\psi|^2\,d\omega d\rho\\
\leq&\: C\int_{M}^{\infty}\int_{\s^2}D(\partial_{\rho} T \psi)^2r^2+r^2Dh(r)^2(T^2\psi)^2\,d\omega d\rho.
\end{split}
\end{equation}
Furthermore, we can decompose
\begin{equation*}
\begin{split}
Dr^{-2}(\partial_{\rho}(Dr^2\partial_{\rho}\psi))^2=&\: Dr^{-2}\left[r^2\partial_{\rho}(D\partial_{\rho}\psi)+2rD\partial_{\rho}\psi\right]^2\\
=&\: r^2D(\partial_{\rho}(D\partial_{\rho}\psi))^2+4D(D\partial_{\rho}\psi)^2+4rD^2\partial_{\rho}\psi\partial_{\rho}(D\partial_{\rho}\psi).
\end{split}
\end{equation*}
By Young's inequality and the standard Poincar\'e inequality on $\s^2$, we have that
\begin{equation*}
\begin{split}
\int_{\s^2}4rD^2\partial_{\rho}\psi\partial_{\rho}(D\partial_{\rho}\psi)\,d\omega\geq& -\frac{1}{2}\int_{\s^2}r^2D(\partial_{\rho}(D\partial_{\rho}\psi))^2\,d\omega -8\int_{\s^2}D^3(\partial_{\rho}\psi)^2\,d\omega\\
\geq& -\frac{1}{2}\int_{\s^2}r^2D(\partial_{\rho}(D\partial_{\rho}\psi))^2\,d\omega -4\int_{\s^2}D^3|\snabla_{\s^2}\partial_{\rho}\psi|^2\,d\omega.
\end{split}
\end{equation*}
We then apply \eqref{eq:mainineqellipticext3} to conclude that
\begin{equation*}
\begin{split}
\int_{M}^{\infty}&\int_{\s^2} Dr^2(\partial_{\rho}(D\partial_{\rho}\psi))^2+Dr^{-2}(\slashed{\Delta}_{\s^2}\psi)^2+D^2|\snabla_{\s^2}\partial_{\rho}\psi|^2\,d\omega d\rho\\
\leq&\: C\int_{M}^{\infty}\int_{\s^2}D(\partial_{\rho} T \psi)^2r^2+r^2Dh(r)^2(T^2\psi)^2\,d\omega d\rho.
\end{split}
\end{equation*}
\end{proof}

\subsection{Pointwise decay estimates}
\label{sec:pdecayest}
In this section we use the energy decay estimates from Section \ref{sec:edecayest} and \ref{sec:hoedecayest} to derive $L^{\infty}$ estimates.
\begin{proposition}
\label{prop:pointdecay}
Let $J\in \N_0$ and assume that for $n=0,1$ and for all $0\leq k\leq 2-n$ and $0\leq j\leq J$:
\begin{align*}
\lim_{v\to \infty}\int_{\s^2}r^{2j}|\snabla_{\s^2}\slashed{\Delta}_{\s^2}^kL^jP_{\geq 1}\Phi_{(n)}|^2\,d\omega|_{u=u_0}<&\:\infty.
\end{align*}
and for all $0\leq j\leq J-1$
\begin{align*}
\lim_{v\to \infty} r^{j+2} L^{j+1}\phi_0(u_0,v)<&\infty.
\end{align*}
Then, for all $\epsilon>0$, there exists a constant $C=C(M,{\Sigma_0},\epsilon,J)>0$ such that
\begin{align}
\label{eq:pdecayl01v1}
||(r-M)^{\frac{1}{2}}T^J \psi_0||_{L^{\infty}(\Sigma_{\tau})}\leq&\: C\cdot\sqrt{E^{\epsilon}_{0;J}[\psi]}  (1+\tau)^{-\frac{3}{2}-J+\frac{\epsilon}{2}},\\
\label{eq:pdecayl01v2}
 ||rT^J \psi_0||_{L^{\infty}(\Sigma_{\tau})}\leq&\:  C\cdot\sqrt{E^{\epsilon}_{0; J}[\psi]}  (1+\tau)^{-1-J+\frac{\epsilon}{2}},\\
\label{eq:pdecaylgeq1v1a}
||(r-M)^{\frac{1}{2}} T^J \psi_{\geq 1}||_{L^{\infty}(\Sigma_{\tau})}\leq&\:C\cdot \sqrt{\sum_{|\alpha|\leq 2}E^{\epsilon}_{1;J}[\Omega^{\alpha}\psi]} (1+\tau)^{-\frac{5}{2}-J+\frac{\epsilon}{2}},\\
\label{eq:pdecaylgeq1v1b}
||(r-M)^{\frac{1}{2}}T^J \psi_{\geq 1}||_{L^{\infty}(\Sigma_{\tau})}\leq&\:C\cdot \sqrt{\sum_{|\alpha|\leq 2}E^{\epsilon}_{2;J}[\Omega^{\alpha}\psi]} (1+\tau)^{-3-J+\frac{\epsilon}{2}},\\
\label{eq:pdecaylgeq1v2a}
||rT^J \psi_{\geq 1}||_{L^{\infty}(\Sigma_{\tau})}\leq&\: C\cdot \sqrt{\sum_{|\alpha|\leq 2}E^{\epsilon}_{1;J}[\Omega^{\alpha}\psi]} (1+\tau)^{-2-J+\frac{\epsilon}{2}},\\
\label{eq:pdecaylgeq1v2b}
 ||rT^J \psi_{\geq 1}||_{L^{\infty}(\Sigma_{\tau})}\leq&\: C\cdot \sqrt{\sum_{|\alpha|\leq 2}E^{\epsilon}_{2;J}[\Omega^{\alpha}\psi]} (1+\tau)^{-\frac{5}{2}-J+\frac{\epsilon}{2}},\\
\label{eq:pdecaylgeq1v3a}
 ||\sqrt{D} T^J \psi_{\geq 1}||_{L^{\infty}(\Sigma_{\tau})}\leq&\:C\cdot \sqrt{\sum_{|\alpha|\leq 2}E^{\epsilon}_{1;J+1}[\Omega^{\alpha}\psi]} (1+\tau)^{-3-J+\frac{\epsilon}{2}},\\
\label{eq:pdecaylgeq1v3b}
 ||\sqrt{D}T^J \psi_{\geq 1}||_{L^{\infty}(\Sigma_{\tau})}\leq&\:C\cdot \sqrt{\sum_{|\alpha|\leq 2}E^{\epsilon}_{2;J+1}[\Omega^{\alpha}\psi]} (1+\tau)^{-\frac{7}{2}-J+\frac{\epsilon}{2}},
\end{align}
\end{proposition}
\begin{proof}
In order to estimate \eqref{eq:pdecayl01v1}, \eqref{eq:pdecaylgeq1v1a} and \eqref{eq:pdecaylgeq1v1b}, we apply the fundamental theorem of calculus along the foliation $\Sigma_{\tau}$ as follows:
\begin{equation*}
\begin{split}
\psi(\tau,\rho,\theta,\varphi)=&- \int_{\rho}^{\infty} \partial_{\rho} \psi (\tau,\rho',\theta,\varphi) \,d\rho'\\
\leq &\:\sqrt{\int_{r}^{\infty} (r'-M)^{-2}\,dr'}\cdot \sqrt{\int_{\rho}^{\infty} Dr^2 (\partial_{\rho}\psi)^2(\tau,\rho',\theta,\varphi) \,d\rho'},
\end{split}
\end{equation*}
where we used that, by the assumptions on the initial data in the proposition and the estimates in Proposition \ref{prop:radfieldsinfLk}, $\psi$ vanishes as $\rho \to \infty$ and moreover, we applied Cauchy--Schwarz. After applying the standard Sobolev inequality on $\s^2$, we therefore have that
\begin{equation*}
||(r-M)^{\frac{1}{2}}T^J \psi||_{L^{\infty}(\Sigma_{\tau})}\lesssim \sqrt{\sum_{|\alpha|\leq 2} \int_{\Sigma_{\tau}} J^T[\Omega^{\alpha}\psi]\cdot \mathbf{n}_{\tau}\,d\mu_{\tau}}.
\end{equation*}
The estimates \eqref{eq:pdecayl01v1}, \eqref{eq:pdecaylgeq1v1a} and \eqref{eq:pdecaylgeq1v1b} then follow from the energy decay estimates in Proposition \ref{prop:hoedaypsi0}, \ref{prop:hoedaypsi1} and \ref{prop:hoedaypsi1b}.

In order to prove the estimates \eqref{eq:pdecayl01v2}, \eqref{eq:pdecaylgeq1v2a} and \eqref{eq:pdecaylgeq1v2b} we can then restrict to ${N}^{\mathcal{I}}_{\tau}$ and ${N}^{\mathcal{H}}_{\tau}$. Let $\chi(r)$ be a cut-off function that is smooth and compactly supported in $r\geq r_{\mathcal{I}}$ away from $r=r_{\mathcal{I}}$. Then we can apply the fundamental theorem of calculus as follows:
\begin{equation*}
\begin{split}
(\chi \phi)^2(u',v,\theta,\varphi)=&\:\int_{v_{r_{\mathcal{I}}}(u')}^{v} 2\chi\phi\cdot L(\chi \phi)|_{u=u'}\,dv'\\
\leq &\: 2 \sqrt{\int_{v_{r_{\mathcal{I}}}(u')}^{v} r^{-2}\phi^2 |_{u=u'}\,dv'}\cdot \sqrt{\int_{v_{r_{\mathcal{I}}}(u')}^{v} r^{2}(L(\chi \phi))^2 |_{u=u'}\,dv'}\\
\lesssim &\:  \sqrt{\int_{v_{r_{\mathcal{I}}}(u')}^{v} (L\chi \phi)^2 |_{u=u'}\,dv'}\cdot \sqrt{\int_{v_{r_{\mathcal{I}}}(u')}^{v} r^{2}(L(\chi \phi))^2 |_{u=u'}\,dv'},
\end{split}
\end{equation*}
where we applied Cauchy--Schwarz to arrive at the second inequality and \eqref{eq:hardyinf} to arrive at the third inequality. If we now redefine $\chi$ to be a smooth, compactly supported cut-off function in $r\leq r_{\mathcal{H}}$ away from $r=r_{\mathcal{H}}$, we can similarly apply the fundamental theorem of calculus, Cauchy--Schwarz and \eqref{eq:hardyhor1} to obtain
\begin{equation*}
\begin{split}
(\chi \phi)^2(u,v',\theta,\varphi) \lesssim &\:  \sqrt{\int_{u_{r_{\mathcal{H}}}(v')}^{u} (\underline{L}\chi \phi)^2 |_{v=v'}\,du'}\cdot \sqrt{\int_{u_{r_{\mathcal{H}}}(v')}^{u} (r-M)^{-2}(\underline{L}(\chi \phi))^2 |_{v=v'}\,du'}.
\end{split}
\end{equation*}
It then follows that
\begin{align*}
 ||rT^J \psi||^2_{L^{\infty}({N}^{\mathcal{I}}_{\tau})}\leq &\: \sqrt{ \sum_{|\alpha|\leq 2}  \int_{{N}^{\mathcal{I}}_{\tau}} r^2(L\Omega^{\alpha}\phi)^2\,d\omega dv\cdot  \int_{\Sigma_{\tau}} J^T[\Omega^{\alpha}\psi]\cdot \mathbf{n}_{\tau}\,d\mu_{\tau}}+ \sum_{|\alpha|\leq 2} \int_{\Sigma_{\tau}} J^T[\Omega^{\alpha}\psi]\cdot \mathbf{n}_{\tau}\,d\mu_{\tau},\\
  ||rT^J \psi||^2_{L^{\infty}({N}^{\mathcal{H}}_{\tau})}\leq &\: \sqrt{ \sum_{|\alpha|\leq 2}  \int_{{N}^{\mathcal{H}}_{\tau}} (r-M)^2(\underline{L}\Omega^{\alpha}\phi)^2\,d\omega du\cdot  \int_{\Sigma_{\tau}} J^T[\Omega^{\alpha}\psi]\cdot \mathbf{n}_{\tau}\,d\mu_{\tau}}+  \sum_{|\alpha|\leq 2} \int_{\Sigma_{\tau}} J^T[\Omega^{\alpha}\psi]\cdot \mathbf{n}_{\tau}\,d\mu_{\tau}.
\end{align*}
We obtain \eqref{eq:pdecayl01v2}, \eqref{eq:pdecaylgeq1v2a} and \eqref{eq:pdecaylgeq1v2b} by applying the energy decay estimates in Proposition \ref{prop:hoedaypsi0}, \ref{prop:hoedaypsi1} and \ref{prop:hoedaypsi1b}.

We are left with proving \eqref{eq:pdecaylgeq1v3a} and \eqref{eq:pdecaylgeq1v3b}. We apply the fundamental theorem of calculus in yet another way:
\begin{equation*}
\begin{split}
\psi^2_{\geq 1} &(\tau,\rho,\theta,\varphi)= -\int_{\rho}^{\infty} 2 \psi_{\geq 1}\cdot \partial_{\rho}\psi_{\geq 1}(\tau,\rho',\theta,\varphi)\,d\rho'\\
\leq &\: \sqrt{\int_{\rho}^{\infty} D^{-2}\psi_{\geq 1}^2(\tau,\rho',\theta,\varphi)\,d\rho'}\cdot \sqrt{\int_{\rho}^{\infty} D^2(\partial_{\rho}\psi_{\geq 1})^2(\tau,\rho',\theta,\varphi)\,d\rho'}\\
\lesssim &\: D^{-1}(\rho)\sqrt{\int_{\rho}^{\infty} D(\partial_{\rho}\psi_{\geq 1})^2(\tau,\rho',\theta,\varphi)r^2\,d\rho'}\cdot \sqrt{\int_{\rho}^{\infty} D^2(\partial_{\rho}\psi_{\geq 1})^2(\tau,\rho',\theta,
\varphi)\,d\rho'}\\
\lesssim &\: D^{-1}(\rho)\sqrt{\int_{\rho}^{\infty} D(\partial_{\rho}\psi_{\geq 1})^2(\tau,\rho',\theta,\varphi)r^2\,d\rho'}\cdot \sqrt{\int_{\rho}^{\infty} [Dr^2(\partial_{\rho}T\psi_{\geq 1})^2+ Dh^2r^2(T^2\psi_{\geq 1})^2](\tau,\rho',\theta,\varphi)\,d\rho'},
\end{split}
\end{equation*}
where we applied Cauchy--Schwarz to arrive at the first inequality and we applied \eqref{eq:hardyhor2} together with the fact that $D^{-2}(r')\lesssim D^{-2}(r)$ for all $r\leq r'$ to obtain the second inequality. The third inequality then follows from an application of the degenerate elliptic estimate in Proposition \ref{prop:degenelliptic}, together with the standard Poincar\'e inequality on $\s^2$. We conclude that
\begin{equation*}
||\sqrt{D} \psi||^2_{L^{\infty}(\Sigma_{\tau})}\lesssim \sqrt{\sum_{|\alpha|\leq 2} \int_{\Sigma_{\tau}} J^T[\Omega^{\alpha}\psi]\cdot \mathbf{n}_{\tau}\,d\mu_{\tau}\cdot \int_{\Sigma_{\tau}} J^T[T\Omega^{\alpha}\psi]\cdot \mathbf{n}_{\tau}\,d\mu_{\tau}}
\end{equation*}
and we apply the energy decay estimates in Proposition \ref{prop:hoedaypsi1} and \ref{prop:hoedaypsi1b} to derive \eqref{eq:pdecaylgeq1v3a} and \eqref{eq:pdecaylgeq1v3b}.
\end{proof}

It will be necessary to use the following \emph{additional} $L^{\infty}$ estimates with stronger weights in the energy norm \emph{either} in $r$ \emph{or} in $(r-M)^{-1}$ compared to the weights appearing in the norms $E^{\epsilon}_{0; J}[\psi]$ and $E^{\epsilon}_{1; J}[\psi]$:
\begin{proposition}
\label{prop:pointdecay2}
Let $J\in \N_0$ and assume that for $n=0,1$ and for all $0\leq k\leq 2-n$ and $0\leq j\leq J$:
\begin{align*}
\lim_{v\to \infty}\int_{\s^2}r^{2j}|\snabla_{\s^2}\slashed{\Delta}_{\s^2}^kL^jP_{\geq 1}\Phi_{(n)}|^2\,d\omega|_{u=u_0}<&\:\infty,
\end{align*}
and for all $0\leq j\leq J-1$
\begin{align*}
\lim_{v\to \infty} r^{j+2} L^{j+1}\phi_0(u_0,v)<&\infty.
\end{align*}
Then, for all $\epsilon>0$, there exists a constant $C=C(M,{\Sigma_0},\epsilon,J)>0$ such that
\begin{align}
\label{eq:pdecayl01v2impI}
 ||rT^J \psi_0||_{L^{\infty}({N}^{\mathcal{I}}_{\tau})}\leq&\: C \cdot \sqrt{E^{\epsilon}_{0,\mathcal{I}; J}[\psi]}  (1+\tau)^{-\frac{3}{2}-J+\frac{\epsilon}{2}},\\ 
\label{eq:pdecayl01v2impH}
 ||rT^J \psi_0||_{L^{\infty}({N}^{\mathcal{H}}_{\tau})}\leq&\: C \cdot \sqrt{E^{\epsilon}_{0,\mathcal{H}; J}[\psi]}  (1+\tau)^{-\frac{3}{2}-J+\frac{\epsilon}{2}},\\
 \label{eq:pdecaylgeq1v2impI}
||rT^J \psi_{\geq 1}||_{L^{\infty}({N}^{\mathcal{I}}_{\tau})}\leq&\: C\cdot \sqrt{\sum_{|\alpha|\leq 2}E^{\epsilon}_{1,\mathcal{I};J}[\Omega^{\alpha}\psi]} (1+\tau)^{-\frac{9}{4}-J+\frac{\epsilon}{2}},\\
\label{eq:pdecaylgeq1v2impH}
 ||rT^J \psi_{\geq 1}||_{L^{\infty}({N}^{\mathcal{H}}_{\tau})}\leq&\: C\cdot \sqrt{\sum_{|\alpha|\leq 2}E^{\epsilon}_{1,\mathcal{H};J}[\Omega^{\alpha}\psi]} (1+\tau)^{-\frac{9}{4}-J+\frac{\epsilon}{2}},
\end{align}
where we additionally assumed \eqref{eq:asmLjphi_02} for \eqref{eq:pdecayl01v2impI}.
\end{proposition}
\begin{proof}
The estimates follow from the proof of Proposition \ref{prop:pointdecay}, where we additionally appeal to the stronger weighted energy decay estimates from Proposition \ref{prop:hoextraendecay}.
\end{proof}

\section{Late-time asymptotics for Type \textbf{C} perturbations}
\label{sec:asympnonzeroconst}
In this section, we will derive the leading-order late-time asymptotics for the spherical mean $\psi_0$ in the case that both $H_0$, the conserved quantity at $\mathcal{H}^+$ and $I_0$, the conserved quantity at $\mathcal{I}^+$ are non-zero. This data is of Type \textbf{C}, as defined in Section \ref{sec:TheTypesOfInitialDataABCD}. We will make use of the pointwise decay estimates for $\psi_0$ derived in Section \ref{sec:decayest}.

\subsection{Late-time asymptotics in the regions $\protect\mathcal{A}^{\mathcal{H}}_{\gamma^{\mathcal{H}}}$ and $\protect\mathcal{A}^{\mathcal{I}}_{\gamma^{\mathcal{I}}}$}
\label{sec:asympradfieldsnonzeroconst}
We introduce the following $L^{\infty}$ norms on the derivatives of $\phi_0$ along the initial hypersurfaces ${{N}_0^{\mathcal{H}}}$ and ${{N}_0^{\mathcal{I}}}$:
\begin{align*}
P_{H_0,\beta;k}[\psi]:=&\max_{0\leq j \leq k}\left|\left| u^{2+j+\beta}\cdot \underline{L}^j\left(\underline{L}\phi_0|_{{{N}_0^{\mathcal{I}}}}-2\frac{H_0[\psi]}{u^2}\right) \right|\right|_{L^{\infty}}\\
P_{I_0,\beta;k}[\psi]:=&\max_{0\leq j \leq k}\left|\left| v^{2+j+\beta}\cdot L^j\left(L\phi_0|_{{{N}_0^{\mathcal{H}}}}-2\frac{I_0[\psi]}{v^2}\right) \right|\right|_{L^{\infty}},
\end{align*}
with $0<\beta\leq 1$.

For the arguments below, we will need to relate decay in terms of the coordinate $r$ to decay in terms of the double null coordinates $u$ and $v$.

\begin{lemma}
\label{lm:relationruv}
Let $N\in \N$. 
\begin{itemize}
\item[\rm (i)]
We can estimate in $\{r\geq r_{\mathcal{I}}\}\cap\{u\geq 0\}$:
\begin{align*}
r-\frac{v-u}{2}-2M\log (v-u)=&\: O_N((v-u)^0),\\
r^{-1}-\frac{2}{v-u}-8M(v-u)^{-2}\log (v-u)=&\: O_N((v-u)^{-2})\\
r^{-2}-\frac{4}{(v-u)^2}-32M(v-u)^{-3}\log (v-u)=& O_N((v-u)^{-3}),\\
r^{-3}(u,v)=&\:O_N((v-u)^{-3}).
\end{align*}
There moreover exists a constant $c_{r_{\mathcal{I}}}>0$ such that
\begin{equation*}
\begin{array}{ll}
c_{r_{\mathcal{I}}} \cdot v \leq v-u-1 \leq v\quad &\textnormal{if}\quad  u_0\leq u\leq \frac{v}{2}+r_*(r_{\mathcal{I}}),\\
c_{r_{\mathcal{I}}}\cdot v\leq u\leq v \quad  &\textnormal{if}\quad   \frac{v}{2}+r_*(r_{\mathcal{I}})\leq u\leq u_{r_{\mathcal{I}}}(v).
\end{array}
\end{equation*}
\item[\rm (ii)]
We can estimate in $\{r\geq r_{\mathcal{I}}\}\cap\{u\geq 0\}$:
\begin{align*}
M^2(r-M)^{-1}-\frac{u-v}{2}-2M\log (u-v)=&\: O_N((u-v)^0),\\
M^{-2}(r-M)-\frac{2}{u-v}-8M(u-v)^{-2}\log (u-v)=&\: O_N((u-v)^{-2})\\
M^{-4}(r-M)^2-\frac{4}{(u-v)^2}-32M(u-v)^{-3}\log (u-v)=& O_N((u-v)^{-3}),\\
M^{-6}(r-M)^3(u,v)=&\:O_N((u-v)^{-3}).
\end{align*}
There moreover exists a constant $c_{r_{\mathcal{H}}}>0$ such that
\begin{equation*}
\begin{array}{ll}
c_{r_{\mathcal{H}}} \cdot v \leq u-v-1 \leq v\quad &\textnormal{if}\quad  v_0\leq v\leq \frac{u}{2}-r_*(r_{\mathcal{H}}),\\
c_{r_{\mathcal{H}}}\cdot u\leq v\leq u \quad  &\textnormal{if}\quad \frac{u}{2}-r_*(r_{\mathcal{H}})\leq v\leq v_{r_{\mathcal{H}}}(u).
\end{array}
\end{equation*}
\end{itemize}
\end{lemma}
\begin{proof}
Observe that
\begin{equation*}
\frac{v-u}{2}=r_*(r)=r-M-M^2(r-M)^{-1}+2M\log\left(\frac{r-M}{M}\right).
\end{equation*}
Hence, we can repeat the proof of Lemma 2.1 in \cite{logasymptotics}, where in the $r\leq r_{\mathcal{H}}$ case, we interchange the roles of $u$ and $v$ and we replace $r$ by $M^{-2}(r-M)^{-1}$.
\end{proof}

\begin{remark}
\label{rmk:Pnorminvrcoords}
By applying Lemma \ref{lm:relationruv} it follows moreover that $P_{H_0,\beta;k}[\psi]<\infty$ if and only if in $(v,r)$ coordinates:
\begin{equation*}
\max_{0\leq j\leq k}\left|\left| (r-M)^{j-\beta}\cdot \partial_r^j\left(\partial_r\phi_0|_{{{N}_0^{\mathcal{H}}}}+M^{-2}H_0[\psi]\right) \right|\right|_{L^{\infty}}<\infty
\end{equation*}
Hence, if $\partial_r\phi_0|_{r=M}=-M^{-2}H_0$, we are guaranteed that $P_{H_0,\beta;k}[\psi]<\infty$ for all $k\in \N_0$ and $\beta=1$, simply by the smoothness assumption on the initial data for $\psi$ together with Taylor's theorem.
\end{remark}

We moreover introduce the following spacetime subregions contained in either the region $\mathcal{A}^{\mathcal{I}}$ or $\mathcal{A}^{\mathcal{H}}$: for $k\in \N_0$ and $\alpha\in (0,1)$ let
\begin{align*}
{\mathcal{A}_{\gamma^{\mathcal{H}}_{\alpha}}^{\mathcal{H}}}:=&\:\{r\leq r_{\mathcal{H}}\}\cap \{0\leq v \leq u-v^{\alpha}+2r_*(r_{\mathcal{H}})\}\subset \mathcal{A}^{\mathcal{H}},\\
{\mathcal{A}_{\gamma^{\mathcal{I}}_{\alpha}}^{\mathcal{I}}}:=&\:\{r\geq r_{\mathcal{I}}\}\cap \{0\leq u \leq v-u^{\alpha}+2r_*(r_{\mathcal{I}})\}\subset \mathcal{A}^{\mathcal{I}}.
\end{align*}

Note that the boundaries $\partial {\mathcal{A}_{\gamma^{\mathcal{H}}_{\alpha}}^{\mathcal{H}}}$ and $\partial {\mathcal{A}_{\gamma^{\mathcal{I}}_{\alpha}}^{\mathcal{I}}}$ contain subsets of, respectively, the following timelike hypersurfaces:
\begin{align*}
\gamma^{\mathcal{H}}_{\alpha}:=&\:\{u-v=v^{\alpha}+2r_*(r_{\mathcal{H}})\},\\
\gamma^{\mathcal{I}}_{\alpha}:=&\:\{v-u=u^{\alpha}+2r_*(r_{\mathcal{I}})\}.
\end{align*}

When the value of $\alpha\in (0,1)$ is not relevant, we will occasionally drop the $\alpha$ subscript for convenience and write $\gamma^{\mathcal{H}}$ and  $\mathcal{A}_{\gamma^{\mathcal{I}}}^\mathcal{I}$.

In the regions ${\mathcal{A}_{\gamma^{\mathcal{H}}_{\alpha}}^{\mathcal{H}}}$ and ${\mathcal{A}_{\gamma^{\mathcal{I}}_{\alpha}}^{\mathcal{I}}}$, we obtain the following additional decay estimates for $r^{-1}$ and $r-M$:
\begin{lemma}
Let $M<r_{\mathcal{H}}<2M$ and $r_{\mathcal{I}}>2M$. Then for all $\eta>0$, we can estimate
\begin{align*}
r^{-1}\lesssim &\:v^{-\alpha}\lesssim u^{-\alpha}  \quad \textnormal{in ${\mathcal{A}_{\gamma^{\mathcal{I}}_{\alpha}}^{\mathcal{I}}}$},\\
r-M\lesssim &\: u^{-\alpha}\lesssim v^{-\alpha} \quad \textnormal{in ${\mathcal{A}_{\gamma^{\mathcal{H}}_{\alpha}}^{\mathcal{H}}}$}.
\end{align*}
\end{lemma}

As a first step towards obtaining the asymptotics of $\phi_0$, we obtain the asymptotics of $L\phi_0$ and more generally $L^{k+1}\phi_0$, for $k\in \N_0$.

\begin{proposition}
\label{prop:asympLphi}
Let $k\in \N_0$ and $\alpha_k\in (\frac{k+2}{k+3},1)$. Take $\epsilon\in (0,\frac{1}{2}(k+3)\alpha_k-\frac{1}{2}(k+2))$. Assume that $E^{\epsilon}_{0}[\psi]<\infty$ and moreover that there exists a $\beta>0$ such that:
\begin{equation*}
P_{I_0,\beta;k}[\psi]<\infty.
\end{equation*}
Then, there exists a constant $C=C(M,{\Sigma_0},r_{\mathcal{H}},r_{\mathcal{I}},\alpha_k,\epsilon,k)>0$ such that
\begin{align*}
|L^{k+1}\phi_0(u,v)-(-1)^k(k+1)!\cdot 2I_0[\psi]\cdot v^{-2-k}|\leq&\: C \sqrt{E^{\epsilon}_{0}[\psi]}\cdot v^{-3\alpha_k+2\epsilon-k}\\
&+P_{I_0,\beta;k}[\psi]\cdot v^{-2-\beta-k}\quad \textnormal{in ${\mathcal{A}_{\gamma^{\mathcal{I}}_{\alpha_k}}^{\mathcal{I}}}$},\\
|\underline{L}^{k+1}\phi_0(u,v)-(-1)^k (k+1)!\cdot 2H_0[\psi]\cdot u^{-2-k}|\leq&\: C \sqrt{E^{\epsilon}_{0}[\psi]}\cdot u^{-3\alpha_k+2\epsilon-k}\\
&+P_{H_0,1;k}[\psi]\cdot  u^{-3-k}\quad \textnormal{in ${\mathcal{A}_{\gamma^{\mathcal{H}}_{\alpha_k}}^{\mathcal{H}}}$}.
\end{align*}
\end{proposition}
\begin{proof}
The equation \eqref{eq:waveequation} for $\psi_0$ on extremal Reissner--Nordstr\"om for can be rewritten as follows in double null coordinates:
\begin{equation*}
\partial_u\partial_v\phi_0=-\frac{1}{4r}DD' \phi_0.
\end{equation*}
From Lemma \ref{lm:relationruv} it therefore follows that we can write
\begin{align*}
\partial_u\partial_v\phi_0=&\:O_N((v-u)^{-3})\phi_0\quad \textnormal{in $\mathcal{A}^{\mathcal{I}}$},\\
\partial_v\partial_u\phi_0=&\: O_N((u-v)^{-3})\phi_0 \quad \textnormal{in $\mathcal{A}^{\mathcal{H}}$}.
\end{align*}
In particular, it follows that the estimates for $L^k\phi_0$ in ${\mathcal{A}_{\gamma^{\mathcal{I}}_{\alpha_k}}^{\mathcal{I}}}$, derived in Proposition 8.3 of \cite{paper2}, apply directly to $\underline{L}^k\phi_0$ in ${\mathcal{A}_{\gamma^{\mathcal{H}}_{\alpha_k}}^{\mathcal{H}}}$, after interchanging the role of $u$ and $v$.
\end{proof}

The next step is to apply the estimates for $L^{k+1}\phi_0$ from Proposition \ref{prop:asympLphi} in order to obtain asymptotics for $\phi_0$ and, more generally, for $T^k\phi_0$ with $k\in \N_0$.
\begin{proposition}
\label{prop:asympradfieldnonzeroIH}
Let $k\in \N_0$. If we additionally restrict $\alpha_k\in [\frac{2k+5}{2k+7},1)$ and $\epsilon\in (0,\frac{1}{6}(1-\alpha_k))$, we can find a constant $C=C(M,{\Sigma_0},r_{\mathcal{H}},r_{\mathcal{I}},\alpha_k,\epsilon,\beta, k)>0$ such that
\begin{align*}
|T^k\phi_0&(u,v)-(-1)^k k!\cdot 2I_0[\psi]\cdot (u^{-1-k}-v^{-1-k})|\\
\leq&\: C \left(\sqrt{E^{\epsilon}_{0;k}[\psi]}+I_0[\psi]\right)\cdot (v+M)^{-\frac{\alpha_k}{2}-\frac{3}{2}+2\epsilon-k}+C\cdot P_{I_0,\beta;k}[\psi]\cdot (u+M)^{-1-\beta-k}\quad \textnormal{in ${\mathcal{A}_{\gamma^{\mathcal{I}}_{\alpha_k}}^{\mathcal{I}}}$},\\
|T^k\phi_0&(u,v)-(-1)^k k!\cdot 2H_0[\psi]\cdot (v^{-1-k}-u^{-1-k})|\\
\leq&\: C \left(\sqrt{E^{\epsilon}_{0;k}[\psi]}+H_0[\psi]\right)\cdot v^{-\frac{\alpha_k}{2}-\frac{3}{2}+2\epsilon-k}+C\cdot P_{H_0,1;k}[\psi]\cdot v^{-2-k}\quad \textnormal{in ${\mathcal{A}_{\gamma^{\mathcal{H}}_{\alpha_k}}^{\mathcal{H}}}$}.
\end{align*}
Furthermore, if we impose $\frac{1}{2}(1-\alpha_k)<\beta +2\epsilon$ and assume $H_0[\psi]\neq 0$ and $I_0[\psi]\neq 0$, then the estimates above provide first-order asymptotics for $\phi$ in the regions $\mathcal{A}_{\gamma^{\mathcal{H}}_{\delta}}$ and $\mathcal{A}_{\gamma^{\mathcal{I}}_{\delta}}$, for $1>\delta>\frac{\alpha_k}{2}+\frac{1}{2}+2\epsilon>\alpha_k+2\epsilon$.
In particular,
\begin{align*}
|T^k\phi_0(u,v)-(-1)^k k!\cdot 2I_0[\psi]\cdot u^{-1-k}|\leq&\: C \left(\sqrt{E^{\epsilon}_{0;k}[\psi]}+I_0[\psi]\right)\cdot u^{-\frac{\alpha_k}{2}-\frac{3}{2}+2\epsilon-k}\\
&+C\cdot P_{I_0,\beta;k}[\psi]\cdot u^{-1-\beta-k},\\
|T^k\phi_0(u,v)-(-1)^k k!\cdot 2H_0[\psi]\cdot v^{-1-k}|\leq&\: C \left(\sqrt{E^{\epsilon}_{0;k}[\psi]}+H_0[\psi]\right)\cdot v^{-\frac{\alpha_k}{2}-\frac{3}{2}+2\epsilon-k}\\
&+C\cdot P_{H_0,1;k}[\psi]\cdot v^{-2-k}.
\end{align*}
\end{proposition}
\begin{proof}
The proof follows directly from the proof of Proposition 8.4 and 8.5 of \cite{paper2}, where as in case in the proof of Proposition \ref{prop:asympLphi}, we use that the estimates in the region analogous to ${\mathcal{A}_{\gamma^{\mathcal{I}}_{\alpha_k}}^{\mathcal{I}}}$ in \cite{paper2} apply directly to ${\mathcal{A}_{\gamma^{\mathcal{H}}_{\alpha_k}}^{\mathcal{H}}}$ after interchanging $u$ and $v$ and $L$ and $\underline{L}$.
\end{proof}

\subsection{Partial asymptotics for $\partial_{\rho}\psi$ away from $\mathcal{H}^+$ up to $\gamma^{\mathcal{I}}$}
\label{sec:SharpDecayAndAsymptoticsForProtectPartialRhoPsiUpToProtectGammaMathcalI}
Before we discuss the late-time asymptotics of $T^k\psi_0$ for Type \textbf{C} data, we will derive the late-time asymptotics for the derivatives $\underline{L}T^k\psi_0$ and $LT^k\psi_0$ in appropriate subsets of $\mathcal{R}$. We will use the asymptotics for $\phi_0$ along $\mathcal{H}^+$ obtained in Proposition \ref{prop:asympradfieldnonzeroIH}, together with the decay estimates \eqref{eq:edecayl01} and \eqref{eq:pdecayl01v1}.
\begin{proposition}
\label{prop:partasymdrhopsi}
Let $k\in \N_0$. Then there exist $\eta,\epsilon>0$ suitably small, such that in $(v,r)$ coordinates, we have that for all $r\leq r_{\mathcal{I}}$:
\begin{equation}
\label{eq:partasymdrhopsi1}
\begin{split}
|-2r^2\underline{L}T^k\psi_0(v,r)&-(-1)^{k+1} (k+1)!\cdot 4MH_0[\psi]\cdot v^{-2-k}|\\
\leq&\: C\left(\sqrt{E^{\epsilon}_{0; k+1}[\psi]}+H_0[\psi]\right)\left[(r-M)^{-\frac{1}{2}}v^{-\frac{5}{2}+\epsilon}+ v^{-2-k-\epsilon'}\right]\\
&+P_{H_0,1;k}[\psi]\cdot v^{-3-k},
\end{split}
\end{equation}
and in $(u,r)$ coordinates, we can estimate for all $r\geq r_{\mathcal{H}}$:
\begin{equation}
\label{eq:partasymdrhopsi2}
\begin{split}
|2LT^k\psi_0(u,r)&-(-1)^{k+1} (k+1)!\cdot 4MH_0[\psi]\cdot D^{-1}r^{-2}u^{-2-k}|\\
\leq &\: C\left(\sqrt{E^{\epsilon}_{0; k+1}[\psi]}+H_0[\psi]\right) r^{-2}u^{-2-k-\epsilon'}\\
&+\left(\sqrt{E^{\epsilon}_{0; k+1}[\psi]}+H_0[\psi]\right) r^{-\frac{1}{2}}u^{-\frac{5}{2}-k+\epsilon}\\
&+P_{H_0,1;k}[\psi]\cdot r^{-2}u^{-3-k},
\end{split}
\end{equation}
where $C=C(M,{\Sigma_0},r_{\mathcal{H}},r_{\mathcal{I}},\eta,\epsilon,k)>0$ is a constant.
\end{proposition}
\begin{proof}
From \eqref{eq:waveequation} in follows that $T^k\psi_0$ satisfies the following equation in $(v,r)$ coordinates:
\begin{equation}
\label{eq:waveeqvr}
\partial_{r}\left(Dr^2\partial_{r}T^k\psi_0+2rT^{k+1}\phi_0\right)=2T^{k+1}\phi_0.
\end{equation}
See also \eqref{eq:waveeqrhocoord} with $h=0$ applied to $\psi_0$.

By integrating both sides of \eqref{eq:waveeqvr} along constant $v$ hypersurfaces from $r'=M$ to $r'=r\leq \min\{r_{\mathcal{I}},r_{\Sigma_0}(v)\}$, where $r_{\Sigma_0}(v)$ denotes the value of $r$ along the intersection of the hypersurface of constant $v$ with $\Sigma_0\cap (\mathcal{B} \cup \mathcal{A}^{\mathcal{I}})$ (which is non-empty for $v>v_0$), and using that $Dr^2\partial_rT^k\psi$ vanishes at $\mathcal{H}^+$ for any $T^k\psi$ (using that $\psi$ is smooth), we therefore arrive at:
\begin{equation*}
Dr^2\partial_{r}T^k\psi_0(v,r)+2rT^{k+1}\phi(v,r)=2MT^{k+1}\phi|_{\mathcal{H}^+}(v)+\int_{M}^r2T^{k+1}\phi(v,r')\,dr'.
\end{equation*}

We first apply Cauchy--Schwarz and \eqref{eq:hardyhor2}, together with \eqref{eq:edecayl01} and \eqref{eq:pdecayl01v1}, to estimate
\begin{equation*}
\begin{split}
\left|\int_{M}^r2T^{k+1}\phi_0(v,r')\,dr'\right|\lesssim&\: \sqrt{\int_{M}^r \int_{\s^2}(T^{k+1}\psi_0)^2 \,d\omega dr} \cdot \sqrt{\int_M^{r_{\mathcal{I}}} \,dr'}\\
\lesssim &\: \sqrt{\int_{M}^{ \min\{r_{\mathcal{I}},r_{\Sigma_0}(v)\}}\int_{\s^2}(r-M)^2(\partial_r(T^{k+1}\psi_0))^2 \,d\omega dr+ (T^{k+1}\psi_0)^2(v,r_{\mathcal{I}})}\cdot  \sqrt{r-M}\\
\lesssim &\:\sqrt{ \int_{\Sigma_{\tau(v, \min\{r_{\mathcal{I}},r_{\Sigma_0}(v)\})}} J^T[T^{k+1}\psi]\cdot \mathbf{n}_{\Sigma_0}\,d\mu_{\Sigma_{\tau(v, \min\{r_{\mathcal{I}},r_{\Sigma_0}(v)\})}}}\cdot \sqrt{r-M}\\
\lesssim &\: \sqrt{E^{\epsilon}_{0; k+1}[\psi]} (r-M)^{\frac{1}{2}}(1+\tau(v,\min\{r_{\mathcal{I}},r_{\Sigma_0}(v)\}))^{-\frac{5}{2}+\epsilon}\\
\lesssim &\: \sqrt{E^{\epsilon}_{0; k+1}[\psi]} (r-M)^{\frac{1}{2}}v^{-\frac{5}{2}+\epsilon}.
\end{split}
\end{equation*}
for $r\leq \min\{r_{\mathcal{I}},r_{\Sigma_0}(v)\}$ and $v\geq v_0$, where in the third inequality we used the conservation property of the $T$-energy flux.

By \eqref{eq:pdecayl01v2} we moreover have that for all $v\geq 0$:
\begin{equation*}
\left|2rT^{k+1}\phi_0(v,r)\right|\lesssim \sqrt{E^{\epsilon}_{0;k+ 1}[\psi]} r^2(r-M)^{-\frac{1}{2}}v^{-\frac{5}{2}-k+\epsilon}.
\end{equation*}

Furthermore, by Proposition \ref{prop:asympradfieldnonzeroIH}, we have that for $\epsilon>0$ suitably small, there exists an $\epsilon'>0$ such that we can estimate
\begin{equation*}
\begin{split}
\left|T^{k+1}\phi_0|_{\mathcal{H}^+}(v)-(-1)^{k+1} (k+1)!\cdot 2H_0[\psi]\cdot v^{-2-k}\right|\lesssim&\: \left(\sqrt{E^{\epsilon}_{0; k+1}[\psi]}+H_0[\psi]\right)\cdot v^{-2-k-\epsilon'}\\
&+P_{H_0,1;k}[\psi]\cdot v^{-3-k}.
\end{split}
\end{equation*}

Combining all the above decay estimates, we can therefore infer that for all $v\geq v_0$ and $r\leq \min\{r_{\mathcal{I}},r_{\Sigma_0}(v)\}$:
\begin{equation}
\label{eq:asymlbarpsidr}
\begin{split}
|Dr^2\partial_rT^k\psi_0(v,r)&-(-1)^{k+1} (k+1)!\cdot 4MH_0[\psi]\cdot v^{-2-k}|\\
\lesssim&\: \left(\sqrt{E^{\epsilon}_{0; k+1}[\psi]}+H_0[\psi]\right)\left[(r-M)^{-\frac{1}{2}}v^{-\frac{5}{2}+\epsilon}+ v^{-2-k-\epsilon'}\right]\\
&+P_{H_0,1;k}[\psi]\cdot v^{-3-k},
\end{split}
\end{equation}
or equivalently, since we can express $\underline{L}=-\frac{D}{2}\partial_r$ in $(v,r)$ coordinates, we can write
\begin{equation}
\label{eq:asymlbarpsi}
\begin{split}
|-2r^2\underline{L}T^k\psi_0(v,r)&-(-1)^{k+1} (k+1)!\cdot 4MH_0[\psi]\cdot v^{-2-k}|\\
\lesssim&\: \left(\sqrt{E^{\epsilon}_{0; k+1}[\psi]}+H_0[\psi]\right)\left[(r-M)^{-\frac{1}{2}}v^{-\frac{5}{2}+\epsilon}+ v^{-2-k-\epsilon'}\right]\\
&+P_{H_0,1;k}[\psi]\cdot v^{-3-k}.
\end{split}
\end{equation}

By using that $T=\underline{L}+L$ and applying once again \eqref{eq:pdecayl01v1}, we can rewrite \eqref{eq:asymlbarpsi} at $r=r_0\geq r_{\mathcal{H}}$ as follows:
\begin{equation}
\begin{split}
\label{eq:LpsiestrI}
|2r_{\mathcal{I}}^2LT^k\psi_0|_{r=r_{0}}(u)&-(-1)^{k+1} (k+1)!\cdot 4MH_0[\psi]\cdot u^{-2-k}|\\
\lesssim&\: \left(\sqrt{E^{\epsilon}_{0; k+1}[\psi]}+H_0[\psi]\right) u^{-2-k-\epsilon'}\\
&+P_{H_0,1;k}[\psi]\cdot u^{-3-k}.
\end{split}
\end{equation}
Let us now switch to $(u,r)$ coordinates in the region $r\geq r_{\mathcal{I}}$. From \eqref{eq:waveequation} in follows that $\psi$ satisfies the following equation in $(u,r)$ coordinates:
\begin{equation}
\label{eq:waveequr}
\partial_{r}\left(2r^2LT^k\psi_0-2rT^{k+1}\phi_0\right)=-2T^{k+1}\phi_0.
\end{equation}
We integrate both sides of \eqref{eq:waveequr} along constant $u$ hypersurfaces from $r'=r_{\mathcal{I}}$ to $r'=r>r_{\mathcal{I}}$ to arrive at:
\begin{equation*}
2r^2LT^k\psi_0(u,r)=2r_{\mathcal{I}}^2LT^k\psi_0(u,r_{\mathcal{I}})-2r_{\mathcal{I}}T^{k+1}\phi_0(u,r_{\mathcal{I}})+2rT^{k+1}\phi_0(u,r)-2\int_{r_{\mathcal{I}}}^rT^{k+1}\phi_0(u,r')\,dr'.
\end{equation*}
First of all, we apply \eqref{eq:pdecayl01v2} to estimate
\begin{equation*}
|2rT^{k+1}\phi_0(u,r)|\lesssim  \sqrt{E^{\epsilon}_{0;k+1}[\psi]} r^{\frac{3}{2}}(1+\tau)^{-\frac{5}{2}-k+\epsilon}
\end{equation*}
for $r\geq r_{\mathcal{I}}$.

We moreover  apply Cauchy--Schwarz together with \eqref{eq:hardyinf} to estimate
\begin{equation*}
\begin{split}
\left|\int_{r_{\mathcal{I}}}^rT^{k+1}\phi_0(u,r')\,dr'\right| \lesssim&\: \sqrt{\int_{{{N}^{\mathcal{I}}}_u} r^{-2} (T^{k+1} \phi_0)^2\,d\omega dr}\sqrt{\int_{r_{\mathcal{I}}}^r r'^2\,dr'} \\
\lesssim&\: \left[\sqrt{\int_{{{N}^{\mathcal{I}}}_u} (\partial_rT^{k+1}\phi_0)^2\,d\omega dr+(T^{k+1}\phi)^2(u,r_{\mathcal{I}})}\right]\cdot r^{\frac{3}{2}}\\
\lesssim&\:   \sqrt{E^{\epsilon'}_{0;k+ 1}[\psi]} r^{\frac{3}{2}}(1+\tau)^{-\frac{5}{2}-k+\epsilon'}.
\end{split}
\end{equation*}
Hence, using that $\partial_r=2D^{-1}L$ in $(u,r)$ coordinates, we have that
\begin{equation*}
\begin{split}
|\partial_rT^k\psi_0(u,r)&-(-1)^{k+1} (k+1)!\cdot 4MH_0[\psi]\cdot D^{-1}r^{-2}u^{-2-k}|\\
\lesssim&\: \left(\sqrt{E^{\epsilon}_{0; k+1}[\psi]}+H_0[\psi]\right) r^{-2}u^{-2-k-\epsilon'}\\
&+\left(\sqrt{E^{\epsilon}_{0; k+1}[\psi]}+H_0[\psi]\right) r^{-\frac{1}{2}}u^{-\frac{5}{2}-k+\epsilon}\\
&+P_{H_0,1;k}[\psi]\cdot r^{-2}u^{-3-k},
\end{split}
\end{equation*}
for all $r\geq r_{\mathcal{I}}$.
\end{proof}
\textbf{We note that the estimate \eqref{eq:partasymdrhopsi2} provides in particular the late-time asymptotics of $\partial_rT^k\psi_0$ in the region $\{r_{\mathcal{H}}\leq r\leq r_{\mathcal{I}}\}$ \underline{even if $I_0[\psi]=0$},} so it will be also be relevant when investigating the asymptotics of Type \textbf{A} data in Section \ref{sec:asympzeroconst} below.
\subsection{Late-time asymptotics in $\mathcal{R}$}
\label{sec:latetimeasympsi}
In this section, we obtain the asymptotics for $T^k\psi_0$, using fundamentally that both $H_0\neq 0$ and $I_0\neq 0$ in the case of Type \textbf{C} data. The next result completes the proof of theorem \ref{prop:asympsitheo}.  
\begin{proposition}
\label{prop:asympsi}
Let $k\in \N_0$. Then there exist $\eta,\epsilon>0$ suitably small, such that we obtain the following \emph{global} estimate:
\begin{equation}
\label{eq:asympsi}
\begin{split}
\Bigg|T^k\psi_0(u,v)&-4\left(I_0[\psi]+ \frac{M}{r\sqrt{D}}H_0[\psi]\right)T^k\left(\frac{1}{u\cdot v}\right)\Bigg|\\
\leq&\: C\left(\sqrt{E_{0;k+1}^{\epsilon}[\psi]}+I_0[\psi]+P_{I_0,\beta;k}[\psi]\right)v^{-1}u^{-1-k-\eta}\\
&+C\left(\sqrt{E_{0;k+1}^{\epsilon}[\psi]}+H_0[\psi]+P_{H_0,1;k}[\psi]\right)D^{-\frac{1}{2}}u^{-1}v^{-1-k-\eta},
\end{split}
\end{equation}
where $C=C(M,{\Sigma_0},r_{\mathcal{H}},r_{\mathcal{I}},\eta,\epsilon, \beta, k)>0$ is a constant.
\end{proposition}

\emph{Outline of proof:}\\
\\
For simplicity let use take $k=0$.
\begin{itemize}
\item[\textbf{Step 1}:]
We use the asymptotics for $\phi_0$ along $\mathcal{H}^+$ obtained in Proposition \ref{prop:asympradfieldnonzeroIH}, together with the decay estimates in \eqref{eq:edecayl01} and  \eqref{eq:pdecayl01v1} to obtain precise decay estimates for $L\psi_0$ and $\underline{L}\psi_0$ (and hence for the radial derivative). This step has been carried out in Proposition \ref{prop:partasymdrhopsi}.
\item[\textbf{Step 2}:]
Using Proposition \ref{prop:asympradfieldnonzeroIH}, we derive the asymptotics for $\psi_0$ in the region $\mathcal{A}^{\mathcal{I}}_{\gamma^{\mathcal{I}}_{\delta}}$, with $\delta<1$ suitably close to 1 and we use the estimates for $L\psi_0$ from \textbf{Step 1} to extend the asymptotics of $\psi_0$ to $\mathcal{A}^{\mathcal{I}}$.
\item[\textbf{Step 3}:]
Similarly, we apply Proposition \ref{prop:asympradfieldnonzeroIH} to obtain the asymptotics for $\psi_0$ in $\mathcal{A}^{\mathcal{H}}_{\gamma^{\mathcal{H}}_{\alpha}}$. We then integrate $\partial_r\psi_0$ from $r=r_{\mathcal{I}}$ in the direction of decreasing $r$ to obtain moreover the asymptotics for $\psi_0$ in $\mathcal{B}\cup \mathcal{A}^{\mathcal{H}}\setminus \mathcal{A}^{\mathcal{H}}_{\gamma^{\mathcal{H}}_{\alpha}}$.
\end{itemize}

\begin{proof}

\textbf{Step 2:}\\
\\
In order to obtain the late-time asymptotics of $\psi_0$ in $\mathcal{A}^{\mathcal{I}}$, we partition the region $\mathcal{A}^{\mathcal{I}}$ into $\mathcal{A}^{\mathcal{I}}_{\gamma^{\mathcal{I}}_{\delta}}=\{r\geq r_{\gamma^{\mathcal{I}}_{\delta}}(u)\}$ and $\mathcal{A}^{\mathcal{I}}\setminus \mathcal{A}^{\mathcal{I}}_{\gamma^{\mathcal{I}}_{\delta}}=\{r< r_{\gamma^{\mathcal{I}}_{\delta}}(u)\}$, with $\delta<1$, where we will choose $1-\delta$ to be suitably small. 

We first use the following identity:
\begin{equation}
\label{eq:identityuminusv}
u^{-1-k}-v^{-1-k}=\frac{v-u}{v\cdot u^{k+1}}\sum_{j=0}^k\left(\frac{u}{v}\right)^j=(v-u)(-1)^k\frac{1}{k!}T^k\left(\frac{1}{u\cdot v}\right),
\end{equation}
together with Proposition \ref{prop:asympradfieldnonzeroIH} and Lemma \ref{lm:relationruv}, to find $\eta,\epsilon>0$ suitably small, so that we can estimate:
\begin{equation}
\label{eq:asympsiinfnonzeroI0}
 \begin{split}
\Bigg|T^k\psi(u,r)-4I_0[\psi]T^k\left(\frac{1}{v\cdot u}\right)\Bigg|\lesssim &\: \left(\sqrt{E^{\epsilon}_{0;k}[\psi]}+I_0[\psi]\right) v^{-1}u^{-1-k-\eta}+P_{I_0,\beta;k}[\psi]\cdot v^{-1}u^{-1-k-\beta}
\end{split}
\end{equation}
in $\mathcal{A}^{\mathcal{I}}_{\gamma^{\mathcal{I}}_{\delta}}$.

Note that this implies in particular that
\begin{equation}
\label{eq:typecasymppsigammaI}
\begin{split}
|T^k\psi_0(u,r_{\gamma^{\mathcal{I}}_{\delta}}(u))-4(-1)^k (k+1)!\cdot I_0[\psi]\cdot u^{-2-k}|\lesssim&\:  \left(\sqrt{E^{\epsilon}_{0;k}[\psi]}+I_0[\psi]\right)\cdot u^{-2-k-\eta}\\
&+P_{I_0,\beta;k}[\psi]\cdot u^{-2-\beta-k}.
\end{split}
\end{equation}

We then integrate in the $-\partial_r$ direction, starting from $r=r_{\gamma^{\mathcal{I}}_{\delta}}(u)$ and we apply \eqref{eq:partasymdrhopsi2} from Step 1. By choosing $\eta>0$ and $\epsilon>0$ suitably small, we obtain:
\begin{equation}
\label{eq:typecasymppsileftgammaI}
 \begin{split}
\Bigg|T^k\psi(u,r)&-T^k\psi(u,r_{\gamma^{\mathcal{I}}_{\delta}}(u))+(-1)^{k+1} (k+1)!\cdot 4MH_0[\psi]u^{-2-k}\int_{r}^{r_{\gamma^{\mathcal{I}}_{\delta}}(u)}\frac{1}{(r'-M)^2}\,dr'\Bigg|\\
\lesssim&\: \left(\sqrt{E^{\epsilon}_{0; k+1}[\psi]}+H_0[\psi]\right) u^{-2-k-\eta}+P_{H_0,1;k}[\psi]\cdot r^{-1}u^{-3-k},
\end{split}
\end{equation}
with $r_{\mathcal{I}}\leq r\leq r_{\gamma^{\mathcal{I}}_{\delta}}(u)$.

By combining \eqref{eq:asympsiinfnonzeroI0}, \eqref{eq:typecasymppsigammaI} and \eqref{eq:typecasymppsileftgammaI}, we conclude that in $\mathcal{A}^{\mathcal{I}}$:
\begin{equation}
\label{eq:asymppsirI}
 \begin{split}
\Bigg|T^k\psi(u,r)&-\left(4MH_0[\psi]\frac{T^k(u^{-2})}{r-M}+4I_0[\psi]T^k\left(\frac{1}{uv}\right)\right)\Bigg|\\
\lesssim&\: \left(\sqrt{E^{\epsilon}_{0; k+1}[\psi]}+H_0[\psi]+I_0[\psi]\right) v^{-1}u^{-1-k-\eta}\\
&+P_{H_0,1;k}[\psi]\cdot r^{-1}u^{-3-k}+P_{I_0,\beta;k}[\psi]\cdot v^{-1}u^{-1-k-\beta}.
\end{split}
\end{equation}

\textbf{Step 3:}\\
\\
We now turn to the region $\mathcal{A}^{\mathcal{H}}\cup\mathcal{B}=\{r\leq r_{\mathcal{I}}\}$. We will partition this region into the region $\mathcal{A}^{\mathcal{H}}_{\gamma^{\mathcal{H}}_{\alpha}}=\{r\leq r_{\gamma^{\mathcal{H}}_{\alpha}(v)}\}$ and $\{r_{\gamma^{\mathcal{H}}_{\alpha}(v)}\leq r\leq r_{\mathcal{I}}\}$ with $\alpha<1$ and $1-\alpha$ suitably small.

Let us first consider the region $\mathcal{A}^{\mathcal{H}}_{\gamma^{\mathcal{H}}_{\alpha}}$. By using the identity
\begin{equation}
\label{eq:identityvminusu}
v^{-1-k}-u^{-1-k}=\frac{u-v}{u\cdot v^{k+1}}\sum_{j=0}^k\left(\frac{v}{u}\right)^j=(u-v)(-1)^k\frac{1}{k!}T^k\left(\frac{1}{u\cdot v}\right)
\end{equation}
together with Lemma \ref{lm:relationruv} and Proposition \ref{prop:asympradfieldnonzeroIH}, we have that for $1-\alpha$ suitably small we can estimate in $r\leq r_{\gamma^{\mathcal{H}}_{\alpha}}(v)$:
\begin{equation}
\label{eq:asympsihornonzeroH0}
\begin{split}
\Bigg|T^k\psi_0(u,v)&-4\frac{M}{r\sqrt{D}}H_0[\psi]T^k\left(\frac{1}{u\cdot v}\right)\Bigg|\\
\leq&\: C\left(\sqrt{E_{0;k+1}^{\epsilon}[\psi]}+H_0[\psi]+P_{H_0,1;k}[\psi]\right)D^{-\frac{1}{2}}u^{-1}v^{-1-k-\eta}.
\end{split}
\end{equation}

We now consider the region  $\{r_{\gamma^{\mathcal{H}}_{\alpha}(v)}\leq r\leq r_{\mathcal{I}}\}$ and use \eqref{eq:asymlbarpsidr} to integrate from $r'=r_{\mathcal{I}}$ to $r'=r\geq r_{\gamma^{\mathcal{H}}_{\alpha}(v)}$ along constant $v$ hypersurfaces with $v\geq v|_{\Sigma_0}(r_{\mathcal{I}})$, for $1-\alpha>0$ suitably small: we have that there exist  $\epsilon,\eta>0$ suitably small such that
\begin{equation*}
\begin{split}
\Bigg|&T^k\psi_0(v,r)- T^k\psi_0(v,r_{\mathcal{I}})+(-1)^{k+1}\frac{(k+1)!4MH_0}{v^{k+2}}\int_{r}^{r_{\mathcal{I}}}(r'-M)^{-2}\,dr'\Bigg|\\
\lesssim&\: \left(\sqrt{E^{\epsilon}_{0; k+1}[\psi]}+H_0[\psi]\right)(r-M)^{-1}v^{-2-k-\eta}\\
&+P_{H_0,1;k}[\psi]\cdot v^{-3-k} .
\end{split}
\end{equation*}
Combined with \eqref{eq:asymppsirI} this implies that for $1-\alpha>0$ suitably small: there exists an $\eta>0$ and $\epsilon>0$ suitably small such that for all  $r_{\gamma^{\mathcal{H}}_{\alpha}(v)}\leq r \leq r_{\mathcal{I}}$:
\begin{equation}
\label{eq:asymppsirH}
\begin{split}
\Bigg|&T^k\psi_0(v,r)-(-1)^{k}(k+1)!\left(\frac{4MH_0}{r\sqrt{D}}+4I_0[\psi]\right)v^{-k-2}\Bigg|\\
\lesssim&\: \left(\sqrt{E^{\epsilon}_{0; k+1}[\psi]}+H_0[\psi]+I_0[\psi]\right)(r-M)^{-1}v^{-2-k-\eta}\\
&+(r-M)^{-1}P_{H_0,1;k}[\psi]\cdot v^{-3-k}+P_{I_0,\beta;k}[\psi]\cdot v^{-2-\beta-k} .
\end{split}
\end{equation}

By combining the estimates for $T^k\psi_0$ in the regions $r_{\mathcal{I}}\leq r\leq r_{\gamma^{\mathcal{I}}_{\delta}}(u)$, $r\geq r_{\gamma^{\mathcal{I}}_{\delta}}(u)$, $r_{\gamma^{\mathcal{H}}_{\alpha}(v)}\leq r \leq r_{\mathcal{I}}$ and  $r\leq r_{\gamma^{\mathcal{H}}_{\alpha}}(v)$, we arrive at \eqref{eq:asympsi}.
\end{proof}

\begin{remark}
We can alternatively consider $\underline{\psi}:=M^{-1}(r-M)\psi=\frac{r}{M}\sqrt{D}\psi$ and reverse the roles of $\mathcal{H}^+$ and $\mathcal{I}^+$ in order to obtain the asymptotics for $\psi_0$ in Proposition \ref{prop:asympsi}; see also the proof of Proposition \ref{prop:explicitexprtimeint2}.
\end{remark}

\section{Time inversion theory}
\label{sec:timeint}
In this section, we will construct an auxilliary ``time integral'' function ${\psi}_0^{(1)}: \mathcal{R}\setminus \mathcal{H}^+\to \R$, which satisfies $T\psi_0^{(1)}=\psi_0$ and $\square_g\psi_0^{(1)}=0$. This construction is fundamental to obtaining asymptotics for $\psi_0$ in Section \ref{sec:asympzeroconst}, when the initial data is of Type \textbf{A}, \textbf{B} or \textbf{D}; that is to say, when $H_0[\psi]$ or $I_0[\psi]$ vanish.

\subsection{Regular time inversion in $\protect\mathring{\mathcal{R}}$}
\label{sec:TheRegularTimeInversionConstructionInOversetOMathcalR}
Consider $\mathring{\mathcal{R}}= \mathcal{R}\setminus \partial \mathcal{R}$. We have the following

\begin{definition}
\label{def:timeint}
Let $\psi_0$ be the spherical mean of a solution $\psi$ to \eqref{eq:waveequation} on extremal Reissner--Nordstr\"om with $I_0[\psi]$ a well-defined limit. We then define the \emph{time integral} $\psi^{(1)}_0$ of $\psi_0$ to be the function ${\psi}^{(1)}: \mathring{\mathcal{R}}\rightarrow \mathbb{R}$, such that 
\begin{itemize}
\item[\emph{(i)}]
$T\psi^{(1)}_0=\psi_0$,
\item[\emph{(ii)}]
 $\square_g\psi^{(1)}_0=0$,
\item[\emph{(iii)}]
 $\lim_{v\to \infty}\psi^{(1)}_0(0,v)=0$,
 \item[\emph{(iv)}]
 $\lim_{u\to \infty}\underline{L}\psi_0^{(1)}(u,v_0)=0$.
\end{itemize}
\end{definition}

As we shall see explicitly in the next proposition, $\psi_0^{(1)}$ is well-defined because it is the unique solution to a boundary-value problem for an inhomogeneous ODE (with boundary conditions at $r=M$ and $r=+\infty$).

\begin{proposition}
\label{prop:explicitexprtimeint}
The time integral $\psi_0^{(1)}$ of the spherical mean $\psi_0$ of a solution to \eqref{eq:waveequation} on extremal Reissner--Nordstr\"om satisfies the following identities:
\begin{align*}
2r^2\underline{L}\psi_0^{(1)}(u,v_0)=&\: 2\int_{u}^{\infty} r\underline{L}\phi_0(u',v_0)\,du' \quad \textnormal{on ${{N}_0^{\mathcal{H}}}\cap \mathring{\mathcal{R}}$},\\
Dr^2\partial_{\rho}\psi_0^{(1)}(0,\rho)=&\:\int_{r_{\mathcal{H}}}^{\rho}\left[-2(1-h\cdot D)r\partial_{\rho}\phi_0+(2-h\cdot D)rhT\phi_0+r\cdot (hD)' \phi_0\right]|_{\Sigma_0}(\rho')d\rho' \\
&+h\cdot Dr^2\phi_0|_{\Sigma_0}(\rho=r_{\mathcal{H}})-2\int_{u_{r_{\mathcal{H}}}(v_0)}^{\infty} r\underline{L}\phi_0(u',v_0)\,du'\quad \textnormal{on $\Sigma_{\tau}\cap \{r_{\mathcal{H}}\leq r\leq r_{\mathcal{I}}\}$},\\
2r^2L\psi_0^{(1)}(u_0,v)=&\:C_0+2\int_{v_{r_{\mathcal{I}}}(u_0)}^{v} rL\phi_0(u_0,v')\,dv'\quad \textnormal{on ${{N}_0^{\mathcal{I}}}$},
\end{align*}
where $h$ is the function of $r$ given by \eqref{definitionh}, and we use the shorthand notation
\begin{equation*}
\begin{split}
4\pi C_0[\psi]:=2\int_{N^{\mathcal{H}}_0} r\underline{L}\phi_0\,d\omega du'+\int_{\Sigma_0\cap \mathcal{B}} n_{\Sigma_0}(\psi)\,d\mu_{0}+4\pi r\phi_0|_{N^{\mathcal{H}}_0} (r=r_{\mathcal{H}})+4\pi r\phi_0|_{N^{\mathcal{H}}_0} (r=r_{\mathcal{I}}).
\end{split}
\end{equation*}
If $\lim_{r\to \infty}r^3\partial_r\phi_0|_{{N^{\mathcal{I}}_0}}<\infty$ (\underline{and therefore $I_0[\psi]=0$}), then we can further express,
\begin{equation*}
r\psi^{(1)}_0|_{{{N}_0^{\mathcal{I}}}}(r)=-r\left[C_0[\psi]+2\int_{v_{r_{\mathcal{I}}}(u_0)}^{\infty} rL\phi_0(u_0,v')\,dv'\right](r-M)^{-1}+2r\int_{r}^{\infty} (r'-M)^{-2}\int_{r'}^{\infty}r\partial_r\phi_0|_{{{N}_0^{\mathcal{I}}}}(r'')\,dr''dr',
\end{equation*}
and we have that in $(u,r)$ coordinates
\begin{equation*}
\begin{split}
I_0[\psi^{(1)}]=&MC_0[\psi]+2M\int_{v_{r_{\mathcal{I}}}(u_0)}^{\infty} rL\phi_0(u_0,v')\,dv'- \lim_{r\to \infty}r^3\partial_r\phi_0|_{{{N}_0^{\mathcal{I}}}}\\
=& M^2\phi_0|_{\mathcal{H}^+} (v=v_0)+\frac{M}{4\pi}\int_{N^{\mathcal{H}}_0} \underline{L}\psi_0\,r^2d\omega du'+\frac{M}{4\pi}\int_{\Sigma_0\cap \mathcal{B}} n_{\Sigma_0}(\psi)\,d\mu_{0}\\
&+M r\phi_0|_{N^{\mathcal{H}}_0} (r=r_{\mathcal{I}})+\frac{2M}{4\pi}\int_{N^{\mathcal{I}}_0} rL\phi_0\,d\omega dv'- \lim_{r\to \infty}r^3\partial_r\phi_0|_{{{N}_0^{\mathcal{I}}}}.
\end{split}
\end{equation*}
\end{proposition}
\begin{proof}
Note that we can write
\begin{equation*}
\square_g\psi_0^{(1)}(\tau,\rho)=\square_g\psi_0^{(1)}(0,\rho)+\int_0^{\tau} \square_g\psi_0(\tau',\rho)\,d\tau'=\square_g\psi_0^{(1)}(0,\rho),
\end{equation*}
and therefore $\square_g\psi_0^{(1)}=0$ in $\mathring{\mathcal{R}}$ if and only if $\square_g\psi_0^{(1)}(0,\rho)=0$ for $\rho>M$, which is equivalent to the following equation:
\begin{equation*}
\underline{L}(r^2\underline{L}\psi_0^{(1)})=r\underline{L}\phi_0.
\end{equation*}
We therefore obtain the following identity on ${{N}_0^{\mathcal{H}}}\cap\mathring{\mathcal{R}}$:
\begin{equation*}
2r^2\underline{L}\psi_0^{(1)}(u,v_0)= \lim_{u\to \infty}2r^2\underline{L}\psi_0^{(1)}(u,v_0)-2\int_{u}^{\infty} r\underline{L}\phi_0(u',v_0)\,du',
\end{equation*}
where the first term on the right-hand side is zero by definition of $\psi^{(1)}$.

Recall that $\partial_{\rho}=-2D^{-1}\underline{L}+h\cdot T=2D^{-1}L+(h-2D^{-1})\cdot T$, so by the above, we have that
\begin{equation*}
D(r_{\mathcal{H}})r_{\mathcal{H}}^2\partial_{\rho}\psi_0^{(1)}(u_{r_{\mathcal{H}}}(v_0),v_0)=2\int_{u_{r_{\mathcal{H}}}(v_0)}^{\infty} r\underline{L}\phi_0(u',v_0)\,du'+hD(r_{\mathcal{H}}) r_{\mathcal{H}}\phi_0(u_{r_{\mathcal{H}}}(v_0),v_0).
\end{equation*}
We compute
\begin{equation*}
\partial_{\rho}(Dr^2\partial_{\rho}\psi_0^{(1)})=-2(1-h\cdot D)r\partial_{\rho}\phi_0+(2-h\cdot D)rhT\phi_0+r\cdot (hD)' \phi_0,
\end{equation*}
so, by using all the above estimates, we can conclude that everywhere on $\Sigma_0\cap \mathcal{B}$:
\begin{equation*}
\begin{split}
Dr^2\partial_{\rho}\psi_0^{(1)}(0,\rho)=&\:2\int_{u_{r_{\mathcal{H}}}(v_0)}^{\infty} r\underline{L}\phi_0(u',v_0)\,du'+hD(r_{\mathcal{H}}) r_{\mathcal{H}}\phi_0(u_{r_{\mathcal{H}}}(v_0),v_0)\\
&+\int_{r_{\mathcal{H}}}^{\rho}[-2(1-h\cdot D)r\partial_{\rho}\phi_0+(2-h\cdot D)rhT\phi_0+r\cdot (hD)' \phi_0]|_{\Sigma_0}(\rho')d\rho'.
\end{split}
\end{equation*}
By $2L=D\partial_{\rho}+(2-hD)\cdot T$ we also obtain the following expression for $\psi^{(1)}$ on ${{N}_0^{\mathcal{I}}}$:
\begin{equation*}
\begin{split}
2r_{\mathcal{I}}^2L\psi_0^{(1)}(u_0,v_{r_{\mathcal{I}}}(u_0))=&2\int_{u_{r_{\mathcal{H}}}(v_0)}^{\infty} r\underline{L}\phi_0(u',v_0)\,du'+hD(r_{\mathcal{H}}) r_{\mathcal{H}}\phi_0(u_{r_{\mathcal{H}}}(v_0),v_0)\\
&+\int_{r_{\mathcal{H}}}^{r_{\mathcal{I}}}[-2(1-h\cdot D)r\partial_{\rho}\phi_0+(2-h\cdot D)rhT\phi_0+r\cdot (hD)' \phi_0]|_{\Sigma_0}(\rho')d\rho'\\
&+(2-h(r_{\mathcal{I}})D(r_{\mathcal{I}}))r_{\mathcal{I}}\phi_0(u_0,v_{r_{\mathcal{I}}}(u_0))=:C_0[\psi].
\end{split}
\end{equation*}
By using that the normal $n_{\Sigma_0}$ to $\Sigma_0\cap \mathcal{B}$ can be expressed as follows:
\begin{equation*}
\sqrt{\det g_{\Sigma_0\cap \mathcal{B}}}n_{\Sigma_0}=r^2\sin \theta\left[(hD-1)\partial_{\rho} +h(2-hD)T\right],
\end{equation*}
we can rewrite
\begin{equation*}
\begin{split}
[-2&(1-h\cdot D)r\partial_{\rho}\phi_0+(2-h\cdot D)rhT\phi_0+r\cdot (hD)' \phi_0]\sin \theta\\
=&\:2(h D-1)r^2\sin \theta\partial_{\rho}\psi+2(hD-1)\sin \theta \phi_0+(2-hD)h r^2\sin \theta T\psi_0+r(hD)' \sin \theta \phi_0\\
=&\:\partial_{\rho}( (hD-1)r^2 \psi_0)\sin \theta+\sqrt{\det g_{\Sigma_0\cap \mathcal{B}} }n_{\Sigma_0}(\psi).
\end{split}
\end{equation*}
Hence, we obtain
\begin{equation*}
\begin{split}
4\pi C_0=&2\int_{N^{\mathcal{H}}_0} r\underline{L}\phi_0\,d\omega du'+\int_{\Sigma_0\cap \mathcal{B}} n_{\Sigma_0}(\psi)\,d\mu_{0}+4\pi r\phi_0|_{N^{\mathcal{H}}_0} (r=r_{\mathcal{H}})+4\pi r\phi_0|_{N^{\mathcal{H}}_0} (r=r_{\mathcal{I}})\\
=&4\pi M\phi_0|_{\mathcal{H}^+} (v=v_0)-\int_{N^{\mathcal{H}}_0} \underline{L}\psi_0\,r^2d\omega du'+\int_{\Sigma_0\cap \mathcal{B}} n_{\Sigma_0}(\psi)\,d\mu_{0}+4\pi r\phi_0|_{N^{\mathcal{H}}_0} (r=r_{\mathcal{I}}).
\end{split}
\end{equation*}
Since $\psi^{(1)}$ satisfies
\begin{equation*}
L(r^2L\psi_0^{(1)})=rL\phi_0,
\end{equation*}
we therefore  conclude that everywhere on ${{N}_0^{\mathcal{I}}}$ we can write
\begin{equation*}
2r^2L\psi_0^{(1)}(u_0,v)=C_0[\psi]+2\int_{v_{r_{\mathcal{I}}}(u_0)}^{v} rL\phi_0(u_0,v')\,dv'.
\end{equation*}

In particular, if $I_0[\psi]=0$, we have that

\begin{equation*}
\left|\int_{v_{r_{\mathcal{I}}}(u_0)}^{\infty} rL\phi_0(u_0,v')\,dv'\right|<\infty
\end{equation*}
so we can switch to $(u,r)$ coordinates in order to express:
\begin{equation*}
\psi^{(1)}_0|_{{{N}_0^{\mathcal{I}}}}(r)=-2\left[C_0[\psi]+2\int_{v_{r_{\mathcal{I}}}(u_0)}^{\infty} rL\phi_0(u_0,v')\,dv'\right](r-M)^{-1}-2\int_{r}^{\infty} (r'-M)^{-1}\int_{r'}^{\infty}r\partial_r\phi_0|_{{{N}_0^{\mathcal{I}}}}(r'')\,dr''
\end{equation*}
The expression for $I_0[\psi_0^{(1)}]$ then follows from multiplying both sides by $r$, using that $\lim_{r\to \infty}r^3\partial_r\phi_0|_{{{N}_0^{\mathcal{I}}}}<\infty$ and taking an $r$ derivative.
\end{proof}

\begin{proposition}
\label{prop:explicitexprtimeint2}
Let $\psi_0$ be the spherical mean of a solution $\psi$ to \eqref{eq:waveequation} on extremal Reissner--Nordstr\"om. If $H_0[\psi]=0$, then the time integral $\psi_0^{(1)}$ of $\psi_0$ satisfies moreover
\begin{equation*}
r\psi^{(1)}_0|_{{{N}_0^{\mathcal{H}}}}(r)=-\tilde{r}\left[\underline{C}_0[\psi]+2\int_{u_{r_{\mathcal{H}}}(v_0)}^{\infty} \tilde{r}\underline{L}\phi_0(u',v_0)\,dv'\right](\tilde{r}-M)^{-1}+2\tilde{r}\int_{\tilde{r
}}^{\infty} (\tilde{r}'-M)^{-2}\int_{\tilde{r}'}^{\infty}\tilde{r}\partial_{\tilde{r}}\phi_0|_{{{N}_0^{\mathcal{H}}}}(\tilde{r}'')\,d\tilde{r}''d\tilde{r}',
\end{equation*}
with
\begin{align*}
\tilde{r}=&\:M+M^2(r-M)^{-1},\\
4\pi \underline{C}_0[\psi]:=&\:2\int_{N^{\mathcal{I}}_0}\frac{M}{r-M} \cdot r{L}\phi_0\,d\omega dv+\int_{\Sigma_0\cap \mathcal{B}} n_{\Sigma_0}\left(\frac{M}{r-M}\cdot \psi\right)\,d\mu_{0}\\
&+4\pi \frac{M}{r-M} \cdot r\phi_0|_{N^{\mathcal{H}}_0} (r=r_{\mathcal{H}})+4\pi \frac{M}{r-M} \cdot r\phi_0|_{N^{\mathcal{H}}_0} (r=r_{\mathcal{I}}).
\end{align*}
and we have that in $(v,r)$ coordinates
\begin{equation*}
H_0[\psi^{(1)}]=M\underline{C}_0[\psi]+2M\int_{u_{r_{\mathcal{H}}}(v_0)}^{\infty} \frac{Mr}{r-M}\underline{L}\phi_0(u',v_0)\,du'+M^4\partial_{r}^2\phi_0|_{{{N}_0^{\mathcal{H}}}}(r=M).
\end{equation*}
\end{proposition}
\begin{proof}
Consider now the rescaled functions $\underline{\psi}_0=M^{-1}(r-M)\psi_0$ and $\underline{\psi}^{(1)}_0=M^{-1}(r-M)\psi_0^{(1)}$. Note that by Proposition \ref{prop:explicitexprtimeint}, we have that
\begin{equation*}
\lim_{u\to \infty}\underline{\psi}^{(1)}_0(u,v_0)=0,\\
\lim_{v\to \infty} L \underline{\psi}^{(1)}_0(u_0,v)=0,
\end{equation*}
so $\underline{\psi}^{(1)}_0$ satisfies analogous boundary conditions to $\psi_0^{(1)}$, but with $u$ and $v$ and $L$ and $\underline{L}$ interchanged. We introduce the notation $\tilde{r}=M+M^2(r-M)^{-1}$ and $\tilde{D}(\tilde{r})=(1-M\tilde{r}^{-1})^2$, we have that $\tilde{r} \underline{\psi}_0=r\psi_0$ and $\tilde{r} \underline{\psi}_0^{(1)}=r\psi_0^{(1)}$, and $\underline{\psi}_0^{(1)}$ satisfies the equations:
\begin{align*}
L(\tilde{r}^2L\underline{\psi}^{(1)}_0)=&\tilde{r}L \phi_0,\\
\underline{L}(\tilde{r}^2\underline{L}\underline{\psi}^{(1)}_0)=\:&\tilde{r}\underline{L}\phi_0.
\end{align*}

We can therefore repeat the arguments above, starting the integration along ${{N}_0^{\mathcal{I}}}$ rather than ${{N}_0^{\mathcal{H}}}$, to obtain the following expressions:
\begin{align*}
2\widetilde{r}^2L\underline{\psi}^{(1)}_0(u_0,v)=&\: 2\int_{v}^{\infty} \widetilde{r}L\phi_0(u_0,v)\,dv' \quad \textnormal{on ${{N}_0^{\mathcal{I}}}$},\\
\tilde{D}(\tilde{r})\tilde{r}^2\partial_{\tilde{\rho}}\underline{\psi}^{(1)}_0(0,\tilde{\rho})=&\:\int^{\tilde{\rho}}_{\tilde{r}(r_{\mathcal{I}})}\left[-2(1-\tilde{h}\cdot \tilde{D})\tilde{r}\partial_{\tilde{\rho}}\phi_0+(2-\tilde{h}\cdot \tilde{D})\tilde{r}\tilde{h}T\phi_0+\tilde{r}\cdot \frac{d(\tilde{h}\tilde{D})}{d\tilde{r}} \phi_0\right]|_{\Sigma_0}(\tilde{\rho}')d\tilde{\rho}' \\
&+\tilde{h}\cdot \tilde{D}(\tilde{r})\tilde{r}^2\phi_0|_{\Sigma_0}(\rho=r_{\mathcal{I}})+2\int_{v_{r_{\mathcal{I}}}(u_0)}^{\infty} \tilde{r}{L}\phi_0(u_0,v')\,dv'\quad \textnormal{on ${S}\cap \{r_{\mathcal{H}}\leq r\leq r_{\mathcal{I}}\}$},\\
2\tilde{r}^2\underline{L}\underline{\psi}^{(1)}_0(u,v_0)=&\:\underline{C}_0[\psi]+2\int_{u_{r_{\mathcal{H}}}(u_0)}^{u} \tilde{r}\underline{L}\phi_0(u',v_0)\,du'\quad \textnormal{on ${{N}_0^{\mathcal{H}}}\cap \mathring{\mathcal{M}}$},
\end{align*}
with
\begin{align*}
\tilde{h}(\tilde{r})=&\:(2\tilde{D}^{-1}-h M^2(\tilde{r}-M)^{-2})=(2D^{-1}-h) M^2(\tilde{r}-M)^{-2},\\
2-\tilde{h}\tilde{D}=&\:hM^2\tilde{r}^{-2}=hD,\\
\partial_{\tilde{\rho}}=&-M^{-2}(r-M)^{2}\partial_{\rho}=-M^{2}(\tilde{r}-M)^{-2}\partial_{\rho},\\
\underline{C}_0[\psi]:=&\:(2-\tilde{h}\cdot \tilde{D})\tilde{r}\phi_0|_{\Sigma_0}(\rho=r_{\mathcal{H}})\\
&-\int^{\tilde{r}(r_{\mathcal{H}})}_{\tilde{r}(r_{\mathcal{I}})}\left[2(1-\tilde{h}\cdot \tilde{D})\tilde{r}\partial_{\tilde{\rho}}\phi_0-(2-\tilde{h}\cdot \tilde{D})\tilde{r}\tilde{h}T\phi_0-\tilde{r}\cdot \frac{d(\tilde{h}\tilde{D})}{d\tilde{r}} \phi_0\right]|_{\Sigma_0}(\tilde{\rho}')d\tilde{\rho}'\\
&+\tilde{h}\cdot \tilde{D}(\tilde{r})\tilde{r}\phi_0|_{\Sigma_0}(\rho=r_{\mathcal{I}})+2\int_{v_{r_{\mathcal{I}}}(u_0)}^{\infty} \tilde{r}L\phi_0(u_0,v')\,dv'\\
=&\:hD\tilde{r}\phi_0|_{\Sigma_0}(\rho=r_{\mathcal{H}})\\
&-\int_{r_{\mathcal{H}}}^{r_{\mathcal{I}}}\left[2(1-hD)\tilde{r}\partial_{\rho}\phi_0-(2-hD)\tilde{r}hT\phi_0-\tilde{r}\cdot \frac{d(hD)}{dr} \phi_0\right]|_{\Sigma_0}(\rho')d\rho'\\
&+(2-hD)\tilde{r}\phi_0|_{\Sigma_0}(\rho=r_{\mathcal{I}})+2\int_{v_{r_{\mathcal{I}}}(u_0)}^{\infty} \tilde{r}L\phi_0(u_0,v')\,dv'\\
=&\:\:hD\tilde{r}\phi_0|_{\Sigma_0}(\rho=r_{\mathcal{H}})\\
&-\int_{r_{\mathcal{H}}}^{r_{\mathcal{I}}}\frac{M}{r-M}\left[2(1-hD)r\partial_{\rho}\phi_0-(2-hD)rhT\phi_0-r\cdot \frac{d(hD)}{dr} \phi_0\right]|_{\Sigma_0}(\rho')d\rho'\\
&+(2-hD)\tilde{r}\phi_0|_{\Sigma_0}(\rho=r_{\mathcal{I}})+2\int_{v_{r_{\mathcal{I}}}(u_0)}^{\infty} \frac{Mr}{r-M}L\phi_0(u_0,v')\,dv'.
\end{align*}

We recall from the proof of Proposition \ref{prop:explicitexprtimeint} that
\begin{equation*}
\begin{split}
\frac{M}{r-M}[-2&(1-h\cdot D)r\partial_{\rho}\phi_0+(2-h\cdot D)rhT\phi_0+r\cdot (hD)' \phi_0]\sin \theta\\
=&\:\frac{M}{r-M}\partial_{\rho}( (hD-1)r^2 \psi_0)\sin \theta+\sqrt{\det g_{\Sigma_0\cap \mathcal{B}} }\frac{M}{r-M}n_{\Sigma_0}(\psi)\\
=&\:\partial_{\rho}\left(\frac{M}{r-M} (hD-1)r^2 \psi_0\right)\sin \theta+\sqrt{\det g_{\Sigma_0\cap \mathcal{B}} } n_{\Sigma_0}\left(\frac{M}{r-M}\psi\right)
\end{split}
\end{equation*}
and hence,
\begin{equation*}
\begin{split}
4\pi \underline{C}_0[\psi]=&\:2\int_{N^{\mathcal{I}}_0}\frac{M}{r-M} \cdot r{L}\phi_0\,d\omega du'+\int_{\Sigma_0\cap \mathcal{B}} n_{\Sigma_0}\left(\frac{M}{r-M}\cdot \psi\right)\,d\mu_{0}\\
&+4\pi \frac{M}{r-M} \cdot r\phi_0|_{N^{\mathcal{H}}_0} (r=r_{\mathcal{H}})+4\pi \frac{M}{r-M} \cdot r\phi_0|_{N^{\mathcal{H}}_0} (r=r_{\mathcal{I}}).
\end{split}
\end{equation*}
Hence, if $H_0[\psi]=0$, we have that
\begin{equation*}
\underline{\psi}^{(1)}_0|_{{{N}_0^{\mathcal{H}}}}(\tilde{r})=-\left[\underline{C}_0[\psi]+2\int_{u_{r_{\mathcal{H}}}(v_0)}^{\infty} \tilde{r}\underline{L}\phi_0(u',v_0)\,dv'\right](\tilde{r}-M)^{-1}+2\int_{\tilde{r
}}^{\infty} (\tilde{r}'-M)^{-2}\int_{\tilde{r}'}^{\infty}\tilde{r}\partial_{\tilde{r}}\phi_0|_{{{N}_0^{\mathcal{H}}}}(\tilde{r}'')\,d\tilde{r}''d\tilde{r}',
\end{equation*}
and the expression for $H_0[\psi_0^{(1)}]$ is derived as above.
\end{proof}

The above propositions motivate the following definitions:
\begin{definition}
We define the \textbf{time-inverted constants} $I_0^{(1)}[\psi]$ and $H_0^{(1)}[\psi]$ as follows:
\begin{align*}
I_0^{(1)}[\psi]:=&\:I_0[\psi_0^{(1)}]\quad \textnormal{if $\lim_{v\to \infty}r^3L\psi_0(u_0,v)<\infty$ (and therefore $I_0[\psi]=0$)},\\
H_0^{(1)}[\psi]:=&\:H_0[\psi_0^{(1)}]\quad \textnormal{if $H_0[\psi]=0$}.
\end{align*}
\end{definition}

\begin{remark}
If we assume the qualitative decay statements: $r\psi_0|_{\mathcal{I}^+}(u)\to 0$ as $u\to \infty$ and $r\psi_0|_{\mathcal{H}^+}(v)\to 0$ as $v\to \infty$, we can use the results of Proposition \ref{prop:explicitexprtimeint} and \ref{prop:explicitexprtimeint2} to obtain the following alternative expressions for $I_0^{(1)}[\psi]$ and $H_0^{(1)}[\psi]$: if $\lim_{r \to \infty} r^3L\phi_0|_{\Sigma_0}<\infty$, then
\begin{equation*}
\begin{split}
I_0^{(1)}[\psi]=&-M \lim_{r\to \infty} r\psi_0^{(1)}|_{\Sigma_0} (r)-2\lim_{r\to \infty} r^3 L\phi_0|_{\Sigma_0}\\
=&\: M \int_{u_0}^{\infty} r\psi_0|_{\mathcal{I}^+}(u)\,du-2\lim_{r\to \infty} r^3 L\phi_0|_{\Sigma_0},
\end{split}
\end{equation*}
and if $H_0[\psi]=0$, we obtain
\begin{equation*}
\begin{split}
H_0^{(1)}[\psi]=&-M \lim_{r\downarrow M} r\psi_0^{(1)}|_{\Sigma_0} (r)-M^4Y^2\phi_0|_{\Sigma_0}(r=M)\\
=&\: M \int_{v_0}^{\infty} r\psi_0|_{\mathcal{H}^+}(v)\,dv-M^4Y^2\phi_0|_{\Sigma_0}(r=M).
\end{split}
\end{equation*}
We will recover the above decay assumption for $r\psi_0^{(1)}$ in Proposition \ref{prop:decayesttimeint}. See also the discussion in Section 1.6 of \cite{paper2} for an analogous expression for $I_0^{(1)}[\psi]$ in the sub-extremal setting.
\end{remark}

\subsection{Extension of the time integral $\protect\psi^{(1)}$ in $\protect\mathcal{R}$ }
\label{sec:extendedtimeinversionr}
In this section, we will investigate the regularity properties of the \emph{continuous extensions} of the time integral functions $\psi^{(1)}$ defined in Section \ref{sec:TheRegularTimeInversionConstructionInOversetOMathcalR} to the full spacetime region $\mathcal{R}$. We will moreover discuss the singular properties of (derivatives) of the radiation field at $\mathcal{I}^+$.

\subsubsection{Regular extension in $\protect\mathcal{R}$ for Type \textbf{B} perturbations}
\label{sec:TheSingularTimeInversionInProtectMathcalRB}
We first consider the case of Type \textbf{B} data.
\begin{proposition}
\label{eq:smoothextTypeBtimeint}
Let $\psi_0^{(1)}$ be the time integral of a smooth solution $\psi_0$ to \eqref{eq:waveequation} corresponding to initial data of Type \textbf{B}. Then $\psi_0^{(1)}$ can be extended uniquely as a smooth function to $\mathcal{R}$. Furthermore, $I_0[\psi^{(1)}]$ is well-defined.
\end{proposition}
\begin{proof}
By Proposition \ref{prop:explicitexprtimeint}, we have that in $(v,r)$ coordinates
\begin{equation}
\label{eq:explicitdrpsi1NH}
\partial_r\psi_0^{(1)} (v_0,r)=\frac{2}{(r-M)^2}\int_M^r r'\partial_r\phi_0(v_0,r')\,dr'
\end{equation}
along $\mathcal{N}^{\mathcal{H}}_0$.

If we assume that $H_0[\psi]=0$, we can use smoothness of $\phi_0$ together with Taylor's theorem to obtain the following: for any $N\in \N$, we can decompose for $r\leq r_{\mathcal{H}}$: $r\partial_r\phi_0(v_0,r)=\sum_{k=1}^Np_k(r-M)^k+(r-M)^{N+1}f_N(v,r)$, for some smooth function $f_N: [M,r_{\mathcal{H}})\to \R$ and coefficients $p_k\in \R$, with $k\in 1,\ldots,N$. 

Hence,
\begin{equation*}
\partial_r\psi_0^{(1)} (v_0,r)=\sum_{k=1}^N \frac{2p_k}{(k+1)(r-M)^2}(r-M)^{k+1}+\frac{2}{(r-M)^2}\int_M^r (r'-M)^{N+1}f_N(r')\,dr'.
\end{equation*}
It is clear that $\partial_r^{N}\psi_0$ attains a finite limit at $r=M$. Using that $N$ can be taken to be arbitrarily large and $T\psi_0^{(1)}=\psi_0$, we can conclude that $\psi^{(1)}$ extends smoothly to $r=M$. The second part of the proposition follows immediately from Proposition \ref{prop:explicitexprtimeint}, using that $\lim_{v\to \infty} r^3L\phi(u_0,v)$ is well-defined for Type \textbf{B} data.
\end{proof}

\textbf{By Proposition \ref{eq:smoothextTypeBtimeint}, all the estimates in Section \ref{sec:asympnonzeroconst} can be applied \underline{without modification} when $\psi_0$ is replaced by $\psi^{(1)}_0$ in the case of Type \textbf{B} data!}

\subsubsection{Singular horizon extension in $\protect\mathcal{R}$ for Type \textbf{A} perturbations}
\label{sec:TheSingularTimeInversionInProtectMathcalRA}
We now consider data of Type \textbf{A} and show that, due to the non-vanishing of $H_0[\psi]$, the time integral $\psi^{(1)}$ displays singular behaviour at $\mathcal{H}^+$.
\begin{proposition}
\label{lm:someregproptimeint}
Let $\psi_0^{(1)}$ be the time integral of a smooth solution $\psi_0$ to \eqref{eq:waveequation} corresponding to initial data of Type \textbf{A}. Then $\psi_0^{(1)}$ cannot be extended as a continuous function to $\mathcal{R}$. However, $I_0[\psi^{(1)}]$ is well-defined. More precisely, we can decompose
\begin{align}
\label{eq:blowupdrpsi1typeA}
\partial_{r}\psi^{(1)}_0(v,r)=&-2M^{-1}H_0[\psi](r-M)^{-1}+f(v,r),\\
\label{eq:blowupdrpsi1typeA2}
\psi^{(1)}_0(v,r)=& -2M^{-1}H_0[\psi]\log(r-M)+\widetilde{f}(v,r),
\end{align}
for some smooth, spherically symmetric functions $f,\widetilde{f}$ on $\mathcal{R}$.

Furthermore, for $\epsilon>0$ arbitrarily small, we can estimate
\begin{align}
\label{eq:finiteTenergypsi1}
\int_{\Sigma_0}& J^T[\psi^{(1)}]\cdot \mathbf{n}_{\Sigma_0}\,d\mu_{\Sigma_0}+\int_{{{N}_0^{\mathcal{H}}}}(r-M)^{-1+\epsilon}(\underline{L}\phi^{(1)})^2\,du+\int_{{{N}_0^{\mathcal{I}}}}r^{1-\epsilon}({L}\phi^{(1)})^2\,dv<\infty,\\
\label{eq:blowupPenergypsi1}
\int_{\Sigma_0}& \frac{1}{\sqrt{D}}\cdot J^T[\psi^{(1)}]\cdot \mathbf{n}_{\Sigma_0}\,d\mu_{\Sigma_0}=\infty.
\end{align}
\end{proposition}
\begin{proof}
We can use smoothness of $\phi_0$ together with Taylor's theorem and the definition of $H_0[\psi]$ to obtain the following: for any $N\in \N$, we can decompose 
\begin{equation*}
\partial_r\phi_0(r,v_0)=-M^{-2}H_0+\sum_{k=1}^Np_k(r-M)^k+(r-M)^{N+1}f_N(v,r).
\end{equation*}
The equation \eqref{eq:blowupdrpsi1typeA} then follows immediately after plugging the above equation into the right-hand side of \eqref{eq:explicitdrpsi1NH}. Then, \eqref{eq:finiteTenergypsi1} and \eqref{eq:blowupPenergypsi1} follow directly.
\end{proof}

\subsubsection{Singular radiation field  $r\protect\protect\psi^{(1)}|_{\mathcal{I}^{+}}$ for Type \textbf{D} perturbations}
\label{sec:TheSingularTimeInversionInProtectMathcalRD}
We now consider data of Type \textbf{D} and show that the radiation field $r\psi^{(1)}|_{\mathcal{I}^+}$ and Newman--Penrose constant $I_0[\psi^{(1)}]$ are ill-defined in this case.
\begin{proposition}
\label{eq:singextTypeDtimeint}
Let $\psi_0^{(1)}$ be the time integral of a smooth solution $\psi_0$ to \eqref{eq:waveequation} corresponding to initial data of Type \textbf{D} such that moreover
\begin{equation*}
\partial_r\phi_0(u,r)=I_0[\psi]r^{-2}+O(r^{-3})
\end{equation*}
along $N^{\mathcal{I}}_0$.

Then $\psi_0^{(1)}$ can be extended uniquely as smooth function to $\mathcal{R}$. However, $r\psi^{(1)}$ and $I_0[\psi^{(1)}]$ are ill-defined at $\mathcal{I}^+$; more precisely,
\begin{align*}
r\psi_0^{(1)}(u,r)=&\:2I_0[\psi]\log r +O(r^0),\\
r^2\partial_r(r\psi_0^{(1)})(u,r)=&\: 2I_0[\psi] r+O(r^0).
\end{align*}
\end{proposition}
\begin{proof}
By the estimates in the proof of Proposition \ref{prop:explicitexprtimeint2}, it follows that
\begin{equation*}
D\tilde{r}^2\partial_r\underline{\psi}_0^{(1)}(u,r)=2\int_{r}^{\infty} \tilde{r}\partial_r\phi_0(u,r')\,dr',
\end{equation*}
where $\tilde{r}-M=M^2(r-M)^{-1}$ and $\underline{\psi}_0=\tilde{r}^{-1}\phi_0$. Hence
\begin{align*}
\partial_r\underline{\psi}_0^{(1)}(u,r)=&\:2M^{-1}I_0[\psi]r^{-1}+O(r^{-2}),\\
\underline{\psi}_0^{(1)}(u,r)=&\:2M^{-1}I_0[\psi]\log r+ O(r^0)
\end{align*}
so we obtain
\begin{equation*}
\partial_r\phi_0^{(1)}(u,r)=\partial_r(\tilde{r} \psi_0^{(1)})(u,r)=2I_0[\psi]r^{-1}+ O(r^{-2})
\end{equation*}
and therefore,
\begin{equation*}
\phi_0^{(1)}=2I_0[\psi]\log r+ O(r^0).
\end{equation*}
\end{proof}

\subsubsection{Decay estimates for $\psi^{(1)}$}
\label{sec:DecayForPsi1}
We now establish some preliminary decay estimates for the time integral $\psi_0^{(1)}$ of $\psi_0$.

\begin{proposition}
\label{prop:decayesttimeint}
Let $\psi_0^{(1)}$ be the time integral of $\psi_0$ and let $\epsilon>0$ be arbitrarily small. Then there exists a constant $C=C(M,{\Sigma_0},r_{\mathcal{H}},r_{\mathcal{I}},\epsilon)>0$ such that
\begin{equation}
\label{eq:edecaypsi1v1}
\begin{split}
\int_{\Sigma_{\tau}}J^T[\psi_0^{(1)}]\cdot \mathbf{n}_{\tau}\,d\mu_{{\tau}}\leq&\: C(1+\tau)^{-1+\epsilon}\cdot \Bigg[\int_{\Sigma_0} J^T[\psi_0^{(1)}]\cdot \mathbf{n}_{\Sigma_0}\,d\mu_{0}+\int_{{{N}_0^{\mathcal{H}}}}(r-M)^{-1+\epsilon}(\underline{L}\phi_0^{(1)})^2\,d\omega du\\
&+\int_{{{N}_0^{\mathcal{I}}}}r^{1-\epsilon}({L}\phi_0^{(1)})^2\,d\omega dv\Bigg].
\end{split}
\end{equation}
We can further estimate
\begin{align*}
|r\cdot \psi^{(1)}|\leq&\: C\cdot \sqrt{E^{\epsilon}_{0, \mathcal{I}}[\psi]} (1+\tau)^{-\frac{1}{2}+\epsilon}\quad \textnormal{in $\mathcal{A}^{\mathcal{I}}$ if $\lim_{v\to \infty} v^3\partial_v\phi_0|_{N_0^{\mathcal{I}}}<\infty$,} \\
|r\cdot \psi^{(1)}|\leq&\: C\cdot \sqrt{E^{\epsilon}_{0,\mathcal{H}}[\psi]}    (1+\tau)^{-\frac{1}{2}+\epsilon}\quad \textnormal{in $\mathcal{A}^{\mathcal{H}}$ if $H_0[\psi]=0$.}
\end{align*}
\end{proposition}
\begin{proof}
We apply the $r$-weighted estimates from Proposition \ref{prop:rpestell01} with $n=0$ and $p=1-\epsilon$, together with the Morawetz estimates (see Appendix \ref{sec:EnergyBounds}) to conclude that there exists a sequence of times $(\tau_k)$ along which we can estimate:
\begin{equation*}
\begin{split}
\int_{N_{\tau_k}^{\mathcal{H}}}& (r-M)^{\epsilon}(\underline{L}\phi^{(1)}_0)^2\,du+\int_{{\Sigma}_{\tau}\cap \mathcal{B}} J^T[\psi_0^{(1)}]\cdot \mathbf{n}_{\tau}\,d\mu_{\tau}+\int_{N_{\tau_k}^{\mathcal{I}}}r^{-\epsilon}(L\phi^{(1)}_0)^2\,dv\\
\lesssim&\: \tau_k^{-1}\left[\int_{\Sigma_0} J^T[\psi^{(1)}]\cdot \mathbf{n}_{\Sigma_0}\,d\mu_{\Sigma_0}+\int_{{{N}_0^{\mathcal{H}}}}(r-M)^{-1+\epsilon}(\underline{L}\phi_0^{(1)})^2\,du+\int_{{{N}_0^{\mathcal{I}}}}r^{1-\epsilon}({L}\phi_0^{(1)})^2\,dv\right].
\end{split}
\end{equation*}
Hence, we estimate in the region $\{M^2\tau_k^{-1}\leq r-M\leq\tau_k\}$:
\begin{equation*}
\begin{split}
&\int_{{\Sigma_0}_{\tau_k}\cap\{M^2\tau_k^{-1}\leq r-M\leq\tau_k\}} J^T[\psi_0^{(1)}]\cdot \mathbf{n}_{\tau}\,d\mu_{\tau}\\
\lesssim&\:\tau_k^{-1+\epsilon}\left[\int_{\Sigma_0} J^T[\psi_0^{(1)}]\cdot \mathbf{n}_{\Sigma_0}\,d\mu_{\Sigma_0}+\int_{{{N}_0^{\mathcal{H}}}}(r-M)^{-1+\epsilon}(\underline{L}\phi_0^{(1)})^2\,d\omega du+\int_{{{N}_0^{\mathcal{I}}}}r^{1-\epsilon}({L}\phi_0^{(1)})^2\,d\omega dv\right].
\end{split}
\end{equation*}
In the region $\{r-M\leq M^2\tau_k^{-1}\}\cup \{r-M\geq \tau_k\}$ we use again Proposition \ref{prop:rpestell01} with $n=0$ and $p=1-\epsilon$ to estimate
\begin{equation*}
\begin{split}
\int_{N_{\tau_k}^{\mathcal{H}}}& (r-M)^{-1+\epsilon}(\underline{L}\phi^{(1)}_0)^2\,du+\int_{N_{\tau_k}^{\mathcal{I}}} r^{-\epsilon}(L\phi^{(1)}_0)^2\,dv\\
\lesssim&\:\left[\int_{\Sigma_0} J^T[\psi_0^{(1)}]\cdot \mathbf{n}_{\Sigma_0}\,d\mu_{\Sigma_0}+\int_{{{N}_0^{\mathcal{H}}}}(r-M)^{-1+\epsilon}(\underline{L}\phi_0^{(1)})^2\,du+\int_{{{N}_0^{\mathcal{I}}}}r^{1-\epsilon}({L}\phi_0^{(1)})^2\,dv\right]
\end{split}
\end{equation*}
to estimate
\begin{equation*}
\begin{split}
&\int_{{\Sigma}_{\tau_k}\cap\{r-M\leq M^2\tau_k^{-1}\}\cup \{r-M\geq \tau_k\}} J^T[\psi_0^{(1)}]\cdot \mathbf{n}_{\tau}\,d\mu_{\tau}\\
\lesssim&\:\tau_k^{-1+\epsilon}\left[\int_{\Sigma_0} J^T[\psi_0^{(1)}]\cdot \mathbf{n}_{\Sigma_0}\,d\mu_{\Sigma_0}+\int_{{{N}_0^{\mathcal{H}}}}(r-M)^{-1+\epsilon}(\underline{L}\phi_0^{(1)})^2\,du+\int_{{{N}_0^{\mathcal{I}}}}r^{1-\epsilon}({L}\phi_0^{(1)})^2\,dv\right].
\end{split}
\end{equation*}
Hence, by applying \eqref{eq:degenbound} we can conclude that \eqref{eq:edecaypsi1v1} must hold (for all times $\tau\geq 0$). In particular, this implies that
\begin{equation*}
\lim_{v\to \infty}\psi_0^{(1)}(v,r)=0.
\end{equation*}
The remaining (quantitative) estimates in the proposition then follow immediately from integrating the pointwise decay estimates for $\psi_0=T\psi_0^{(1)}$ obtained in Proposition \ref{prop:pointdecay}.
\end{proof}

We can moreover relate the relevant initial pointwise norms of $\psi^{(1)}$ to analogous pointwise norms of $\psi$.
\begin{lemma}
\label{lm:estPnormpsi1}
For all $0\leq \beta\leq 1$ and $k\in \N_0$,  we can estimate
\begin{align*}
P_{H_0,\beta;k}[\psi_0^{(1)}]\lesssim&\: H_0^{(1)}[\psi]+P_{H_0,\beta;k}[\psi]\quad \textnormal{if $H_0[\psi]=0$},\\
P_{I_0,\beta;k}[\psi_0^{(1)}]\lesssim&\: I_0^{(1)}[\psi]+P_{I_0,\beta;k}[\psi]\quad \textnormal{if $\lim_{v\to \infty}r^3L\psi_0(u_0,v)<\infty$},\\
E^{\epsilon}_{0;k+1}[\psi_0^{(1)}]\lesssim&\: E^{\epsilon}_{0, \mathcal{H};k}[\psi]+E^{\epsilon}_{0, \mathcal{I};k}[\psi]\quad \textnormal{if $H_0[\psi]=0$ and $\lim_{v\to \infty}r^3L\psi_0(u_0,v)<\infty$}.
\end{align*}
\end{lemma}
\begin{proof}
The estimates follow from the expressions for $\psi_0^{(1)}$ in Proposition \ref{prop:explicitexprtimeint} and \ref{prop:explicitexprtimeint2}.
\end{proof}

\section{Late-time asymptotics for Type \textbf{A} perturbations}
\label{sec:asympzeroconst}
In this section we will use the time integral construction from the previous section to obtain late-time asymptotics for $\psi$ arising from Type \textbf{A} data. Recall that Type \textbf{A} data includes \emph{generic smooth and  compactly supported } data on $\Sigma_0$.
\subsection{Conditional asymptotics for $r\psi^{(1)}$  in  $\protect\mathcal{A}^{\mathcal{I}}_{\gamma^{\mathcal{I}}}$}
\label{sec:andConditionalAsymptoticsForRPsi}
We will first obtain estimates in the region $\protect\mathcal{A}^{\mathcal{I}}_{\gamma^{\mathcal{I}}}$. These may be thought of as the analogues of the estimates in Proposition \ref{prop:asympLphi} and \ref{prop:asympradfieldnonzeroIH} applied to $\phi_0^{(1)}$ rather than $\phi_0$. However, it is important to note the estimates for $\phi_0^{(1)}$ are not as strong as the estimates for $\phi_0$ arising from Type \textbf{C} data due to the fact that the upper bound decay estimates that we have for $\phi_0^{(1)}$ (Proposition \ref{prop:decayesttimeint}) are \emph{weaker} than the upper bound decay estimates for $\phi_0$ from Proposition \ref{prop:pointdecay}.

\begin{proposition}
\label{prop:improvedasymderphi}
Let $k\in \N_0$ and let  $\alpha>0$ such that $1-\alpha$ is suitably small. Then, there exists an $\eta>0$ and $\epsilon>0$ suitably small and a constant $C=C(M,\Sigma,r_{\mathcal{H}},r_{\mathcal{I}},\alpha,\epsilon,\beta,\eta)>0$ such that in $\mathcal{A}_{\gamma^{\mathcal{I}}_{\alpha}}^{\mathcal{I}}$:
\begin{align*}
|LT^k\phi_0^{(1)}&(u,v)-(-1)^{k}(k+1)! 2I_0^{(1)}[\psi]\cdot v^{-2-k}|\\
\leq&\: C \left[\sqrt{E^{\epsilon}_{0, \mathcal{I}; k}[\psi]}\cdot v^{-2-k-\eta}+\left(P_{I_0,\beta;k}[\psi]+I_0^{(1)}[\psi]\right)\cdot v^{-2-\beta-k}\right]
\end{align*}
and moreover,
\begin{align*}
|LT^k\phi_0&(u,v)-(-1)^{k+1}(k+2)! 2I_0^{(1)}[\psi]\cdot v^{-3-k}|\\
\leq&\: C \left[\sqrt{E^{\epsilon}_{0, \mathcal{I}; k+1}[\psi]}\cdot v^{-3-k-\eta}+\left(P_{I_0,\beta;k+1}[\psi]+I_0^{(1)}[\psi]\right)\cdot v^{-3-\beta-k}\right].
\end{align*}
\end{proposition}
\begin{proof}
We apply Proposition \ref{prop:asympLphi} to $\phi_0^{(1)}$ instead of $\phi_0$, replacing $k$ with $k+1$, where we only consider the region $\mathcal{A}_{\gamma^{\mathcal{I}}_{\alpha}}^{\mathcal{I}}$. We need the pointwise decay estimates in Proposition \ref{prop:pointdecay2} and Proposition \ref{prop:decayesttimeint}, rather than the decay estimates in Proposition \ref{prop:pointdecay}. We also apply Lemma \ref{lm:estPnormpsi1} in order to have only norms involving $\psi$ on the right-hand side of our estimates.
\end{proof}

\begin{proposition}
\label{prop:improvedasymphi}
Let $k\in \N_0$ and let $\alpha>0$ such that $1-\alpha$ is suitably small, then there exists a constant $C=C(M,\Sigma,r_{\mathcal{H}},r_{\mathcal{I}},\alpha,\epsilon, \beta, k)>0$ and $\eta>0$ such that in $\mathcal{A}_{\gamma^{\mathcal{I}}_{\alpha}}^{\mathcal{I}}$:
\begin{align*}
\Bigg|T^k\phi_0(u,v)&-T^k\phi_0(u,v_{\gamma^{\mathcal{I}}_{\alpha}}(u))- (-1)^{k+1}(k+1)! 2I_0^{(1)}[\psi]\left[u^{-2-k}-v^{-2-k}\right]\Bigg|\\
\leq&\: C\left[\sqrt{E^{\epsilon}_{0, \mathcal{I};k+1}[\psi]}+P_{I_0,\beta;k+1}[\psi]+I_0^{(1)}[\psi]\right]\cdot  \frac{v-u}{vu^{2+k+\eta}}.
\end{align*}
\end{proposition}
\begin{proof}
The estimates in the proposition follow from Proposition \ref{prop:improvedasymderphi} in the same way as the estimates in Proposition \ref{prop:asympradfieldnonzeroIH} follow from Proposition \ref{prop:asympLphi}, but we do not estimate $|T^k\phi_0(u,v_{\gamma^{\mathcal{I}}_{\alpha}}(u))|$.
\end{proof}

\subsection{Asymptotics for  $\protect\partial_{\rho}\psi$ away from $\protect\h$ up to $\protect\gamma^{\mathcal{I}}$}
\label{sec:AsymptoticsForPsiAndProtectPartialRhoPsiUpToProtectGammaMathcalIOrProtectGammaMathcalH}
In order to obtain late-time asymptotics from the estimates in Proposition \ref{prop:improvedasymphi}, we first need to determine, independently, the late-time asymptotics of $T^k\phi_0|_{\gamma^{\mathcal{I}}_{\alpha}}$. This involves a derivation of late-time asymptotics for $L\psi_0$ that are valid all the way up to $\gamma^{\mathcal{I}}_{\alpha}$.\footnote{While the estimate \eqref{eq:partasymdrhopsi2} can be used to obtain asymptotics for $L\psi_0$ in spacetime regions of bounded $r$ also in the case of Type \textbf{A} data, it fails to provide asymptotics along the curves $\gamma^{\mathcal{I}}_{\alpha}$.}

\begin{lemma}
\label{prop:improvedasympLpsiTypeA}
Let $k\in \N_0$ and let $\alpha>0$ such that $1-\alpha$ is arbitrarily small. Then, there exists an $\eta>0$ and $\epsilon>0$ suitably small and a constant $C=C(M,\Sigma,r_{\mathcal{H}},r_{\mathcal{I}},\alpha,\epsilon,\eta)>0$, such that in $\mathcal{A}^{\mathcal{I}}\setminus \mathcal{A}_{\gamma^{\mathcal{I}}_{\alpha}}^{\mathcal{I}}$:
\begin{align}
\label{eq:asympLpsigammaTypeA}
\Bigg|2r^2LT^k\psi_0&(u,v)-(-1)^{k+1}(k+1)!4MH_0[\psi]\cdot u^{-2-k}\Bigg|\\ \nonumber
\leq&\: C\left[\sqrt{E^{\epsilon}_{0, \mathcal{I}; k+1}[\psi]}+P_{H_0,1;k}[\psi]+H_0[\psi]\right]\cdot u^{-2-k-\eta}.
\end{align}
\end{lemma}
\begin{proof}
We apply the fundamental theorem of calculus in the $L$-direction, together with the equation $L(r^2LT^k\psi_0)=rLT^{k+1}\phi_0$, to obtain: for all $v_{r_{\mathcal{I}}(u)}\leq v\leq v_{\gamma_{\alpha}}(u)$,
\begin{equation*}
r^2LT^k\psi_0(u,v)=r_{\mathcal{I}}^2LT^k\psi_0(u,v_{r_{\mathcal{I}}}(u))+\int_{v_{r_{\mathcal{I}}}(u)}^{v} rLT^{k+1}\phi_0(u,v')\,dv'.
\end{equation*}
By Cauchy--Schwarz, together with Proposition \ref{prop:hoextraendecay}, we can estimate
\begin{equation*}
\begin{split}
\int_{v_{r_{\mathcal{I}}}(u)}^{v_{\gamma_{\alpha}}(u)} r|LT^{k+1}\phi_0|(u,v')\,dv'\lesssim&\: \sqrt{\int_{v_{r_{\mathcal{I}}}(u)}^{v_{\gamma_{\alpha}}(u)} \,dv'}\cdot\sqrt{\int_{v_{r_{\mathcal{I}}}(u)}^{\infty}r^{2} (L T^{k+1}\phi_0)^2\,dv'}\\
\lesssim &\:  u^{-\frac{5}{2}+\frac{\epsilon}{2}-k+\frac{\alpha}{2}}\cdot \sqrt{E^{\epsilon}_{0, \mathcal{I}; k+1}[\psi]},
\end{split}
\end{equation*}
and we will take $\alpha+\epsilon<1$.

Now, we appeal to \eqref{eq:LpsiestrI} to estimate:
\begin{equation*}
\begin{split}
\Bigg|2r_{\mathcal{I}}^2LT^k\psi_0|_{r=r_{\mathcal{I}}}(u)&-(-1)^{k+1}(k+1)!\cdot 4MH_0[\psi]\cdot u^{-2-k}\Bigg |\\
\lesssim&\: \left(\sqrt{E^{\epsilon}_{0, \mathcal{I};k+1}[\psi]}+H_0[\psi]+ P_{H_0,1;k}[\psi]\right) u^{-2-k-\eta},
\end{split}
\end{equation*}
for some $\eta>0$. By combining the above estimates, we arrive at \eqref{eq:asympLpsigammaTypeA}.
\end{proof}

\begin{proposition}
\label{prop:improvedasymphigamma}
Let $k\in \N_0$ and let $\alpha>0$ such that $1-\alpha$ is arbitrarily small. Then, there exists an $\eta>0$ and $\epsilon>0$ suitably small and a constant $C=C(M,\Sigma,r_{\mathcal{H}},r_{\mathcal{I}},\alpha,\epsilon,\beta,\eta)>0$, such that
\begin{equation}
\label{eq:asymppsigammaIzero}
\begin{split}
\Bigg|T^k\phi_0&(u,v_{\gamma^{\mathcal{I}}_{\alpha}}(u))-(-1)^k(k+1)!4MH_0[\psi]\cdot u^{-2-k}\Bigg|\\ \nonumber
\leq&\: C\Bigg[\sqrt{E^{\epsilon}_{0, \mathcal{I}; k+1}[\psi]}+P_{I_0,\beta;k+1}[\psi]+P_{H_0,1;k}[\psi]+I_0^{(1)}[\psi]+H_0[\psi]\Bigg]\cdot u^{-2-k-\eta}.
\end{split}
\end{equation}
\end{proposition}
\begin{proof}
We split:
\begin{equation*}
\frac{D}{2}rT^k\psi_0(u,v_{\gamma_{\alpha}}(u))=rLT^k\phi_0(u,v_{\gamma_{\alpha}}(u))-r^2LT^k\psi_0(u,v_{\gamma_{\alpha}}(u)).
\end{equation*}
Now, we apply Proposition \ref{prop:improvedasymderphi} together with the estimate $r\lesssim u^{\alpha}$ in $\mathcal{A}_{\gamma^{\mathcal{I}}_{\alpha}}^{\mathcal{I}}$ to estimate for $\epsilon>0$ suitably small:
\begin{equation*}
\begin{split}
r\cdot& |LT^k\phi_0|(u,v_{\gamma_{\alpha}}(u))\leq C\left[\sqrt{E^{\epsilon}_{0, \mathcal{I}; k+1}[\psi]} +P_{I_0,\beta;k+1}[\psi]+I_0^{(1)}[\psi]\right]\cdot u^{-3-k+\alpha}.
\end{split}
\end{equation*}
Now, we apply \eqref{eq:asympLpsigammaTypeA} to arrive at \eqref{eq:asymppsigammaIzero}.
\end{proof}

\subsection{Global asymptotics for $\psi$ in $\protect\mathcal{R}$}
\label{sec:AsymptoticsInProtectMathcalR}
In this section, we derive the asymptotics of $\psi$ in the \emph{full} spacetime region, for Type \textbf{A} initial data.

\begin{proposition}
\label{prop:mainasymptypeA}
Let $k\in \N_0$ and assume that $\lim_{v\to \infty}v^3 L\phi_0(u_0,v)<\infty$.

Then there exists an $\eta>0$ and $\epsilon>0$ suitably small, so that we can estimate:
\begin{equation*}
\begin{split}
\Bigg|T^k\psi_0(u,v)&-4\left[ I_0^{(1)}[\psi]T^{k+1}\left(\frac{1}{u\cdot v}\right)+\frac{M}{r \sqrt{D}}H_0[\psi]T^{k}\left(\frac{1}{u(v+4M-2r)}\right)\right]\Bigg|\\
\leq&\: C\Bigg[\sqrt{E^{\epsilon}_{0, \mathcal{I}; k+1}[\psi]}+P_{I_0,\beta;k+1}[\psi]+P_{H_0,1;k}[\psi]+H_0[\psi]+I_0^{(1)}[\psi] \Bigg]\\
&\cdot \left( v^{-1}u^{-2-k-\eta}+D^{-\frac{1}{2}}u^{-1}v^{-1-k-\eta}\right).
\end{split}
\end{equation*}
\end{proposition}
\begin{proof}
By combining the estimates in Proposition \ref{prop:improvedasymphigamma} and Proposition \ref{prop:improvedasymphi}, we arrive at the following estimates for $r\cdot \psi_0(u,v)$: let $\alpha>0$ be sufficiently close to 1, then there exists an $\eta>0$ such that in $\mathcal{A}_{\gamma^{\mathcal{I}}_{\alpha}}^{\mathcal{I}}$:
\begin{align}
\label{eq:asympIzeronearinf}
\Bigg|T^k\phi_0(u,v)&- (-1)^{k+1}(k+1)!\left[ 2I_0^{(1)}[\psi]\left(u^{-2-k}-v^{-2-k}\right)-4MH_0[\psi]u^{-2-k}\right]\Bigg|\\ \nonumber
\leq&\: C\left[\sqrt{E^{\epsilon}_{0, \mathcal{I}; k+1}[\psi]}+P_{I_0,\beta;k+1}[\psi]+P_{H_0,1;k}[\psi]+H_0[\psi]+I_0^{(1)}[\psi] \right]\cdot\frac{v-u}{vu^{2+k+\eta}}.
\end{align}

By applying moreover Lemma \ref{lm:relationruv} together with \eqref{eq:identityuminusv} and \eqref{eq:identityvminusu}, we can rewrite \eqref{eq:asympIzeronearinf} as follows:
\begin{align*}
\Bigg|T^k\psi_0(u,v)&-\left[ 4I_0^{(1)}[\psi]T^{k+1}\left(\frac{1}{u\cdot v}\right)+4Mr^{-1}H_0[\psi]T^{k}(u^{-2})\right]\Bigg|\\
\leq&\: C\left[\sqrt{E^{\epsilon}_{0, \mathcal{I}; k+1}[\psi]}+P_{I_0,\beta;k+1}[\psi]+P_{H_0,1;k}[\psi]+H_0[\psi]+I_0^{(1)}[\psi] \right]\cdot  \frac{1}{vu^{2+k+\eta}}.
\end{align*}

To obtain a global estimate for $\psi_0$, we first combine the above estimates with \eqref{eq:asympsihornonzeroH0} in the region where $r\leq r_{\gamma^{\mathcal{H}}_{\alpha}}(v)$, \eqref{eq:asymppsirH} in the region where $r_{\gamma^{\mathcal{H}}_{\alpha}}(v)\leq r\leq r_{\mathcal{I}}$. 

To obtain late-time asymptotics in the remaining region $r_{\mathcal{I}}\leq r\leq r_{\gamma^{\mathcal{I}}_{\alpha}}(u)$, we use \eqref{eq:asymppsigammaIzero} and we integrate the estimate \eqref{eq:asympLpsigammaTypeA} from $r=r_{\gamma^{\mathcal{I}}_{\alpha}}(u)$ to any $r\geq r_{\mathcal{I}}$. 

We then obtain:
\begin{equation*}
\begin{split}
\Bigg|T^k\psi_0(u,v)&-4\left[ I_0^{(1)}[\psi]T^{k+1}\left(\frac{1}{u\cdot v}\right)+\frac{M}{r \sqrt{D}}H_0[\psi]T^{k}\left(\frac{1}{u(v+4M-2r)}\right)\right]\Bigg|\\
\leq&\: C\left[\sqrt{E^{\epsilon}_{0, \mathcal{I}; k+1}[\psi]}+P_{I_0,\beta;k+1}[\psi]+P_{H_0,1;k}[\psi]+H_0[\psi]+I_0^{(1)}[\psi] \right]\cdot \left( v^{-1}u^{-2-k-\eta}+D^{-\frac{1}{2}}u^{-1}v^{-1-k-\eta}\right),
\end{split}
\end{equation*}
everywhere in $\mathcal{R}$. Note that $v+4M-2r$ has the property that it approaches $u$ as we increase $v$ and keep $u$ constant, but it remains finite as we approach $r=M$; indeed, we have that everywhere in $\{r\geq 2M\}$, $v-2r+2M\geq u$ and in $\{r\leq 2M\}$, $v\leq v-2r+4M\leq v+2M$.
\end{proof}

\section{Asymptotics for Type \textbf{B} and \textbf{D} perturbations}
\label{sec:AsymptoticsForTypeDPerturbations}
In this section, we treat the remaining types of initial data: Type \textbf{B} and \textbf{D}. The late-time asymptotics for Type \textbf{B} data follow immediately from Proposition \ref{prop:asympsi} applied to $\psi_0^{(1)}$, where we use the regularity properties of $\psi_0^{(1)}$ that follow from Proposition \ref{eq:smoothextTypeBtimeint}.
\begin{corollary}
\label{cor:asympsizeroIH}
Let $k\in \N_0$. If $\lim_{v\to \infty} r^3L\phi_0(u_0,v)<\infty$ and $H_0[\psi]=0$, then there exists an $\eta>0$ and $\epsilon>0$ suitably small, such that we obtain the following \emph{global} estimate:
\begin{equation}
\label{eq:asympsizeroIH}
\begin{split}
\Bigg|T^k\psi_0(u,v)&-4\left(I_0^{(1)}[\psi]+ \frac{M}{r\sqrt{D}}H_0^{(1)}[\psi]\right)T^{k+1}\left(\frac{1}{v\cdot u}\right)\Bigg|\\
\leq&\: C\left(\sqrt{E_{0, \mathcal{H};k+1}^{\epsilon}[\psi]+E_{0, \mathcal{I};k+1}^{\epsilon}[\psi]}+I_0^{(1)}[\psi]+P_{I_0,\beta;k+1}[\psi]\right)v^{-1}u^{-2-k-\eta}\\
&+C\left(\sqrt{E_{0, \mathcal{H};k+1}^{\epsilon}[\psi]+E_{0, \mathcal{I};k+1}^{\epsilon}[\psi]}+H_0^{(1)}[\psi]+P_{H_0,1;k+1}[\psi]\right)D^{-\frac{1}{2}}u^{-1}v^{-2-k-\eta},
\end{split}
\end{equation}
where $C=C(M,{\Sigma_0},r_{\mathcal{H}},r_{\mathcal{I}},\eta,\epsilon,\beta,k)>0$ is a constant.
\end{corollary}
\begin{proof}
We apply Proposition \ref{prop:asympsi} with $k$ replaced by $k+1$ and $\psi_0$ replaced by ${\psi}_0^{(1)}$. We also use Lemma \ref{lm:estPnormpsi1}.
\end{proof}

We are left with Type \textbf{D} data. We obtain asymptotics by following arguments analogous to those for Type \textbf{A} data in Section \ref{sec:asympzeroconst}, so we will omit most of the proofs, unless a different argument is needed, compared to the Type \textbf{A} data case.

\begin{proposition}
\label{prop:improvedasymderphiD}
Let $k\in \N_0$ and assume that $H_0[\psi]=0$. Let  $\alpha>0$ such that $1-\alpha$ is suitably small. Then, there exists an $\eta>0$ and $\epsilon>0$ suitably small and a constant $C=C(M,\Sigma,r_{\mathcal{H}},r_{\mathcal{I}},\alpha,\epsilon,\eta,k)>0$ such that in $\mathcal{A}_{\gamma^{\mathcal{H}}_{\alpha}}^{\mathcal{H}}$:
\begin{align*}
|\underline{L}T^k\phi_0^{(1)}&(u,v)-(-1)^{k}(k+1)! 2H_0^{(1)}[\psi]\cdot u^{-2-k}|\\
\leq&\: C \left[\sqrt{E^{\epsilon}_{0, \mathcal{H}; k}[\psi]}\cdot u^{-2-k-\eta}+\left(P_{H_0,1;k}[\psi]+H_0^{(1)}[\psi]\right)\cdot u^{-2-\beta-k}\right]
\end{align*}
and moreover,
\begin{align*}
|\underline{L}T^k\phi_0&(u,v)-(-1)^{k+1}(k+2)! 2H_0^{(1)}[\psi]\cdot u^{-3-k}|\\
\leq&\: C \left[\sqrt{E^{\epsilon}_{0, \mathcal{H}; k+1}[\psi]}\cdot u^{-3-k-\eta}+\left(P_{H_0,1;k+1}[\psi]+H_0^{(1)}[\psi]\right)\cdot u^{-3-\beta-k}\right].
\end{align*}
\end{proposition}
\begin{proof}
We repeat the steps in the proof of Proposition \ref{prop:improvedasymderphi} to the region $\mathcal{A}_{\gamma^{\mathcal{H}}_{\alpha}}^{\mathcal{H}}$ instead of $\mathcal{A}_{\gamma^{\mathcal{I}}_{\alpha}}^{\mathcal{I}}$ and interchange the roles of $u$ and $v$.
\end{proof}

\begin{proposition}
\label{prop:improvedasymphiD}
Let $k\in \N_0$ and assume that $H_0[\psi]=0$. Let  $\alpha>0$ such that $1-\alpha$ is suitably small. Then, there exists an $\eta>0$ and $\epsilon>0$ suitably small and a constant $C=C(M,\Sigma,r_{\mathcal{H}},r_{\mathcal{I}},\alpha,\epsilon,\eta)>0$ such that in $\mathcal{A}_{\gamma^{\mathcal{H}}_{\alpha}}^{\mathcal{H}}$:
\begin{align*}
\Bigg|T^k\phi_0(u,v)&-T^k\phi_0(u_{\gamma^{\mathcal{H}}_{\alpha}}(v),v)- (-1)^{k+1}(k+1)! 2H_0^{(1)}[\psi]\left[v^{-2-k}-u^{-2-k}\right]\Bigg|\\
\leq&\: C \left[\sqrt{E^{\epsilon}_{0, \mathcal{H}; k+1 }[\psi]}+P_{H_0,1;k+1}[\psi]+H_0^{(1)}[\psi]\right] \cdot \frac{u-v}{u v^{2+k+\eta}}.
\end{align*}
\end{proposition}
\begin{proof}
We repeat the steps in the proof of Proposition \ref{prop:improvedasymphi} to the region $\mathcal{A}_{\gamma^{\mathcal{H}}_{\alpha}}^{\mathcal{H}}$ instead of $\mathcal{A}_{\gamma^{\mathcal{I}}_{\alpha}}^{\mathcal{I}}$ and interchange the roles of $u$ and $v$.
\end{proof}

In contrast with Lemma \ref{prop:improvedasympLpsiTypeA}, we cannot yet obtain asymptotics for $\partial_r\psi_0$ in the region $\mathcal{A}^{\mathcal{H}}\setminus \mathcal{A}_{\gamma^{\mathcal{H}}_{\alpha}}^{\mathcal{H}}$ for Type \textbf{D} data. Instead, we consider $\partial_r((r-M)\cdot \psi_0)$, which, as we will show, is sufficient for our purposes. See however Proposition \ref{prop:asympdrpsiTypeD} at the end of the section, where we do obtain asymptotics for $\partial_r\psi_0$.

\begin{lemma}
\label{prop:improvedasympLpsiTypeD}
Let $k\in \N_0$ and assume that $H_0[\psi]=0$. Let  $\alpha>0$ such that $1-\alpha$ is arbitrarily small. Then, there exists an $\eta>0$ and $\epsilon>0$ suitably small and a constant $C=C(M,\Sigma,r_{\mathcal{H}},r_{\mathcal{I}},\alpha,\epsilon,\eta,\beta,k)>0$, such that in $\mathcal{A}^{\mathcal{H}}\setminus \mathcal{A}_{\gamma^{\mathcal{H}}_{\alpha}}^{\mathcal{H}}$:

\begin{equation}
\label{eq:asympLpsigammaTypeD}
\begin{split}
|M&\partial_r((r-M)\cdot T^k\psi_0)(u,v)-(-1)^{k}(k+1)!4MI_0v^{-2-k}|\\
\leq&\:  C\left[\sqrt{E^{\epsilon}_{0,\mathcal{H};k+1}[\psi]}+P_{H_0,1;k+1}[\psi]+P_{I_0,\beta;k}[\psi]+I_0[\psi]\right]v^{-2-k-\eta}.
\end{split}
\end{equation}
\end{lemma}
\begin{proof}
Note that we can rewrite \eqref{eq:waveeqvr} as follows in $(v,r)$ coordinates:
\begin{equation*}
\partial_r^2((r-M)\cdot T^k\psi_0)=-2(r-M)^{-1}\partial_rT^{k+1}\phi_0.
\end{equation*}
Using the above equation, together with the fundamental theorem of calculus in the $\underline{L}$ direction, we arrive at the following estimate:
\begin{equation*}
\begin{split}
|&\partial_r((r-M)\cdot T^k\psi_0)(v,r_{\gamma^{\mathcal{H}}_{\alpha}}(v))-\partial_r((r-M)\cdot T^k\psi_0)(v,r_{\mathcal{H}})|\\
\lesssim &\: \int^{u_{\gamma^{\mathcal{H}}_{\alpha}(v)}}_{u(r_{\mathcal{H}})} (r-M)^{-1} |\underline{L} T^{k+1}\phi_0|(u',v)\,du'\\
\lesssim &\: \sqrt{ \int^{u_{\gamma^{\mathcal{H}}_{\alpha}(v)}}_{u(r_{\mathcal{H}})}\,du'}\cdot \sqrt{\int^{u_{\gamma^{\mathcal{H}}_{\alpha}(v)}}_{u(r_{\mathcal{H}})}  (r-M)^{-2}(\underline{L} T^{k+1}\phi_0)^2(u',v)\,du'}\\
\lesssim &\: \sqrt{E^{\epsilon}_{0, \mathcal{I}; k+1}[\psi]} \cdot v^{-\frac{5}{2}-k+\frac{\epsilon}{2}+\frac{\alpha}{2}},
\end{split}
\end{equation*}
where we applied Proposition \ref{prop:extraendecay} together with the estimate $(r-M)^{-1}\lesssim v^{\alpha}$ to obtain the last inequality.

We moreover have that
\begin{equation*}
\partial_r((r-M)\cdot T^k\psi_0)(v,r_{\mathcal{H}})=T^k\psi_0(v,r_{\mathcal{H}})+ (r_{\mathcal{H}}-M)\partial_r T^k\psi_0(v,r_{\mathcal{H}}).
\end{equation*}
By \eqref{eq:asymlbarpsidr} it follows that there exists an $\eta>0$ such that
\begin{equation*}
(r_{\mathcal{H}}-M)|\partial_rT^k\psi_0|(v,r_{\mathcal{H}})\lesssim \left(\sqrt{E^{\epsilon}_{0;k+1}[\psi]} +P_{H_0,1;k}[\psi]\right)v^{-2-k-\eta}.
\end{equation*}
Therefore, we can use \eqref{eq:asymppsirH} at $r=r_{\mathcal{H}}$ to estimate
\begin{equation*}
\begin{split}
|\partial_r&((r-M)\cdot T^k\psi_0)(v,r_{\mathcal{H}})-(-1)^{-k}(k+1)!4I_0v^{-k-2}|\\
\lesssim&\:  \left(\sqrt{E^{\epsilon}_{0;k+1}[\psi] }+P_{H_0,1;k}[\psi]+P_{I_0,\beta;k}+I_0[\psi]\right)v^{-2-k-\eta}.
\end{split}
\end{equation*}
By combining the estimates above, we arrive at \eqref{eq:asympLpsigammaTypeD}.
\end{proof}
\begin{proposition}
\label{prop:improvedasymphigammaB}
Let $k\in \N_0$ and assume $H_0[\psi]=0$. Let $\alpha>0$ such that $1-\alpha$ is arbitrarily small. Then, there exists an $\eta>0$ and $\epsilon>0$ suitably small and a constant $C=C(M,\Sigma,r_{\mathcal{H}},r_{\mathcal{I}},\alpha,\epsilon,\eta,\beta,k)>0$, such that
\begin{align}
\label{eq:asymppsigammaHzero}
\Bigg|T^k\phi_0&(u_{\gamma^{\mathcal{H}}_{\alpha}}(v),v)-(-1)^k(k+1)!4MI_0[\psi]\cdot v^{-2-k}\Bigg|\\ \nonumber
\lesssim&\: \Bigg[\sqrt{E^{\epsilon}_{0, \mathcal{H};k+1}[\psi]}+P_{H_0,1;k+1}[\psi]+P_{I_0,\beta;k}[\psi]+H_0^{(1)}[\psi]+I_0[\psi]\Bigg]\cdot v^{-2-k-\eta}.
\end{align}
\end{proposition}
\begin{proof}
We can split in $(v,r)$ coordinates
\begin{equation*}
\begin{split}
M^2r^{-1}T^k\psi_0(v,r_{\gamma^{\mathcal{H}}_{\alpha}}(v))=&\:M\partial_r(r^{-1}(r-M))\cdot rT^k\psi_0(v,r_{\gamma^{\mathcal{H}}_{\alpha}}(v))\\
=&\: M\partial_r((r-M) \cdot T^k\psi_0)(v,r_{\gamma^{\mathcal{H}}_{\alpha}}(v))-Mr^{-1}(r-M)\partial_rT^k\phi_0(v,r_{\gamma^{\mathcal{H}}_{\alpha}}(v))\\
=&\:   M\partial_r((r-M) \cdot T^k\psi_0)(v,r_{\gamma^{\mathcal{H}}_{\alpha}}(v))+2Mr(r-M)^{-1}\underline{L}T^k\phi_0(v,r_{\gamma^{\mathcal{H}}_{\alpha}}(v)).
\end{split}
\end{equation*}

We apply Proposition \ref{prop:improvedasymderphiD} together with the estimate $(r-M)^{-1}\lesssim v^{\alpha}$ in $\mathcal{B}^\mathcal{H}_{\alpha}$ to estimate
\begin{equation*}
\begin{split}
2Mr(r-M)^{-1}|\underline{L}T^k\phi_0|(v,r_{\gamma^{\mathcal{H}}_{\alpha}(v)})\lesssim&\: \Bigg[\sqrt{E^{\epsilon}_{0, \mathcal{H}; k+1}}+P_{H_0,1;k+1}[\psi^{(1)}]+H_0^{(1)}[\psi]\Bigg]\cdot v^{-3-k+\alpha}.
\end{split}
\end{equation*}
The estimate \eqref{eq:asymppsigammaHzero} then follows by applying \eqref{eq:asympLpsigammaTypeA}.
\end{proof}

\begin{proposition}
\label{prop:mainasymptypeD}
Let $k\in \N_0$ and assume that $H_0[\psi]=0$. Let  $\alpha>0$ such that $1-\alpha$ is arbitrarily small. Then, there exists an $\eta>0$ and $\epsilon>0$ suitably small and a constant $C=C(M,\Sigma,r_{\mathcal{H}},r_{\mathcal{I}},\alpha,\epsilon,\eta, \beta, k)>0$, such that
\begin{equation}
\label{eq:mainasymptypeD}
\begin{split}
\Bigg|T^k\psi_0(u,v)&-4\left[ \frac{1}{\sqrt{D}}H_0^{(1)}[\psi]T^{k+1}\left(\frac{1}{u\cdot v}\right)+I_0[\psi]T^{k}\left(\frac{1}{v(u+2M-2M^2(r-M)^{-1})}\right)\right]\Bigg|\\
\leq&\: C\Bigg[\sqrt{E^{\epsilon}_{0, \mathcal{H}; k+1}[\psi]}+P_{I_0,\beta;k}[\psi]+P_{H_0,1;k+1}[\psi]+I_0[\psi]+H_0^{(1)}[\psi] \Bigg]\\
&\cdot \left( v^{-1}u^{-1-k-\eta}+D^{-\frac{1}{2}}u^{-1}v^{-2-k-\eta}\right).
\end{split}
\end{equation}
\end{proposition}
\begin{proof}
We apply the previous propositions in this section, together with the asymptotics in derived in Section \ref{sec:latetimeasympsi} to arrive at \eqref{eq:mainasymptypeD}, analogously to what is done in the proof of Proposition \ref{prop:mainasymptypeA} (with the roles of $u$ and $v$ reversed). 

We moreover used that in $\{r\leq 2M\}$ we can estimate $u+2M-2M^2(r-M)^{-1}\geq v$ and in $\{r\geq 2M\}$, $u+2M-2M^2(r-M)^{-1}\geq u$.
\end{proof}

For completeness, we will also derive the precise late-time asymptotics for $\partial_r\psi$ for Type \textbf{B} data, \emph{and show that the leading order term decays one power faster compared to the Type \textbf{A} and \textbf{C} cases}. We will restrict here to a bounded region $\{r\leq r_{\mathcal{I}}\}$ for the sake of convenience, but we note that the estimates providing late-time asymptotics can in principle be extended to the full region $\mathcal{R}$.
\begin{proposition}
\label{prop:asympdrpsiTypeD}
Let $k\in \N_0$ and assume that $H_0[\psi]=0$. Let  $\alpha>0$ such that $1-\alpha$ is arbitrarily small. Then, there exists an $\eta>0$ and $\epsilon>0$ suitably small and a constant $C=C(M,\Sigma,r_{\mathcal{H}},r_{\mathcal{I}},\alpha,\epsilon,\eta,\beta,k)>0$, such that
\begin{equation}
\label{eq:asympdrpsihzeroB}
\begin{split}
\Bigg|Dr^2\partial_rT^k\psi_0(v,r)&-8MH_0^{(1)}[\psi]T^k(u^{-3})-8I_0[\psi](r^2-M^2)T^k(v^{-3})\Bigg|\\
\leq&\: C\Bigg[\sqrt{E^{\epsilon}_{0, \mathcal{H}; k+2}[\psi]}+P_{I_0,\beta;k+1}[\psi]+P_{H_0,1;k+2}[\psi]+I_0[\psi]+H_0^{(1)}[\psi] \Bigg]\cdot v^{-3-\eta-k},
\end{split}
\end{equation}
in $(v,r)$ coordinates, for all $r\leq r_{\mathcal{I}}$.
\end{proposition}
\begin{proof}
We apply \eqref{eq:waveeqvr} to obtain in $(v,r)$ coordinates
\begin{equation}
\label{eq:maineqestdrpsitypeD}
Dr^2\partial_rT^k\psi_0(v,r)=2M^2T^{k+1}\psi_0(v,M)-2r^2 T^{k+1}\psi_0(v,r)+\int_{M}^r 2r T^{k+1}\psi_0(v,r')\,dr'.
\end{equation}
We will first estimate $2M^2T^{k+1}\psi_0(v,M)-2r^2 T^{k+1}\psi_0(v,r)$.  If $r\leq r_{\gamma^{\mathcal{H}}_{\alpha}}(v)$, we apply Proposition \ref{prop:improvedasymphiD} together with Proposition \ref{prop:improvedasymphigammaB}, with $k$ replaced by $k+1$. If $r\geq r_{\gamma^{\mathcal{H}}_{\alpha}}(v)$, we apply Proposition \ref{prop:improvedasymphigammaB} and we integrate the estimate in Lemma \ref{prop:improvedasympLpsiTypeD}. We then arrive at the following expressions:
\begin{align*}
\Big|&M T^{k+1}\psi_0(v,M)+ 8MI_0[\psi] T^k(v^{-3})- 4H_0^{(1)}[\psi]T^k(v^{-3})\Big| \\
\lesssim&\: \Bigg[\sqrt{E^{\epsilon}_{0, \mathcal{H}; k+2}[\psi]}+P_{I_0,\beta;k+1}[\psi]+P_{H_0,1;k+2}[\psi]+I_0[\psi]+H_0^{(1)}[\psi] \Bigg] v^{-3-k-\eta},\\
\Big|&r^2 T^{k+1}\psi_0(v,r)+ 8r^2I_0[\psi] T^k(v^{-3})- 4H_0^{(1)}[\psi]T^k(v^{-3}-u^{-3})\Big| \\
\lesssim&\: \Bigg[\sqrt{E^{\epsilon}_{0, \mathcal{H}; k+2}[\psi]}+P_{I_0,\beta;k+1}[\psi]+P_{H_0,1;k+2}[\psi]+I_0[\psi]+H_0^{(1)}[\psi] \Bigg] v^{-3-k-\eta}.
\end{align*}

Hence, we obtain
\begin{equation*}
\begin{split}
\Bigg| 2M^2&T^{k+1}\psi_0(v,M)-2r^2T^{k+1}\psi_0(v,r)-8MH_0^{(1)}[\psi]T^k(u^{-3})-16(r^2-M^2)I_0[\psi]T^k(v^{-3})\Bigg|\\
\leq &\: C \Bigg[\sqrt{E^{\epsilon}_{0, \mathcal{H}; k+2}[\psi]}+P_{I_0,\beta;k+1}[\psi]+P_{H_0,1;k+2}[\psi]+I_0[\psi]+H_0^{(1)}[\psi] \Bigg] \cdot v^{-3-k-\eta}.
\end{split}
\end{equation*}

In order to estimate the integral on the right-hand side of \eqref{eq:maineqestdrpsitypeD}, we apply  Proposition \ref{prop:improvedasymphiD} together with Proposition \ref{prop:improvedasymphigammaB} and \eqref{eq:asymppsirH}:
\begin{equation*}
\Bigg| \int_{M}^r 2r T^{k+1}\psi_0(v,r')+2r'\cdot  8I_0[\psi] T^k(v^{-3})\,dr'\Bigg|\leq C \left[\int_M^{M+v^{-\alpha}}\textnormal{Err}_1\,dr' +\int^r_{M+v^{-\alpha}}\textnormal{Err}_2\,dr'\right],
\end{equation*}
where we take $r>M+v^{-\alpha}$ without loss of generality, and where
\begin{align*}
\textnormal{Err}_1:=&\: \left[\sqrt{E^{\epsilon}_{0, \mathcal{H}; k+2}[\psi]}+P_{I_0,\beta;k+1}[\psi]+P_{H_0,1;k+2}[\psi]+I_0[\psi]+H_0^{(1)}[\psi] \right] v^{-3-k-\eta}, \\
\textnormal{Err}_2:=&\: \left[\sqrt{E^{\epsilon}_{0;k+2}[\psi]}+I_0[\psi]+P_{H_0,1;k+1}[\psi]+P_{I_0,\beta;k+1}[\psi]\right] (r-M)^{-1}v^{-3-k-2\eta}.
\end{align*}

It follows immediately that (note that the logarithmic term from integrating Err$_2$ can be easily absorbed by the $v$ power)
\begin{equation*}
\begin{split}
\int_M^{M+v^{-\alpha}}&\textnormal{Err}_1\,dr' +\int^r_{M+v^{-\alpha}}\textnormal{Err}_2\,dr'\\
\leq &\: C  \Bigg[\sqrt{E^{\epsilon}_{0, \mathcal{H}; k+2}[\psi]}+P_{I_0,\beta;k+1}[\psi]+P_{H_0,1;k+2}[\psi]+I_0[\psi]+H_0^{(1)}[\psi] \Bigg]\cdot v^{-3-k-\eta}.
\end{split}
\end{equation*}
Finally, we have that
\begin{equation*}
 8I_0[\psi]  T^k(v^{-3}) \int_{M}^r  2r'\,dr'= (r^2-M^2)8I_0[\psi]  T^k(v^{-3}).
\end{equation*}
When we combine the estimates above, we obtain \eqref{eq:maineqestdrpsitypeD}.
\end{proof}

\section{Higher-order asymptotics and logarithmic corrections}
\label{sec:hoasymp}
In this section, we derive refined asymptotics along $\mathcal{H}^+$ for data with $H_0[\psi]\neq 0$ and along $\mathcal{I}^+$ for data with $I_0[\psi]\neq 0$. The derivation proceeds in a very similar manner to the arguments in \cite{logasymptotics}.

We first introduce the following additional weighted $L^{\infty}$ norms: we define with respect to $(u,r)$ coordinates,
\begin{align*}
P_{\mathcal{I}}[\psi]:=&\:\left| \left| Dr^3\left(\partial_r\phi_0-\frac{I_0[\psi]}{r^2}\right)\right|\right|_{L^{\infty}(\Sigma_0)},\\
P_{\mathcal{I},T}[\psi]:=&\:\left| \left| Dr^4\partial_r\left(D\partial_r\phi_0-D\frac{I_0[\psi]}{r^2}\right)\right|\right|_{L^{\infty}(\Sigma_0)}.
\end{align*}
And we define with respect to $(v,r)$ coordinates:
\begin{align*}
P_{\mathcal{H}}[\psi]:=&\:\left| \left| D^{-\frac{1}{2}}\left(\partial_r\phi_0+M^2H_0[\psi]\right)\right|\right|_{L^{\infty}(\Sigma_0)},\\
P_{\mathcal{H},T}[\psi]:=&\:\left| \left| \partial_r^2\phi_0\right|\right|_{L^{\infty}(\Sigma_0)}.
\end{align*}

\begin{proposition}
\label{prop:asymdvphi}
For all $\epsilon>0$, there exists a constant $C=C(M,\Sigma,r_{\mathcal{I}},\epsilon)>0$ such that for all $(u,v)$ in $\mathcal{A}^{\mathcal{I}}$ we can estimate:
\begin{itemize}
\item[ \rm(i)]
\begin{equation}
\label{eq:2ndoasympdvphiinf}
\begin{split}
&\Bigg|\partial_v(r\psi)(u,v)-2I_0[\psi]v^{-2}-16M I_0[\psi]v^{-3}\log v+8MI_0[\psi]uv^{-3}(v-u)^{-1}\\
&+8MI_0[\psi]v^{-3} \log \left(\frac{vu}{v-u}\right)\Bigg|\\
\leq&\: C(I_0[\psi]+H_0[\psi]+\sqrt{E_{0; 1}^{\epsilon}[\psi]}+P_{\mathcal{I}}[\psi]+P_{\mathcal{H}}[\psi])\textnormal{Err}_{\mathcal{I}}(u,v),
\end{split}
\end{equation}
where
\begin{equation*}
\textnormal{Err}_{\mathcal{I}}(u,v):=v^{-3}+v^{-2-\epsilon}\cdot (v-u)^{-1}+v^{-2}\cdot (v-u)^{-2+\eta},
\end{equation*}
with $\eta>0$ arbitrarily small.
\item[ \rm(ii)]
For all $\epsilon>0$, there exists a constant $C=C(M,\Sigma,r_{\mathcal{H}},\epsilon)>0$ such that for all $(u,v)$ in $\mathcal{A}^{\mathcal{H}}$ we can estimate:
\begin{equation}
\label{eq:2ndoasympdvphihor}
\begin{split}
&\Bigg|\partial_u(r\psi)(u,v)-2H_0[\psi]u^{-2}-16M H_0[\psi]u^{-3}\log u+8MH_0[\psi]vu^{-3}(u-v)^{-1}\\
&+8MH_0[\psi]u^{-3} \log \left(\frac{vu}{u-v}\right)\Bigg|\\
\leq&\: C(I_0[\psi]+H_0[\psi]+\sqrt{E^{\epsilon}_{0 ;1}[\psi]}+P_{\mathcal{H}}[\psi]+P_{\mathcal{I}}[\psi])\textnormal{Err}_{\mathcal{H}}(u,v),
\end{split}
\end{equation}
where
\begin{equation*}
\textnormal{Err}_{\mathcal{H}}(u,v):=u^{-3}+u^{-2-\epsilon}\cdot (u-v)^{-1}+u^{-2}\cdot (u-v)^{-2+\eta},
\end{equation*}
with $\eta>0$ arbitrarily small.
\end{itemize}
\end{proposition}
\begin{proof}
By applying the relations between $r$, $u$ and $v$ from Lemma \ref{lm:relationruv}, we obtain in $\mathcal{A}^{\mathcal{I}}$:
\begin{equation}
\label{eq:maineq2ndasympinf}
\partial_u\partial_v(r\psi)(u,v)=\Big[-2M(v-u)^{-2}+O((v-u)^{-3+\eta})\Big]\cdot \psi\\
\end{equation}
with $\eta>0$ arbitrarily small, and hence, by Proposition \ref{prop:asympsi} we have that there exists an $\epsilon>0$ such that
\begin{equation*}
\begin{split}
\left|\partial_u\partial_v(r\psi)(u,v)+\frac{8MI_0[\psi]}{vu}(v-u)^{-2}\right|\leq &\:C(I_0[\psi[+H_0[\psi])(v-u)^{-3+\eta}v^{-1}u^{-1}\\
&+ C(I_0[\psi]+H_0[\psi]+\sqrt{E^{\epsilon}_{0; 1}[\psi]}+P_{\mathcal{H}}[\psi]+P_{\mathcal{I}}[\psi])(v-u)^{-2}v^{-1}u^{-1-\epsilon}.
\end{split}
\end{equation*}
The estimate \eqref{eq:2ndoasympdvphiinf} now follows by repeating the arguments in the proof of Proposition 3.1 of \cite{logasymptotics}

Now, we apply the relations between $r-M$, $u$ and $v$ from Lemma \ref{lm:relationruv} in $\mathcal{A}^{\mathcal{H}}$ to obtain:
\begin{equation}
\label{eq:maineq2ndasymphor}
\partial_v\partial_u(r\psi)(u,v)=\Big[-2M(v-u)^{-2}+O((v-u)^{-3+\eta})\Big]\cdot \sqrt{D} \psi\\
\end{equation}
By using \eqref{eq:maineq2ndasymphor} together with the estimate for $\sqrt{D}\cdot \psi$ from Proposition \ref{prop:asympsi}, we can similarly find an $\epsilon>0$ such that for all $(u,v)$ in $\mathcal{A}^{\mathcal{H}}$
\begin{equation*}
\begin{split}
\left|\partial_v\partial_u(r\psi)(u,v)+\frac{8MH_0}{uv}(u-v)^{-2}\right|\leq &\:C(I_0+H_0)(u-v)^{-3+\eta}u^{-1}v^{-1}\\
&+ C(I_0+H_0+\sqrt{E^{\epsilon}_{0; 1}[\psi]}+P_{\mathcal{H}}[\psi]+P_{\mathcal{I}}[\psi])(u-v)^{-2}u^{-1}v^{-1-\epsilon}.
\end{split}
\end{equation*}
We obtain \eqref{eq:2ndoasympdvphihor} by once again repeating the arguments in the proof of Proposition 3.1 of \cite{logasymptotics} and moreover interchanging the roles of $u$ and $v$ (and $I_0$ and $H_0$).
\end{proof}

\begin{proposition}
\label{prop:2ndasymphi}
For all $\epsilon>0$, there exists a constant $C=C(M,\Sigma,r_{\mathcal{I}},r_{\mathcal{H}},\epsilon)>0$ such that we can estimate:
\begin{equation}
\label{eq:2ndasymphiinf}
\begin{split}
\Bigg|r\psi(u,v)&-2I_0[\psi([u^{-1}-v^{-1})+4MI_0[\psi]u^{-2}\log u-4MI_0[\psi]v^{-2}\log u\\
&+8MI_0[\psi]v^{-2}\log v+4MI_0[\psi](u^{-2}+v^{-2})\log\left(\frac{v-u}{v}\right)\Bigg|\\
\leq&\:   C\left(I_0[\psi]+H_0[\psi]+\sqrt{E^{\epsilon}_{0; 1}[\psi]}+P_{\mathcal{H}}[\psi]+P_{\mathcal{I}}[\psi]\right)u^{-2} \quad \textnormal{in $\mathcal{A}^{\mathcal{I}}$}.
\end{split}
\end{equation}
and
\begin{equation}
\label{eq:2ndasymphinearhor}
\begin{split}
\Bigg|r\psi(u,v)&-2H_0[\psi](v^{-1}-u^{-1})+4MH_0[\psi]v^{-2}\log v-4MH_0[\psi]u^{-2}\log v\\
&+8MH_0[\psi]u^{-2}\log u+4MH_0[\psi](v^{-2}+u^{-2})\log\left(\frac{u-v}{u}\right)\Bigg|\\
\leq&\:   C\left(I_0[\psi]+H_0[\psi]+\sqrt{E^{\epsilon}_{0; 1}[\psi]}+P_{\mathcal{H}}[\psi]+P_{\mathcal{I}}[\psi]\right)v^{-2} \quad \textnormal{in $\mathcal{A}^{\mathcal{H}}$}.
\end{split}
\end{equation}
In particular,
\begin{equation}
\label{eq:2ndasymphinullinf}
\left|r\psi|_{\mathcal{I}^+}(u)-2I_0[\psi]u^{-1}+4MI_0[\psi]u^{-2}\log u\right|\leq C\left(I_0[\psi]+H_0[\psi]+\sqrt{E^{\epsilon}_{0; 1}[\psi]}+P_{\mathcal{H}}[\psi]+P_{\mathcal{I}}[\psi]\right)u^{-2}
\end{equation}
and
\begin{equation}
\label{eq:2ndasymphihor}
\left|r\psi|_{\mathcal{H}^+}(v)-2H_0[\psi]v^{-1}+4MH_0[\psi]v^{-2}\log v\right|\leq  C\left(I_0[\psi]+H_0[\psi]+\sqrt{E^{\epsilon}_{0;1}[\psi]}+P_{\mathcal{H}}[\psi]+P_{\mathcal{I}}[\psi]\right)v^{-2}.
\end{equation}
\end{proposition}
\begin{proof}
In order to prove \eqref{eq:2ndasymphiinf} we integrate $\partial_v(r\psi)$ from $(u,v=u+2r_*(r_{\mathcal{I}}))$ to $(u,v=v')$, estimating the contribution of \eqref{eq:2ndoasympdvphiinf} as in Proposition 3.2 of \cite{logasymptotics} and moreover using  \eqref{eq:asympsi} to estimate
\begin{equation*}
|r\psi|(u,v=u+2r_*(r_{\mathcal{I}}))\lesssim \left(I_0[\psi]+H_0[\psi]+\sqrt{E_{0; 1}[\psi]}+P_{\mathcal{H}}[\psi]+P_{\mathcal{I}}[\psi]\right)u^{-2}.
\end{equation*}
We similarly obtain \eqref{eq:2ndasymphihor} by integrating $\partial_u(r\psi)$ from $(u=v-2r_*(r_{\mathcal{H}}),v)$ to $(u=u', v)$, estimating the contribution of \eqref{eq:2ndoasympdvphihor} as in Proposition 3.2 of \cite{logasymptotics}, but with the roles of $u$ and $v$ interchanged, and moreover using  \eqref{eq:asympsi} to estimate
\begin{equation*}
|r\psi|(u=v-2r_*(r_{\mathcal{H}}),v)\lesssim \left(I_0[\psi]+H_0[\psi]+\sqrt{E_{0; 1}[\psi]}+P_{\mathcal{H}}[\psi]+P_{\mathcal{I}}[\psi]\right)v^{-2}.
\end{equation*}
The estimates \eqref{eq:2ndasymphinullinf} and \eqref{eq:2ndasymphihor} then follow simply by taking respectively the limit $v\to \infty$ and $u\to \infty$.
\end{proof}

We can moreover obtain refined late-time asymptotics along $\mathcal{H}^+$ and $\mathcal{I}^+$ in the case when \emph{both} $I_0[\psi]$ and $H_0[\psi]=0$ are vanishing (i.e.\ for Type \textbf{B} data).

\begin{proposition}
\label{prop:2ndasymphitimeint}
Suppose $H_0[\psi]=0$ and $I_0[\psi]=0$. For all $\epsilon>0$, there exists a constant $C=C(M,\Sigma,r_{\mathcal{I}},r_{\mathcal{H}},\epsilon)>0$ such that we can estimate:
\begin{equation}
\label{eq:2ndasymphinullinftimeintmain}
\begin{split}
\Big|&r\psi|_{\mathcal{I}^+}(u)+2I_0^{(1)}[\psi]u^{-2}-8MI_0^{(1)}[\psi]u^{-3}\log u\Big|\\
\leq&\: C\left(I_0^{(1)}[\psi]+H_0^{(1)}[\psi]+\sqrt{E^{\epsilon}_{0,\mathcal{H}; 1}[\psi]+E^{\epsilon}_{0,\mathcal{I}; 1}[\psi]}+P_{\mathcal{H}, T}[\psi]+P_{\mathcal{I}, T}[\psi]\right)u^{-3},
\end{split}
\end{equation}
and the following estimate holds on $\mathcal{H}^{+}$
\begin{equation}
\label{eq:2ndasymphihotimintmain}
\begin{split}
\Big|&r\psi|_{\mathcal{H}^+}(v)+2H_0^{(1)}[\psi]v^{-2}-4MH_0^{(1)}[\psi]v^{-3}\log v\Big|\\
\leq&\:  C\left(I_0^{(1)}[\psi]+H_0^{(1)}[\psi]+\sqrt{E^{\epsilon}_{0,\mathcal{H}; 1}[\psi]+E^{\epsilon}_{0,\mathcal{I}; 1}[\psi]}+P_{\mathcal{H}, T}[\psi]+P_{\mathcal{I}, T}[\psi]\right)v^{-3}.
\end{split}
\end{equation}\end{proposition}
\begin{proof}
By $H_0[\psi]=0$ and $I_0[\psi]=0$, together with the assumption that $P_{\mathcal{I}, T}[\psi]<\infty$, it follows by Proposition \ref{eq:smoothextTypeBtimeint} that $\psi_0^{(1)}$ is smooth. We can then apply the arguments in the proof of Proposition 3.3 of \cite{logasymptotics} in $\mathcal{A}^{\mathcal{I}}$, and similar arguments with the roles of $u$ and $v$ reversed in  $\mathcal{A}^{\mathcal{H}}$, to derive the late-time asymptotics of $T\psi_0^{(1)}=\psi_0$. We omit the details of the proof.
\end{proof}

\appendix
 
\section{Basic inequalities on ERN}
\label{sec:BasicInequalitiesOnERN}
In this section, we will list some basic inequalities that are used throughout the paper.

\subsection{Hardy inequalities}
\label{sec:HardyInequalities}

\begin{lemma}[Hardy inequalities]
\label{lm:hardy}
Let $q\in \R\setminus \{-1\}$ and $r_1>M$. Let $f: \times [v_0,\infty)_v\times [M,\infty)_r \to \R$ be a $C^1$ function. Then for all $M<r_1<r_2\leq \infty$ and $u'\geq 0$
\begin{equation}
\label{eq:hardyinf}
\int_{v_{r_1}(u')}^{v_{r_2}(u')} r^qf^2|_{u=u'}\,dv\lesssim (q+1)^{-2} \int_{v_{r_1}(u')}^{v_{r_2}(u')} r^{q+2}(Lf)^2|_{u=u'}\,dv+ 2r_2^{q+1}f^2(u',v_{r_2}(u')),
\end{equation}
where in the $r_2=\infty$ case, the second term on the right-hand side is interpreted as follows:
\begin{equation*}
2\lim_{r\to \infty }r^{q+1}f^2(u,v_{r}(u))
\end{equation*}
 
Furthermore, for all $M\leq r_0<r_1<\infty$ and $v'\geq 0$
\begin{equation}
\label{eq:hardyhor1}
\int_{u_{r_1}(v')}^{u_{r_0}(v')} (r-M)^{-q}f^2|_{v=v'}\,du\lesssim (q+1)^{-2}\int_{u_{r_1}(v')}^{u_{r_0}(v')} (r-M)^{-q-2}(\underline{L}f)^2|_{v=v'}\,du+ 2(r_0-M)^{-q-1}f^2(u_{r_0}(v'),v'),
\end{equation}
where in the $r_0=M$ case, the second term on the right-hand side is interpreted as follows:
\begin{equation*}
2\lim_{u\to \infty}(r-M)^{-q-1}f^2|_{v=v'},
\end{equation*}
or equivalently, in $(v,r,\theta,\varphi)$ coordinates, for all $v'\geq 0$
\begin{equation}
\label{eq:hardyhor2}
\int_{r_0}^{r_1} (r-M)^{-2-q}f^2|_{v=v'}\,dr\lesssim (q+1)^{-2}\int_{r_0}^{r_1} (r-M)^{-q}(\partial_rf)^2|_{v=v'}\,du+ 2(r_0-M)^{-q-1}f^2(v',r_0).
\end{equation}
\end{lemma}

\subsection{Interpolation estimates}
\label{sec:Interpolationestimates}
We will make use of the following interpolation estimates.
\begin{lemma}[Interpolation estimates]
\label{lm:interpolation}
Let $f: \{(u,v)\in \R^2\,|\, u\in [u_0,\infty)\quad v\in [v_{r_{\mathcal{I}} }(u),\infty)\}\to \R$ be a function such that the following inequalities hold: there exist $u$-independent constants $\mathcal{E}_1,\mathcal{E}_2>0$, such that
\begin{align*}
\int_{v_{r_{\mathcal{I}} }(u)}^{\infty}r^{p-\epsilon}f^2(u,v)\,dv\leq &\:\mathcal{E}_1u^{-q},\\
\int_{v_{r_{\mathcal{I}} }(u)}^{\infty}r^{p+1-\epsilon}f^2(u,v)\,dv\leq&\: \mathcal{E}_2u^{-q+1},
\end{align*}
with $q\in \R$ and $\epsilon\in (0,1)$.

Then
\begin{equation}
\label{eq:interp3I}
\int_{v_{r_{\mathcal{I}} }(u)}^{\infty}r^{p}f^2(u,v)\,dv\leq C \max\{\mathcal{E}_1,\mathcal{E}_2\}u^{-q+\epsilon},
\end{equation}
with $C>0$ a constant depending only on $M$, $\Sigma_0$ and $r_{\mathcal{I}}$.

Furthermore, let $\underline{f}: \{(u,v)\in \R^2\,|\, v\in [v_0,\infty)\quad u\in [u_{r_{\mathcal{H}} }(v),\infty)\}\to \R$ be a function such that the following inequalities hold: there exist $v$-independent constants $\mathcal{E}_1,\mathcal{E}_2>0$, such that
\begin{align*}
\int_{u_{r_{\mathcal{H}} }(v)}^{\infty}(r-M)^{-p+\epsilon}\underline{f}^2(u,v)\,du\leq &\:\mathcal{E}_1v^{-q},\\
\int_{u_{r_{\mathcal{H}} }(v)}^{\infty}(r-M)^{-p-1+\epsilon}\underline{f}^2(u,v)\,du\leq&\: \mathcal{E}_2v^{-q+1},
\end{align*}
with $q\in \R$ and $\epsilon\in (0,1)$.

Then
\begin{equation}
\label{eq:interp3H}
\int_{u_{r_{\mathcal{H}} }(v)}^{\infty}(r-M)^{-p}\underline{f}^2(u,v)\,du\leq C \max\{\mathcal{E}_1,\mathcal{E}_2\}v^{-q+\epsilon},
\end{equation}
with $C>0$ a constant depending only on $M$, $\Sigma_0$ and $r_{\mathcal{H}}$.
\end{lemma}
\begin{proof}
See the proof of Lemma 2.6 of \cite{paper1} for the derivation of \eqref{eq:interp3I}. The estimate \eqref{eq:interp3H} follows after replacing $r$ with $(r-M)^{-1}$ and reversing the roles of $u$ and $v$.
\end{proof}

\subsection{Basic energy estimates}
\label{sec:EnergyBounds}
The following energy boundedness estimate holds for all solutions $\psi$ to the wave equation \eqref{eq:waveequation} on ERN:
\begin{equation}
\label{eq:degenbound}
\int_{\Sigma_{\tau}}J^T[\psi]\cdot \mathbf{n}_{\tau}\,d\mu_{\tau}\leq \int_{\Sigma_{0}}J^T[\psi]\cdot \mathbf{n}_{0}\,d\mu_{0}
\end{equation}
and it follows straightforwardly from the Killing property of the vector field $T$, together with the non-negativity of the $T$-energy flux through $\mathcal{H}^+$ and $\mathcal{I}^+$ (in a limiting sense).

We next give an overview of the main Morawetz or \emph{integrated local energy decay} estimates that we will make use of throughout the remainder of the paper. A proof of these estimates can be found in \cite{aretakis1,aretakis2}. 
\begin{theorem}
\label{thm:ileds}
Let $M<r_1<r_2<2M<r_3<r_4<\infty$. Let $N\in \N_0$ and $0\leq \tau_1\leq \tau_2\leq \infty$. Then:
\begin{enumerate}
\item[1.)] There exists a constant $C=C(\Sigma,r_1,r_2,r_3,r_4)>0$ such that
\begin{equation}
\label{eq:iledwayps}
\sum_{0\leq k+l+m\leq N+1}\int_{\tau_1}^{\tau_2} \left[\int_{\Sigma_{\tau}\cap (\{r_1\leq r\leq r_2\}\cup \{r_3\leq r\leq r_4\})} |\snabla_{\s^2}^k\partial_v^{l}\partial_r^m\psi|^2\,d\mu_{\Sigma_{\tau}}\right]\,d\tau\leq C\sum_{k\leq N} \int_{\Sigma_{\tau}} J^T[T^k\psi]\cdot \mathbf{n}_{\tau_1}\,d\mu_{\Sigma_{\tau_1}}.
\end{equation}
\item[2.)] There exists a constant $C=C(\Sigma,r_1,r_4)>0$ such that
\begin{equation}
\label{eq:iledlossder}
\sum_{0\leq k+l+m\leq N+1}\int_{\tau_1}^{\tau_2} \left[\int_{\Sigma_{\tau}\cap (\{r\geq r_1\}\cap \{r\leq r_4\})} |\snabla_{\s^2}^k\partial_v^{l}\partial_r^m\psi|^2\,d\mu_{\Sigma_{\tau}}\right]\,d\tau\leq C\sum_{k\leq N+1} \int_{\Sigma_{\tau}} J^T[T^k\psi]\cdot \mathbf{n}_{\tau_1}\,d\mu_{\Sigma_{\tau_1}}.
\end{equation}
\item[3.)] There exists a constant $C=C(\Sigma,r_1,r_4)>0$ such that
\begin{equation}
\label{eq:iledpsi0}
\sum_{0\leq l+m\leq N+1}\int_{\tau_1}^{\tau_2} \left[\int_{\Sigma_{\tau}\cap (\{r\geq r_1\}\cap \{r\leq r_4\})} |\partial_v^{l}\partial_r^m\psi_0|^2\,d\mu_{\Sigma_{\tau}}\right]\,d\tau\leq C\sum_{k\leq N} \int_{\Sigma_{\tau}} J^T[T^k\psi]\cdot \mathbf{n}_{\tau_1}\,d\mu_{\Sigma_{\tau_1}}.
\end{equation}
\end{enumerate}
\end{theorem}

\footnotesize

\bibliographystyle{plain}

\footnotesize{

Department of Mathematics, University of California, Los Angeles, CA 90095, United States, yannis@math.ucla.edu

\bigskip

Princeton University, Department of Mathematics, Fine Hall, Washington Road, Princeton, NJ 08544, United States, aretakis@math.princeton.edu

\bigskip

Department of Mathematics, University of Toronto Scarborough 1265 Military Trail, Toronto, ON, M1C 1A4, Canada, aretakis@math.toronto.edu

\bigskip

Department of Mathematics, University of Toronto, 40 St George Street, Toronto, ON, Canada, aretakis@math.toronto.edu

\bigskip

Department of Pure Mathematics and Mathematical Statistics, University of Cambridge, Wilberforce Road, Cambridge CB3 0WB, United Kingdom, dg405@cam.ac.uk }

\end{document}